\documentclass{article} %preprint, review  3p,twocolumn,review

\usepackage[utf8x]{inputenc}
\usepackage{amsmath,amssymb,amstext,amsthm}
\usepackage{amsfonts}
\usepackage{graphicx,epstopdf,epsfig,multirow,epic,bm}
\usepackage{color}
\usepackage{multicol}
\usepackage{algorithm}
\usepackage{algorithmic}
\usepackage{paralist}%% [asparaenum]

\theoremstyle{definition}
\newtheorem{thm}{Theorem}
\newtheorem{cor}[thm]{Corollary}
\newtheorem{lem}[thm]{Lemma}
\newtheorem{prop}[thm]{Proposition}
\theoremstyle{definition}
\newtheorem{defn}{Definition}
\theoremstyle{definition}
\newtheorem{rem}{Remark}

\theoremstyle{definition}
\newtheorem{problem}{Problem}

\theoremstyle{definition}
\newtheorem{conj}{Conjecture}
\theoremstyle{definition}
\newtheorem{exa}{Example}
\theoremstyle{definition}

\DeclareGraphicsExtensions{.eps,.eps.gz}
\normalsize \rm
\usepackage{titlesec}

\newcommand{\qqed}{\hfill $\Box$}

\begin{document}

\begin{center}
{\huge \textbf{Graphic Lattices and Matrix Lattices \\
Of Topological Coding}}
\end{center}

\vskip 1cm

\pagenumbering{roman}
\tableofcontents

\newpage

\setcounter{page}{1}
\pagenumbering{arabic}

\thispagestyle{empty}

\begin{center}
{\huge \textbf{Graphic Lattices and Matrix Lattices Of Topological Coding}}\\

\vskip 0.6cm

{\Large Bing Yao}\\

\vskip 0.6cm

{\Large College of Mathematics and Statistics,
 Northwest Normal University, Lanzhou, 730070, CHINA, yybb918@163.com}
\end{center}

\vskip 1cm

\begin{quote}
\textbf{Abstract:} Lattice-based Cryptography is considered to have the characteristics of classical computers and quantum attack resistance. We will design various graphic lattices and matrix lattices based on knowledge of graph theory and topological coding, since many problems of graph theory can be expressed or illustrated by (colored) star-graphic lattices. A new pair of the leaf-splitting operation and the leaf-coinciding operation will be introduced, and we combine graph colorings and graph labellings to design particular proper total colorings as tools to build up various graphic lattices, graph homomorphism lattice, graphic group lattices and Topcode-matrix lattices. Graphic group lattices and (directed) Topcode-matrix lattices enable us to build up connections between traditional lattices and graphic lattices. We present mathematical problems encountered in researching graphic lattices, some problems are: Tree topological authentication, Decompose graphs into Hanzi-graphs, Number String Decomposition Problem, $(p,s)$-gracefully total numbers. \\[6pt]
\textbf{Keywords:} Lattice; cryptosystem; graphic lattice; total coloring; matrices; graphic group; graph homomorphism lattice; topological coding.
\end{quote}

%\section{Introduction and preliminary}
%%\input{1-section/Intro-preliminary-1}

\section{Introduction and preliminary}

The reality is an infinite set of random events with changing rules. Whenever there is a major development, some rules must be rewritten, and then all participants have to adjust to follow it or die. Scientific research is such a process based on ``model'' to constantly update themselves. At the same time, we may never be able to develop it ``ultimate truth'', but we can approach a better model infinitely. Lattice of cryptosystems can bring new perspective and new technology for part of research objects and problems in graph theory.

\subsection{Research background}

We recall some investigations on ``Resisting classical computers and quantum computers''.

\subsubsection{Cryptosystems resisting classical computers and quantum computers}

The authors in \cite{Bernstein-Buchmann-Dahmen-Quantum-2009} point: ``There are many important classes of cryptographic
systems beyond RSA and DSA and ECDSA, and they are believed to resist classical computers and quantum computers, such as Hash-based cryptography, Code-based cryptography, Lattice-based cryptography, Multivariate-quadratic-equations cryptography, Secret-key cryptography''. Notice that the lattice difficulty problem is not only a classical number theory, but also an important research topic of computational complexity theory. Researchers have found that lattice theory has a wide range of applications in cryptanalysis and design. Many difficult problems in lattice have been proved to be NP-hard. So, this kind of cryptosystems are generally considered to have the characteristics of quantum attack resistance (Ref. \cite{Wang-Xiao-Yun-Liu-2014}).

\subsubsection{Homomorphic encryption}

Because the homomorphic encryption can calculate the ciphertext arbitrarily without decryption, and can solve the problem of data privacy security immediately. For example, a user wants to process a data, but his computer computing power is weak, this user can use homomorphic encryption to encrypt his data and store it in the cloud environment, since the cloud cannot obtain the content of the data. After the cloud process the encrypted data directly, he receives the processing results. In a word, the homomorphic encryption not only protects the data by encryption, but also does not lose the computability.

\subsubsection{Lattice encryption}

The idea of writing this article is motivated from Lattice-based Cryptography. A \emph{lattice} $\textrm{\textbf{L}}(\textbf{B})$ is defined as the set of all integer combinations
\begin{equation}\label{eqa:popular-lattice}
\textrm{\textbf{L}}(\textbf{B}) =\left \{\sum^n_{i=1}x_i\textbf{b}_i : x_i \in Z, 1 \leq i \leq n\right \}
\end{equation}
of $n$ linearly independent vectors $\textbf{b}_1, \textbf{b}_2,\dots , \textbf{b}_n$ in $R^m$ with $n\leq m$, where $Z$ is the integer set, $m$ is the \emph{dimension} and $n$ is the \emph{rank} of the lattice, and the vector group $\textbf{b}_1, \textbf{b}_2,\dots , \textbf{b}_n$ is called a \emph{lattice base}. A lattice is a set of discrete points with periodic structure in $R^m$, and it can be expressed by different lattice bases. For no confusion, we call $\textrm{\textbf{L}}(\textbf{B})$ defined in (\ref{eqa:popular-lattice}) \emph{traditional lattice} in this article.

\subsubsection{Encryption optical chip}In December 20,2019, the King Abdullah University of science and technology in Saudi Arabia has developed an encryption optical chip, which uses a one-time key to realize information transmission between users, such that the key used to unlock a message will never be stored and associated with the message, or even recreated by the user.

\subsubsection{Graphs like Lattices}

As known, a connected Euler's graph is a union of cycles in graph theory, which can be expressed mathematically as $\odot ^n_{k=1}a_kC_k$ with each $a_k$ is a non-negative integer and $\sum ^n_{k=1}a_k\geq 1$, where $C_k$ is a cycle of $k$ vertices, and ``$\odot $'' is the vertex-coinciding operation of graphs, i.e. gluing together cycles. Similarly, each graph of a set $\{\odot^n_{k=1}a_kP_k\}$ is connected, where $P_k$ is a path of $k$ vertices. And, the forms $\odot ^n_{k=1}a_kC_k$ and $\odot ^n_{k=1}a_kP_k$ are like $\sum a_k\textbf{\textrm{b}}_k$ shown in that defined in (\ref{eqa:popular-lattice}). The authors in \cite{YAO-SUN-WANG-SU-XU2018arXiv} guessed: A maximal planar graph is 4-colorable if and only it can be tiled by the every-zero graphic group $\{F_{\textrm{inner}\bigtriangleup};\oplus\}$ shown in Fig.\ref{fig:4-color-planar-tile-lattice}(c). There are many sets of colored graphs such that each of them can be written in the form lattices like a traditional lattice (\ref{eqa:popular-lattice}) in this article. For example, each uncolored planar graph $H$ with each inner face to be a triangle is isomorphic to $G=H\bigtriangleup ^4_{k=1}a_kT^r_k$ with $a_k\in Z^0$ and $\sum a_k\geq 1$ defined in the planar graphic lattice (\ref{eqa:planar-lattice}). For investigating ``Topsnut-gpws'' that is the abbreviation of ``Graphical passwords based on the idea of topological structure plus number theory'' (Ref. \cite{Wang-Xu-Yao-2016}, \cite{Wang-Xu-Yao-Key-models-Lock-models-2016} and \cite{Yao-Sun-Zhang-Li-Zhang-Xie-Yao-2017-Tianjin-University}), Wang \emph{et al.}, in \cite{Wang-Xu-Yao-Ars-2018}, have constructed some spaces $\{T\odot |^p_{i=1}H_i\}$ of particular Topsnut-gpws made by a group of disjoint Topsnut-gpws $H_1,H_2,\dots ,H_p$. Topsnut-gpws belong to ``Topological Coding'', a combinatoric subbranch of graph theory and cryptography. The domain of topological coding involves millions of things, and a graph of topological coding connects things together to form a complete ``story'' under certain constraints.

\subsubsection{Our works}

Motivated from the traditional lattices and the topological authentications (see Fig.\ref{fig:authentication-system}), we will define so-called \emph{graphic lattices}, \emph{graphic group lattices}, \emph{Topcode-matrix lattices}, \emph{matching-type graphic lattices}, \emph{graphic lattice sequences} and other type lattices made by various graph operations, matrix operations and group operations in the following sections and subsections. In fact, a graphic group consisted of particular graphs and some special graph operations. Graphic lattices and matrix lattices are combination of traditional lattice and topological coding. Part of our works here are cited from \cite{Yao-Wang-Su-Sun-ITOEC2020, Wang-Su-Sun-Yao-submitted-ITOEC2020, Wang-Yao-Star-type-Lattices-2020, Bing-Yao-Hongyu-Wang-graph-homomorphisms-2020} directly.

\begin{figure}[h]
\centering
\includegraphics[width=15cm]{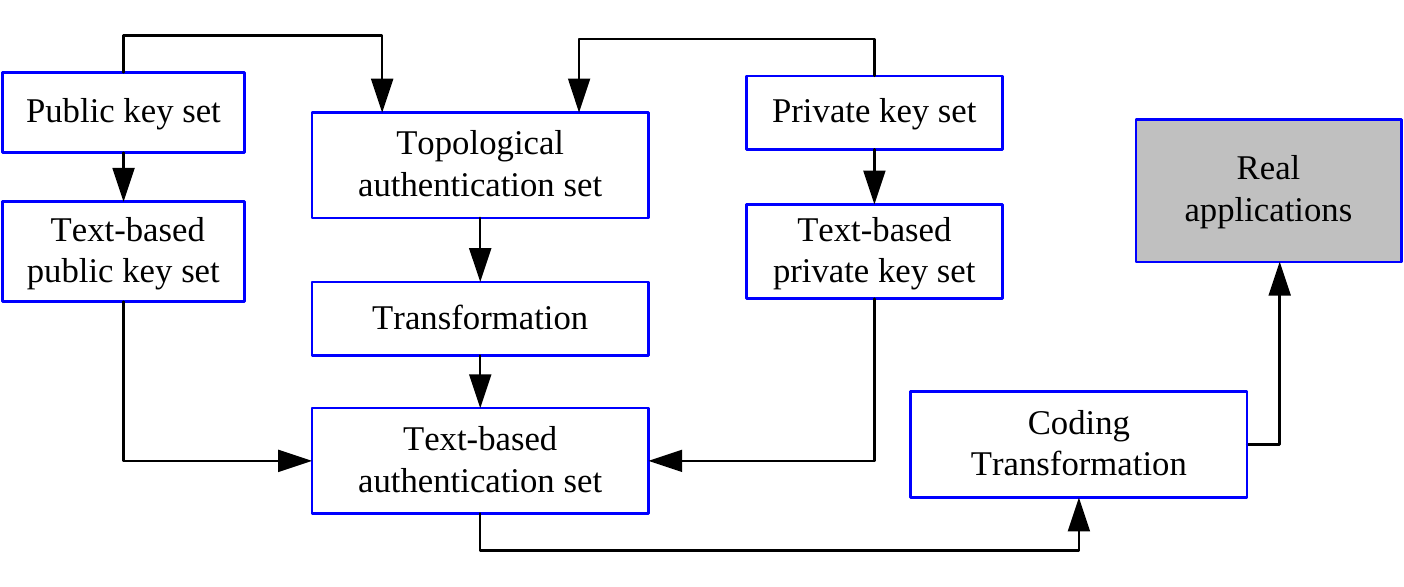}
\caption{\label{fig:authentication-system}{\small A topological authentication system cited from \cite{Wang-Su-Yao-mixed-difference-2019}.}}
\end{figure}

\subsection{An example for graphic lattices}

Before listing our main works in this article, let us see examples. We have four Hanzi-graphs $G_{\textrm{4476}}$, $G_{\textrm{4734}}$, $G_{\textrm{4610}}$ and $G_{\textrm{2511}}$ shown in Fig.\ref{fig:2-txwg}, where the lower right code ``abcd'' in $G_{\textrm{abcd}}$ can be found in \cite{GB2312-80}. Then $G=G_{\textrm{4476}}\cup G_{\textrm{4734}}\cup G_{\textrm{4610}}\cup G_{\textrm{2511}}$ is called a \emph{disconnected graph}. In English, $G$ means ``the whole world as one community''.

\begin{figure}[h]
\centering
\includegraphics[width=13cm]{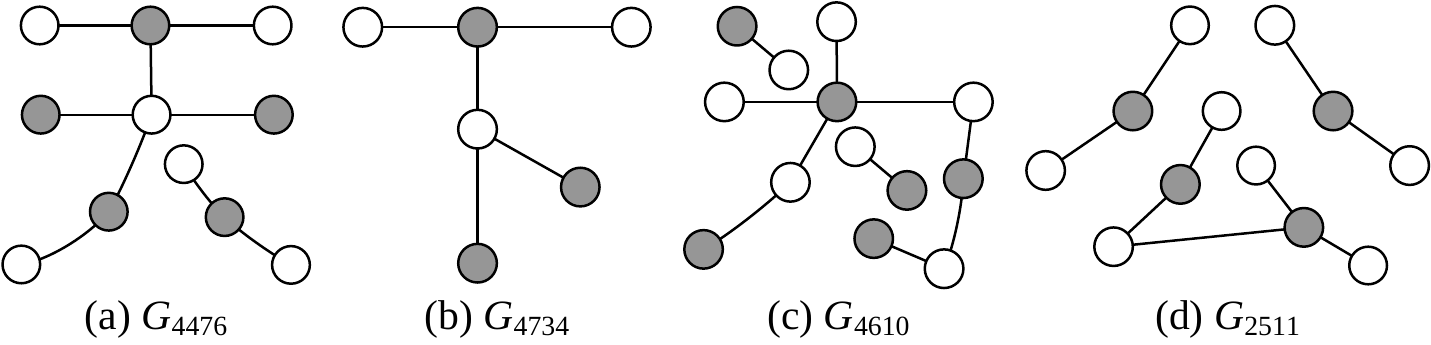}
\caption{\label{fig:2-txwg}{\small A disconnected graph $G=G_{\textrm{4476}}\cup G_{\textrm{4734}}\cup G_{\textrm{4610}}\cup G_{\textrm{2511}}$, where (a)-(d) are four Hanzi-graphs, in which all are disconnected, except (b).}}
\end{figure}

We call: (i) Four Hanzi-graphs shown in Fig.\ref{fig:2-txwg} a group of \emph{linearly independent graphic vectors} (or \emph{graphic base}); (ii) $\textbf{\textrm{T}}=(T_{\textrm{4476}},T_{\textrm{4734}},T_{\textrm{4610}},T_{\textrm{2511}})$ shown in Fig.\ref{fig:2-txwg-graceful-disconnected} a group of \emph{linearly independent colored graphic vectors} (or \emph{colored graphic base}). In fact, there are many Topsnut-gpws $H_i$ like that $H$ shown in Fig.\ref{fig:2-txwg-graceful} to form vertex-coincided Topsnut-gpws $H_i\odot \textbf{\textrm{T}}$, and then we put these vertex-coincided Topsnut-gpws $H_i$ and $H$ into a set $F$. We say the following set
\begin{equation}\label{eqa:c3xxxxx}
{
\begin{split}
\textbf{\textrm{L}}(\textbf{\textrm{T}}\odot F)=\{H_i\odot \textbf{\textrm{T}}:H_i\in F\}= \{H_i\odot (T_{\textrm{4476}}\cup T_{\textrm{4734}}\cup T_{\textrm{4610}}\cup T_{\textrm{2511}}):~H_i\in F\}
\end{split}}
\end{equation}
a \emph{colored graphic lattice}, where ``$\odot$'' is the vertex-coinciding operation between graphs, see an example shown in Fig.\ref{fig:2-txwg-graceful}. We have a connected Topsnut-gpw $H\odot \textbf{\textrm{T}}$ obtained by doing the vertex-coinciding operation ``$\odot$'' on a Topsnut-gpw and a graphic-vector base, see Fig.\ref{fig:2-txwg-graceful-disconnected} and Fig.\ref{fig:2-txwg-graceful}. In Fig.\ref{fig:2-txwg-graceful-disconnected}, the disconnected Topsnut-gpw $W$ admits a flawed graceful labelling $f$ (see Definition \ref{defn:flawed-odd-graceful-labelling}). A disconnected Topsnut-gpw $H$ shown in Fig.\ref{fig:2-txwg-graceful} admits a coloring $g$. So, the connected Topsnut-gpw $H\odot \textbf{\textrm{T}}$ shown in Fig.\ref{fig:2-txwg-graceful} admits a proper graceful labelling $f\odot g$.

\begin{figure}[h]
\centering
\includegraphics[width=14.4cm]{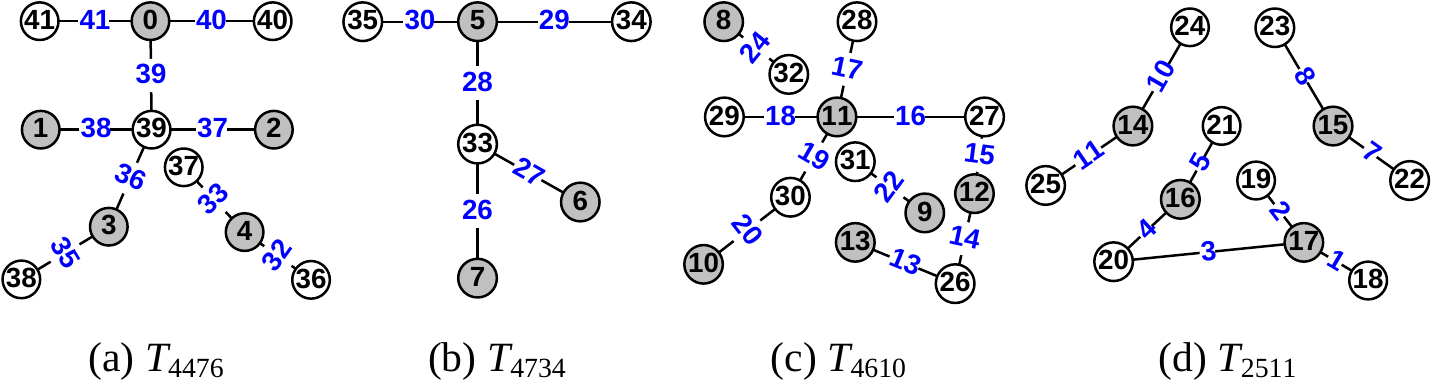}
\caption{\label{fig:2-txwg-graceful-disconnected}{\small A disconnected Topsnut-gpw $W$ consists of four Topsnut-gpws $T_{\textrm{4476}}$, $T_{\textrm{4734}}$, $T_{\textrm{4610}}$ and $T_{\textrm{2511}}$.}}
\end{figure}

\begin{figure}[h]
\centering
\includegraphics[width=16.4cm]{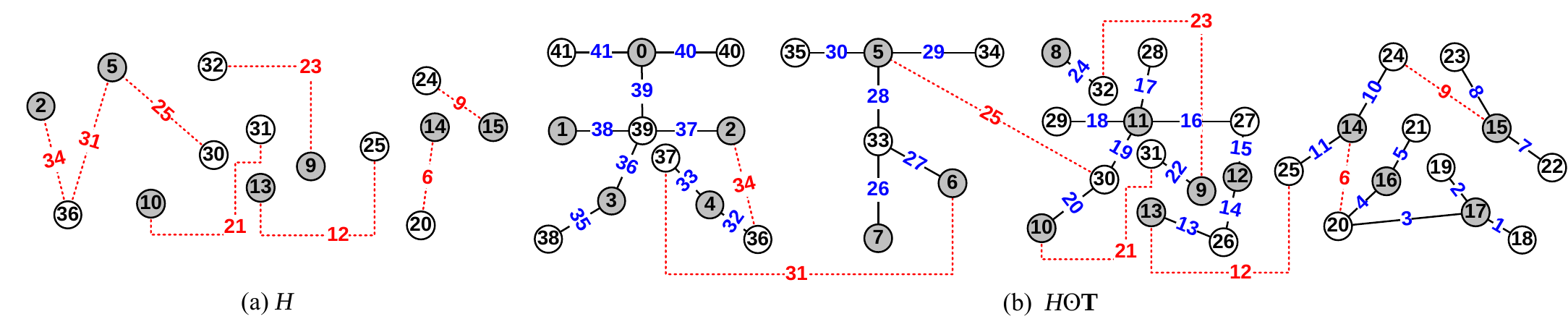}
\caption{\label{fig:2-txwg-graceful}{\small (a) $H$; (b) a connected Topsnut-gpw $H\odot \textbf{\textrm{T}}$ obtained by doing the vertex-coinciding operation on two Topsnut-gpws $G$ shown in Fig.\ref{fig:2-txwg-graceful-disconnected} and $H$ shown in (a).}}
\end{figure}
\begin{problem}\label{qeu:preface}
We, for characterize our colored graphic lattice $\textbf{\textrm{L}}(\textbf{\textrm{T}}\odot F)$, need to clarify the following questions:
\begin{asparaenum}[\quad \textbf{\textrm{Que}}-1. ]
\item \textbf{Determine} Topsnut-gpws induced by the graph $G$ shown in Fig.\ref{fig:2-txwg}, like the disconnected Topsnut-gpw $W$ shown in Fig.\ref{fig:2-txwg-graceful-disconnected}, since this disconnected graph $G$ admits two or more different colorings.
\item Since each Topsnut-gpw like $W$ corresponds a set like $F$ containing colored graphs like the disconnected Topsnut-gpw $H$ shown in Fig.\ref{fig:2-txwg-graceful} (a), \textbf{determine} such sets $F$ for $\textbf{\textrm{L}}(\textbf{\textrm{T}}\odot F)$.
\item \textbf{Determine} the cardinality of a colored graphic lattice $\textbf{\textrm{L}}(\textbf{\textrm{T}}\odot F)$ and particular Topsnut-gpws $H\odot \textbf{\textrm{T}}$ in $\textbf{\textrm{L}}(\textbf{\textrm{T}}\odot F)$, such as $H\odot \textbf{\textrm{T}}$ has the shortest diameter, Hamilton cycle, spanning trees with maximal leaves, scale-free behavior, clustering coefficient, and so on.
\item \textbf{Apply} $\textbf{\textrm{L}}(\textbf{\textrm{T}}\odot F)$ to the real world. For example, we can consider the disconnected graph $G$ shown in Fig.\ref{fig:2-txwg} as a \emph{public key}, and the disconnected Topsnut-gpw $W$ shown in Fig.\ref{fig:2-txwg-graceful-disconnected} as a \emph{private key}, then the connected Topsnut-gpw $H\odot \textbf{\textrm{T}}$ shown in Fig.\ref{fig:2-txwg-graceful}(b) is just a \emph{topological coloring authentication}.

\item Joining nine components $G_1,G_2,\dots ,G_9$ shown in Fig.\ref{fig:2-txwg-decompose} by a graph $H_i$ and doing the vertex-coinciding operation to them produce a connected graph, denoted as $H_i\odot ^9_{j=1}G_j$, see such examples shown in Fig.\ref{fig:2-txwg-3-joins}. We get a uncolored graphic lattice $\textbf{\textrm{L}}(G\odot F)$, \textbf{characterize} it. Since the disconnected graph $G$ shown in Fig.\ref{fig:2-txwg} consists of four Hanzi-graphs, we call $\textbf{\textrm{L}}(G\odot F)$ a \emph{Hanzi-graphic lattice}.

\begin{figure}[h]
\centering
\includegraphics[width=15cm]{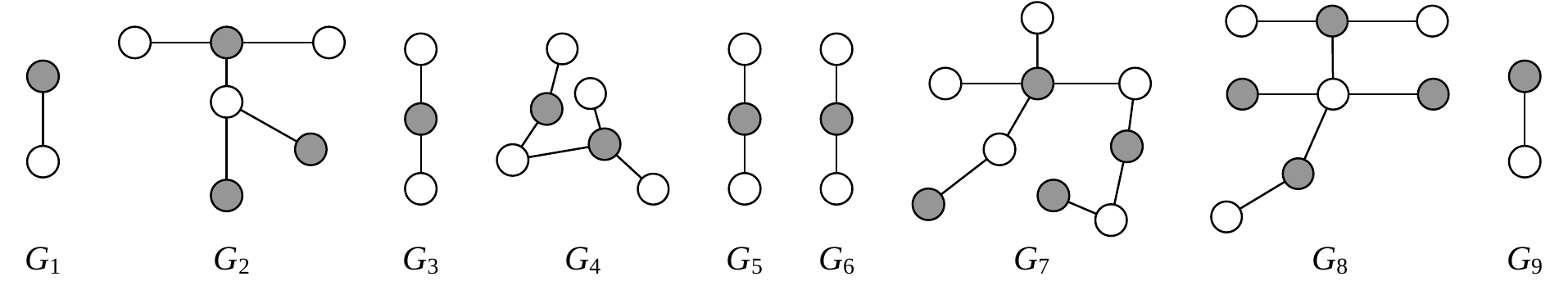}
\caption{\label{fig:2-txwg-decompose}{\small Nine connected graphs $G_1,G_2,\dots ,G_9$ are the components of the disconnected graph $G$ shown in Fig.\ref{fig:2-txwg}.}}
\end{figure}

\begin{figure}[h]
\centering
\includegraphics[width=16.4cm]{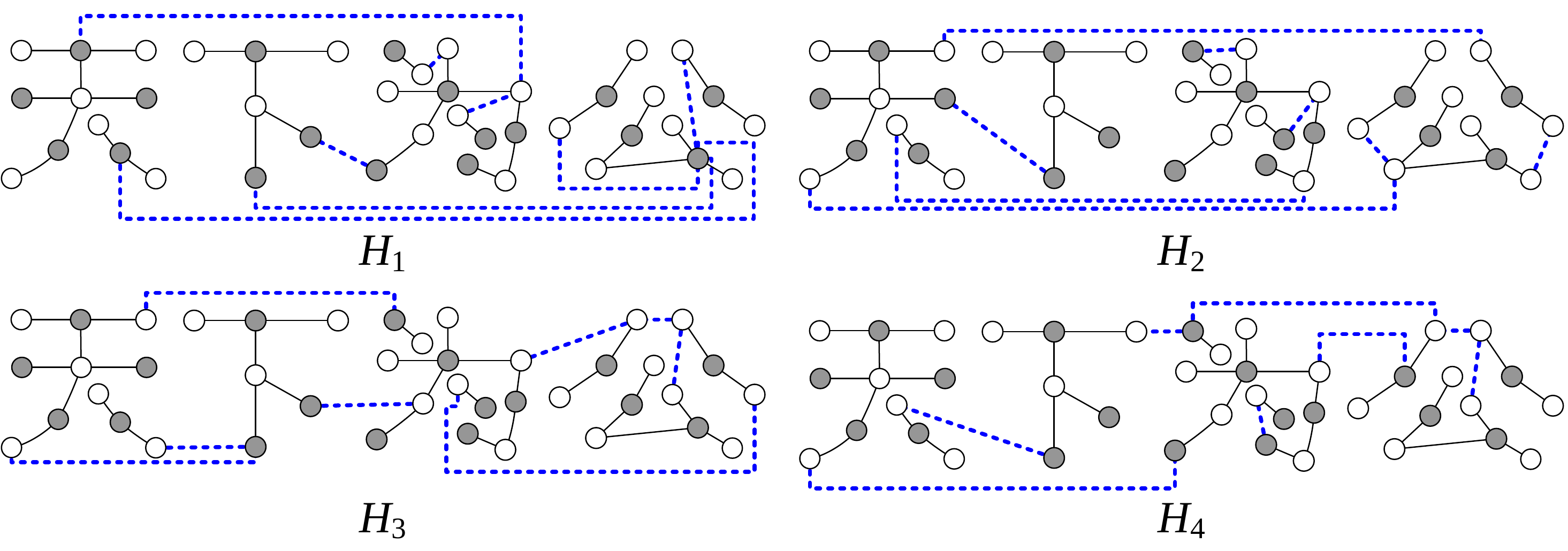}
\caption{\label{fig:2-txwg-3-joins}{\small Four connected graphs $H_1,H_2,H_3,H_4$ obtained by joining nine connected graphs $G_1,G_2,\dots ,G_9$ shown in Fig.\ref{fig:2-txwg-decompose} together.}}
\end{figure}

\item \textbf{Determine} other groups of (colored) Hanzi-graphs differ from the disconnected graph $G$ shown in Fig.\ref{fig:2-txwg} and the group $W$ shown in Fig.\ref{fig:2-txwg-graceful-disconnected}, see examples shown in Fig.\ref{fig:2-txwg-other-groups}, we can see: (i) these groups of Hanzi-graphs have the same number of Hanzi-strokes; (ii) these groups of Hanzi-graphs make different (colored) Hanzi-graphic lattices $\textbf{\textrm{L}}(G\odot F)$. For example, $H_1,H_2,H_3,H_4\in \textbf{\textrm{L}}(G\odot F)$, see Fig.\ref{fig:2-txwg-3-joins}.

\item \textbf{Decomposing graphs into Hanzi-graphs.} We can vertex-split $J_2$ into $J_1$ shown in Fig.\ref{fig:2-txwg-splittings} into groups of Hanzi-graphs shown in Fig.\ref{fig:2-txwg} and Fig.\ref{fig:2-txwg-other-groups}. Splitting a connected graph into some groups of Hanzi-graphs, such that: (i) these groups differ from each other; (ii) each group of Hanzi-graphs form a Chinese sentence, or a Chinese paragraph to be meaningful in Chinese. Notice that there are two kinds of Chinese characters: one is traditional Chinese characters, and another is simplified Chinese characters (see \cite{GB2312-80}).

\begin{figure}[h]
\centering
\includegraphics[width=12cm]{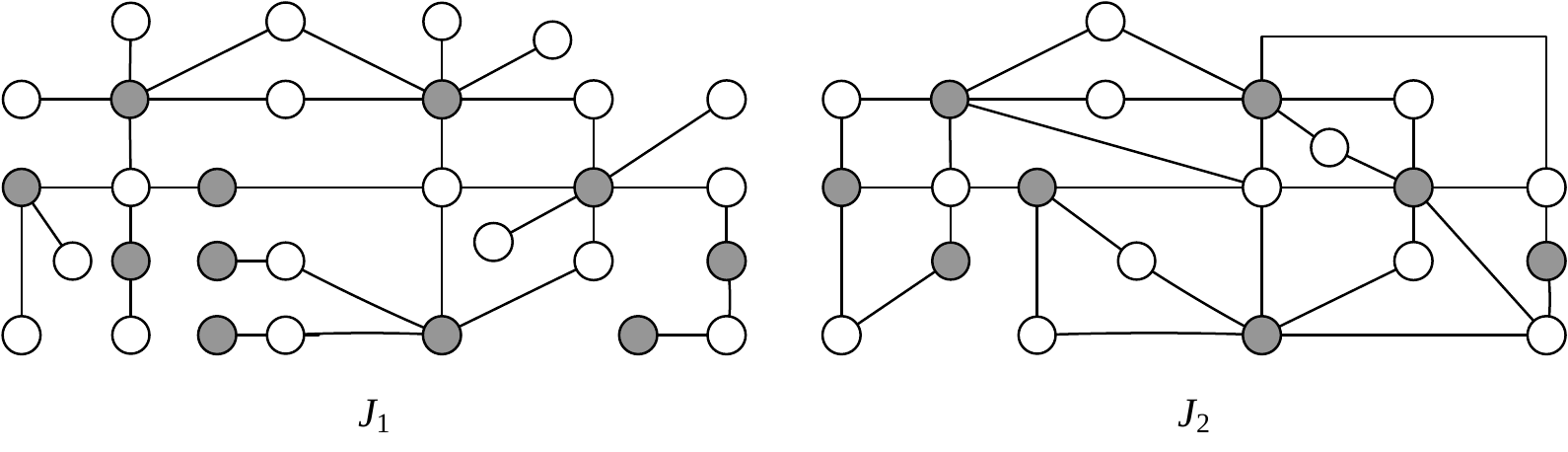}
\caption{\label{fig:2-txwg-splittings}{\small Two connected graphs $J_1$ and $J_2$ obtained by doing vertex-coinciding operation on $G$ shown in Fig.\ref{fig:2-txwg}.}}
\end{figure}

\begin{figure}[h]
\centering
\includegraphics[width=16.4cm]{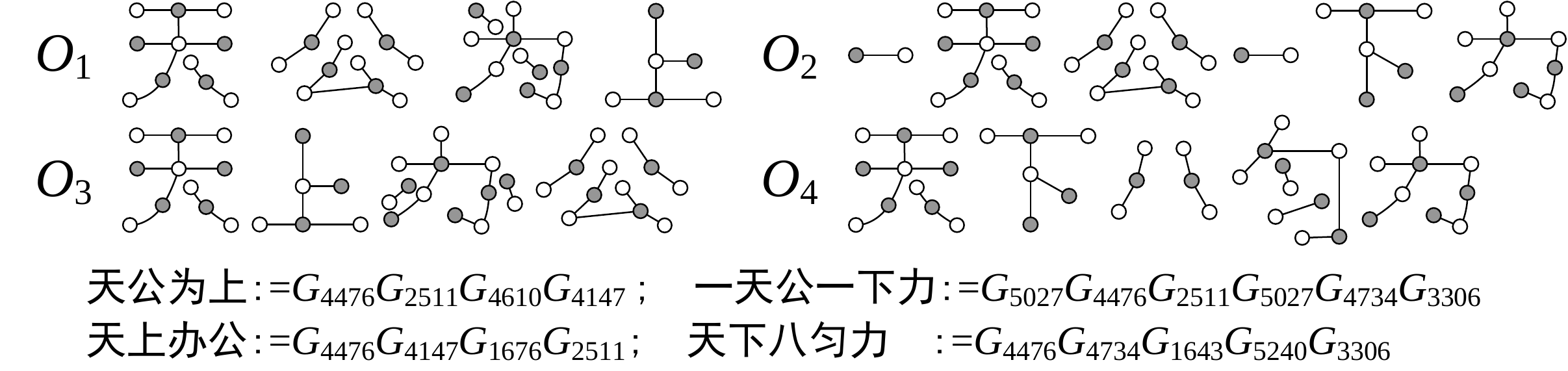}
\caption{\label{fig:2-txwg-other-groups}{\small Four groups $O_1,O_2,O_3,O_4$ of Hanzi-graphs differ from the disconnected graph $G$ shown in Fig.\ref{fig:2-txwg}.}}
\end{figure}

\quad Notice that $J_1\not \neq J_2$, however, both $J_1$ and $J_2$ can be vertex-split into the same groups of Hanzi-graphs shown in Fig.\ref{fig:2-txwg} and Fig.\ref{fig:2-txwg-other-groups}, so it is difficult and complex in Hanzi-graph authentication, but it means that Hanzi-graph authentication has greatly application potential in the ear of supercomputers and quantum computers.\qqed
\end{asparaenum}
\end{problem}

\begin{figure}[h]
\centering
\includegraphics[width=16cm]{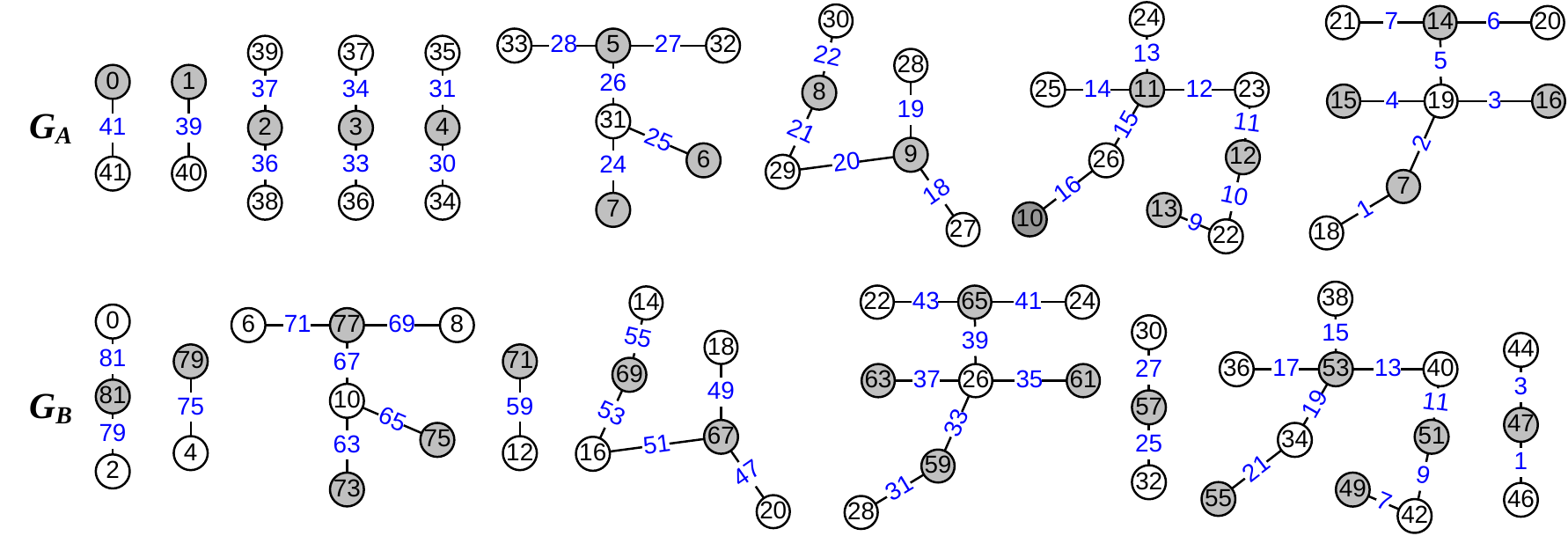}
\caption{\label{fig:Dpermutation-graceful-odd}{\small A flawed graceful labelling and a flawed odd-graceful labelling for two permutations $G_A$ and $G_B$ of nine connected graphs $G_1,G_2,\dots ,G_9$ shown in Fig.\ref{fig:2-txwg-decompose}.}}
\end{figure}

\subsubsection{Topological authentication problems}

Yao \emph{et al.} in \cite{Yao-Wang-Su-Ma-Wang-Sun-ITNEC-2020} and \cite{Yao-Wang-Ma-Su-Wang-Sun-ITNEC-020} have investigated \emph{multiple authentication} (also, multiple color-valued graphic authentication) and topological authentication by especial total colorings and real-valued total colorings. These techniques are based on Topsnut-gpws (graphic passwords made by topological structure and number theory) proposed first by Wang \emph{et al.} in \cite{Wang-Xu-Yao-2016} and \cite{Wang-Xu-Yao-Key-models-Lock-models-2016}, since Topsnut-gpws have many excellent properties: (i) each one is a composition of a topological structure and a mathematical restriction; (ii) each one can be saved in computer by a popular matrix of algebra, and runs quickly in computation; (iii) each one is easily to produce text-based passwords with longer bytes; (iv) they are suitable to design various topological authentications, such as one-vs-more authentication, more-vs-more authentication, and so on.

\begin{problem}\label{qeu:isomorphic}
In \cite{Wang-Su-Yao-mixed-difference-2019}, Wang \emph{et al.} propose the \emph{topological coloring isomorphic problem} consisted of graph isomorphism and coloring isomorphism, which will induce more complex topological authentications. However, we are facing the following mathematical problems in the topological authentication:

\begin{asparaenum}[\quad \textbf{\textrm{Iso}}-1. ]
\item \textbf{Find} a multivariate function $\theta$ of vertex colors and edges colors of each particular subgraph $H$ of a graph $G$ such that $G$ admits a proper total coloring $h$ holding $\theta(h(V(H)\cup E(H)))=$ a constant for each particular subgraph $H$, where

\quad 1-1. $H$ may be an edge, or a face $f_i$ having bound $B(f_i)$, or a cycle $C_n$, or a path $P_n$, and so on.

\quad 1-2. Find more multivariate functions $\theta$ of vertex colors and edge colors such that $\theta$ to be a constant under a proper total coloring of $G$.

\item Given a set $V_{co}$ of colored vertices and a set $E_{co}$ of colored edges, \textbf{how} to assemble all elements of two sets into graphs $G_i$ such that $G_i$ is just colored by a $W$-type total coloring $f$ holding $f(V(G_i)\cup E(G_i))\subseteq V_{co}\cup E_{co}$.
\item \textbf{$J$-graphic isomorphic problem.} Let $G$ and $H$ be two graphs of $p$ vertices, and let $J$ be a particular graph. Suppose that each vertex of $G$ is in some particular graph $J\subset G$, so is each vertex of $H$ in $J\subset H$. If $G-V(J)\cong H-V(J)$ for each particular graph $J$ of $G$ and $H$, can we claim $G\cong H$? Here, $J$ may be a path of $p$ vertices, or a cycle of $p$ vertices, or a complete graph of $p$ vertices, \emph{etc.} Recall, let $G$ and $H$ be two graphs that have the same number of vertices. If there exists a bijection $\varphi: V (G)\rightarrow V (H)$ such that $uv \in E(G)$ if and only if $\varphi(u)\varphi(v) \in E(H)$, then we say both graphs $G$ and $H$ to be isomorphic to each other, denoted by $G\cong H$ in \cite{Bondy-2008}. A long-standing Kelly-Ulam's Reconstruction Conjecture (1942): Let both $G$ and $H$ be graphs with $n$ vertices. If there is a bijection $f: V (G)\rightarrow V (H)$ such that $G-u\cong H -f(u)$ for each vertex $u \in V (G)$, then $G\cong H$. This conjecture supports some cryptosystems consisted of graphic isomorphism to be ``Resisting classical computers and quantum computers''.
\item A \emph{topological coloring isomorphism} consists of graph isomorphism and coloring isomorphism. For two colored graphs $G$ admitting a $W$-type total coloring $f$ and $H$ admitting a $W$-type total coloring $g$, if there is a mapping $\varphi$ such that $w'=\varphi(w)$ for each element $w\in V(G)\cup E(G)$ and each element $w'\in V(H)\cup E(H)$, then we say they are \emph{isomorphic} to each other, and write this case by $G\cong H$, and moreover if $g(w')=f(w)$ for $w'=\varphi(w)$, we say they are subject to \emph{coloring isomorphic} to each other, so we denoted $G=H$ for expressing the combination of topological isomorphism and coloring isomorphism.
\item In \cite{YAO-SUN-WANG-SU-XU2018arXiv}, the authors defined: Let ``\emph{$W$-type labelling}'' be a given graph labelling, and let a connected graph $G$ admit a $W$-type labelling. If every connected proper subgraph of $G$ also admits a labelling to be a $W$-type labelling, then we call $G$ a \emph{perfect $W$-type labelling graph}. Caterpillars are perfect $W$-type labelling graphs if these $W$-type labellings are listed in Theorems \ref{thm:connections-several-labellings}, \ref{thm:connection-flawed-labellings} and \ref{thm:flawed-(k,d)-labellings}, and each lobster is a perfect (odd-)graceful labelling graph. Conversely, we ask for: \emph{If every connected proper subgraph of a connected graph $G$ admits a $W$-type labelling, then \textbf{does} $G$ admit this $W$-type labelling too}?
\item \textbf{How} to construct a matrix $A_{3\times q}$ by a given integer character string such that $A_{3\times q}$ is just a Topsnut-matrix of some Topsnut-gpw?
\item Let $P_3\times P_q$ be a lattice in $xOy$-plane. There are points $(i,j)$ on the lattice $P_3\times P_q$ with $i\in[1,3]$ and $j\in [1,q]$. If a fold-line $L$ with initial point $(a,b)$ and terminal point $(c,d)$ on $P_3\times P_q$ is internally disjoint and contains all points $(i,j)$ of $P_3\times P_q$, we call $L$ a \emph{total TB-paw line}. \textbf{Find} all possible total TB-paw lines of $P_3\times P_q$. In general, let $\{L_i\}^m_1=\{L_1,L_2,\dots, L_m\}$ be a set of $m$ disjoint fold-lines on $P_3\times P_q$, where each $L_i$ has own initial point $(a_i,b_i)$ and terminal point $(c_i,d_i)$. If the fold-line set $\{L_i\}^m_1$ contains all points $(i,j)$ of $P_3\times P_q$, we call $\{L_i\}^m_1$ a \emph{group of TB-paw lines}, here it is not allowed $(a_i,b_i)=(c_i,d_i)$ for any fold-line $L_i$. Find all possible groups $\{L_i\}^m_1$ of TB-paw lines for $m\in [1,3q]$.\qqed
\end{asparaenum}
\end{problem}

\begin{defn}\label{defn:definition-graph-homomorphism}
\cite{Bondy-2008} A \emph{graph homomorphism} $G\rightarrow H$ from a graph $G$ into another graph $H$ is a mapping $f: V(G) \rightarrow V(H)$ such that $f(u)f(v)\in E(H)$ for each edge $uv \in E(G)$. (see examples shown in Fig.\ref{fig:1-homomorphism}.)\qqed
\end{defn}
\begin{figure}[h]
\centering
\includegraphics[width=12cm]{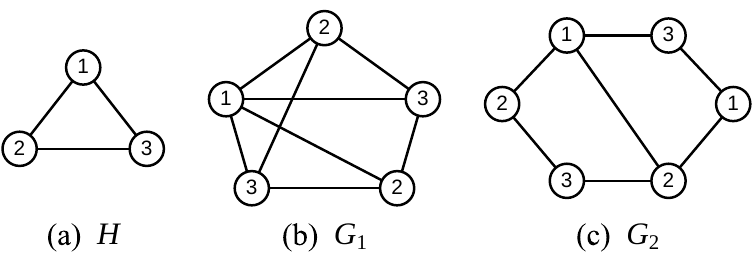}
\caption{\label{fig:1-homomorphism}{\small Two graph homomorphisms $\theta_i: G_i\rightarrow H$ for $i=1,2$.}}
\end{figure}

The comprehensive survey by Zhu \cite{Zhu-X-Discrete-Math2001} contains many other intriguing problems about graph homomorphism. By \cite{Pavol-Hell-Cambridge-2003}, we have the following concepts:
\begin{asparaenum}[(a) ]
\item A homomorphism from a graph $G$ to itself is called an \emph{endomorphism}. An isomorphism from $G$ to $H$ is a particularly graph homomorphism from $G$ to $H$, also, they are \emph{homomorphically equivalent}.
\item Two graphs are homomorphically equivalent if each admits a homomorphism to the other, denoted as $G\leftrightarrow H$ which contains a homomorphism $G\rightarrow H$ from $G$ to $H$, and another homomorphism $H\rightarrow G$ from $H$ to $G$.
\item A homomorphism to the complete graph $K_n$ is exactly an $n$-coloring, so a homomorphism of $G$ to $H$ is also called an \emph{$H$-coloring} of $G$. The \textbf{homomorphism problem} for a fixed graph $H$, also called the $H$-coloring problem, asks whether or not an input graph $G$ admits a homomorphism to $H$.
\item By analogy with classical colorings, we associate with each $H$-coloring $f$ of $G$ a partition of $V (G)$ into the sets $S_h = f^{-1}(h)$, $h \in V (H)$. It is clear that a mapping $f: V (G) \rightarrow V (H)$ is a homomorphism of $G$ to $H$ if and only if the associated partition satisfies the following two constraints:

\quad (a-1) if $hh$ is not a loop in $H$, then the set $S_h$ is independent in $G$; and

\quad (a-2) if $hh'$ is not an edge (arc) of $H$, then there are no edges (arcs) from the set $S_h$ to the set $S_{h'}$ in $G$.

\quad Thus for a graph $G$ to admit an \emph{$H$-coloring is equivalent} to admitting a partition satisfying (a-1) and (a-2).
\item If $H$,$H'$ are homomorphically equivalent, then a graph $G$ is $H$-colorable if and only if it is $H'$-colorable.
\item Suppose that $H$ is a subgraph of $G$. We say that $G$ retracts to $H$, if there exists a homomorphism $f:G\rightarrow H$, called a \emph{retraction}, such that $f(u)=u$ for any vertex of $H$. A \emph{core} is a graph which does not retract to a proper subgraph. Any graph is homomorphically equivalent to a core.
\end{asparaenum}

\subsection{Preliminary}

\subsubsection{Notation and terminology} Standard notation and terminology of graph theory will be used in this article and can be found in \cite{Bondy-2008} and \cite{Gallian2019}. Graphs mentioned are simple, that is, they have no loops and multiple edges, hereafter.

\begin{asparaenum}[$\star$ ]
\item A $(p,q)$-graph is a graph having $p$ vertices and $q$ edges.
\item The \emph{cardinality} of a set $X$ is denoted as $|X|$, so the \emph{degree} of a vertex $x$ in a $(p,q)$-graph $G$ is $\textrm{deg}_G(x)=|N(x)|$, where $N(x)$ is the set of neighbors of the vertex $x$.
\item A vertex $x$ is called a \emph{leaf} if its degree $\textrm{deg}_G(x)=1$.
\item The symbol $[a,b]$ stands for an integer set $\{a,a+1,a+2,\dots, b\}$ with two integers $a,b$ subject to $0<a<b$, and $[a,b]^o$ denotes an \emph{odd-set} $\{a,a+2,\dots, b\}$ with odd integers $a,b$ with respect to $1\leq a<b$.
\item A \emph{text-based password} is abbreviated as \emph{TB-paw}. A password made by ``topological structure and number theory'' is simply written as \emph{Topsnut-gpw}.
\item A \emph{text string} $D=t_1t_2\cdots t_m$ has its own \emph{reciprocal text string} defined by $D^{-1}=t_mt_{m-1}\cdots t_2t_1$, also, we say that $D$ and $D^{-1}$ match with each other.
\item All non-negative integers are collected in the set $Z^0$.
\item A graph $G$ admits a \emph{labelling} $f:V(G)\rightarrow [a,b]$ means that $f(x)\neq f(y)$ for any pair of distinct vertices $x,y\in V(G)$.
\item A graph $G$ admits a \emph{coloring} $g:V(G)\rightarrow [a,b]$ means that $g(u)= g(v)$ for some two distinct vertices $u,v\in V(G)$.
\item For a mapping $f:S\subset V(G)\cup E(G)\rightarrow [1,M]$, we write $f(S)=\{f(w):w\in S\}$.
\item A proper total coloring $f:V(G)\cup E(G)\rightarrow [1,M]$ of a simple graph $G$ holds $f(u)\neq f(v)$ for each edge $uv\in E(G)$ and $f(uv)\neq f(uw)$ for distinct neighbors $v,w\in N(u)$. The number $\chi''(G)=\min_f\{M:~f\textrm{ is a proper total coloring of }G\}$ is called the \emph{total chromatic number} of $G$.
\end{asparaenum}

\subsubsection{Graph operations} Graph operation is not only very important in graph theory, but also useful and efficient in application of network security.

\begin{asparaenum}[\quad \textbf{\textrm{Oper}}-1. ]
\item \emph{Vertex-splitting operation.} Let $x$ be a vertex of a graph $G$ with its degree $\textrm{deg}_G(x)=d\geq 2$, and its neighbor set $N(x)=\{x_1,x_2,\dots, x_d\}$. We vertex-split the vertex $x$ into two vertices $x',x''$ such that $N(x)=N(x')\cup N(x'')$, where $N(x')=\{x_1,x_2,\dots, x_k\}$ and $N(x'')=\{x_{k+1},x_{k+2},\dots, x_d\}$ with $1\leq k<d$, and $N(x')\cap N(x'')=\emptyset$. There resultant graph is denoted as $G\wedge x$, and the process of obtaining $G\wedge x$ is called \emph{vertex-splitting operation} (see Fig.\ref{fig:1-vertex-splitting}).
\item \emph{Vertex-coinciding operation.} Suppose that two vertices $x'$ and $x''$ of a graph $H$ hold $N(x')\cap N(x'')=\emptyset$, then we vertex-coincide these two vertices $x'$ and $x''$ into one vertex $x$, and write the resultant graph as $H(x'\odot x'')$, and call the process of obtaining $H(x'\odot x'')$ as \emph{vertex-coinciding operation}. Since $|E(H)|=|E(H(x'\odot x''))|$, we call this vertex-coinciding operation as \emph{edge-protected vertex-coinciding operation} (see Fig.\ref{fig:1-vertex-splitting}).

\begin{figure}[h]
\centering
\includegraphics[width=11cm]{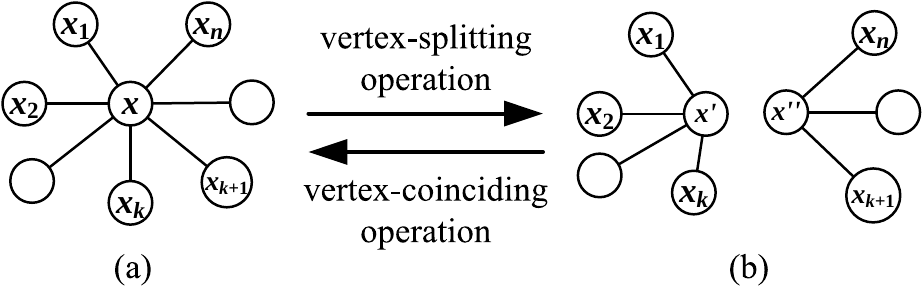}
\caption{\label{fig:1-vertex-splitting}{\small A scheme for the vertex-splitting operation and the vertex-coinciding operation.}}
\end{figure}

\quad Let $G_1$ and $G_2$ be two disjoint graphs. We take vertices $x_{i,1},x_{i,2},\dots ,x_{i,m}$ of $G_i$ with $i=1,2$, and vertex-coincide the vertex $x_{1,j}$ with the vertex $x_{2,j}$ into one $w_j=x_{1,j}\odot x_{2,j}$ with $j\in [1,m]$, the resultant graph is denoted as $G_1\odot G_2$, called the \emph{vertex-coincided graph}. Conversely, we vertex-split each vertex $w_j=x_{1,j}\odot x_{2,j}$ into two vertices $x_{1,j}$ and $x_{2,j}$ with $j\in [1,m]$, the vertex-coincided graph $G_1\odot G_2$ is vertex-split into two disjoint graphs $G_1$ and $G_2$, and we denote the process of vertex-splitting $G_1\odot G_2$ into $G_1$ and $G_2$ by $(G_1\odot G_2)\wedge \{w_j\}^m_1$.

\quad The authors in \cite{Yao-Zhang-Sun-Mu-Sun-Wang-Wang-Ma-Su-Yang-Yang-Zhang-2018arXiv, Yao-Mu-Sun-Sun-Zhang-Wang-Su-Zhang-Yang-Zhao-Wang-Ma-Yao-Yang-Xie2019} introduced the \emph{vertex-splitting operation} and the \emph{vertex-coinciding operation}, these two operations are a pair of mutually inverse operations.

\item \emph{Substitution operation.} Let $x$ be a vertex of a graph $G$, and $N(x)=\{x_1$, $x_2$, $\dots $, $ x_d\}$ be the neighbor set of the vertex $x$, where $d=\textrm{deg}(x)$. A vertex-substitution operation is defined as: For a graph $H$ with vertex set $\{y_1,y_2,\dots ,y_n\}$ with $n\geq d$, we remove the vertex $x$ from $G$, and add $H$ to the remainder graph $G-x$ by joining $y_i$ and $x_i$ together by an edge $y_ix_i$ with $i\in [1,d]$. The resultant graph is called a \emph{vertex-substitution graph}, written as $(G-x)\ominus H$.

\quad In general, we take a vertex subset $V'=\{u_1,u_2,\dots,u_m\}$ of a graph $G$, here, each neighbor set $N(u_j)=\{v_{j,1},v_{j,2},\dots, v_{j,d_j}\}$ with $d_j=\textrm{deg}(u_j)$ for $j\in [1,m]$. Let each $H_j\in S^*=\{H_1,H_2,\dots,H_m\}$ be a graph having vertex set $\{w_{j,1},w_{j,2},\dots, w_{j,n_j}\}$ with $n_j\geq d_j$ for $j\in [1,m]$, we delete the vertices of $V'$ from $G$, and add $H_j$ to the remainder graph $G-V'$ by joining $w_{j,i}$ and $u_i$ together by an edge $w_{j,i}u_i$ for $j\in [1,m]$. The \emph{vertex-substitution graph} is denoted as $(G-V')\ominus S^*$.
\item \cite{Wang-Yao-Star-type-Lattices-2020} \emph{Leaf-splitting and leaf-coinciding operations.} Let $uv$ be an edge of a graph $G$ with a (proper) total coloring $f$, and $\textrm{deg}_G(u)\geq 2$, $\textrm{deg}_G(v)\geq 2$. A \emph{leaf-splitting operation} is defined as: Remove the edge $uv$ from $G$, the resulting graph is denoted as $G-uv$. Add a new leaf $v'$, and join it with the vertex $u$ of $G-uv$ by a new edge $uv'$, and then add another new leaf $u'$ to join it with the vertex $v$ of $G-uv$ by another new edge $vu'$, the resultant graph is written as $H=G(uv\prec)$. Defined a (proper) total coloring $g$ of $H$ as: $g(w)=f(w)$ for each element $w\in [V(H)\cup E(H)]\setminus \{u',v',uv',vu'\}$, $g(u')=f(u)$, $g(v')=f(v)$, $g(uv')=f(uv)$ and $g(vu')=f(uv)$. See Fig.\ref{fig:leaf-splitting} from (a) to (b). Conversely, a \emph{leaf-coinciding operation} is defined by vertex-coinciding two leaves $uv'$ and $vu'$ of $H=G(uv\prec)$ admitting a (proper) total coloring $g$ into one edge $uv=uv'\ominus vu'$ if $g(u)=g(u')$, $g(v)=g(v')$ and $g(uu')=g(vv')$. The resultant graph is written as $G=H(uv'\ominus vu')$. And define a (proper) total coloring $f$ of $G$ as: $f(w)=g(w)$ for each element $w\in [V(G)\cup E(G)]\setminus \{uv\}$, $f(uv)=g(uv')=g(vu')$. For understanding this leaf-coinciding operation see Fig.\ref{fig:leaf-splitting} from (b) to (a), also, this operation is very similar with the connection of two train hooks.
\end{asparaenum}

\begin{figure}[h]
\centering
\includegraphics[width=13cm]{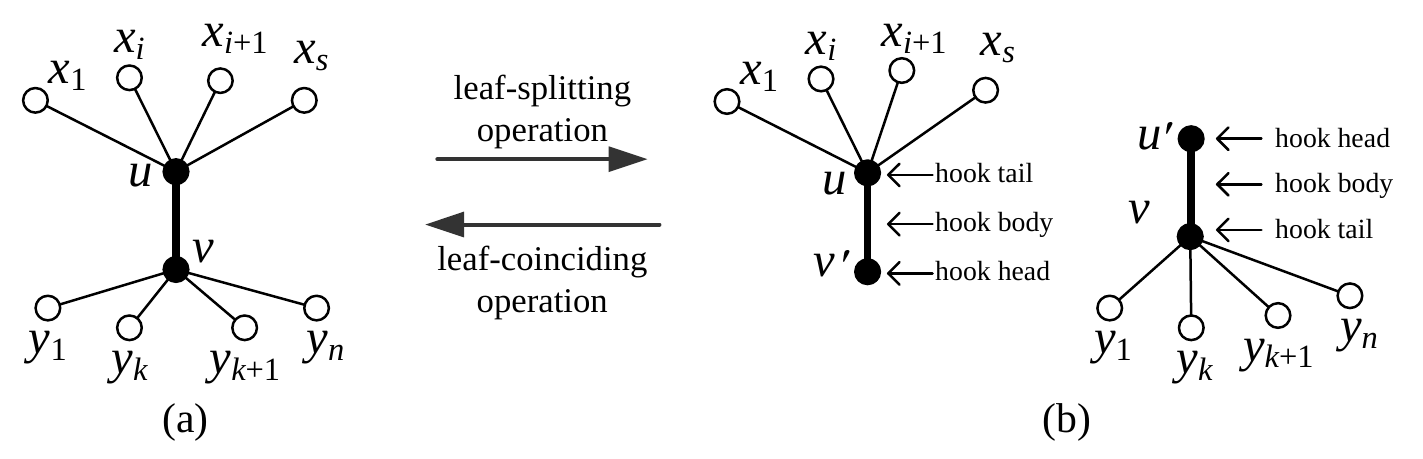}\\
\caption{\label{fig:leaf-splitting} {\small Leaf-splitting operation and leaf-coinciding operation (also, the connection of two train hooks).}}
\end{figure}

\subsubsection{Particular proper total colorings} As known, there are many intriguing colorings/labellings of graphs (Ref. \cite{Bondy-2008, Gallian2019, Wang-Xu-Yao-2017-Twin, Yao-Sun-Zhang-Mu-Sun-Wang-Su-Zhang-Yang-Yang-2018arXiv, Yao-Zhang-Sun-Mu-Sun-Wang-Wang-Ma-Su-Yang-Yang-Zhang-2018arXiv, Yao-Zhao-Zhang-Mu-Sun-Zhang-Yang-Ma-Su-Wang-Wang-Sun-arXiv2019}). Here, a graph $G$ admitting a ``\emph{$W$-type coloring}'' means one of particular colorings and graph labellings of graph theory hereafter. A \emph{proper total coloring} $f$ of a graph $G$ is a mapping: $f: V(G)\cup E(G)\rightarrow [1,M]$, such that $f(x)\neq f(y)$ for any pair of adjacent vertices $x,y\in V(G)$ and $f(uv)\neq f(uw)$ for any pair of adjacent edges $uv,uw\in E(G)$. We restate several particular $W$-type labellings as follows:

\begin{defn} \label{defn:proper-bipartite-labelling-ongraphs}
\cite{Gallian2019, Yao-Cheng-Yao-Zhao-2009, Zhou-Yao-Chen-Tao2012, Yao-Sun-Zhang-Mu-Sun-Wang-Su-Zhang-Yang-Yang-2018arXiv} Suppose that a connected $(p,q)$-graph $G$ admits a mapping $\theta:V(G)\rightarrow \{0,1,2,\dots \}$. For edges $xy\in E(G)$
the induced edge labels are defined as $\theta(xy)=|\theta(x)-\theta(y)|$. Write vertex color set $\theta(V(G))=\{\theta(u):u\in V(G)\}$ and edge color set
$\theta(E(G))=\{\theta(xy):xy\in E(G)\}$. There are the following constraints:
\begin{asparaenum}[(a)]
\item \label{Proper01} $|\theta(V(G))|=p$.
\item \label{Proper02} $|\theta(E(G))|=q$.
\item \label{Graceful-001} $\theta(V(G))\subseteq [0,q]$, $\min \theta(V(G))=0$.
\item \label{Odd-graceful-001} $\theta(V(G))\subset [0,2q-1]$, $\min \theta(V(G))=0$.
\item \label{Graceful-002} $\theta(E(G))=\{\theta(xy):xy\in E(G)\}=[1,q]$.
\item \label{Odd-graceful-002} $\theta(E(G))=\{\theta(xy):xy\in E(G)\}=[1,2q-1]^o$.
\item \label{set-ordered-definition-1} $G$ is a bipartite graph with the bipartition
$(X,Y)$ such that $\max\{\theta(x):x\in X\}< \min\{\theta(y):y\in
Y\}$ ($\max \theta(X)<\min \theta(Y)$ for short).
\item \label{Graceful-matching} $G$ is a tree containing a perfect matching $M$ such that
$\theta(x)+\theta(y)=q$ for each edge $xy\in M$.
\item \label{Odd-graceful-matching} $G$ is a tree having a perfect matching $M$ such that
$\theta(x)+\theta(y)=2q-1$ for each edge $xy\in M$.
\end{asparaenum}

We have: a \emph{graceful labelling} $\theta$ satisfies (\ref{Proper01}), (\ref{Graceful-001}) and (\ref{Graceful-002}); a \emph{set-ordered
graceful labelling} $\theta$ holds (\ref{Proper01}), (\ref{Graceful-001}), (\ref{Graceful-002}) and (\ref{set-ordered-definition-1}) true;
a \emph{strongly graceful labelling} $\theta$ holds (\ref{Proper01}), (\ref{Graceful-001}), (\ref{Graceful-002}) and
(\ref{Graceful-matching}) true; a \emph{strongly set-ordered graceful labelling} $\theta$ holds (\ref{Proper01}), (\ref{Graceful-001}), (\ref{Graceful-002}), (\ref{set-ordered-definition-1}) and (\ref{Graceful-matching}) true. An \emph{odd-graceful labelling} $\theta$ holds (\ref{Proper01}), (\ref{Odd-graceful-001}) and (\ref{Odd-graceful-002}) true; a \emph{set-ordered odd-graceful labelling} $\theta$ holds (\ref{Proper01}), (\ref{Odd-graceful-001}), (\ref{Odd-graceful-002})
and (\ref{set-ordered-definition-1}) true; a \emph{strongly odd-graceful labelling} $\theta$ holds (\ref{Proper01}), (\ref{Odd-graceful-001}),
(\ref{Odd-graceful-002}) and (\ref{Odd-graceful-matching}) true; a \emph{strongly set-ordered odd-graceful labelling} $\theta$ holds (\ref{Proper01}), (\ref{Odd-graceful-001}), (\ref{Odd-graceful-002}), (\ref{set-ordered-definition-1}) and (\ref{Odd-graceful-matching}) true.\qqed
\end{defn}

We introduce a group of particular total colorings, in which some are very similar with that in \cite{Gallian2019, Yao-Zhang-Sun-Mu-Sun-Wang-Wang-Ma-Su-Yang-Yang-Zhang-2018arXiv, Zhou-Yao-Chen2013}, as follows:

\begin{defn} \label{defn:new-graceful-strongly-colorings}
$^*$ Suppose that a connected $(p,q)$-graph $G$ admits a proper total coloring $f:V(G)\cup E(G)\rightarrow [1,M]$, and there are $f(x)=f(y)$ for some pairs of vertices $x,y\in V(G)$. Write $f(S)=\{f(w):w\in S\}$ for any non-empty set $S\subseteq V(G)\cup E(G)$. We have a group of constraints as follows:
\begin{asparaenum}[(1$^\circ$)]
\item \label{27TProper01} $|f(V(G))|< p$.
\item \label{27TProper02} $|f(E(G))|=q$.
\item \label{27TGraceful-001} $f(V(G))\subseteq [1,M]$, $\min f(V(G))=1$.
\item \label{27TOdd-graceful-001} $f(V(G))\subset [1,2q+1]$, $\min f(V(G))=1$.
\item \label{27TGraceful-002} $f(E(G))=[1,q]$.
\item \label{27Tmodulo-01} $f(E(G))=[0,q-1]$.
\item \label{27TOdd-graceful-002} $f(E(G))=[1,2q-1]^o$.
\item \label{27Tedge-set-odd-sum} $f(E(G))=[1,2q-1]^o$.
\item \label{27Teven-edge-set-11} $f(E(G))=[2, 2q]^e$.
\item \label{27Tedge-set-odd-sum} $f(E(G))=[1,2q-1]^o$.
\item \label{27Tsequential-edge-set} $f(E(G))=[c,c+q-1]$.
\item \label{27Tgraceful-002} $f(uv)=|f(u)-f(v)|$.
\item \label{27Tfelicitous-002} $f(uv)=f(u)+f(v)$.
\item \label{27Tedge-labels-even-odd} $f(uv)=f(u)+f(v)$ when $f(u)+f(v)$ is even, and $f(uv)=f(u)+f(v)+1$ when $f(u)+f(v)$ is odd.
\item \label{27Tmodulo-00} $f(uv)=f(u)+f(v)~(\textrm{mod}~q)$.
\item \label{27Tmodulo-11} $f(uv)=f(u)+f(v)~(\textrm{mod}~2q)$.
\item \label{27Tedge-difference} $f(uv)+|f(u)-f(v)|=k$.
\item \label{27Tgraceful-difference} $\big |f(uv)-|f(u)-f(v)|\big |=k$.
\item \label{27Tedge-magic} $f(u)+f(uv)+f(v)=k$.
\item \label{27Tmodulo-ordered} There exists an integer $k$ so that $\min \{f(u),f(v)\}\leq k <\max\{f(u),f(v)\}$.
\item \label{27TSet-ordered} $(X,Y)$ is the bipartition of a bipartite graph $G$ such that $\max f(X)< \min f(Y)$.
\end{asparaenum}

We then have a \emph{$W$-type} coloring $f$ to be:
\begin{asparaenum}[(1)]
\item a \emph{gracefully total coloring} if (\ref{27TProper01}$^\circ$), (\ref{27TGraceful-001}$^\circ$), (\ref{27TGraceful-002}$^\circ$) and (\ref{27Tgraceful-002}$^\circ$) hold true.
\item a \emph{set-ordered
gracefully total coloring} if (\ref{27TProper01}$^\circ$), (\ref{27TGraceful-001}$^\circ$), (\ref{27TGraceful-002}$^\circ$), (\ref{27Tgraceful-002}$^\circ$) and (\ref{27TSet-ordered}$^\circ$) hold true.
\item an \emph{odd-gracefully total coloring} if (\ref{27TProper01}$^\circ$), (\ref{27TOdd-graceful-001}$^\circ$), (\ref{27TOdd-graceful-002}$^\circ$) and (\ref{27Tgraceful-002}$^\circ$) hold true.
\item a \emph{set-ordered odd-gracefully total coloring} if (\ref{27TProper01}$^\circ$), (\ref{27TOdd-graceful-001}$^\circ$), (\ref{27TOdd-graceful-002}$^\circ$), (\ref{27Tgraceful-002}$^\circ$) and (\ref{27TSet-ordered}$^\circ$) hold true.
\item a \emph{felicitous total coloring} if (\ref{27TGraceful-001}$^\circ$),(\ref{27Tmodulo-00}$^\circ$) and (\ref{27Tmodulo-01}$^\circ$) hold true.
\item a \emph{set-ordered felicitous total coloring} if (\ref{27TGraceful-001}$^\circ$),(\ref{27Tmodulo-00}$^\circ$), (\ref{27Tmodulo-01}$^\circ$) and (\ref{27TSet-ordered}$^\circ$) hold true.
\item an \emph{odd-elegant total coloring} if (\ref{27TOdd-graceful-001}$^\circ$), (\ref{27Tmodulo-11}$^\circ$) and (\ref{27TOdd-graceful-002}$^\circ$) hold true.
\item a \emph{set-ordered odd-elegant total coloring} if (\ref{27TOdd-graceful-001}$^\circ$), (\ref{27Tmodulo-11}$^\circ$), (\ref{27TOdd-graceful-002}$^\circ$) and (\ref{27TSet-ordered}$^\circ$) hold true.
\item a \emph{harmonious total coloring} if (\ref{27TGraceful-001}$^\circ$), (\ref{27Tmodulo-00}$^\circ$) and (\ref{27Tmodulo-01}$^\circ$) hold true, and when $G$ is a tree, exactly one edge label may be used on two vertices.
\item a \emph{set-ordered harmonious total coloring} if (\ref{27TGraceful-001}$^\circ$), (\ref{27Tmodulo-00}$^\circ$), (\ref{27Tmodulo-01}$^\circ$) and (\ref{27TSet-ordered}$^\circ$) hold true.
\item a \emph{strongly harmonious total coloring} if (\ref{27TGraceful-001}$^\circ$), (\ref{27Tmodulo-00}$^\circ$), (\ref{27Tmodulo-01}$^\circ$) and (\ref{27Tmodulo-ordered}$^\circ$) hold true.
\item a \emph{properly even harmonious total coloring} if (\ref{27TOdd-graceful-001}$^\circ$), (\ref{27Tmodulo-11}$^\circ$) and (\ref{27Tedge-set-odd-sum}$^\circ$) hold true.
\item a \emph{$c$-harmonious total coloring} if (\ref{27TGraceful-001}$^\circ$), (\ref{27Tfelicitous-002}$^\circ$) and (\ref{27Tsequential-edge-set}$^\circ$) hold true.
\item an \emph{even sequential harmonious total coloring} if (\ref{27TOdd-graceful-001}$^\circ$), (\ref{27Tedge-labels-even-odd}$^\circ$) and (\ref{27Teven-edge-set-11}$^\circ$) hold true.
\item a \emph{pan-harmonious total coloring} if (\ref{27TProper02}$^\circ$) and (\ref{27Tfelicitous-002}$^\circ$) hold true.
\item an \emph{edge-magic total coloring} if (\ref{27Tedge-magic}$^\circ$) holds true.
\item a \emph{set-ordered edge-magic total coloring} if (\ref{27Tedge-magic}$^\circ$) and (\ref{27TSet-ordered}$^\circ$) hold true.
\item a \emph{graceful edge-magic total coloring} if (\ref{27TGraceful-002}$^\circ$) and (\ref{27Tedge-magic}$^\circ$) hold true.
\item a \emph{set-ordered graceful edge-magic total coloring} if (\ref{27TGraceful-002}$^\circ$), (\ref{27Tedge-magic}$^\circ$) and (\ref{27TSet-ordered}$^\circ$) hold true.
\item an \emph{edge-difference magic total coloring} if (\ref{27Tedge-difference}$^\circ$) holds true.
\item a \emph{set-ordered edge-difference magic total coloring} if (\ref{27Tedge-difference}$^\circ$) and (\ref{27TSet-ordered}$^\circ$) hold true.
\item a \emph{graceful edge-difference magic total coloring} if (\ref{27TGraceful-002}$^\circ$) and (\ref{27Tedge-difference}$^\circ$) hold true.
\item a \emph{set-ordered graceful edge-difference magic total coloring} if (\ref{27TGraceful-002}$^\circ$), (\ref{27Tedge-difference}$^\circ$) and (\ref{27TSet-ordered}$^\circ$) hold true.
\item an \emph{ev-difference magic total coloring} if (\ref{27Tgraceful-difference}$^\circ$) holds true.
\item a \emph{set-ordered ev-difference magic total coloring} if (\ref{27Tgraceful-difference}$^\circ$) and (\ref{27TSet-ordered}$^\circ$) hold true.
\item a \emph{graceful ev-difference magic total coloring} if (\ref{27TGraceful-002}$^\circ$) and (\ref{27Tgraceful-difference}$^\circ$) hold true.
\item a \emph{set-ordered graceful ev-difference magic total coloring} if (\ref{27TGraceful-002}$^\circ$), (\ref{27Tgraceful-difference}$^\circ$) and (\ref{27TSet-ordered}$^\circ$) hold true.
\end{asparaenum}

\vskip 0.2cm

We call $\chi''_{W,M}(G)=\min_f\{M:~f(V(G))\subseteq [1,M]\}$ over all $W$-type colorings $f$ of $G$ for a fixed $W$ as \emph{$W$-type total chromatic number} of $G$, and the number $v_{W}(G)=\min_f\{|f(V(G))|\}$ over all $W$-type colorings $f$ of $G$ as \emph{$W$-type total splitting number}.\qqed
\end{defn}

\begin{rem}\label{rem:graceful-labelling-vs-graceful-coloring}
Clearly, determining a $W$-type total chromatic number $\chi''_{W,M}(G)$ for a given graph $G$ could be difficult, since $\chi''_{W,M}(G)\geq \chi''(G)$, and the Total Coloring Conjecture $\chi''(G)\leq \Delta(G)+2$ is open now. It is also not slight to determine whether a graph admits a $W$-type total coloring defined in Definition \ref{defn:new-graceful-strongly-colorings}. Also, computing $v_{W}(G)=\min_f\{|f(V(G))|\}$ will meet difficult cases since there are many conjectures of graph labellings. For each integer $m$ subject to $v_{W}(G)<m\leq p-1$, does there exist a $W$-type total coloring $g$ holding $|g(V(G))|=m$?

Comparing Definition \ref{defn:proper-bipartite-labelling-ongraphs} with Definition \ref{defn:new-graceful-strongly-colorings}, a gracefully total coloring $f$ is weaker than a graceful labelling $g$ holding $|g(V(G))|=|V(G)|$, since $|f(V(G))|<|V(G)|$, and this gracefully total coloring $f$ is stronger than the traditional total coloring because of $f(E(G))=[1,q]$. So, there are more graphs admitting (set-ordered) $W$-type total colorings than with admitting (set-ordered) $W$-type labellings.

We meet $f(E(G))=\{k\}^q_{k=1}$ or $f(E(G))=\{2k-1\}^q_{k=1}$ in Definition \ref{defn:new-graceful-strongly-colorings}, so we can consider some $W$-type total colorings with $f(E(G))=\{a_n\}^q_{k=1}$, where $\{a_n\}^q_{k=1}$ is a strict increasing sequence of positive integers, and we call them \emph{$\{a_n\}^q_{k=1}$-type proper total colorings}.\qqed
\end{rem}

\begin{figure}[h]
\centering
\includegraphics[width=16.2cm]{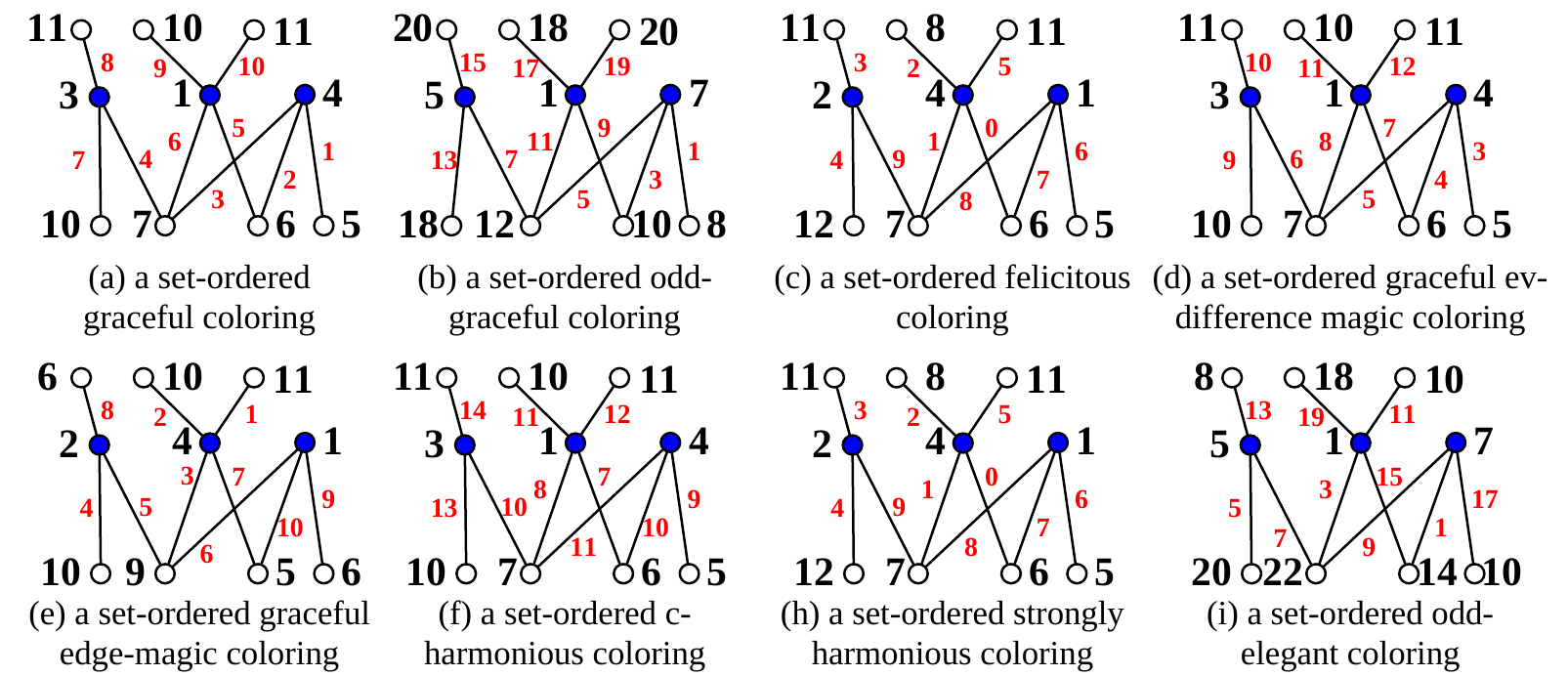}\\
\caption{\label{fig:graceful-colorings} {\small Part of examples for understanding Definition \ref{defn:new-graceful-strongly-colorings}.}}
\end{figure}

\begin{defn}\label{defn:flawed-odd-graceful-labelling}
\cite{Yao-Mu-Sun-Sun-Zhang-Wang-Su-Zhang-Yang-Zhao-Wang-Ma-Yao-Yang-Xie2019} Let $H=E^*+G$ be a connected graph, where $E^*$ is a non-empty set of edges and $G=\bigcup^m_{i=1}G_i$ is a disconnected graph, where $G_1,G_2,\dots, G_m$ are disjoint connected graphs. If $H$ admits a (set-ordered) graceful labelling (resp. a (set-ordered) odd-graceful labelling) $f$, then we call $f$ a \emph{flawed (set-ordered) graceful labelling} (resp. a \emph{flawed (set-ordered) odd-graceful labelling}) of $G$.\qqed
\end{defn}

We will define particular proper total colorings in Definition \ref{defn:combinatoric-definition-total-coloring}, these colorings are combinatory of traditional proper total colorings and graph labellings.

\begin{defn} \label{defn:combinatoric-definition-total-coloring}
$^*$ For a proper total coloring $f:V(G)\cup E(G)\rightarrow [1,M]$ of a simple graph $G$, we define an \emph{edge-function} $c_f(uv)$ with three non-negative integers $a,b,c$ for each edge $uv\in E(G)$, and have a parameter
\begin{equation}\label{eqa:edge-difference-total-coloring}
B^*_{\alpha}(G,f,M)=\max_{uv \in E(G)}\{c_f(uv)\}-\min_{xy \in E(G)}\{c_f(xy)\}.
\end{equation}
If $B^*_{\alpha}(G,f,M)=0$, we call $f$ a \emph{$\alpha$-proper total coloring} of $G$, the smallest number
\begin{equation}\label{eqa:minimum}
\chi''_{\alpha}(G) =\min_f \{M:~B^*_{\alpha}(G,f,M)=0\}
\end{equation}
over all $\alpha$-proper total colorings of $G$ is called \emph{$\alpha$-proper total chromatic number}, and $f$ is called a \emph{perfect $\alpha$-proper total coloring} if $\chi''_{\alpha}(G)=\chi''(G)$. Moreover
\begin{asparaenum}[\textrm{\textbf{Tcoloring}}-1. ]
\item We call $f$ a \emph{(perfect) \textbf{edge-magic proper total coloring}} of $G$ if $c_f(uv)=f(u)+f(v)+f(uv)$, rewrite $B^*_{\alpha}(G,f,M)=B^*_{emt}(G,f$, $M)$, and $\chi''_{\alpha}(G)=\chi''_{emt}(G)$ is called \emph{edge-magic total chromatic number} of $G$.
\item We call $f$ a \emph{(perfect) \textbf{edge-difference proper total coloring}} of $G$ if $c_f(uv)=f(uv)+|f(u)-f(v)|$, rewrite $B^*_{\alpha}(G,f,M)=B^*_{edt}(G,f$, $M)$, and $\chi''_{\alpha}(G)=\chi''_{edt}(G)$ is called \emph{edge-difference total chromatic number} of $G$.
\item We call $f$ a \emph{(perfect) \textbf{felicitous-difference proper total coloring}} of $G$ if $c_f(uv)=|f(u)+f(v)-f(uv)|$, rewrite $B^*_{\alpha}(G,f,M)=B^*_{fdt}(G,f,M)$, and $\chi''_{\alpha}(G)=\chi''_{fdt}(G)$ is is called \emph{ felicitous-difference total chromatic number} of $G$.
\item We refer to $f$ a \emph{(perfect) \textbf{graceful-difference proper total coloring}} of $G$ if $c_f(uv)=\big ||f(u)-f(v)|-f(uv)\big |$, rewrite $B^*_{\alpha}(G,f,M)=B^*_{gdt}(G,f,M)$, and $\chi''_{\alpha}(G)=\chi''_{gdt}(G)$ is called \emph{graceful-difference total chromatic number} of $G$.\qqed
\end{asparaenum}
\end{defn}

\begin{rem}\label{rem:ABC-conjecture}
The proper total colorings of Definition \ref{defn:combinatoric-definition-total-coloring} have been discussed in \cite{Wang-Su-Yao-mixed-difference-2019, Wang-Su-Yao-Total-2019, Wang-Yao-edge-difference-2019, Wang-Yao-Star-type-Lattices-2020}.

(i) The form $B^*_{\alpha}(G,f,M)=0$ appeared in Definition \ref{defn:combinatoric-definition-total-coloring} means that there exists a constant $k$ such that $c_f(uv)=k$ for each edge $uv\in E(G)$, also, $f$ is \emph{edge-magic} in the view of graph theory. Moreover,
$$\sum_{u\in V(G)}\sum_{v\in N(u)}c_f(uv)=\sum_{u\in V(G)}k\cdot \textrm{deg}_G(u)=k\cdot 2|E(G)|.$$

Obviously, the proper total chromatic number $\chi''(G)\leq \chi''_{\gamma}(G)$ for $\gamma\in P_{ara}=\{$emt, edt, fdt, gdt$\}$. It is difficult to determine the exact values of $\chi''_{\gamma}(G)$ for $\gamma\in P_{ara}$, since the total chromatic number $\chi''(G)\leq \Delta(G)+2$ is not settled down up to now.

(ii) We add three parameters for generalizing Definition \ref{defn:combinatoric-definition-total-coloring} if $G$ is bipartite, and get another group of particular total colorings as follows:

\begin{defn} \label{defn:combinatoric-definition-total-coloring-abc}
$^*$ Suppose that a bipartite graph $G$ admits a proper total coloring $f:V(G)\cup E(G)\rightarrow [1,M]$. We define an \emph{edge-function} $c_f(uv)(a,b,c)$ with three non-negative integers $a,b,c$ for each edge $uv\in E(G)$, and have a parameter
\begin{equation}\label{eqa:edge-difference-total-coloring}
B^*_{\alpha}(G,f,M)(a,b,c)=\max_{uv \in E(G)}\{c_f(uv)(a,b,c)\}-\min_{xy \in E(G)}\{c_f(xy)(a,b,c)\}.
\end{equation}
If $B^*_{\alpha}(G,f,M)(a,b,c)=0$, we call $f$ a \emph{parameterized $\alpha$-proper total coloring} of $G$, the smallest number
\begin{equation}\label{eqa:minimum}
\chi''_{\alpha}(G)(a,b,c) =\min_f \{M:~B^*_{\alpha}(G,f,M)(a,b,c)=0\}
\end{equation}
over all parameterized $\alpha$-proper total colorings of $G$ is called \emph{parameterized $\alpha$-proper total chromatic number}, and $f$ is called a \emph{perfect $\alpha$-proper total coloring} if $\chi''_{\alpha}(G)(a,b,c)=\chi''(G)$. Moreover
\begin{asparaenum}[\textrm{\textbf{TCol}}-1. ]
\item We call $f$ a \emph{(perfect) \textbf{parameterized edge-magic proper total coloring}} of $G$ if $c_f(uv)=af(u)+bf(v)+cf(uv)$, rewrite $B^*_{\alpha}(G,f,M)(a,b,c)=B^*_{emt}(G,f$, $M)(a,b,c)$, and $\chi''_{\alpha}(G)(a,b,c)=\chi''_{emt}(G)(a,b,c)$ is called \emph{parameterized edge-magic total chromatic number} of $G$.
\item We call $f$ a \emph{(perfect) \textbf{parameterized edge-difference proper total coloring}} of $G$ if $c_f(uv)=cf(uv)+|af(u)-bf(v)|$, rewrite $B^*_{\alpha}(G,f,M)(a,b,c)=B^*_{edt}(G,f$, $M)(a,b,c)$, and $\chi''_{\alpha}(G)(a,b,c)=\chi''_{edt}(G)(a,b,c)$ is called \emph{parameterized edge-difference total chromatic number} of $G$.
\item We call $f$ a \emph{(perfect) \textbf{parameterized felicitous-difference proper total coloring}} of $G$ if $c_f(uv)=|af(u)+bf(v)-cf(uv)|$, rewrite $B^*_{\alpha}(G,f,M)(a,b,c)=B^*_{fdt}(G,f,M)(a,b,c)$, and $\chi''_{\alpha}(G)(a,b,c)=\chi''_{fdt}(G)(a,b,c)$ is is called \emph{parameterized felicitous-difference total chromatic number} of $G$.
\item We refer to $f$ a \emph{(perfect) \textbf{parameterized graceful-difference proper total coloring}} of $G$ if $c_f(uv)=\big ||af(u)-bf(v)|-cf(uv)\big |$, rewrite $B^*_{\alpha}(G,f,M)(a,b,c)=B^*_{gdt}(G,f,M)(a,b,c)$, and $\chi''_{\alpha}(G)(a,b,c)=\chi''_{gdt}(G)(a,b,c)$ is called \emph{parameterized graceful-difference total chromatic number} of $G$.\qqed
\end{asparaenum}
\end{defn}

We can put forward various requirements for $(a,b,c)$ in Definition \ref{defn:combinatoric-definition-total-coloring-abc} to increase the difficulty of attacking our topological coding, since the ABC-conjecture (or Oesterl\'{e}-Masser conjecture, 1985) involves the equation $a+b=c$ and the relationship between prime numbers. Proving or disproving the ABC-conjecture could impact many Diophantine (polynomial) math problems including Tijdeman's theorem, Vojta's conjecture, Erd\"{o}s-Woods conjecture, Fermat's last theorem, Wieferich prime and Roth's theorem \cite{Cami-Rosso2017Abc-conjecture}.

(iii) We remove ``proper'' from Definition \ref{defn:combinatoric-definition-total-coloring} as: A simple graph $G$ admits a total coloring $f:V(G)\cup E(G)\rightarrow [1,M]$ such that $f(u)\neq f(v)$ for each edge $uv\in E(G)$, and $f(xy)\neq f(xw)$ for two adjacent edges $xy,xw\in E(G)$. So, this particular total coloring allows $f(u)=f(uv)$ for some edge $uv\in E(G)$, and is weak than that in Definition \ref{defn:combinatoric-definition-total-coloring}. Similarly, removing ``proper'' from Definition \ref{defn:combinatoric-definition-total-coloring-abc} produces four parameterized $\alpha$-proper total colorings weak than that in Definition \ref{defn:combinatoric-definition-total-coloring-abc}.\qqed
\end{rem}

\begin{problem}\label{qeu:flawed-abc-total-colorings}
(i) It is natural based on Definition \ref{defn:flawed-odd-graceful-labelling}, the authors, in \cite{Yao-Mu-Sun-Sun-Zhang-Wang-Su-Zhang-Yang-Zhao-Wang-Ma-Yao-Yang-Xie2019}, conjecture: ``\emph{Each forest $T=\bigcup ^m_{i=1}T_i$ with disjoint trees $T_1,T_2,\dots ,T_m$ admits a flawed graceful/odd-graceful labelling}''.  Determine integers $A_e$ and $B_e$ such that $H=E^*+T$ admits a (set-ordered) graceful/odd-graceful labelling as $A_e\leq |E^*|\leq B_e$.

(ii) For a bipartite graph $G$, finding three parameters $a,b,c$ holding $(a,b,c)\neq (1,1,1)$ under a proper total coloring $f:V(G)\cup E(G)\rightarrow [1,M]$ realizes $B^*_{\alpha}(G,f,M)(a,b,c)=0$ holding each one of the parameterized edge-magic proper total coloring, the parameterized edge-difference proper total coloring, the parameterized felicitous-difference proper total coloring and the parameterized graceful-difference proper total coloring defined in Definition \ref{defn:combinatoric-definition-total-coloring-abc}, .

In a parameterized edge-magic proper total coloring $f$, $B^*_{\alpha}(G,f,M)=0$ means that $c_f(uv)=af(u)+bf(v)+cf(uv)=k$ for each edge $uv\in E(G)$. If there are $(a_0,b_0,c_0)\neq (1,1,1)$ holding $c_f(uv)=a_0f(u)+b_0f(v)+c_0f(uv)=k$, then we have $c_f(uv)=\beta a_0f(u)+\beta b_0f(v)+\beta c_0f(uv)=\beta k$ for each edge $uv\in E(G)$ with $\beta > 0$ and $(\beta a_0,\beta b_0,\beta c_0)\neq (\beta ,\beta ,\beta )$. So, there are infinite group of parameters $a,b,c$ holding $(a,b,c)\neq (1,1,1)$ for the total colorings.
\end{problem}

\begin{exa}\label{exa:felicitous-difference-total-coloring}
\textbf{Duality.} For the felicitous-difference proper total coloring \cite{Wang-Yao-Star-type-Lattices-2020}, we say $f$ to be \emph{edge-ordered} if $f(x)+f(y)\leq f(xy)$ (resp. $f(x)+f(y)\geq f(xy)$) for each edge $xy\in E(G)$. If $G$ admits two felicitous-difference proper total colorings $g$ and $g^c$ holding $g(x)+g^c(x)=C_v>0$ for each vertex $x\in V(G)$ and a constant $C_v$, then $g^c$ is called the \emph{vertex-dual} of $g$, conversely, $g$ is the vertex-dual of $g^c$ too; and moreover if $g(uv)+g^c(uv)=C_e>0$ holds true for each edge $uv\in E(G)$ and a constant $C_e$, we call $g^c$ (resp. $g$) an \emph{all-dual} of $g$ (resp. $g^c$), as well as $g^c$ (resp. $g$) is a \emph{perfect all-dual} of $g$ (resp. $g^c$) if $C_v=C_e$. As an example, a graph $C_5+e$ shown in Fig.\ref{fig:all-dual_total_coloring} admits six felicitous-difference proper total colorings $g_Q$ shown in Fig.\ref{fig:all-dual_total_coloring} (Q) with $Q=$a,b,c,d,e,f, and moreover we observe: (1) $g_a$ and $g_b$ are a pair of vertex-dual colorings, since $g_a(u)+g_b(u)=6$ for each vertex $u\in V(C_5+e)$; (2) $g_d$ and $g_e$ are a pair of perfect all-dual total colorings, since $g_d(u)+g_e(u)=8$ for each vertex $u\in V(C_5+e)$, and $g_d(xy)+g_e(xy)=8$ for each edge $xy\in E(C_5+e)$; (3) $g_a$ and $g_d$ are edge-ordered; (4) $\chi''_{fdt}(C_5+e)=7$ according to $g_f$; (5) $|g_a(x)+g_a(y)-g_a(xy)|=0$, $|g_b(x)+g_b(y)-g_b(xy)|=1$, $|g_c(x)+g_c(y)-g_c(xy)|=k$, $|g_d(x)+g_d(y)-g_d(xy)|=4$, $|g_e(x)+g_e(y)-g_e(xy)|=4$ and $|g_f(x)+g_f(y)-g_f(xy)|=1$ for each edge $xy\in E(C_5+e)$.\qqed
\end{exa}

\begin{figure}[h]
\centering
\includegraphics[width=14cm]{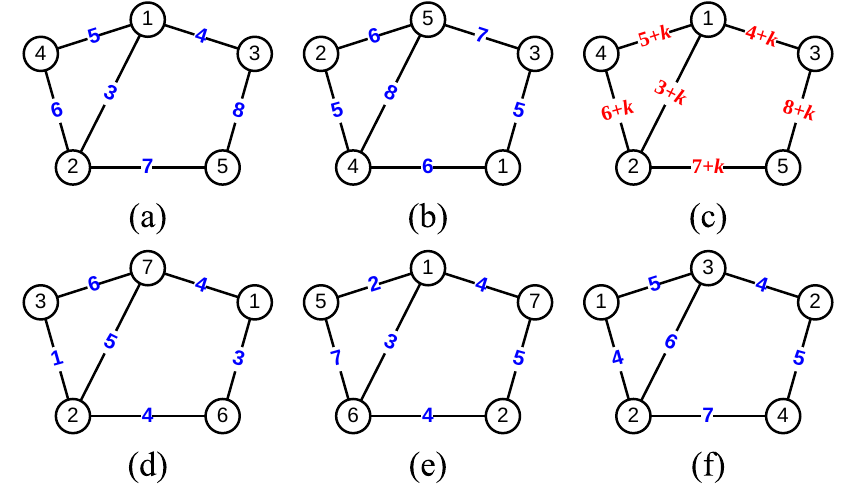}\\
\caption{\label{fig:all-dual_total_coloring} {\small Examples for illustrating the felicitous-difference proper total coloring cited from \cite{Wang-Yao-Star-type-Lattices-2020} .}}
\end{figure}

\begin{thm} \label{thm:perfect-all-dual}
A pair of felicitous-difference proper total colorings $g$ and $g'$ of a graph $G$ is perfect all-dual if and only if there exist two constants $M>0$ and $M'\geq 0$ such that $g(x)+g'(x)=M$ for each vertex $x\in V(G)$, and each edge $uv\in E(G)$ holds $|g(u)+g(v)-g(uv)|=M'$ and $|g'(u)+g'(v)-g'(uv)|=M-M'$ true.
\end{thm}

\begin{cor} \label{thm:unique}
If a graph $G$ holds $\chi''_{fdt}(G)=1+\Delta(G)$, then $G$ admits a unique felicitous-difference proper total coloring.
\end{cor}

\begin{defn} \label{defn:4-dual-total-coloring}
$^*$ We define the \emph{dual total colorings} for the colorings defined in Definition \ref{defn:combinatoric-definition-total-coloring} in the following:
\begin{asparaenum}[\textrm{\textbf{Dual}}-1. ]
\item If $f_{em}$ is an edge-magic proper total coloring of a graph $G$, so there exists a constant $k$ such that $f_{em}(u)+f_{em}(uv)+f_{em}(v)=k$ for each edge $uv\in E(G)$. Let $\max f_{em}=\max \{f_{em}(w):w\in V(G)\cup E(G)\}$ and $\min f_{em}=\min \{f_{em}(w):w\in V(G)\cup E(G)\}$. We have the dual $g_{em}$ of $f_{em}$ defined as: $g_{em}(w)=(\max f_{em}+\min f_{em})-f_{em}(w)$ for each element $w\in V(G)\cup E(G)$, and then
\begin{equation}\label{eqa:f-em}
{
\begin{split}
g_{em}(u)+g_{em}(uv)+g_{em}(v)&=3(\max f_{em}+\min f_{em})-[f_{em}(u)+f_{em}(uv)+f_{em}(v)]\\
&=3(\max f_{em}+\min f_{em})-k=k'
\end{split}}
\end{equation} for each edge $uv\in E(G)$.
\item Suppose that $f_{ed}$ is an edge-difference proper total coloring of a graph $G$, so there exists a constant $k$ such that $f_{ed}(uv)+|f_{ed}(u)-f_{ed}(v)|=k$ for each edge $uv\in E(G)$. Let $\max f_{ed}=\max \{f_{ed}(w):w\in V(G)\cup E(G)\}$ and $\min f_{ed}=\min \{f_{ed}(w):w\in V(G)\cup E(G)\}$. We have the dual $g_{ed}$ of $f_{ed}$ defined by setting $g_{ed}(x)=(\max f_{ed}+\min f_{ed})-f_{ed}(x)$ for $x\in V(G)$ and $g_{ed}(uv)=f_{ed}(uv)$ for $uv\in E(G)$, and then
\begin{equation}\label{eqa:f-ed}
g_{ed}(uv)+|g_{ed}(u)-g_{ed}(v)|=f_{ed}(uv)+|f_{ed}(u)-f_{ed}(v)|=k
\end{equation} for every edge $uv\in E(G)$.
\item When $f_{gd}$ is a graceful-difference proper total coloring of a graph $G$, so there exists a constant $k$ such that $\big ||f_{gd}(u)-f_{gd}(v)|-f_{gd}(uv)\big |=k$ for each edge $uv\in E(G)$. Let $\max f_{gd}=\max \{f_{gd}(w):w\in V(G)\cup E(G)\}$ and $\min f_{gd}=\min \{f_{gd}(w):w\in V(G)\cup E(G)\}$. We have the dual $g_{gd}$ of $f_{gd}$ defined in the way: $g_{gd}(x)=(\max f_{gd}+\min f_{gd})-f_{gd}(x)$ for $x\in V(G)$ and $g_{gd}(uv)=f_{gd}(uv)$ for each edge $uv\in E(G)$, and then
\begin{equation}\label{eqa:f-md}
{
\begin{split}
\big ||g_{gd}(u)-g_{gd}(v)|-g_{gd}(uv)\big |=\big ||f_{gd}(u)-f_{gd}(v)|-f_{gd}(uv)\big |=k
\end{split}}
\end{equation} for each edge $uv\in E(G)$.
\item As $f_{fd}$ is a felicitous-difference proper total coloring of a graph $G$, there exists a constant $k$ such that $|f_{fd}(u)+f_{fd}(v)-f_{fd}(uv)|=k$ for each edge $uv\in E(G)$. Let $\max f_{fd}=\max \{f_{fd}(w):w\in V(G)\cup E(G)\}$ and $\min f_{fd}=\min \{f_{fd}(w):w\in V(G)\cup E(G)\}$. We have the dual $g_{fd}$ of $f_{fd}$ defined as: $g_{fd}(w)=(\max f_{fd}+\min f_{fd})-f_{fd}(w)$ for each element $w\in V(G)\cup E(G)$, and then
\begin{equation}\label{eqa:f-tg}
{
\begin{split}
|g_{fd}(u)+g_{fd}(v)-g_{fd}(uv)|&=|(\max f_{fd}+\min f_{fd})+f_{fd}(u)+f_{fd}(v)-f_{fd}(uv)|\\
&=(\max f_{fd}+\min f_{fd})\pm k
\end{split}}
\end{equation} for each edge $uv\in E(G)$. Here, if $f_{fd}$ is edge-ordered such that $f_{fd}(x)+f_{fd}(y)\geq f_{fd}(xy)$ for each edge $xy\in E(G)$, then $$|g_{fd}(u)+g_{fd}(v)-g_{fd}(uv)|=(\max f_{fd}+\min f_{fd})+k=k'.$$ We have $$|g_{fd}(u)+g_{fd}(v)-g_{fd}(uv)|=(\max f_{fd}+\min f_{fd})-k=k',$$
if $f_{fd}(x)+f_{fd}(y)< f_{fd}(xy)$ for each edge $xy\in E(G)$. \qqed
\end{asparaenum}
\end{defn}

\begin{rem}\label{rem:new-connection-old}
There are connections between graph colorings and graph labellings as follows:
\begin{asparaenum}[\textrm{\textbf{Conn}}-1. ]
\item If a proper graceful-difference total coloring $h$ of $G$ satisfies $h(x)\neq h(y)$ for distinct vertices $x,y\in V(G)$, and $h(uv)\neq h(wz)$ for distinct edges $uv,wz\in E(G)$, and $\max \{h(w): w \in V(G)\cup E(G)\}=1+|E(G)|$, then we get a \emph{graceful labelling} $\alpha$ defined as: $\alpha(x)=h(x)-1$ for $x\in V(G)$. There is a well-known conjecture proposed by Rosa, called \emph{Graceful Tree Conjecture}: ``\emph{Every tree admits a graceful labelling}''. If it is so, then it will settle down a longstanding Ringel-Kotzig Decomposition Conjecture (Gerhard Ringel and Anton Kotzig, 1963; Alexander Rosa, 1967): ``\emph{A complete graph $K_{2n+1}$ can be decomposed into $2n+1$ subgraphs that are all isomorphic to any given tree having $n$ edges}.''
\item For an edge-magic proper total coloring $f$ of $G$ in Definition \ref{defn:combinatoric-definition-total-coloring}, we have $f(u)+f(uv)+f(v)=f(w)+f(wz)+f(z)$ for any pair of distinct edges $uv,wz\in E(G)$. If $f(x)\neq f(y)$ for distinct vertices $x,y\in V(G)$, and $f(uv)\neq f(wz)$ for distinct edges $uv,wz\in E(G)$, so this edge-magic proper total coloring is just an \emph{edge-magic total labelling} (Ref. \cite{Gallian2019}). Anton Kotzig and Alex Rosa, in 1970, conjectured: \emph{Every tree admits an edge-magic total labelling}. Moreover, it was conjectured: \emph{Every tree admits a super edge-magic total labelling}.
\item Let $g$ be an edge-difference proper total coloring of $G$. If $g(uv)+|g(u)-g(v)|=k$ for any edge $uv\in E(G)$, then $g$ will be related with a \emph{$k$-dually graceful labelling} if $g(x)\neq g(y)$ for distinct vertices $x,y\in V(G)$, and $g(uv)\neq g(wz)$ for distinct edges $uv,wz\in E(G)$.
\item Let $\alpha$ be a felicitous-difference proper total coloring of $G$ of $q$ edges. If $|\alpha(u)+\alpha(v)-\alpha(uv)|=0$ for each edge $uv\in E(G)$, and $\alpha:V(G)\rightarrow [0,q-1]$, the edge color set $\{\alpha(uv):uv\in E(G)\}=[c,c+q-1]$, we get a \emph{strongly $c$-harmonious labelling} $\alpha$ of $G$. The generalization of harmonious labellings is a \emph{felicitous labelling} $f:V (G)\rightarrow [0, q-1]$, such that the edge label $f(uv)$ of each edge $uv \in E(G)$ is defined as $f(uv) = f(u) + f(v)~(\bmod ~q)$, and the resultant edge labels are mutually distinct. Similarly with felicitous labelling, a labelling $f: V (G) \rightarrow [0, q]$ is called a \emph{strongly $k$-elegant labelling} if $\{\alpha(uv)~(\bmod ~q+1):uv\in E(G)\}=[k,k+q-1]$.
\item Let $G$ be a bipartite graph and $(X,Y)$ be the bipartition of vertex set $V(G)$. If $G$ admits a $W$-type coloring $f$ holding $\max \{f(x):x\in X\}<\min \{f(y):y\in Y\}$, then we call $f$ a \emph{set-ordered $W$-type coloring} of $G$ (\cite{Yao-Cheng-Yao-Zhao-2009, Zhou-Yao-Chen-Tao2012, Yao-Liu-Yao-2017}). In \cite{Yao-Liu-Yao-2017}, the author show: a set-ordered $W_i$-type coloring is equivalent to another set-ordered $W_j$-type coloring, for example, a bipartite graph $G$ admits a set-ordered graceful labelling if and only if $G$ admits a set-ordered odd-graceful labelling. By technique of set-ordered $W$-type colorings, Zhou \emph{et al.} have proven: (i) each lobster admits an odd-graceful labelling in \cite{Zhou-Yao-Chen-Tao2012}; (ii) each lobster admits an odd-elegant labelling in \cite{Zhou-Yao-Chen2013}.\qqed
\end{asparaenum}
\end{rem}

\begin{defn} \label{defn:rainbow-proper-total-coloring}
$^*$ A \emph{rainbow proper total coloring} $f$ of a connected graph $G$ holds: For any path $x_1x_2x_3x_4x_5\subset G$, edge colors $f(x_{i}x_{i+1})\neq f(x_{j}x_{j+1})$ with $i,j\in [1,4]$, and each $f(x_{j}x_{j+1})$ is one of $f(x_{j}x_{j+1})=f(x_i)+f(x_ix_{i+1})+f(x_{i+1})$, $f(x_{j}x_{j+1})=f(x_ix_{i+1})+|f(x_i)-f(x_{i+1})|$, $f(x_{j}x_{j+1})=|f(x_i)+f(x_{i+1})-f(x_ix_{i+1})|$ and $f(x_{j}x_{j+1})=\big ||f(x_i)-f(x_{i+1})|-f(x_ix_{i+1})\big |$.\qqed
\end{defn}

\begin{problem}\label{qeu:especial-total-colorings}
We propose the following problems for further research on particular proper total colorings:
\begin{asparaenum}[P-1. ]
\item \textbf{Estimate} the bounds of the constant $k_i$ with $i\in [1,4]$ in each especial proper total coloring defined in Definition \ref{defn:combinatoric-definition-total-coloring}, where $k_1=f(u)+f(uv)+f(v)$, $k_2=f(uv)+|f(u)-f(v)|$, $k_3=|f(u)+f(v)-f(uv)|$ and $k_4=\big ||f(u)-f(v)|-f(uv)\big |$.
\item For any group of positive integers $k_1,k_2,k_3,k_4$, \textbf{find} a connected graph $G$ admitting a $\{k_i\}^4_1$-magic proper total coloring $h$, such that there are edges $u_iv_i\in E(G)$ with $i\in [1,4]$ holding $k_1=h(u_1)+h(u_1v_1)+h(v_1)$, $k_2=h(u_2v_2)+|h(u_2)-h(v_2)|$, $k_3=|h(u_3)+h(v_3)-h(u_3v_3)|$ and $k_4=\big ||h(u_4)-h(v_4)|-h(u_4v_4)\big |$ true.
\item \textbf{Color-valued graphic authentication problem}: For a given connected non-tree $(p,q)$-graph $G$, we have two graph sets: A \emph{\textbf{public-key set}} $S_v$ and a \emph{\textbf{private-key set}} $S_e$, each graph $H_i$ of $S_v$ admits a proper vertex coloring, each graph $L_j$ of $S_e$ admits a proper edge coloring, and $|E(G)|=|E(H_i)|=|E(L_j)|$. Can we find a graph $H_i\in S_v$ and another graph $L_j\in S_e$, and do the vertex-coinciding operation to $H_i$ and $L_j$ respectively, such that the resulting graphs $H'_i$ and $L'_j$ hold $G\cong H'_i$ and $G\cong L'_j$, and two colorings of $H'_i$ and $L'_j$ induce just a proper total coloring of $G$ (as an authentication)? Since we can vertex-split the vertices of $G$ into at least $q-p+1$ different connected graphs, so $S_v\neq \emptyset $ and $S_e\neq \emptyset $.
\item \textbf{Find} a simple and connected graph $G$ admitting a proper total coloring $f:V(G)\cup E(G)\rightarrow [1,M]$ and inducing an edge-function $c_f(uv)$ for each edge $uv\in E(G)$ according to Definition \ref{defn:combinatoric-definition-total-coloring}, and find constants $k_1,k_2,\dots ,k_m$, such that each edge $uv\in E(G)$ corresponds some $k_i$ holding $c_f(uv)=k_i$ true, and each constant $k_j$ corresponds at least one edge $xy$ holding $c_f(xy)=k_j$.
\item For any integer sequence $\{k_i\}^n_1$ with $k_i<k_{i+1}$, \textbf{find} a simple and connected graph $G$ such that each $k_i$ corresponds a proper total coloring $f_i:V(G)\cup E(G)\rightarrow [1,M]$ defined in Definition \ref{defn:combinatoric-definition-total-coloring}, and $f_i$ induces an edge-function $c_{f_i}(uv)=k_i$ for each edge $uv\in E(G)$.
\item \textbf{Splitting-coinciding problem}: Given two connected graphs $W$ and $U$ with $\chi''(W)=\chi''(U)$, \textbf{does} doing vertex-splitting and vertex-coinciding operations to $W$ (resp. $U$) produce $U$ (resp. $W$)?
\end{asparaenum}
\end{problem}

%\section{Graphic lattices}
%%\input{2-section/Graphic-lattices-2}

\section{Graphic lattices}

We will construct \emph{graphic lattices}, \emph{graphic group lattices}, \emph{Topcode-matrix lattices} and \emph{topological coding lattices} produced by graph operations, matrix operations, group operations. So, we will define two kinds of undirected digraphic lattices and colored digraphic lattices on digraphs (directed graphs). In fact, various graphic lattices are sets of Topsnut-gpws of topological coding.

\subsection{Linearly independent graphic vectors}

We say $n$ disjoint graphs $G_1,G_2,\dots ,G_n$ to be \emph{linearly independent} under a graph operation $(\bullet)$ if there is no tree $T$ with $r~(\geq 2)$ vertices such that $G_j=T(\bullet)\{G_{i_s}\}^r_{s=1,i_s\neq j}$ for each $j\in [1,n]$. In a (colored) \emph{graph-vector} group $\textbf{\textrm{H}}^c=(H_1,H_2,\dots, H_n)$ with $H_i\cong H_j$ and each graph $H_i$ admits a $W$-type coloring $f_i$, if there exists no operation ``$(\bullet)$'' such that $H_j=(\bullet)\{H_{i_s}\}^r_{s=1,i_s\neq j}$, we say $\textbf{\textrm{H}}^c$ to be linearly independent.

\subsection{Graphic lattices subject to a graph operation}

Let $\textbf{\textrm{H}}=(H_1,H_2,\dots, H_n)$ be a group of $n$ \emph{linearly independent graphic vectors} (also, a \emph{graphic base}) under a graph operation ``$(\bullet)$'', where each $H_i$ is a colored/uncolored graph, and $F_{p,q}$ is a set of colored/uncolored graphs of $\lambda$ vertices and $\mu$ edges with respect to $\lambda \leq p$, $\mu \leq q$ and $2n-2\leq p$. We write the result graph obtained by doing a graph operation $(\bullet)$ on $G$ and the base $\textbf{\textrm{H}}$ with $a_i\in Z^0$, denoted as $\textbf{\textrm{H}}(\bullet)G=G(\bullet)^n_{i=1}a_iH_i$. In general, we call the following graph set
\begin{equation}\label{eqa:graphic-lattice-graph-operation}
\textbf{\textrm{L}}(\textbf{\textrm{H}}(\bullet)F_{p,q})=\{G(\bullet)^n_{i=1}a_iH_i:~a_i\in Z^0,~G\in F_{p,q}\}
\end{equation}
with $\sum^n_{i=1} a_i\geq 1$ a \emph{graphic lattice} (or \emph{colored graphic lattice}), $\textbf{\textrm{H}}$ a \emph{graphic lattice base}, $p$ is the \emph{dimension} and $n$ is the \emph{rank} of $\textbf{\textrm{L}}(\textbf{\textrm{H}}(\bullet)F_{p,q})$. Moreover, $\textbf{\textrm{L}}(\textbf{\textrm{H}}(\bullet)F_{p,q})$ is called a \emph{linear graphic lattice} if every $G\in F_{p,q}$, each base $H_i$ of $\textbf{\textrm{H}}$ and $G(\bullet)^n_{i=1}a_iH_i$ are forests or trees. An uncolored tree-graph lattice, or a colored tree-graph lattice is \emph{full-rank} $p=n$ in the equation (\ref{eqa:graphic-lattice-graph-operation}).

\begin{rem}\label{rem:Hanzi-directed-lattices}
Especially, if each $H\in \textbf{\textrm{H}}$ is a (colored) \emph{Hanzi-graph}, we call $\textbf{\textrm{L}}(\textbf{\textrm{H}}(\bullet)F_{p,q})$ a (colored) \emph{Hanzi-lattice}.

Let $\overrightarrow{F}_{p,q}$ be a set of directed graphs of $p$ vertices and $q$ arcs with $n\leq p$, and let $\overrightarrow{\textbf{\textrm{H}}}=(\overrightarrow{H}_1,\overrightarrow{H}_2,\dots, \overrightarrow{H}_n)$ be a group of $n$ linearly independent directed-graphic vectors, where each $\overrightarrow{H}_i$ is a directed graph. By an operation ``$(\bullet)$'' on directed graphs, we have a directed-graphic lattice (or \emph{colored directed-graphic lattice}) as follows
\begin{equation}\label{eqa:directed-graphic-lattice-graph-operation}
\overrightarrow{\textbf{\textrm{L}}}(\overrightarrow{\textbf{\textrm{H}}}(\bullet)\overrightarrow{F}_{p,q})=\left \{\overrightarrow{G}(\bullet)^n_{i=1}a_i\overrightarrow{H}_i:~a_i\in Z^0,~\overrightarrow{G}\in \overrightarrow{F}_{p,q}\right \}
\end{equation}
with $\sum^n_{i=1} a_i\geq 1$.\qqed
\end{rem}

\begin{problem}\label{qeu:extrems-in-lattices}
We propose the following problems:
\begin{asparaenum}[\textrm{A}-1. ]
\item \textbf{Characterize} the connection between the graphic lattice base $\textbf{\textrm{H}}$ and the graph set $F_{p,q}$, that is, the graphic lattice $\textbf{\textrm{L}}(\textbf{\textrm{H}}(\bullet)F_{p,q})$ is not empty as $F_{p,q}$ holds what conditions.
\item \textbf{Find} a graph $G^*$ of a graphic lattice $\textbf{\textrm{L}}(\textbf{\textrm{H}}(\bullet)F_{p,q})$, such that $G^*$ has the shortest diameter, or $G^*$ is Hamiltonian, or $G^*$ has a spanning tree with the most leaves in $\overrightarrow{\textbf{\textrm{L}}}(\overrightarrow{\textbf{\textrm{H}}}(\bullet)\overrightarrow{F}_{p,q})$, and so on.
\item \textbf{Does} there exist a Hanzi-graphic lattice containing any Chinese essay with $m$ Chinese letters?
\end{asparaenum}
\end{problem}

\subsection{Graphic lattices subject to the vertex-coinciding operation}

\subsubsection{Uncolored graphic lattices} Let $\textbf{\textrm{T}}=(T_1,T_2,\dots, T_n)$ be a group of $n$ \emph{linearly independent graphic vectors} under the vertex-coinciding operation, also, a \emph{graphic base}, and let $H\in F_{p,q}$ be a connected graph of vertices $u_{1},u_{2},\dots, u_{m}$. We write the result graph obtained by vertex-coinciding a vertex $v_i$ of some base $T_i$ with some vertex $u_{i_j}$ of the connected graph $H$ into one vertex $w_i=u_{i_j}\odot v_i$ as $H\odot \textbf{\textrm{T}}=H\odot |^n_{i=1}a_iT_i$ with $a_i\in Z^0$ and $\sum^n_{i=1} a_i\geq 1$. Since there are two or more vertices of the graphic lattice base $T_i$ that can be vertex-coincided with some vertex of the connected graph $H$, so $H\odot \textbf{\textrm{T}}$ is not unique in general, in other word, these graphs $H\odot \textbf{\textrm{T}}$ forms a set. We call the following set
\begin{equation}\label{eqa:graphic-lattice-00}
\textbf{\textrm{L}}(\textbf{\textrm{T}}\odot F_{p,q})=\{H\odot |^n_{i=1}a_iT_i:~a_i\in Z^0,~H\in F_{p,q}\}
\end{equation} with $\sum^n_{i=1} a_i\geq 1$ a \emph{graphic lattice}, and $p$ is the \emph{dimension}, and $n$ is the \emph{rank} of the graphic lattice. Moreover $\textbf{\textrm{L}}(\textbf{\textrm{T}}\odot F_{p,q})$ is called a \emph{linear graphic lattice} if every $H\in F_{p,q}$, each base $T_i$ of the lattice base $\textbf{\textrm{T}}$ and $H\odot \textbf{\textrm{T}}$ are forests or trees. We have several obvious facts:

(1) A graphic lattice can be expressed by different graphic bases. See five groups of Hanzi-graphic vectors shown in Fig.\ref{fig:2-txwg}, Fig.\ref{fig:2-txwg-3-joins} and Fig.\ref{fig:2-txwg-other-groups}.

(2) If a graph of a graphic lattice $\textbf{\textrm{L}}(\textbf{\textrm{T}}\odot F_{p,q})$ is connected, and $H$ has just $n-1$ edges, then each graphic vector $T_i$ with $i\in [1,n]$ is connected, and $p\geq n$, as well as $q\geq p-1$.

There are many ways to construct graphs $G\odot ^n_{i=1}a_iT_i$ with $a_i\in Z^0$ and $\sum^n_{i=1} a_i\geq 1$. Here, we discuss mainly two ways: One-vs-one by the vertex-coinciding operation, and String $T_1,T_2,\dots ,T_n$ together by the vertex-coinciding operation.

\textbf{1. One-vs-one by the vertex-coinciding operation.} For a graph $G\odot ^n_{i=1}T_i$, we suppose that each $T_i$ is vertex-coincided with one vertex $u_i$ of $G$, and any pair of two $T_i$ and $T_j$ are vertex-coincided with two distinct vertices $u_i$ and $u_j$ of $G$. So, we have:
\begin{asparaenum}[\textrm{Case 1.}1 ]
\item There are ${p\choose n}$ groups of vertices $u_{k,1},u_{k,2},\cdots ,u_{k,n}$ for a graph $G$ of $p$ vertices.

\item There are $n!$ permutations $u_{k,i_1}u_{k,i_2}\cdots u_{k,i_n}$ for each group of vertices $u_{k,1}$, $u_{k,2}$, $\cdots $, $u_{k,n}$ of $G$.

\item There are $n!$ permutations $T_{i_1}T_{i_2}\cdots T_{i_n}$ of graphic vectors $T_1,T_2,\dots ,T_n$ of a graphic lattice base $\textbf{\textrm{T}}$. For each permutation $u_{k,i_1}u_{k,i_2}\cdots u_{k,i_n}$, a vertex $x_{i_j}$ of $T_{i_j}$ is vertex-coincided with the vertex $u_{k,i_j}$ with $j\in [1,n]$, such that two graphic vectors $T_{i_j}$ and $T_{i_s}$ are vertex-coincided two distinct vertices $u_{k,i_j}$ and $u_{k,i_s}$ of $G$.

\item For each $T_{i_j}$ of a permutation $T_{i_1}T_{i_2}\cdots T_{i_n}$, there are $|T_{i_j}|$ vertices being vertex-coincided with the vertex $u_{k,i_j}$ with $j\in [1,n]$.
\end{asparaenum}

Thereby, under one-vs-one vertex-coinciding operation, we have ${p\choose n}\cdot (n!)^2\cdot \prod ^n_{i=1}|T_i|$ possible graphs $G\odot ^n_{i=1}T_i$ in total.

\textbf{2. String $T_1,T_2,\dots ,T_n$ together by the vertex-coinciding operation.} We consider to string a permutation $T_{i_1}T_{i_2}\cdots T_{i_n}$ of $T_1,T_2,\dots ,T_n$ together by $(n-1)$ edges $u_{j}v_{j}$ of $G$, such that a vertex $y_{i_j}$ of $T_{i_j}$ is joined by the vertex $u_{j}$, a vertex $x_{i_{j+1}}$ of $T_{i_{j+1}}$ is joined by the vertex $v_{j}$ for $j\in [1,n-1]$, the resultant graph is just a bunch of $T_{i_1},T_{i_2},\cdots ,T_{i_n}$, denoted as
\begin{equation}\label{eqa:string-base}
G\ominus \textbf{\textrm{T}}=T_{i_1}+u_{1}v_{1}+T_{i_2}+u_{2}v_{2}+T_{i_3}+\cdots +T_{i_{n-1}}+u_{n-1}v_{n-1}+T_{i_n}.
\end{equation}

\begin{asparaenum}[\textrm{Case 2.}1 ]
\item There are ${q\choose n-1}$ groups of edges $u_{k,1}v_{k,1},u_{k,2}v_{k,2},\cdots ,u_{k,n-1}v_{k,n-1}$ for a connected $(p,q)$-graph $G$.
\item There are $n!$ permutations $T_{i_1}T_{i_2}\cdots T_{i_n}$ of graphic vectors $T_1,T_2,\dots ,T_n$ of a graphic lattice base $\textbf{\textrm{T}}$.
\item In a string graph
\begin{equation}\label{eqa:string-base-11}
T_{i_1}+u_{k,1}v_{k,1}+T_{i_2}+u_{k,2}v_{k,2}+T_{i_3}+\cdots +T_{i_{n-1}}+u_{k,n-1}v_{k,n-1}+T_{i_n},
\end{equation}there are $|T_{i_1}|$ vertices of the graphic vector $T_{i_1}$ to join the vertex $u_{k,1}$; each of $|T_{i_2}|$ vertices of the graphic vector $T_{i_2}$ can be joined with the vertex $v_{k,1}$, and with the vertex $u_{k,2}$, respectively, so we have $|T_{i_2}|^2$ cases; go on in this way, we get $|T_{i_1}|\cdot (\prod ^{n-1}_{s=2} |T_{i_s}|^2)\cdot |T_{i_n}|$ possible string graphs in total.
\end{asparaenum}

Thereby, we have $n_s(G)$ possible string graphs $G\ominus \textbf{\textrm{T}}$ shown in (\ref{eqa:string-base}) from $G\odot ^n_{i=1}T_i$, where
$$n_s(G)={q\choose n-1}\cdot n!\cdot |T_{i_1}|\cdot \left (\prod ^{n-1}_{s=2} |T_{i_s}|^2\right )\cdot |T_{i_n}|$$

\begin{problem}\label{qeu:structures-lattice}
We are interesting on the structures and properties of a graphic lattice $\textbf{\textrm{L}}(\textbf{\textrm{T}}\odot F_{p,q})$, such as:
\begin{asparaenum}[\textrm{B}-1. ]
\item \textbf{Determine} a graph $G^*$ of $\textbf{\textrm{L}}(\textbf{\textrm{T}}\odot F_{p,q})$ such that two diameters $D(G^*)\leq D(G)$ for any $G\in \textbf{\textrm{L}}(\textbf{\textrm{T}}\odot F_{p,q})$.

\item \textbf{Estimate} the \emph{cardinality} of $\textbf{\textrm{L}}(\textbf{\textrm{T}}\odot F_{p,q})$, however this will be related with the graph isomorphic problem, a NP-hard problem.

\item \textbf{Dose} $H'\odot |^n_{i=1}T_i\cong H''\odot |^n_{i=1}T_i$ if $H'\cong H''$?
\item Let $\overline{T}_i$ be the \emph{complementary graph} of each base $T_i$ of a graphic base $\textbf{\textrm{T}}$, and let $\overline{H}$ be the complement of $H\in F_{p,q}$. \textbf{Is} $\overline{H}\odot |^n_{i=1}\overline{T}_i$ the complementary graph of $H\odot |^n_{i=1}T_i$?\qqed
\end{asparaenum}
\end{problem}

\subsubsection{Colored graphic lattices} Let $F^c_{p,q}$ be a set of colored graphs of $\lambda $ vertices and $\mu $ edges with respect to $\lambda \leq p$, $\mu \leq q$ and $2n-2\leq p$, where each graph $H^c\in F^c_{p,q}$ is colored by a $W$-type coloring $f$, and let $\textbf{\textrm{T}}^c=(T^c_1,T^c_2,\dots, T^c_n)$ with $n\leq p$ be a \emph{linearly independent colored graphic base} under the vertex-coinciding operation, we have two particular cases: (i) each graph $T^c_i$ of $\textbf{\textrm{T}}^c$ admits a $W_i$-type coloring $g_i$; (ii) the union graph $\bigcup ^n_{i=1}T^c_i$ admits a flawed $W_i$-type coloring. Vertex-coinciding a vertex $x_i$ of some base $T^c_i$ with some vertex $y_{i_j}$ of the colored graph $H^c\in F^c_{p,q}$ into one vertex $z_i=y_{i_j}\odot x_i$ produces a vertex-coincided graph $H^c\odot |^n_{i=1}T^c_i$ admitting a coloring induced by $g_1, g_2, \dots ,g_n$ and $f$, here, two vertices $x_i,y_{i_j}$ are colored the same color $\gamma$, then the vertex $z_i=y_{i_j}\odot x_i$ is colored with the color $\gamma$ too. We call the following set
\begin{equation}\label{eqa:graphic-lattice-colored}
\textbf{\textrm{L}}(\textbf{\textrm{T}}^c\odot F^c_{p,q})=\left \{H^c\odot |^n_{i=1}a_iT^c_i:~a_i\in Z^0,~H^c\in F^c_{p,q}\right \}
\end{equation} with $\sum^n_{i=1} a_i\geq 1$ a \emph{colored graphic lattice}, where $p$ is the \emph{dimension}, and $n$ is the \emph{rank} of the colored graphic lattice. We call the colored graphic lattice $\textbf{\textrm{L}}(\textbf{\textrm{T}}^c\odot F^c_{p,q})$ a \emph{linear colored graphic lattice} if every colored graph $H^c\in F^c_{p,q}$, each graph $T^c_i$ of $\textbf{\textrm{T}}^c$ and the graph $H^c\odot |^n_{i=1}a_iT^c_i$ are colored forests, or colored trees. Clearly, each element of the lattice $\textbf{\textrm{L}}(\textbf{\textrm{T}}^c\odot F^c_{p,q})$, each graph $T^c_i$ and each colored graph $H^c\in F^c_{p,q}$ may admit the same $W$-type colorings.

\begin{exa}\label{exa:2222}
In Fig.\ref{fig:10-year-flawed-graceful}, the graphs $M_1,M_2,\dots ,M_{14}$ are a block permutation of a group of Hanzi-graphs $G_{\textrm{4214}}, G_{\textrm{3674}}, G_{\textrm{4287}}, G_{\textrm{3630}}, G_{\textrm{1657}}, G_{\textrm{3674}}, G_{\textrm{4287}}, G_{\textrm{4043}}$ composed by 14 blocks (Ref. \cite{GB2312-80} ). Let $O=\bigcup ^{14}_{k}M_k$. The disconnected graph admits a flawed graceful labelling defined in Definition \ref{defn:flawed-odd-graceful-labelling}. Let $O+A_i$ with $i=1,2$, where two colored disconnected graphs $A_1,A_2$ are shown in Fig.\ref{fig:10-year-flawed-graceful}. So, each connected graph $O+A_i$ with $i=1,2$ admits a graceful labelling $g_i$ shown in Fig.\ref{fig:10-year-flawed-graceful-join}. we can see that $O+A_1$ is the result of string $M_1,M_2,\dots ,M_{14}$ together by the graph $A_1$ according to the vertex-coinciding operation ``$\odot$'', $O+A_i$ is the result of vertex-coinciding $O$ with $A_i$, so we can rewrite $O+A_i$ as $O\odot A_i$. Observe $O=\bigcup ^{14}_{k}M_k$ carefully, we are conscious of there are many colored graphs $A_j$ like $A_1$ and $A_2$, such that vertex-coincided graphs $A_j\odot ^{14}_{k=1}M_k$ are connected and admit graceful labellings. \qqed
\end{exa}

\begin{figure}[h]
\centering
\includegraphics[width=16.2cm]{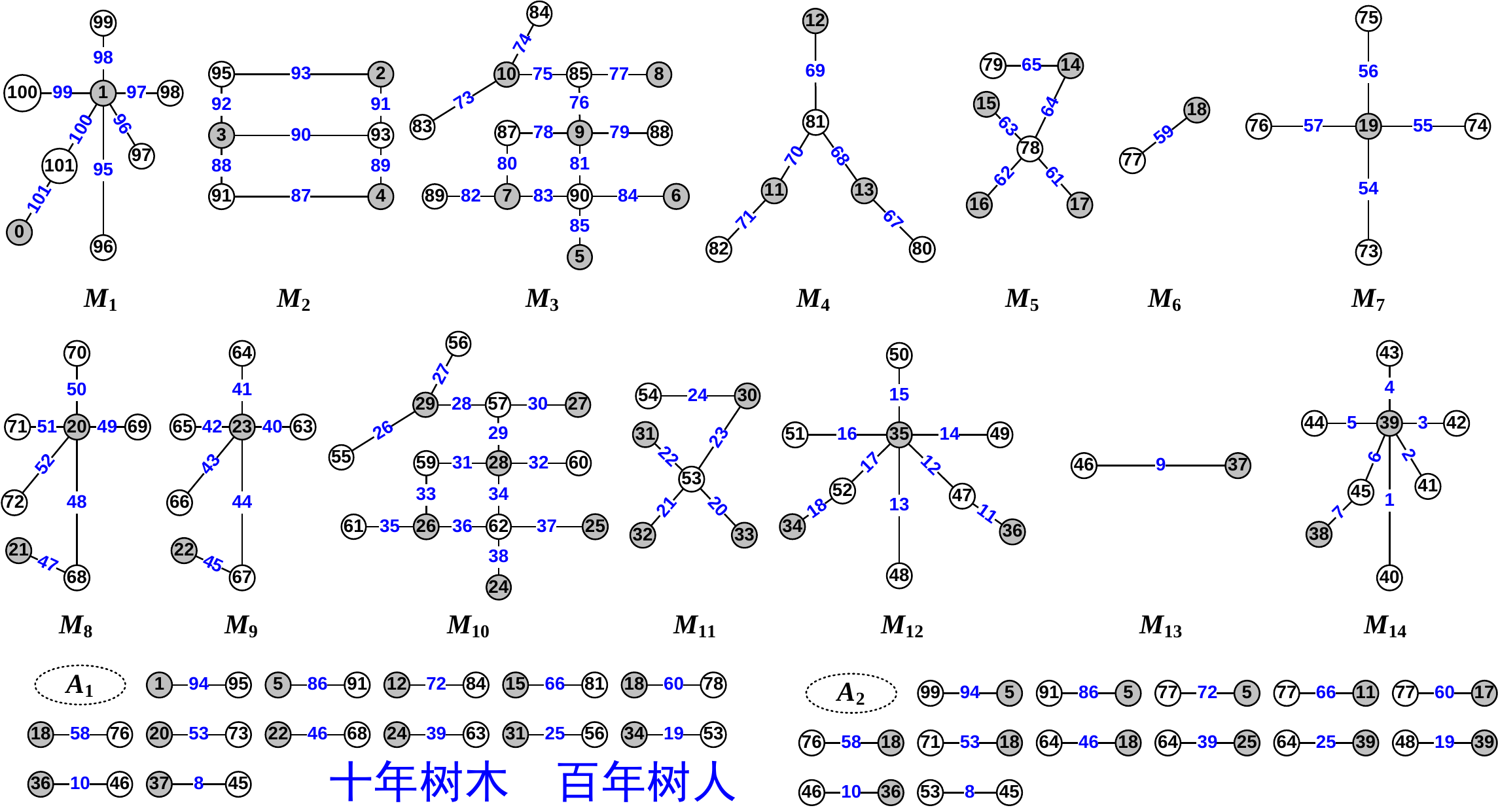}
\caption{\label{fig:10-year-flawed-graceful} {\small A group of Hanzi-graphs $G_{\textrm{4214}}, G_{\textrm{3674}}, G_{\textrm{4287}}, G_{\textrm{3630}}, G_{\textrm{1657}}, G_{\textrm{3674}}, G_{\textrm{4287}}, G_{\textrm{4043}}$ induces 14 block graphs $M_1,M_2,\dots ,M_{14}$, which form a disconnected graph $O=\bigcup ^{14}_{k}M_k$ admitting a flawed graceful labelling.}}
\end{figure}

\begin{figure}[h]
\centering
\includegraphics[width=14.2cm]{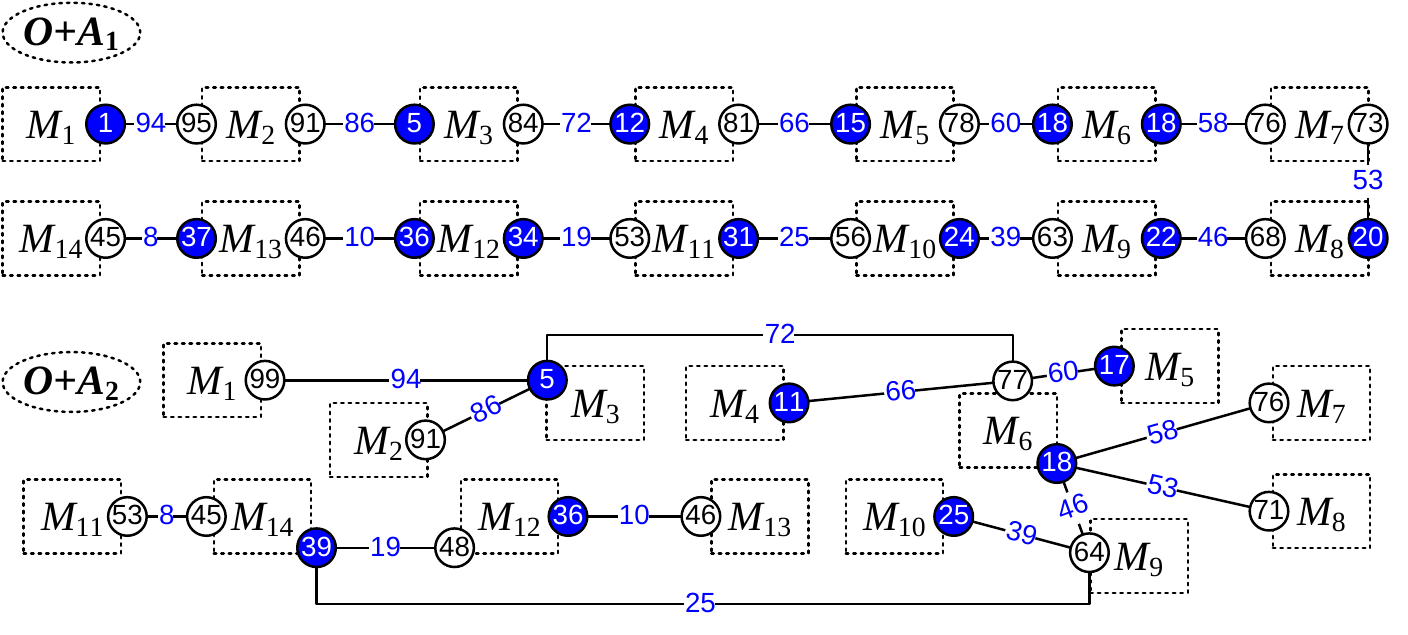}
\caption{\label{fig:10-year-flawed-graceful-join} {\small Two connected Topsnut-gpws obtained from $M_1,M_2,\dots ,M_{14}$ and $A_1,A_2$ shown in Fig.\ref{fig:10-year-flawed-graceful}.}}
\end{figure}

\begin{problem}\label{qeu:complex-problems}
For a colored graphic lattice $\textbf{\textrm{L}}(\textbf{\textrm{T}}^c\odot F^c_{p,q})$, we may consider the following complex problems:
\begin{asparaenum}[\textrm{C}-1. ]
\item \textbf{Classify} a colored graphic lattice $\textbf{\textrm{L}}(\textbf{\textrm{T}}^c\odot F^c_{p,q})$ into some particular subsets $L^{sub}_k$ with $k\in [1,m]$, find particular subsets, such as each graph of $L^{sub}_i$ is a tree, or an Euler's graph, or a Hamiltonian graph; if the colored graphic lattice base $\textbf{\textrm{T}}^c$ admits a flawed $W$-type coloring, then each graph of $L^{sub}_j$ admits a $W$-type coloring too.
\item \textbf{Find} a graph of $\textbf{\textrm{L}}(\textbf{\textrm{T}}^c\odot F^c_{p,q})$ with the shortest diameter $D(G^*)$, such that $D(G^*)\leq D(G)$ for any graph $G\in \textbf{\textrm{L}}(\textbf{\textrm{T}}^c\odot F^c_{p,q})$.
\item \textbf{List} possible $W$-type colorings for constructing a colored graphic lattice $\textbf{\textrm{L}}(\textbf{\textrm{T}}^c\odot F^c_{p,q})$.
\item \textbf{Do} we have the topological coloring isomorphism $H^c\odot |^n_{i=1}T^c_i= G^c\odot |^n_{i=1}T^c_i$ in a colored graphic lattice $\textbf{\textrm{L}}(\textbf{\textrm{T}}^c\odot F^c_{p,q})$ when $H^c\cong G^c$ or $H^c\not \cong G^c$?
\item If $T^c_i\cong T^c_j$ and $T^c_i\cong H^c$ for distinct $i,j\in [1,n]$, \textbf{characterize} $H^c\odot |^n_{i=1}T^c_i$.
\item Since a tree admits a set-ordered graceful labelling if and only if it admits a set-ordered odd-graceful labelling, we consider: For two colored graphic lattices $\textbf{\textrm{L}}(\textbf{\textrm{T}}^c_i\odot F^c_{p,q})$ and two bases $\textbf{\textrm{T}}^c_i=(T^c_{i,1}$, $T^c_{i,2}$, $\dots $, $T^c_{i,n})$ with $i=1,2$, each $H^c\odot^n_{j=1}T^c_{i,j}$ admits a $W_i$-type coloring, if both $X_i$ and $X_{3-i}$ are equivalent to each other. \textbf{Is} $\textbf{\textrm{L}}(\textbf{\textrm{T}}^c_i\odot F^c_{p,q})$ equivalent to $\textbf{\textrm{L}}(\textbf{\textrm{T}}^c_{3-i}\odot F^c_{p,q})$ with $i=1,2$?
\item \cite{Yao-Sun-Zhao-Li-Yan-2017, Yao-Mu-Sun-Zhang-Wang-Su-Ma-IAEAC-2018, Sun-Zhang-Zhao-Yao-2017} If the graphic lattice base $\textbf{\textrm{T}}^c=(T^c_1,T^c_2,\dots, T^c_n)$ forms an every-zero graphic group based under a $W$-type coloring, \textbf{does} the corresponding colored graphic lattice form a graphic group too?
\item If each graphic base of the graphic lattice base $\textbf{\textrm{T}}^c=(T^c_1,T^c_2,\dots, T^c_n)$ admits a $W$-type coloring defined in Definition \ref{defn:combinatoric-definition-total-coloring}, \textbf{determine} a subset $S(\textbf{\textrm{L}})$ of the graphic lattice $\textbf{\textrm{L}}(\textbf{\textrm{T}}^c\odot F^c_{p,q})$, such that each connedted graph of $S(\textbf{\textrm{L}})$ admits a rainbow proper total coloring defined in Definition \ref{defn:rainbow-proper-total-coloring}.
\item \textbf{Find} a graph $G\in \textbf{\textrm{L}}(\textbf{\textrm{T}}^c\odot F^c_{p,q})$, such that for any $H\in \textbf{\textrm{L}}(\textbf{\textrm{T}}^c\odot F^c_{p,q})$, we have (1) the proper total chromatic numbers satisfy $\chi''(G)\leq \chi''(H)$; or (2) the edge-magic total chromatic number $\chi''_{emt}$, the edge-difference total chromatic number $\chi''_{edt}$, the felicitous-difference total chromatic number $\chi''_{fdt}$ and the graceful-difference total chromatic number $\chi''_{gdt}$ hold $\chi''_{\varepsilon}(G)\leq \chi''_{\varepsilon}(H)$ for $\varepsilon\in \{$\emph{emt, edt, fdt, gdt}$\}$; and or (3) the diameters obey $D_{\textrm{iameter}}(G)\leq D_{\textrm{iameter}}(H)$.\qqed
\end{asparaenum}
\end{problem}

\begin{exa}\label{exa:3333}
In Fig.\ref{fig:2-txwg-other-groups}, there are four groups of Hanzi-graphs as follows:

$O_1=G_{\textrm{4476}}\cup G_{\textrm{2511}}\cup G_{\textrm{4610}}\cup G_{\textrm{4147}}$, $O_2=G_{\textrm{5027}}\cup G_{\textrm{4476}}\cup G_{\textrm{2511}}\cup G_{\textrm{5027}}\cup G_{\textrm{4734}}\cup G_{\textrm{3306}}$

$O_3=G_{\textrm{4476}}\cup G_{\textrm{4147}}\cup G_{\textrm{1676}}\cup G_{\textrm{2511}}$, $O_4=G_{\textrm{4476}}\cup G_{\textrm{4734}}\cup G_{\textrm{1643}}\cup G_{\textrm{5240}}\cup G_{\textrm{3306}}$

We use these four groups of Hanzi-graphs to get the following groups of linearly independent Hanzi-graphic vectors

$\textbf{\textrm{O}}_1=(G_{\textrm{4476}}, G_{\textrm{2511}}, G_{\textrm{4610}}, G_{\textrm{4147}})$, $\textbf{\textrm{O}}_2=(G_{\textrm{5027}}, G_{\textrm{4476}}, G_{\textrm{2511}}, G_{\textrm{5027}}, G_{\textrm{4734}}, G_{\textrm{3306}})$

$\textbf{\textrm{O}}_3=(G_{\textrm{4476}}, G_{\textrm{4147}}, G_{\textrm{1676}}, G_{\textrm{2511}})$, $\textbf{\textrm{O}}_4=(G_{\textrm{4476}}, G_{\textrm{4734}}, G_{\textrm{1643}}, G_{\textrm{5240}}, G_{\textrm{3306}})$.

By the above Hanzi-graphic groups $\textbf{\textrm{O}}_k$ with $k\in [1,4]$, we get four Hanzi-graphic lattices $\textbf{\textrm{L}}(\textbf{\textrm{O}}_k\odot F_{p,q})$ with $k\in [1,4]$. Observe four Hanzi-graphic lattices $\textbf{\textrm{L}}(\textbf{\textrm{O}}_k\odot F_{p,q})$ with $k\in [1,4]$ carefully, we can see: ``\emph{A graphic lattice can be expressed by different graphic bases}.''

Let $T^c_{\textrm{abcd}}=T_{\textrm{code}}(G_{\textrm{abcd}})$ be a Topcode-matrix of Hanzi-graph $G_{\textrm{abcd}}$ (see the definition of a Topcode-matrix shown in Definition \ref{defn:topcode-matrix-definition}), we have four colored graphic bases:

$\textbf{\textrm{O}}^c_1=(T^c_{\textrm{4476}}, T^c_{\textrm{2511}}, T^c_{\textrm{4610}}, T^c_{\textrm{4147}})$, $\textbf{\textrm{O}}^c_2=(T^c_{\textrm{5027}}, T^c_{\textrm{4476}}, T^c_{\textrm{2511}}, T^c_{\textrm{5027}}, T^c_{\textrm{4734}}, T^c_{\textrm{3306}})$.

$\textbf{\textrm{O}}^c_3=(T^c_{\textrm{4476}}, T^c_{\textrm{4147}}, T^c_{\textrm{1676}}, T^c_{\textrm{2511}})$, $\textbf{\textrm{O}}^c_4=(T^c_{\textrm{4476}}, T^c_{\textrm{4734}}, T^c_{\textrm{1643}}, T^c_{\textrm{5240}}, T^c_{\textrm{3306}})$

Finally, we obtain four colored Hanzi-graphic lattices $\textbf{\textrm{L}}(\textbf{\textrm{O}}^c_k\odot F^c_{p,q})$ with $k\in [1,4]$. Suppose that each graph of $\textbf{\textrm{L}}(\textbf{\textrm{O}}^c_k\odot F^c_{p,q})$ admits a $W_k$-type coloring with $k\in [1,4]$. Clearly, a colored Hanzi-graphic lattice $\textbf{\textrm{L}}(\textbf{\textrm{O}}^c_i\odot F^c_{p,q})$ differs from another colored Hanzi-graphic lattice $\textbf{\textrm{L}}(\textbf{\textrm{O}}^c_j\odot F^c_{p,q})$ if $\textbf{\textrm{O}}^c_i\neq \textbf{\textrm{O}}^c_j$, and moreover ``\emph{a colored graphic lattice is not expressed by different $W$-type colored-graphic bases}''. See two examples shown in Fig.\ref{fig:Dpermutation-graceful-odd}.\qqed
\end{exa}

\subsection{Graphic lattices subject to the vertex-substituting operation}

Saturated systems are defined by the vertex-replacing and edge-replacing operations on graphs:

(a) Replacing a vertex $x$ of a graph $G$ by another graph $T$: First, remove the vertex $x$ from $G$, and join each vertex $x_i\in N(x)$ with some vertex $y_i$ of $T$ by an edge, and the resultant graph is denoted as $(G-x)\triangleleft T$;

(b) Replacing an edge $xy$ of a graph $G$ by another graph $H$: First, remove the edge $xy$ from $G$, and then join the vertex $x$ with some vertex $u$ of $H$ by an edge, and join the vertex $y$ with some vertex $v$ of $H$ by an edge, and the resultant graph is denoted as $(G-xy)\ominus H$.

Notice that: (a') the graph $G$ can be obtained by contracting $T$ of the graph $(G-x)\triangleleft T$ into a vertex $x$; (b') as $H=K_1$, the graph $(G-xy)\ominus H$ is an \emph{edge-subdivision operation} of graph theory.

\textbf{The fully vertex-replacing operation ``$\overline{\triangleleft} $'':} Replacing a vertex $x$ of a graph $G$ by another graph $T$ holding $|V(T)|\geq |N(x)|$, first remove $x$ from $G$, and then join each $x_i\in N(x)$ with vertex $y_i\in V(T)$ by an edge, such that $y_i\neq y_j$ if $x_i\neq x_j$. The resultant graph is denoted as $(G-x)\overline{\triangleleft} T$.

In a graphic base $\textbf{\textrm{T}}=(T_1,T_2,\dots, T_n)$, graphic vectors $T_1,T_2,\dots, T_n$ are linearly independent under the fully vertex-replacing operation. Let maximum degrees $\Delta(T_i)\leq \Delta(T_{i+1})$ for $i\in [1,n-1]$. For a connected graph $H\in F_{p,q}$ having $m$ vertices, we substitute the first vertex $x_1$ of $H$ by doing a fully vertex-replacing operation with some graphic vector $T_{i_1}\in \textbf{\textrm{T}}$, where $|V(T_{i_1})|\geq |N(x_1)|$, the resultant graph is written as $H_1=(H-x_1)\overline{\triangleleft} T_{i_1}$. Next, we do a fully vertex-replacing operation to a vertex $x_2$ of $H_1$ but $x_2\not \in V(T_{i_1})\subset V(H_1)$ by some graphic vector $T_{i_2}\in \textbf{\textrm{T}}$ with $|V(T_{i_2})|\geq |N(x_2)|$, and then denote the resulting graph as $H_2=(H_1-x_2)\overline{\triangleleft} T_{i_2}$. Go on in this way, we get $H_m=(H_{m-1}-x_{m})\overline{\triangleleft} T_{i_m}$ with $T_{i_m}\in \textbf{\textrm{T}}$ and $|V(T_{i_m})|\geq |N(x_m)|$ after doing a fully vertex-replacing operation to the last vertex $x_m\in V(H)\setminus (\bigcup^{m-1}_{j=1}V(T_{i_j}))$ by some graphic vector $T_{i_m}\in \textbf{\textrm{T}}$ with $|V(T_{i_m})|\geq |N(x_m)|$. For simplicity, we write $H_m$ by $H\triangleleft |^n_{k=1} a_kT_{k}$, and call the following set
\begin{equation}\label{eqa:fully-replacing-graphic-lattice}
\textbf{\textrm{L}}(\textbf{\textrm{T}}\overline{\triangleleft} F_{p,q})=\left \{H\overline{\triangleleft} |^n_{k=1} a_kT_{k}:~a_k\in Z^0,~H\in F_{p,q}\right \}
\end{equation}
a \emph{graphic lattice} under the fully vertex-replacing operation, where $\sum^n_{k=1} a_k\geq 1$.

\begin{exa}\label{exa:edge-difference-replacing}
By the fully vertex-replacing operation ``$\overline{\triangleleft} $'', we present a fully vertex-replacing graph $H'=H\overline{\triangleleft} |^{15}_{k=1} a_kE_{1,6}D_{k}$ shown in Fig.\ref{fig:2-edge-diff-replacing}(b), an \emph{edge-difference ice-flower system} $I_{ce}(E_{1,6}D_k)^{15}_{k=1}$, where $a_k\in Z^0$ and the graphic lattice base $\textbf{\textrm{E}}_{1,6}\textbf{\textrm{D}}=\{E_{1,6}D_{1},E_{1,6}D_{2},\dots ,E_{1,6}D_{15}\}$ shown in Fig.\ref{fig:Dsaturated-edge-difference-even}. Notice that $H'$ admits an edge-difference proper total coloring $f$ holding $f(uv)+|f(u)-f(v)|=16$ for each edge $uv\in E(H')$. Another graph $H''$ shown in Fig.\ref{fig:2-edge-diff-replacing}(c) is obtained by vertex-coinciding some vertices of $H'$ with the same color into one, so $H''$ admits an edge-difference proper total coloring too. Conversely, we can vertex-split some vertices of $H''$ to obtain the original graph $H'$. In the language of graph homomorphism, we have $\varphi:V(H')\rightarrow V(H'')$, and $\varphi^{-1}:V(H'')\rightarrow V(H')$, that is, $H'$ admits a \emph{graph homomorphism} to $H''$ (see Problem \ref{qeu:isomorphic}). Let $F(H',\odot)$ be the set of graphs obtained by vertex-coinciding some vertices of $H'$ with the same color into one, where each $G$ of $F(H',\odot)$ admits an edge-difference proper total coloring $h$ holding $h(uv)+|h(u)-h(v)|=16$ for each edge $uv\in E(G)$ such that $G=H'\odot X_G$ with $X_G \subset V(H')$ and $H'=G\wedge X_G$. Thereby, each of $F(H',\odot)$ can be considered as a private key if we set $H'$ as a public key and $L$ admits a graph homomorphism to $H'$.

Sometimes, we call $H'$ a $\Delta$-saturated graph since degree $\textrm{deg}_{H'}(u)=1$ or $\textrm{deg}_{H'}(u)=\Delta(H')$ for each vertex $u\in V(H')$, also, $H''$ is a $\Delta$-saturated graph too. Based on an \emph{edge-difference ice-flower system} $I_{ce}(E_{1,7}D_k)^{17}_{k=1}$ shown in Fig.\ref{fig:Dsaturated-edge-difference-odd}, another $\Delta$-saturated graph $H^*=H'\overline{\triangleleft} |^{17}_{k=1} a_kE_{1,7}D_{k}$ with $a_k\in Z^0$ shown in Fig.\ref{fig:1-edge-diff-replacing-2}, where the graphic lattice base $\textbf{\textrm{E}}_{1,7}\textbf{\textrm{D}}=\{E_{1,7}D_{1},E_{1,7}D_{2},\dots ,E_{1,7}D_{17}\}$ is shown in Fig.\ref{fig:Dsaturated-edge-difference-odd}.\qqed
\end{exa}

\begin{figure}[h]
\centering
\includegraphics[width=16.2cm]{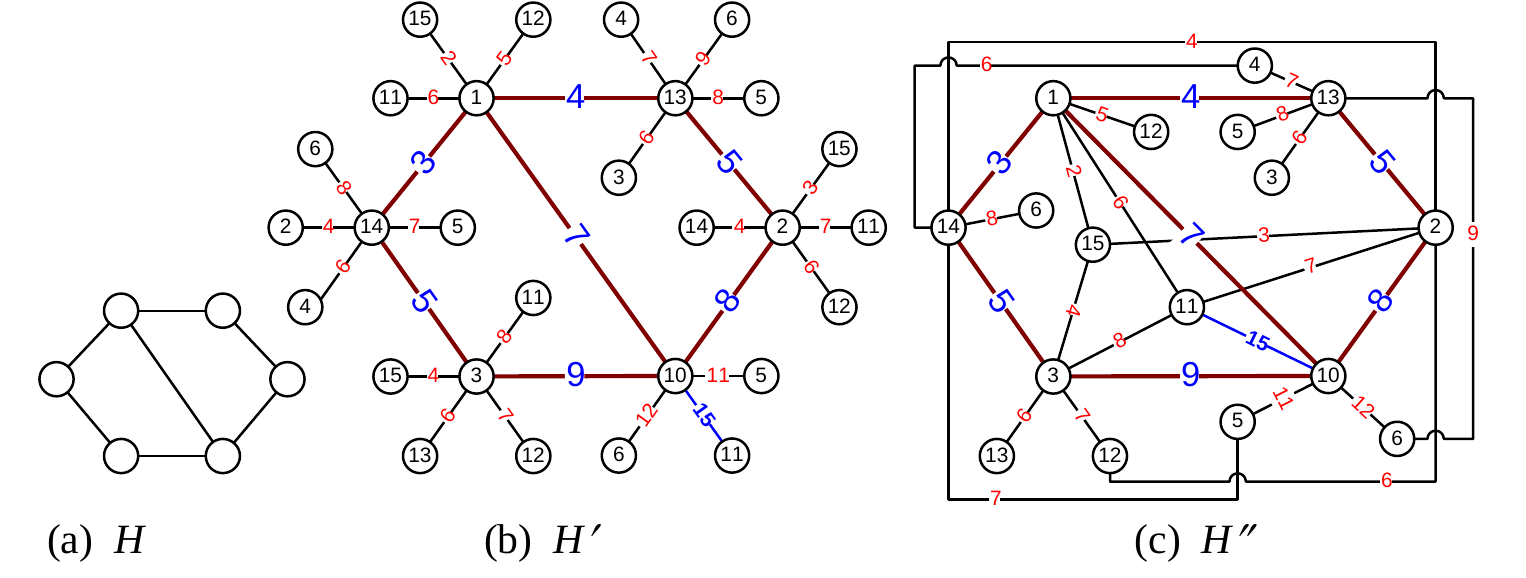}
\caption{\label{fig:2-edge-diff-replacing}{\small An example for illustrating the fully vertex-replacing operation under the edge-difference proper total coloring.}}
\end{figure}

\subsection{Matching-type graphic lattices}

Matching-type graphic lattices are connected with each other by matching (colored) graphs, matching labellings, and matching colorings.

\subsubsection{Matchings made by two or more graphs} In the following discussion, we will use traditional complementary graphs and $G$-complementary graphs to build up graphic lattices.

Traditional graph and its complement. A graph $\overline{G}$ is called the \emph{complement} of a graph $G$ of $n$ vertices if $V(G)=V(\overline{G})=V(K_n)$, $E(G)\cap E(\overline{G})=\emptyset$ and $E(G)\cup E(\overline{G})=E(K_n)$, then we say that $(G,\overline{G})$ is a \emph{complete-graphic matching}. Comparing the graphic lattice (\ref{eqa:graphic-lattice-graph-operation}), we have a \emph{complement graphic lattice}
\begin{equation}\label{eqa:matching-graphic-lattice}
\textbf{\textrm{L}}(\textbf{\textrm{H}}(\bullet)\overline{F}_{p,q})=\left \{\overline{G}(\bullet)^n_{i=1}a_iH_i:~a_i\in Z^0,~\overline{G}\in \overline{F}_{p,q}\right \}
\end{equation} where the graphic lattice base $\textbf{\textrm{H}}=(H_1$, $H_2$, $\dots $, $ H_n)$ is the same as that shown in (\ref{eqa:graphic-lattice-graph-operation}), $\overline{F}_{p,q}$ is the set of all complements of graphs of $F_{p,q}$ shown in (\ref{eqa:graphic-lattice-graph-operation}), and $\sum a_i\geq 1$. Let $\overline{\textbf{\textrm{H}}}=(\overline{H}_1,\overline{H}_2,\dots ,\overline{H}_n)$ be the \emph{complement base} of the graphic lattice base $\textbf{\textrm{H}}$ with the complement $\overline{H}_i$ of $H_i$ for $i\in [1,n]$, we get a \emph{complement base graphic lattice}
\begin{equation}\label{eqa:matching-graphic-lattice-complement}
\textbf{\textrm{L}}(\overline{\textbf{\textrm{H}}}(\bullet)F_{p,q})=\left \{G(\bullet)^n_{i=1}a_i\overline{H}_i:~a_i\in Z^0,~G\in F_{p,q}\right \}.
\end{equation}
with $\sum ^n_{i=1}a_i\geq 1$. Moreover, we obtain a \emph{totally complement graphic lattice} as follows:
\begin{equation}\label{eqa:matching-graphic-lattice-totally-complement}
\textbf{\textrm{L}}(\overline{\textbf{\textrm{H}}}(\bullet)\overline{F}_{p,q})=\left \{\overline{G}(\bullet)^n_{i=1}a_i\overline{H}_i:~a_i\in Z^0,~\overline{G}\in \overline{F}_{p,q}\right \}
\end{equation}
with $\sum ^n_{i=1}a_i\geq 1$. We call $(\textbf{\textrm{L}}(\textbf{\textrm{H}}(\bullet)F_{p,q}), \textbf{\textrm{L}}(\overline{\textbf{\textrm{H}}}(\bullet)\overline{F}_{p,q}))$ a matching of \emph{complementary graphic lattices}. However, let $G^*=G(\bullet)^n_{i=1}a_iH_i$, the complement $\overline{G^*}$ of $G^*$ is not $\overline{G}(\bullet)^n_{i=1}a_i\overline{H}_i$, in general.

A graph $G$ has two proper subgraphs $G_1,G_2$ such that $V(G)=V(G_1)\cup V(G_2)$, $E(G_1)\cap E(G_2)=\emptyset$ and $E(G_1)\cup E(G_2)=E(G)$. Thereby, we call $(G_1,G_2)$ a \emph{$G$-matching}. Correspondingly, we have the \emph{$G$-complementary graphic lattice} like that shown in (\ref{eqa:matching-graphic-lattice-totally-complement}).

\subsubsection{Coloring matchings on a graph} There are many matching labellings or matching colorings in graph theory.

(i) A graph $G$ admits two matchable colorings $f,h$, so we have two colored graphs $G_f,G_h$ holding $G\cong G_f\cong G_h$, where $G_f$ admits the coloring $f$, and $G_h$ admits the coloring $h$. For example, a connected graph $T$ admits a graceful labelling $f$, then its dual labelling $g(x)=\max f+\min f-f(x)$ for $x\in V(T)$ matches with $f$, where $\max f=\max \{f(x):x\in V(T)\}$ and $\min f=\min \{f(x):x\in V(T)\}$. Some matching colorings are introduced in Definition \ref{defn:combinatoric-definition-total-coloring} and Definition \ref{defn:4-dual-total-coloring}.

(ii) A graph $H$ admits a $W$-type coloring $g$, and this coloring $g$ matches with $W_i$-type colorings $g_i$ with $i\in [1,m]$ such that $(g,g_i)$ is a pair of matchable colorings. For example, the authors in \cite{Yao-Liu-Yao-2017} presented that a set-ordered graceful labelling of a tree is equivalent with many different labellings of the tree.

\subsubsection{Matchings made by graphs and colorings} Let $(G^{(k)}_1,G^{(k)}_2)$ be a $G^{(k)}$-matching graph pair based on a graph $G^{(k)}$ with $k\in[1,n]$, and let each $G^{(k)}_i$ admit a $W^{(k)}_i$-type coloring $f^{(k)}_i$ with $i=1,2$. Suppose that $G^{(k)}$ admits a $W$-type coloring $f^{(k)}$ induced by $f^{(k)}_1,f^{(k)}_2$, in other words, $(f^{(k)}_1,f^{(k)}_2)$ is an $f^{(k)}$-matching coloring. We obtain a pair of \emph{matching graphic lattices} below
\begin{equation}\label{eqa:matching-graphic-colorings-lattice}
{
\begin{split}
\textbf{\textrm{L}}(\textbf{\textrm{G}}_i(\bullet)F_{p,q})=\left \{H(\bullet)^n_{k=1}a_kG^{(k)}_i:a_k\in Z^0,~H\in F_{p,q}\right \}
\end{split}}
\end{equation} with $\sum ^n_{k=1}a_k\geq 1$, and the graphic lattice base $\textbf{\textrm{G}}_i=\left (G^{(1)}_i,G^{(2)}_i,\dots, G^{(n)}_i\right )$ for $i=1,2$. Naturally, we call $(\textbf{\textrm{G}}_1,\textbf{\textrm{G}}_2)$ a pair of \emph{matching bases}.

A colored matching $(H, H^*)$ of matchings made by graphs and colorings is shown in Fig.\ref{fig:4-color-planar-dual-graphs}(b). In Fig.\ref{fig:harmonious_odd-graceful}, $T$ admits an odd-graceful labelling $f$, and $G_i$ admits a harmonious labelling $g_i$ with $i=1,2$. So, $(f,g_i)$ is a matching of an odd-graceful labelling and a harmonious labelling for $i=1,2$; $H$ is the topological authentication of the public key $H_1$ and the private key $H_2$.

\begin{figure}[h]
\centering
\includegraphics[width=15cm]{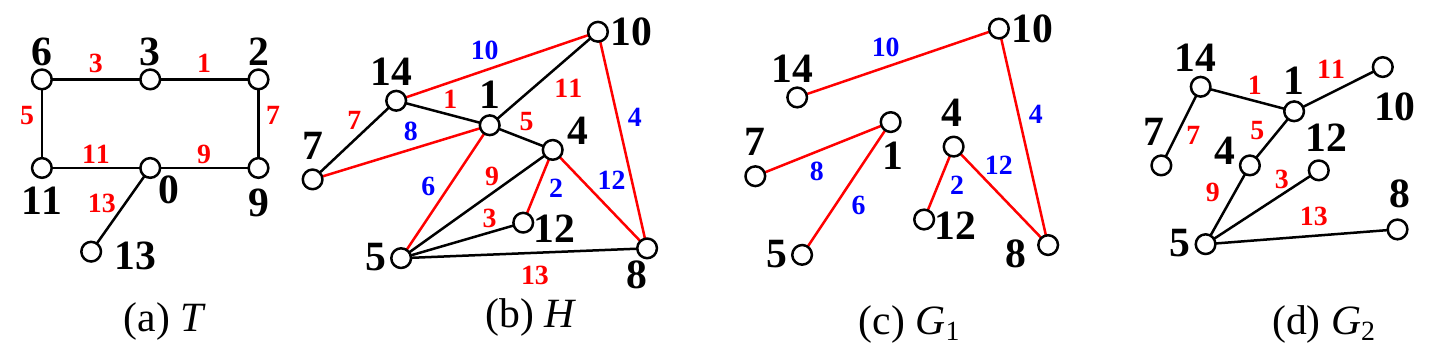}\\
\caption{\label{fig:harmonious_odd-graceful} {\small Harmonious labellings match with odd-graceful labellings.}}
\end{figure}

Each graph of three $(7,7)$-graphs $G_1,G_2,G_3$ shown in Fig.\ref{fig:twin-odd-graceful-00} and Fig.\ref{fig:twin-odd-graceful-11} admits an \emph{odd-graceful labelling} $f_i$ and each graph $H_{i,j}$ admits a \emph{pseudo odd-graceful labelling} $g_{i,j}$ with $i\in [1,3]$ and $j\in [1,6]$, such that $f_i(V(G_i))\cup g_{i,j}(V(H_{i,j}))=[1,14]$, $f_i(E(G_i))=[1,13]^o=g_{i,j}(E(H_{i,j}))$. So, we call $(G_t,H_{i,j})$ a \emph{twin odd-graceful matching}, $(f_t,g_{i,j})$ a pair of \emph{twin odd-graceful labellings} defined in \cite{Wang-Xu-Yao-2017-Twin}.
Notice that two odd-graceful $(7,7)$-graphs $G_1$ and $G_2$ have their twin odd-graceful matchings with $H_{1,j}\cong H_{2,j}$ for $j\in [1,6]$, in other word, the twin odd-graceful matchings of $G_1$ and $G_2$ keep isomorphic configuration, so $G_1$ and $G_2$ are \emph{twisted} under the isomorphic configuration of their own twin odd-graceful matchings. However, the twin odd-graceful matching $H_{3,j}~(j\in [1,6])$ of the odd-graceful $(7,7)$-graph $G_3$ shown in Fig.\ref{fig:twin-odd-graceful-11} are not isomorphic to $H_{1,j}$ and $H_{2,j}$ of $G_1$ and $G_2$ with $j\in [1,6]$. The above examples tell us that finding twin odd-graceful matchings of a non-tree graph is not a slight work.

\begin{figure}[h]
\centering
\includegraphics[width=16cm]{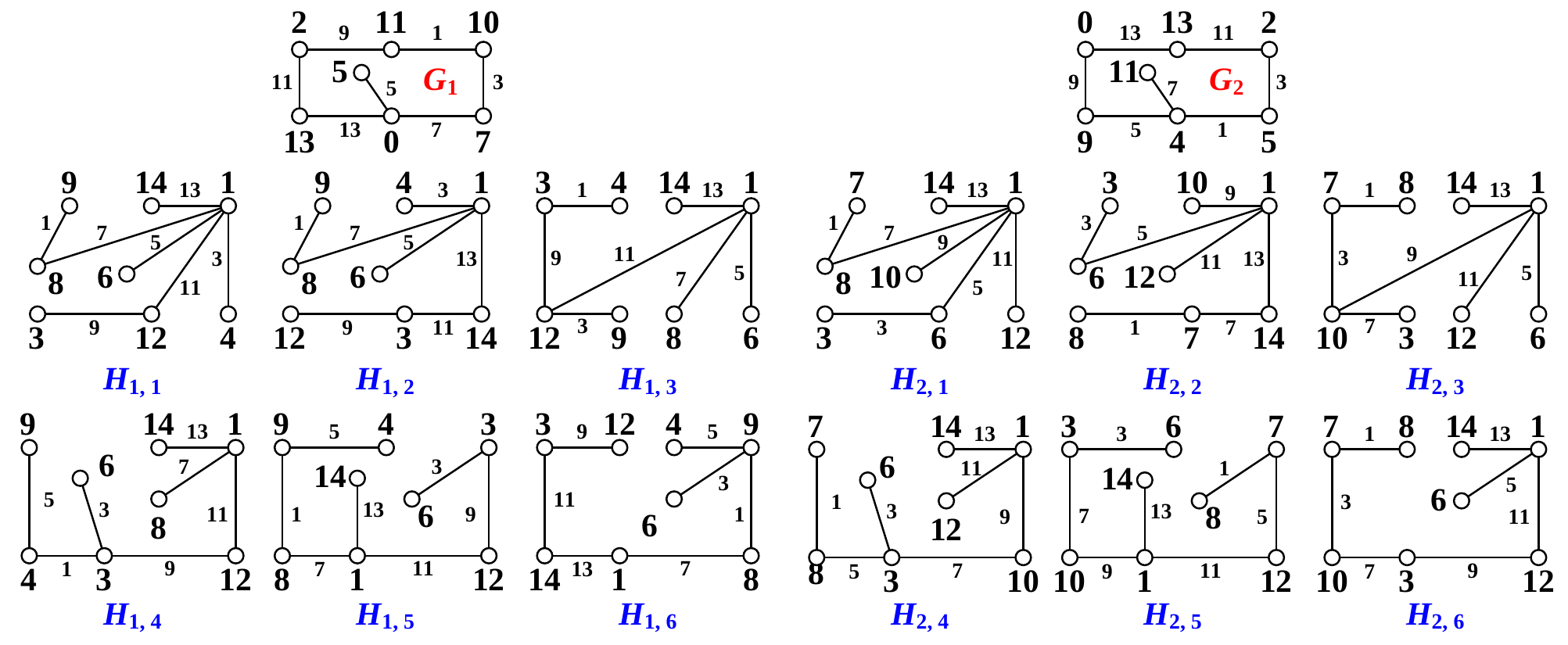}
\caption{\label{fig:twin-odd-graceful-00}{\small Two groups of twin odd-graceful matchings.}}
\end{figure}

\begin{figure}[h]
\centering
\includegraphics[width=16cm]{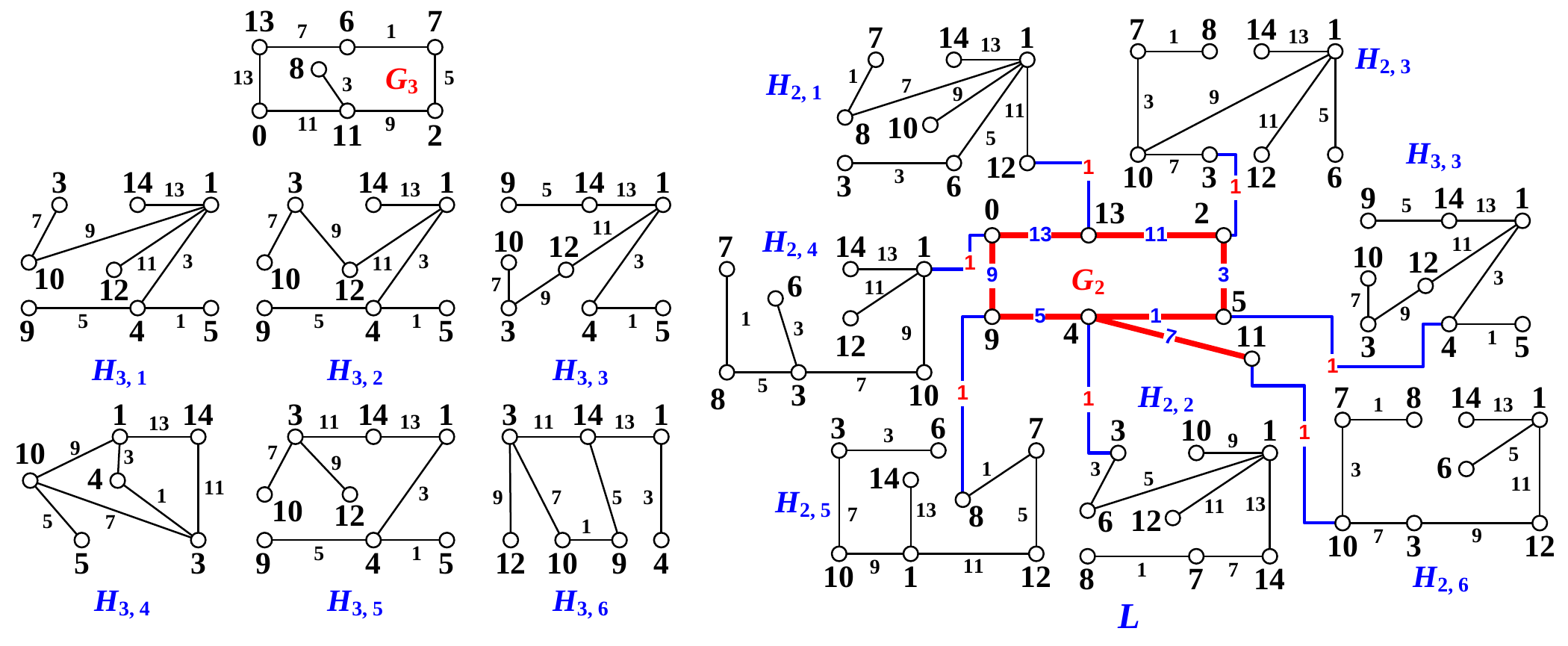}
\caption{\label{fig:twin-odd-graceful-11}{\small A group of twin odd-graceful matchings, and a graph $L=G_t\oslash ^3_{k=1}|^{M_{\textrm{odd}}}_{k,j}a_{k,j}H_{k,j}$.}}
\end{figure}

Let $F_{\textrm{odd}}(G)$ be the set of graphs $G_t$ with $G_t\cong G$ and admitting odd-graceful labellings, and let $\textrm{\textbf{H}}=(H_{i,j})^{M_{\textrm{odd}}}_{i,j}$ be the \emph{base}, where $M_{\textrm{odd}}$ is the number of twin odd-graceful matchings $(G_t,H_{i,j})$ for $G_t\in F_{\textrm{odd}}(G)$. We can join a vertex $x_{t,s}$ of $G_t$ for $s\in [1,|V(G_t)|]$ with a vertex $y^{k}_{i,j}$ of $H_{i,j}$ for $k\in [1,|V(H_{i,j})|]$ by an edge $x_{t,s}y^{k}_{i,j}$, the resultant graph is denoted as $L=G_t\oslash ^3_{k=1}|^{M_{\textrm{odd}}}_{k,j}a_{k,j}H_{k,j}$ with $\sum a_{k,j}=|V(G_t)|$ (see an example shown in Fig.\ref{fig:twin-odd-graceful-11}), and $L$ admits a labelling $h$ defined by $h(w)=f_t(w)$ for each element $w\in V(G_t)\cup E(G_t)$, $h(w)=g_{i,j}(w)$ for each element $w\in V(H_{i,j})\cup E(H_{i,j})$ and $h(x_{t,s}y^{k}_{i,j})=|f_t(x_{t,s})-g_{i,j}(y^{k}_{i,j})|$. We obtain a \emph{twin odd-graceful graphic lattice}
\begin{equation}\label{eqa:twin-odd-graceful-lattice}
\textrm{\textbf{L}}(\textrm{\textbf{H}}\oslash \textbf{F}_{\textrm{odd}}(G)) =\left \{G_t\oslash ^3_{k=1}|^{M_{\textrm{odd}}}_{k,j}a_{k,j}H_{k,j}: a_k\in Z^0, G_t\in F_{\textrm{odd}}(G)\right \}
\end{equation} with $\sum a_{k,j}=|V(G_t)|$, where each matching $(G_t,H_{k,j})$ is a twin odd-graceful matching.

\begin{problem}\label{qeu:especial-total-colorings}
\begin{asparaenum}[\textrm{Twin}-1. ]
\item \textbf{Find} an algorithm for figuring all graphs $H_{i,j}$ of the base $\textrm{\textbf{H}}$, that is, \textbf{find} all twin odd-graceful matchings $(G_t,H_{i,j})$ for each colored graph $G_t\in F_{\textrm{odd}}(G)$, and determine $M_{\textrm{odd}}$.
\item \textbf{Find} the smallest $\sum h(x_{t,s}y^{k}_{i,j})$ in all graphs $L=G_t\oslash ^3_{k=1}|^{M_{\textrm{odd}}}_{k,j}a_{k,j}H_{k,j}$ of a twin odd-graceful graphic lattice $\textrm{\textbf{L}}(\textrm{\textbf{H}}\oslash \textbf{F}_{\textrm{odd}}(G))$.
\end{asparaenum}
\end{problem}

\subsection{Graphic lattice sequences}

Let $F^{(0)}$ be the initial set of graphs. So we have $\textbf{\textrm{L}}^{(1)}(\textbf{\textrm{H}}(\bullet)F^{(0)})=F^{(1)}$, thus, $\textbf{\textrm{L}}^{(2)}(\textbf{\textrm{H}}(\bullet)F^{(1)})=F^{(2)}$, go on in this way, we get a sequence of graphic lattices, denoted as $\{\textbf{\textrm{L}}^{(t)}(\textbf{\textrm{H}}(\bullet)F^{(t-1)})=F^{(t)}\}$, and we call $\{\textbf{\textrm{L}}^{(t)}(\textbf{\textrm{H}}(\bullet)F^{(t-1)})\}$ as a \emph{graphic lattice sequence}.

We see another type of graphic lattice sequences as follows: Let $\textbf{\textrm{T}}^{(t)}=(T^{t}_1,T^{t}_2,\dots T^{t}_n)$, where $T^{t}_j=H_{t-1}(\bullet)^n_{i=1}a_{j,i}T^{t-1}_i$ with $a_{j,i}\in Z^0$ and $H_{t-1}\in F_{p,q}$. We have a graphic lattice sequence $\{\textbf{\textrm{L}}^{(t)}(\textbf{\textrm{T}}^{(t)}(\bullet)F_{p,q})\}$ defined as follows:
\begin{equation}\label{eqa:matching-graphic-lattice-sequence}
{
\begin{split}
\textbf{\textrm{L}}^{(t)}(\textbf{\textrm{T}}^{(t)}(\bullet)F_{p,q})=\{H_{t}(\bullet)^n_{i=1}a_{i}T^{t}_i:a_i\in Z^0,~H_{t}\in F_{p,q}\}.
\end{split}}
\end{equation} with $\sum^n_{i=1}a_{i}\geq 1$. Furthermore, we get
$$H_{t}(\bullet)^n_{j=1}a_{i}T^{t}_i=H_{t}(\bullet)^n_{j=1}a_{i}[H_{t-1}(\bullet)^n_{i=1}a_{j,i}T^{t-1}_i].$$

Clearly, these two graphic lattice sequences $\left \{\textbf{\textrm{L}}^{(t)}(\textbf{\textrm{H}}(\bullet)F^{(t-1)})\right \}$ and $\left \{\textbf{\textrm{L}}^{(t)}(\textbf{\textrm{T}}^{(t)}(\bullet)F_{p,q})\right \}$ form dynamic networks. It may be interesting to connect these two graphic lattices with some topics of researching networks.

\subsection{Planar graphic lattices}

As known, each 4-colorable planar graph $G$ forms an every-zero graphic group $\{F_f(G); \oplus\}$ with $|F_f(G)|=4$. In Figure \ref{fig:4-color-planar-tile-lattice}(a), we tile a colored triangle $T^r_i$ with another colored triangle $T^r_{5-i}$ together by vertex-coinciding an edge $ab$ of $T^r_i$ with an edge $ab$ of $T^r_j$ into one, where $a+b=5$ and $a\neq i$ for $i\in [1,4]$. We use the triangles $T^r_1,T^r_2,T^r_3,T^r_4$ of the every-zero graphic group $\{F_{\textrm{4color}};\oplus\}$ to replace each inner face of a planar graph $H$ having triangular inner faces, such that $T^r_i$ and $T^r_{5-i}$ are tiled correctly, the resulting planar graph $H^*$ is properly colored with four colors, we write $H^*$ as $H^*=H\bigtriangleup ^4_{k=1}a_kT^r_k$ with $a_k\in Z^0$ and $\sum a_k\geq 1$. Let $F_{\textrm{inner}\bigtriangleup}$ be the set of planer graphs having triangular inner faces. Thereby, we get a \emph{planar graphic lattice} as follows:
\begin{equation}\label{eqa:planar-lattice}
{
\begin{split}
\textbf{\textrm{L}}(\textbf{\textrm{T}}^r\bigtriangleup F_{\textrm{inner}\bigtriangleup})=\{H\bigtriangleup ^4_{k=1}a_kT^r_k:a_k\in Z^0,H\in F_{\textrm{inner}\bigtriangleup}\}
\end{split}}
\end{equation} with $\sum^4_{k=1} a_k\geq 1$, where the \emph{planar graphic lattice base} is $\textbf{\textrm{T}}^r=(T^r_1,T^r_2,T^r_3,T^r_4)$. Thereby, each planar graph $G\in \textbf{\textrm{L}}(\textbf{\textrm{T}}^r\bigtriangleup F_{\textrm{inner}\bigtriangleup})$ is a $4$-colorable graph having each inner face to be a triangle.

\begin{figure}[h]
\centering
\includegraphics[width=16.2cm]{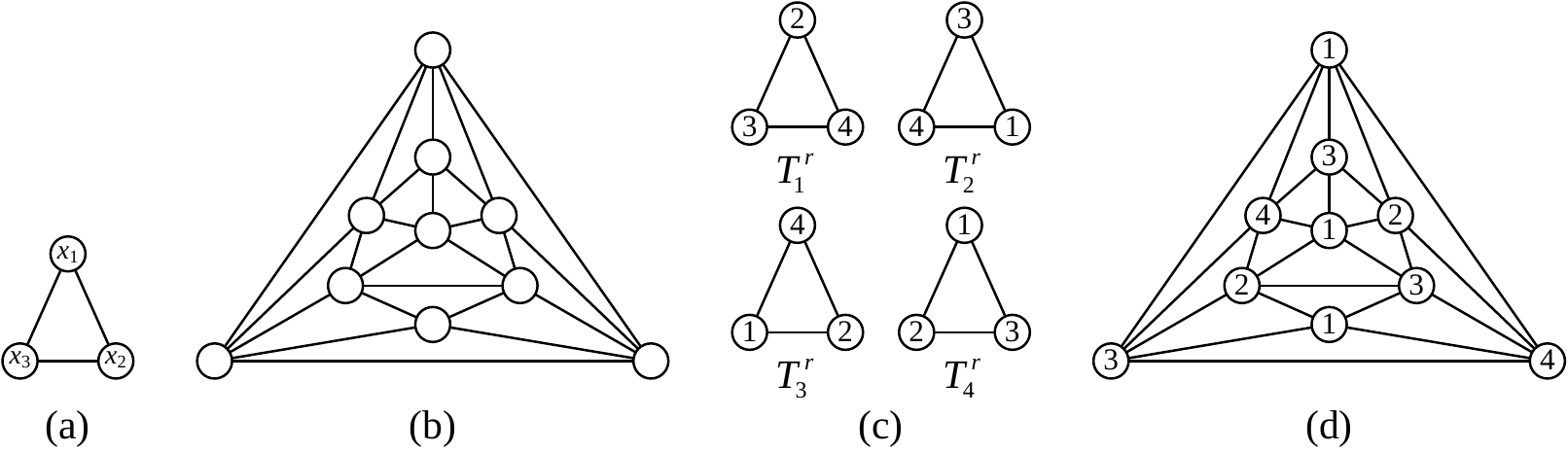}
\caption{\label{fig:4-color-planar-tile-lattice} {\small (a) A triangle; (b) a maximal planar graph $H$; (c) an every-zero graphic group $\{F_{\textrm{inner}\bigtriangleup};\oplus\}$ cited from \cite{YAO-SUN-WANG-SU-XU2018arXiv}; (d) a maximal planar graph $H$ tiled by the every-zero graphic group $\{F_{planar};\oplus\}$ shown in (c).}}
\end{figure}

\begin{conj}\label{conj:velocity}
\cite{YAO-SUN-WANG-SU-XU2018arXiv} A maximal planar graph is 4-colorable if and only it can be tiled by the every-zero graphic group $\{F_{\textrm{inner}\bigtriangleup};\oplus\}$ shown in Fig.\ref{fig:4-color-planar-tile-lattice}(c).
\end{conj}

\begin{figure}[h]
\centering
\includegraphics[width=16cm]{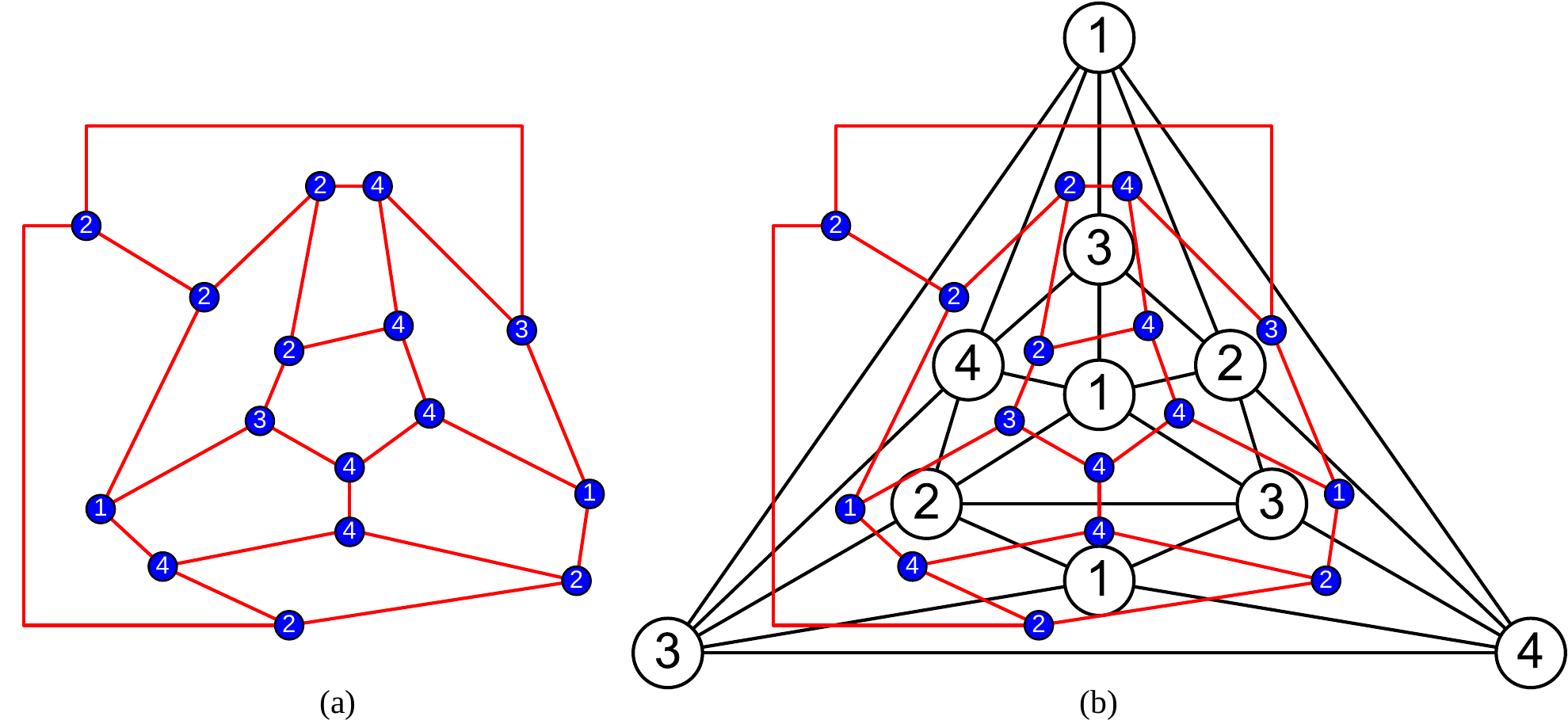}
\caption{\label{fig:4-color-planar-dual-graphs} {\small (a) The dual $H^*$ of a colored maximal planar graph $H$ shown in Fig.\ref{fig:4-color-planar-tile-lattice}(d), and $H^*$ is a 3-regular planar graph; (b) a colored matching $(H, H^*)$.}}
\end{figure}

\begin{problem}\label{qeu:planar-graph}
It may be interesting to consider the following questions:
\begin{asparaenum}[\textrm{4C}-1. ]
\item \textbf{Determine} $\{H\bigtriangleup [F(H)-1]\cdot T^r_k: H\in F_{\textrm{inner}\bigtriangleup}\}$, where $F(H)$ is the face number of $H\in F_{\textrm{inner}\bigtriangleup}$ and a fixed $T^r_k\in \{F_{\textrm{inner}\bigtriangleup};\oplus\}$. In other word, each $H\bigtriangleup [F(H)-1]\cdot T^r_k$ is tiled by one $T^r_k$ only, like $L_1$ shown in Fig.\ref{fig:4C-group-new-00}, so $H$ is 3-colorable. Some results on this question can be founded in \cite{Li-Zhu-Shao-Xu-2016Discrete} and \cite{Li-Zhu-Shao-Xu-2017Computer-Mathematics}.
\item \textbf{Find} conditions for a planar graph $H\in F_{\textrm{inner}\bigtriangleup}$, such that $H$ must be tiled with all elements of the planar graphic lattice base $\textbf{\textrm{T}}^r$ only.
\item \textbf{Estimate} the exact cardinality of a planar graphic lattice $\textbf{\textrm{L}}(\textbf{\textrm{T}}^r\bigtriangleup F_{\textrm{inner}\bigtriangleup})$.
\item For each uncolored planar graph $H\in F_{\textrm{inner}\bigtriangleup}$, \textbf{does} there exist a 4-colorable planar graph $G=T\bigtriangleup ^4_{k=1}a_kT^r_k$ with $a_k\in Z^0$ and $\sum a_k\geq 1$ such that $H\cong G$?\qqed
\item Use the elements of the planar graphic lattice base $\textbf{\textrm{T}}^r$ to tile completely the whole $xOy$-plane, such that the resultant plane, denoted as $P_{\textrm{4C}}$, is $4$-colorable, and the plane $P_{\textrm{4C}}$ contains infinite triangles $T^r_k$ for each $k\in [1,4]$, we call $P_{\textrm{4C}}$ a \emph{$4$-colorable plane}. For any given planar graph $G$, \textbf{prove} $G$ in, or not in one of all $4$-colorable planes of the plane $P_{\textrm{4C}}$. Similarly, we can consider: Any 3-colorable planar graph is in one of all $3$-colorable planes of the plane $P_{\textrm{3C}}$ tiled completely by one element of the planar graphic lattice base $\textbf{\textrm{T}}^r$.
\end{asparaenum}
\end{problem}

Recall Four-Color Problem, we have:

\begin{thm} \label{thm:Four-Color-Problem-equivalent}
\cite{Bondy-2008} The following three assertions are equivalent to each other:

(i) Every planar graph is $4$-vertex colorable, that is, $\chi\leq 4$;

(ii) Every planar graph is $4$-face-colorable; and

(iii) Every simple, $2$-connected and 3-regular planar graph is $3$-edge-colorable.
\end{thm}

By means of computer, Appel, Haken, Koch and Bull, in 1976, proved Four-Color Theorem (4CT); and furthermore Robertson, Sanders, Seymour and Thomas in 1997, proved it again. Unfortunately, no one find a simpler proof of 4CT that can be
written in a fewer papers. Kauman and Saleur \cite{Louis-Kauman-Saleur-1993} pointed out: ``\emph{While it has sometimes been said that the four color problem is an isolated problem in mathematics, we have found that just the opposite is the case. The Four-Color Problem is central to the intersection of algebra, topology, and statistical mechanics}.''

\begin{figure}[h]
\centering
\includegraphics[width=16.2cm]{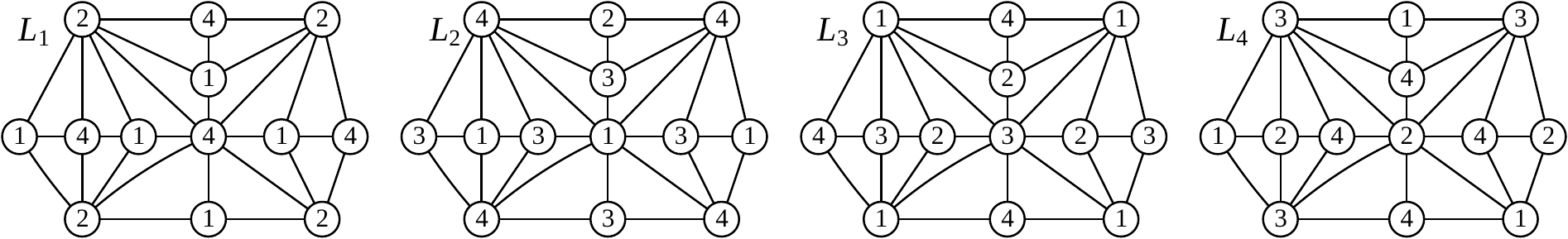}
\caption{\label{fig:4C-group-new-00}{\small According to an every-zero graphic group $\{F_{\textrm{inner}\bigtriangleup};\oplus\}$ shown in Fig.\ref{fig:4-color-planar-tile-lattice}(c): $L_1$ is tiled by $T^r_3$ only; $L_2$ is tiled by $T^r_1$ and $T^r_2$; $T^r_1$ does not appear in $L_3$; and each of $\{F_f(\Delta);\oplus\}$ is in $L_4$.}}
\end{figure}

\subsection{Graphic lattices made by graph labellings}

We focus on some particular colored graphic lattices defined in (\ref{eqa:graphic-lattice-colored}). A colored tree lattice
$\textbf{\textrm{L}}(\textbf{\textrm{H}}^c\odot F^c_{n})=\{T\odot |^n_{i=1}H_i:~T\in F^c_{n}\}$
is made by trees $T\odot |^n_{i=1}H_i$, where each $T$ is a tree/forest with $p~(\leq n)$ vertices and admitting a $W$-type labelling $g_T$, and each base graph-vector $H_i\in F^c_{n}$ is a tree with $i\in [1,n]$, as well as the lattice base $\textbf{H}^c$ formed by $n$ linearly independent disjoint graphs $H_1,H_2,\dots ,H_n$ admits a flawed $W$-type labelling $f_{H}$. So, each tree $T\odot |^n_{i=1}H_i$ admits a labelling $h=g_T\odot f_{H}$. One want to know the labelling $h$ is one of well-defined labellings in \cite{Gallian2019, Yao-Sun-Zhang-Mu-Sun-Wang-Su-Zhang-Yang-Yang-2018arXiv, Yao-Zhang-Sun-Mu-Sun-Wang-Wang-Ma-Su-Yang-Yang-Zhang-2018arXiv, Yao-Mu-Sun-Sun-Zhang-Wang-Su-Zhang-Yang-Zhao-Wang-Ma-Yao-Yang-Xie2019}.

\subsubsection{Graphic lattices on felicitous labellings}

In \cite{Zhang-Yao-Wang-Wang-Yang-Yang-2013}, the authors have shown a felicitous graphic lattice $\{T \odot |^n_{i=1}H_i\}$ with each $T$ of $n$ vertices admitting a set-ordered felicitous labelling and the graphic lattice base $\textbf{H}^c=(H_1,H_2,\dots ,H_n)$ admitting a flawed felicitous labelling by Definition \ref{defn:optimal-set-ordered-felicitous}, Lemma \ref{thm:felicitous-links-identify} and Theorem \ref{thm:felicitous-graphic-lattice} shown in the following:

\begin{defn} \label{defn:optimal-set-ordered-felicitous}
Let $(X, Y)$ be the bipartition of a bipartite $(p, q)$-graph $G$. If $G$ admits a felicitous labelling $f$ such that $\max\{f(x):x\in X\}<b=\min\{f(y):y\in Y\}$, then we call $f$ a \emph{set-ordered felicitous labelling} and $G$ a \emph{set-ordered felicitous graph}, and moreover $f$ is called an \emph{optimal set-ordered felicitous labelling} if $f(V(G))\subseteq [b,b+q-1]$ and $f(E(G))(\bmod~q)=[0,q-1]$.\qqed
\end{defn}

\begin{lem} \label{thm:felicitous-links-identify}
Let $T$ be a tree of $n$ vertices admitting a set-ordered felicitous labelling $g$ and let $G$ be a connected $(p,q)$-graph admitting an optimal set-ordered felicitous labelling (see Definition \ref{defn:optimal-set-ordered-felicitous}). Then we have at least a graph of $\{T \odot |^n_{i=1}H_i\}$ admitting a felicitous labelling, where $H_i\cong G$ with $i\in [1,n]$, and $\textbf{H}^c=(H_1,H_2,\dots ,H_n)$ admits a flawed felicitous labelling induced by the optimal set-ordered felicitous labelling.
\end{lem}

\begin{thm} \label{thm:felicitous-graphic-lattice}
If a tree $T$ of $n$ vertices admits a set-ordered felicitous labelling, and $\textbf{H}^c=(H_1$, $H_2$, $\dots $, $ H_n)$ admits a flawed felicitous labelling, then $\{T \odot |^n_{i=1}H_i\}$ contains graphs admitting felicitous labellings (see \cite{Gallian2019}).
\end{thm}

\subsubsection{Graphic lattices on edge-magic and anti-edge-magic total labellings}

\begin{defn} \label{defn:edge-magic-total-labelling}
\cite{Gallian2019} For a $(p,q)$-graph $G$, if there exist a constant $\lambda$ and a bijection $f:V(G)\cup
E(G)\rightarrow [1, p+q]$ such that $f(u)+f(v)+f(uv)=\lambda$ for every edge $uv\in E(G)$, then $f$ is called an \emph{edge-magic total labelling} and $\lambda$ a \emph{magic constant}.\qqed
\end{defn}

\begin{defn} \label{defn:super-set-ordered-edge-magic-total-labelling}
\cite{Wang-Yao-Yao-2014Information} Let $G$ be a bipartite graph with bipartition $(X, Y)$, and let $G$ admit an edge-magic total labelling $f$. There are two constraints:

(C1) $f(V(G))=[1,p]$; and

(C2) $\max \{f(x): x\in X\}<\min\{f(y): y\in Y\}$ (denoted as $f(X)<f(Y)$).\\
We call $f$ a \emph{super edge-magic total labelling} of $G$ if $f$ holds (C1), and $f$ a \emph{set-ordered edge-magic total labelling} of $G$ if $f$ holds (C2), and $f$ a \emph{super set-ordered edge-magic total labelling}
(super-so-edge-magic total labelling) of $G$ if $f$ holds both (C1) and (C2) true.\qqed
\end{defn}

\begin{defn} \label{defn:generalized-edge-magic-total-labelling}
\cite{Wang-Yao-Yao-2014Information} Let $G$ be a $(p,q)$-graph. If there exist a constant $\mu$ and a mapping $f:V(G)\cup E(G)\rightarrow [1, 2q+1]$ such that $f(u)+f(v)+f(uv)=\mu$ for every edge $uv\in E$, then we say $f$ a \emph{generalized edge-magic total labelling} of $G$, $\mu$ a \emph{generalized magic constant}. Furthermore, if $G$ is a bipartite graph with bipartition $(X,Y)$, and $f$ holds $f(V(G))=[1,q+1]$ and $f(X)<f(Y)$, we call $f$ a \emph{generalized super set-ordered edge-magic total labelling}.\qqed
\end{defn}

Suppose that the bipartition $(X,Y)$ of a tree $T$ of vertices $x_1,x_2,\dots,x_n$ holds $\big| |X|-|Y|\big |\leq 1$ true, and each $H_i$ of disjoint graphs $H_1, H_2,\dots ,H_n$ is a bipartite graph with its bipartition $(X_i,Y_i)$ holding $|X_i|=s$ and $|Y_i|=t$ for two constants $s,t$ and $i\in [1,n]$. We get a graph $T\odot |^n_{i=1}H_i$ obtained by vertex-coinciding a vertex of $H_i$ with the vertex $x_i$ of $T$ into one, so there is a set $\{T\odot |^n_{i=1}H_i\}$ containing at least $(s+t)^n$ graphs of the form $T\odot |^n_{i=1}H_i$. Wang \emph{et al.} in \cite{Wang-Yao-Yao-2014Information} have shown:

\begin{thm} \label{thm:main-theorem3}
\cite{Wang-Yao-Yao-2014Information} If $T$ admits a set-ordered graceful labelling, each $H_i$ admits a (generalized) super set-ordered edge-magic total labelling (see Definition \ref{defn:generalized-edge-magic-total-labelling}), then $\bigcup^n_{i=1} H_i$ (also $\textbf{H}^c=(H_1,H_2,\dots ,H_n)$) admits a flawed (generalized) super edge-magic total labelling. There exists at least a graph $G$ of $\{T\odot |^n_{i=1}H_i\}$ admits a (generalized) super edge-magic total labellings, and moreover $G$ admits a super edge-magic total labelling if $G$ is a tree.
\end{thm}

\begin{defn} \label{defn:anti-edge-magic-total-abelling}
Let $G$ be a $(p,q)$-graph. If there exists a set of arithmetic progression and a
bijection $f : V(G)\cup E(G)\rightarrow [1, p+q]$ such that $f(u)+f(v)+f(uv)\in \{k, k+d, k+2d, \cdots, k+(q-1)d\}$ for
every edge $uv\in E(G)$, and some values of $k,d\in Z^0$, then we say $f$ an anti-edge-magic total labelling of $G$. Furthermore, if $G$ is a bipartite graph with bipartition $(X, Y)$, and $f$ holds $f(V(G))=[1,p]$
and $\max\{f(x): x\in X\}<\min\{f(y): y\in Y\}$, we call $f$ a \emph{super set-ordered anti-edge-magic total labelling}.\qqed
\end{defn}

\begin{thm} \label{thm:anti-edge-magic-total-labelling}
\cite{Wang-Yao-Yang-Yang-Chen-2013} Suppose that $T$ and $H_1,H_2,\cdots, H_p$ are mutually disjoint trees, where $p=|V(T)|$. If $T$ admits a set-ordered graceful labelling, and each tree $H_k$ with $k\in[1, p]$ admits a
super set-ordered anti-edge-magic total labelling and its own bipartition $(X_k,Y_k)$ holding $|X_k|=s$ and $|Y_k|=t$ for two constants $s,t\geq 1$. Then $\{T\odot |^p_{i=1}H_i\}$ contains at least a graph admitting a super set-ordered anti-edge-magic total labelling defined in Definition \ref{defn:anti-edge-magic-total-abelling}.
\end{thm}

\subsubsection{Graphic lattices on $(k,d)$-edge-magic total labellings}
\begin{defn} \label{defn:k-d-edge-magic-total-labelling}
A $(p, q)$-graph $G$ admits a bijection $f:V(G)\cup E(G)\rightarrow \{d,2d, \dots, \mu d, k+(\mu+1)d, k+(p+q-1)d\}$ with $\mu\in[1,p+q-1]$, such that $f(u)+f(v)+f(uv)=\lambda$ for each edge $uv\in E(G)$, we call $f$ a $(k,d)$-edge-magic total labelling, $\lambda$ a magic constant. Moreover, if $G$ is a bipartite graph with bipartition $(X, Y)$, and $f$ holds $f(X)=\{ d,2d, \dots, |X|d\}$, $f(Y)=\{ k+|X|d, k+(|X|+1)d, \dots, k+(|X|+|Y|-1)d$, we call $f$ a super set-ordered $(k,d)$-edge-magic total labelling of $G$.\qqed
\end{defn}

\begin{thm} \label{thm:2}
Suppose that disjoint trees $H_1, H_2,\dots, H_{n}$ admit super set-ordered $(k,d)$-edge-magic total labellings, and $V(H_1)=V(H_2)=\cdots =V(H_n)$, a tree $T$ of $n$ vertices admits a set-ordered graceful labelling and its bipartition $(X,Y)$ holding $\big| |X|-|Y|\big |\leq 1$ true. Then there exists at least a graph $G$ of $\{T\odot |^n_{i=1}H_i\}$ admits a super set-ordered $(k,d)$-edge-magic total labelling defined in Definition \ref{defn:k-d-edge-magic-total-labelling}.
\end{thm}

\subsubsection{Graphic lattices on total graceful labellings}

\begin{defn} \label{defn:generalized-total-graceful-labellings}
\cite{Wang-Xu-Yao-Ars-2018, Wang-Xu-Yao-2019} A \emph{labelling} $\theta$ of a $(p,q)$-graph $G$ is a mapping from a set $V(G)\cup E(G)$ to $[m,n]$. Write
$\theta(V(G))=\{\theta(u):u\in V(G)\}$,
$\theta(E(G))=\{\theta(xy):xy\in E(G)\}$. There are the following
constraints:
\begin{asparaenum}[(a)]
\item \label{Proper01} $|\theta(V(G))|=p$, $|\theta(E(G))|=q$ and $\theta(xy)=|\theta(x)-\theta(y)|$ for every edge $xy\in E(G)$.
\item \label{TGraceful-001} $\theta(V(G))\cup \theta(E(G))= [1,p+q]$.
\item \label{Super-TGraceful} $\theta(E(G))=[1,q]$.
\item \label{Set-ordered} $G$ is a bipartite graph with the bipartition
$(X,Y)$ such that $\max\{\theta(x):x\in X\}< \min\{\theta(y):y\in
Y\}$ ($\theta(X)<\theta(Y)$ for short).
\item \label{General-TGraceful-001} $\theta(V(G))\cup \theta(E(G))= [1,M]$ with $M\geq 2q+1$.
\end{asparaenum}

Then a \emph{total graceful labelling} $\theta$ holds (\ref{Proper01})
and (\ref{TGraceful-001}) true; a \emph{super total graceful labelling} $\theta$ holds (\ref{Proper01}),
(\ref{TGraceful-001}) and (\ref{Super-TGraceful}) true; a \emph{set-ordered
total graceful labelling} $\theta$ holds (\ref{Proper01}),
(\ref{TGraceful-001}) and (\ref{Set-ordered}) true; a \emph{super set-ordered total graceful labelling} $\theta$ holds (\ref{Proper01}), (\ref{TGraceful-001}), (\ref{Super-TGraceful}) and (\ref{Set-ordered}) true.

Moreover, a \emph{generalized total graceful labelling} $\theta$ holds (\ref{Proper01}) and (\ref{General-TGraceful-001}) true; a \emph{super generalized total graceful labelling} $\theta$ holds (\ref{Proper01}), (\ref{General-TGraceful-001}) and (\ref{Super-TGraceful}) true; a \emph{set-ordered generalized total graceful labelling} $\theta$ holds (\ref{Proper01}),
(\ref{General-TGraceful-001}) and (\ref{Set-ordered}) true; a \emph{super set-ordered generalized total graceful labelling} $\theta$ holds (\ref{Proper01}), (\ref{General-TGraceful-001}), (\ref{Super-TGraceful}) and (\ref{Set-ordered}) true.\qqed
\end{defn}

The total graceful labelling was introduced in \cite{Subbiah-Pandimadevi-Chithra2015}. By Definition \ref{defn:generalized-total-graceful-labellings}, we have

\begin{thm} \label{thm:total-graceful-main-theorem}
\cite{Wang-Xu-Yao-2019} Suppose that $T$ and $H_1,H_2,\cdots, H_p$ are mutually disjoint trees, where $p=|V(T)|$. If $T$ and $H_{\lceil\frac{p+1}{2}\rceil}$ are disjoint graceful trees, each tree $H_i$ with $k\in[1,\lceil\frac{p-1}{2}\rceil]\cup[\lceil\frac{p+3}{2}\rceil, p]$ admits a set-ordered total graceful labelling $f_k$ and its own bipartition $(X_k,Y_k)$ holding $|X_k|=s$ and $|Y_k|=t$ with $s+t=|V(H_{\lceil\frac{p+1}{2}\rceil})|$. Then $\{T\odot |^p_{i=1}H_i\}$ contains at least a graph admitting a super total graceful labelling (see Definition \ref{defn:generalized-total-graceful-labellings}).
\end{thm}

\begin{thm} \label{thm:total-graceful-corollary}
\cite{Wang-Xu-Yao-2019} Suppose that $T$ and $H_1,H_2,\cdots, H_p$ are mutually disjoint trees, where $p=|V(T)|$. If $T$ admits a set-ordered total graceful labelling and each tree $H_k$ with $k\in[1, p]$ admits a set-ordered graceful labelling and its own bipartition $(X_k,Y_k)$ holding $|X_k|=s$ and $|Y_k|=t$ for two constants $s,t\geq 1$. Then $\{T\odot |^p_{i=1}H_i\}$ contains at least a graph admitting a super set-ordered total graceful labelling (see Definition \ref{defn:generalized-total-graceful-labellings}).
\end{thm}

\begin{figure}[h]
\centering
\includegraphics[width=16.2cm]{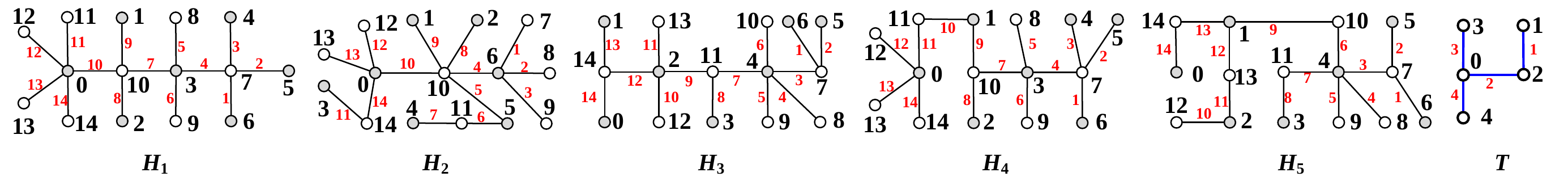}\\
\caption{\label{fig:example-total-graceful-0} {\small A base $\textbf{H}^c=(H_1,H_2,H_3,H_4, H_5)$ and a tree $T$ for illustrating Theorem \ref{thm:total-graceful-corollary} cited from \cite{Wang-Xu-Yao-2019}.}}
\end{figure}

\begin{figure}[h]
\centering
\includegraphics[width=15cm]{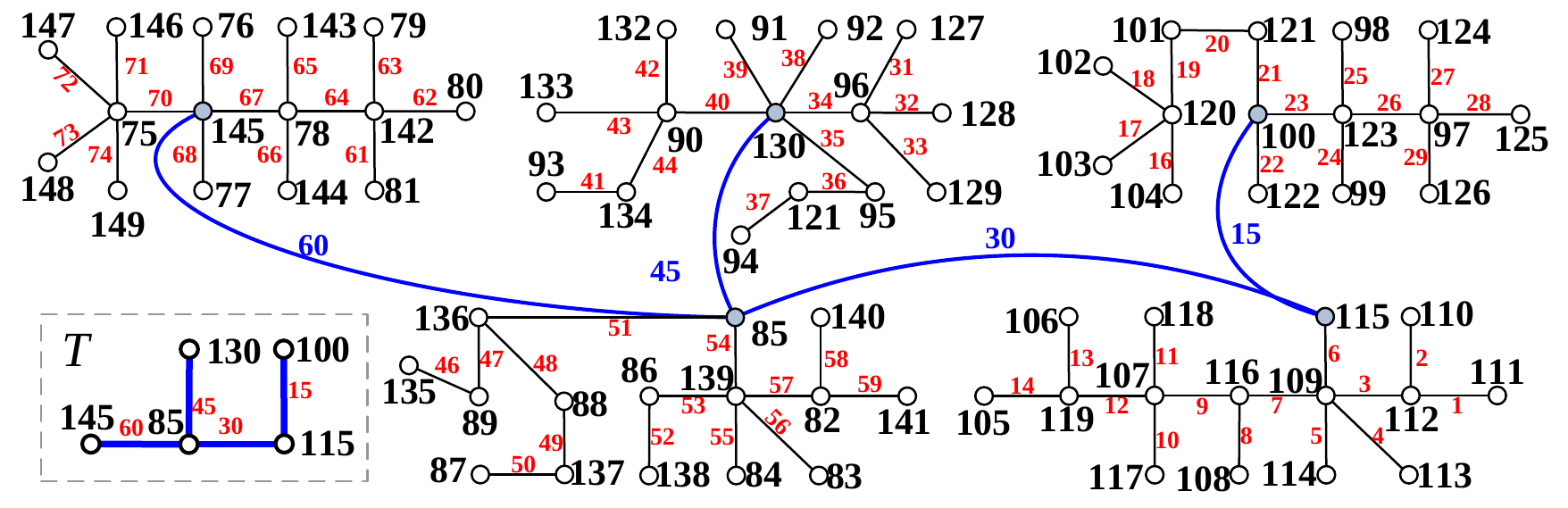}\\
\caption{\label{fig:example-total-graceful-1} {\small A tree $T\odot |^5_{i=1}H_i$ admits a super set-ordered total graceful labelling for illustrating Theorem \ref{thm:total-graceful-corollary} cited from \cite{Wang-Xu-Yao-2019}, where the graphic lattice base $\textbf{H}^c$ is shown in Fig.\ref{fig:example-total-graceful-0}.}}
\end{figure}

\begin{thm} \label{thm:GTGL-main-theorem}
\cite{Wang-Xu-Yao-Ars-2018} Let $T$ be a graceful tree of order $p$. Every connected $(n_k,m)$-graph $H_k$ has a set-ordered graceful labellings $f_k$ with $k\in[1, p]$ and $k\neq \lceil\frac{p+1}{2}\rceil$, and they, except $H_{\lceil\frac{p+1}{2}\rceil}$ which is a connected graceful graph, have the same labelling intervals of vertex bipartition. Then there exists a graph in the form $T\odot |^p_{i=1}H_i$ admits a super generalized total graceful labelling (see Definition \ref{defn:generalized-total-graceful-labellings}).
\end{thm}

\begin{figure}[h]
\centering
\includegraphics[width=16.2cm]{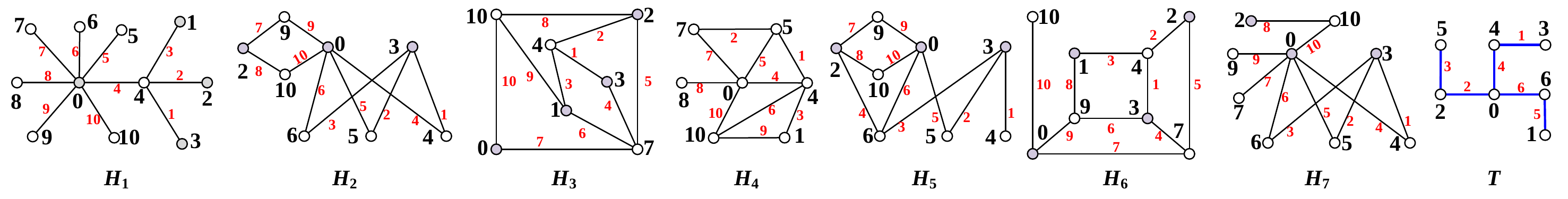}\\
\caption{\label{fig:example-generalized-total-graceful-0} {\small A base $\textbf{H}^c=(H_1,H_2,\dots ,H_7)$ and a tree $T$ for illustrating Theorem \ref{thm:GTGL-main-theorem} cited from \cite{Wang-Xu-Yao-Ars-2018}.}}
\end{figure}

\begin{figure}[h]
\centering
\includegraphics[width=15cm]{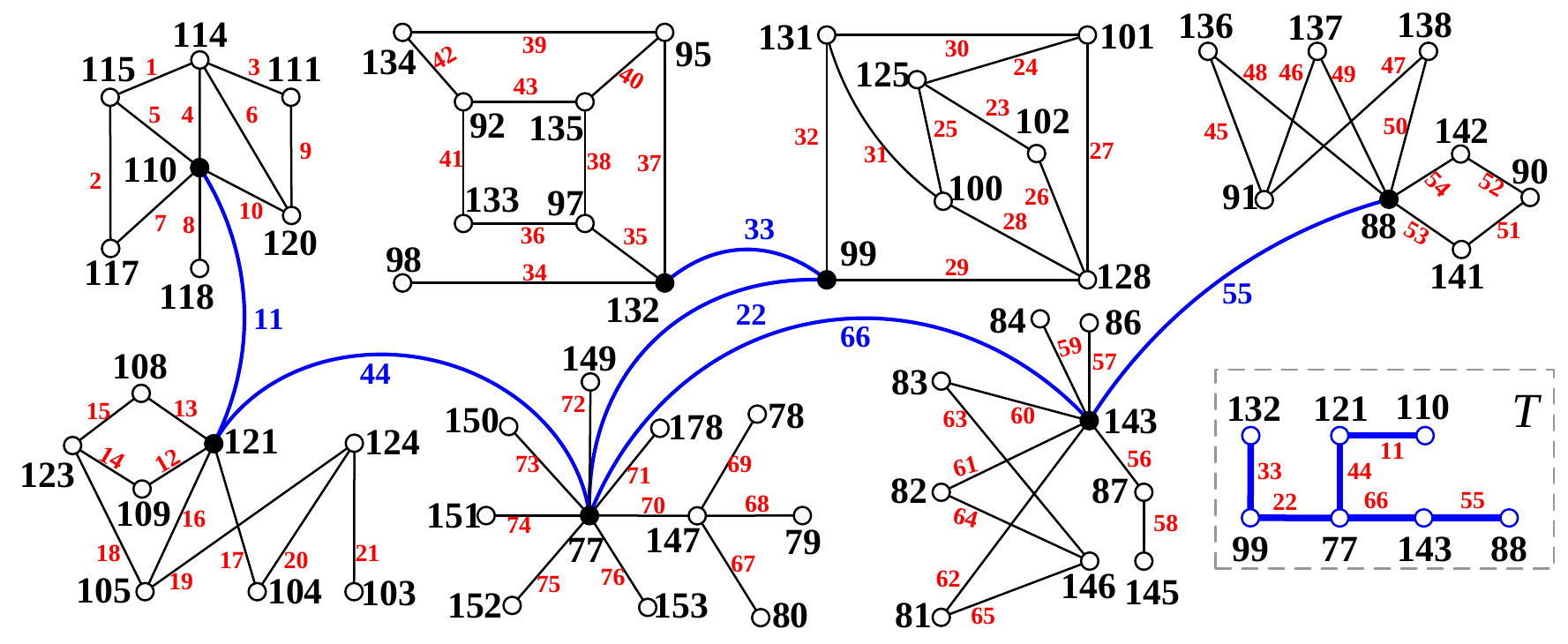}\\
\caption{\label{fig:example-generalized-total-graceful-1} {\small A graph $T\odot |^7_{i=1}H_i$ admits a super generalized total graceful labelling for illustrating Theorem \ref{thm:GTGL-main-theorem} cited from \cite{Wang-Xu-Yao-Ars-2018}, where the graphic lattice base $\textbf{H}^c$ is shown in Fig.\ref{fig:example-generalized-total-graceful-0}.}}
\end{figure}

\subsubsection{Graphic lattices on multiple operations}

\textbf{Vertex-coinciding operation. }We define an ordered graph operation ``$(\theta,\odot)$''. Each connected graph $H_i$ admitting a $W$-type total coloring $f_i:V(H_i)\cup E(H_i)\rightarrow [a, M_i]$ with $i=1,2$, where $f_1(V(H_1))\cap f_2(V(H_2))\neq \emptyset$, we do $(\theta,\odot)$ to these two graphs $H_1,H_2$ in the following process:

Step 1. Let $\theta_j$ be a transformation. Each element $w$ of $V(H_i)\cup E(H_i)$ is colored with $\theta_j(f_i(w))$, such that no two edges $uv,uw$ of the union $H_1\cup H_2$ hold $\theta_j(uv)=\theta_j(uw)$, and there are vertices $x\in V(H_1)$ and $y\in V(H_2)$ hold $\theta_j(x)=\theta_j(y)$. Here, we restrict two colorings $f_i,f_2$ and a transformation $\theta_j$ to be the same $X$-type total coloring, of course, this restriction can be deleted in some particular issues.

Step 2. Vertex-coinciding a vertex $x\in V(H_1)$ with another vertex $y\in V(H_2)$ into one $z=x\odot y$ if $\theta_j(x)=\theta_j(y)$, such that the resultant graph, denoted as $H_1\odot H_2$, is connected, and there are two vertices $s,t\in V(H_1\odot H_2)$ holding $\theta_j(s)=\theta_j(t)$ true.

Suppose that each graph-vector $H_k$ of the graphic lattice base $\textbf{H}=(H_1,H_2,\dots ,H_n)$ made by disjoint connected graphs $H_1,H_2,\dots ,H_n$ admits a $W$-type coloring $g_k$ with $k\in [1,n]$, so we say $\textbf{H}$ admitting the $W$-type coloring. Let $C_{set}$ be a set of graph colorings and labellings. Thereby, we have a \emph{$(\theta,\odot)$-graphic lattice}
\begin{equation}\label{eqa:two-operations-graphic-lattice-00}
{
\begin{split}
\textbf{\textrm{L}}(\textbf{H}(\theta,\odot)\textbf{\textrm{C}}_{set})=\{(\theta_j,\odot)^n_{k=1}a_kH_k:~a_k\in Z^0,~\theta_j\in C_{set}\}.
\end{split}}
\end{equation} Notice that $(\theta_j,\odot)^n_{k=1}a_kH_k$ in (\ref{eqa:two-operations-graphic-lattice-00}) is a set of graphs admitting the same $W$-type colorings for each fixed transformation $\theta_j\in C_{set}$. For understanding this fact, see a graph shown in Fig.\ref{fig:set-ordered-grow-1}.

\begin{thm} \label{thm:two-operations-graphic-lattices}
$^*$ Suppose that $C_{set}$ is a set containing graph colorings being equivalent set-ordered graceful labellings. Then each graph of the $(\theta,\odot)$-graphic lattice $\textbf{\textrm{L}}(\textbf{\textrm{C}}_{set}(\theta,\odot)\textbf{H})$ admits a set-ordered $W$-type coloring if the graphic lattice base $\textbf{H}^c=(H_1,H_2,\dots ,H_n)$ admits the set-ordered $X$-type coloring. (see an example shown in Fig.\ref{fig:set-ordered-grow}, Fig.\ref{fig:set-ordered-grow-0} and Fig.\ref{fig:set-ordered-grow-1})
\end{thm}

\begin{figure}[h]
\centering
\includegraphics[width=16.2cm]{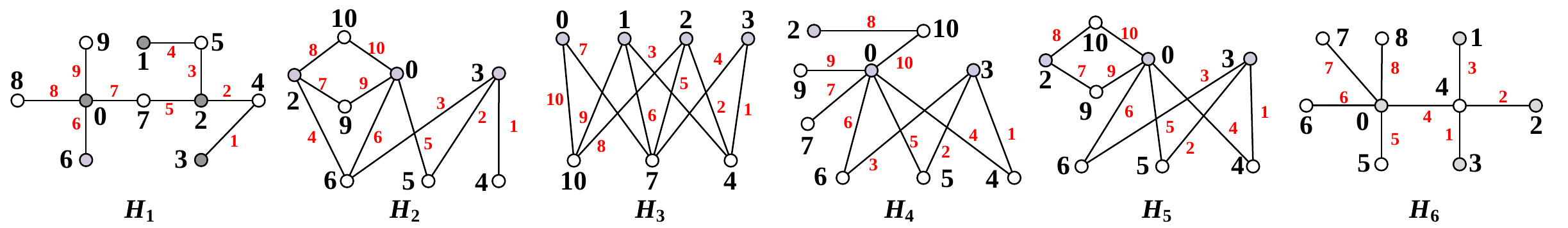}\\
\caption{\label{fig:set-ordered-grow} {\small Each connected graph $H_i$ of a base $\textbf{H}=(H_1,H_2,\dots ,H_6)$ admits a set-ordered graceful labelling $f_i$ with $i\in [1,6]$.}}
\end{figure}

\begin{figure}[h]
\centering
\includegraphics[width=16.2cm]{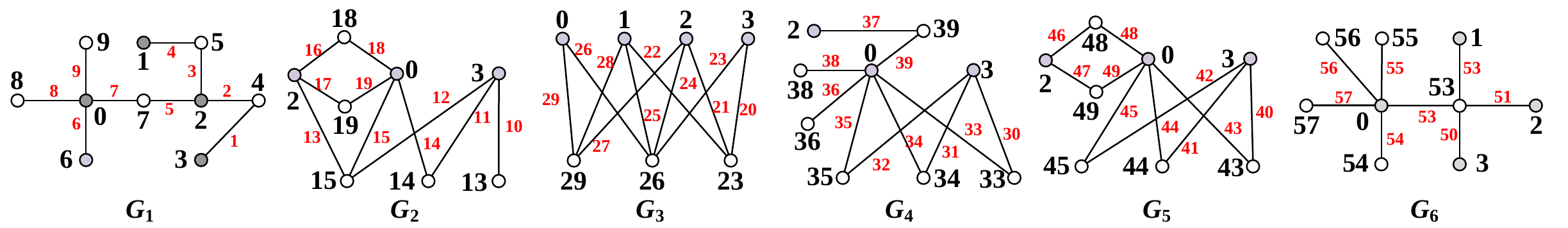}\\
\caption{\label{fig:set-ordered-grow-0} {\small Each connected graph $G_i$ admits a set-ordered labelling $g_i$ obtained by $\theta(f_i)$ with $i\in [1,6]$, each $f_i$ is shown in Fig.\ref{fig:set-ordered-grow}.}}
\end{figure}

\begin{figure}[h]
\centering
\includegraphics[width=15cm]{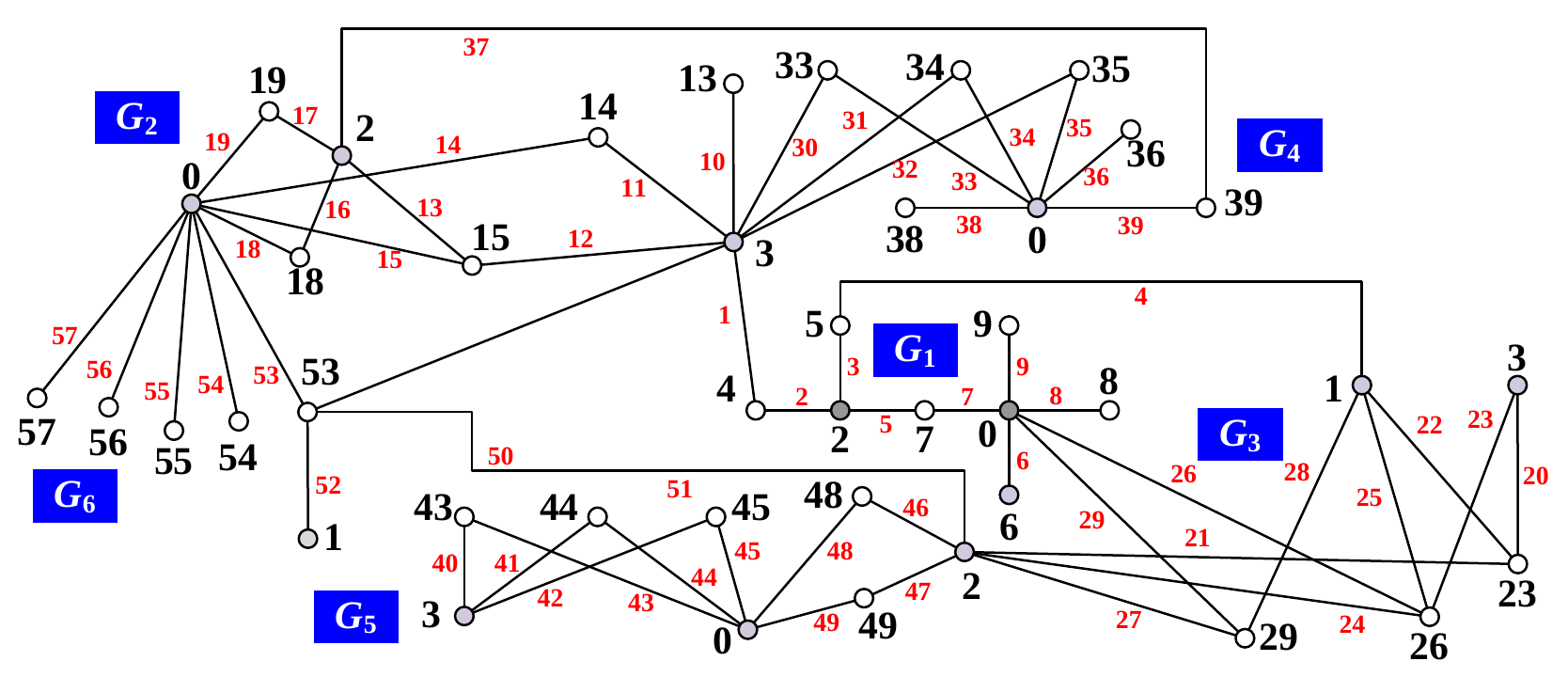}\\
\caption{\label{fig:set-ordered-grow-1} {\small A connected graph $G=(\theta,\odot)^6_{k=1}H_k$ made by doing the operation $(\theta,\odot)$ to the graphic lattice base $\textbf{H}=(H_1,H_2,\dots ,H_6)$ and graphs $G_1,G_2,\dots ,G_6$ shown in Fig.\ref{fig:set-ordered-grow} and Fig.\ref{fig:set-ordered-grow-0}.}}
\end{figure}

In Fig.\ref{fig:set-ordered-grow-1}, a connected graph $G=(\theta,\odot)^6_{k=1}H_k$ admits a set-ordered gracefully total coloring $f$ since $f(x)=f(y)$ for some two vertices $x,y\in V(G)$.

\textbf{Edge joining operation. }There is another ordered graph operation ``$(\varphi,\ominus)$'' defined as: Let a forest $T$ of $m$ vertices $x_1,x_2,\dots ,x_m$ admit a (flawed) $W$-type total labelling $f$, and let a base $\textbf{H}^c=(H_1,H_2,\dots ,H_n)$ with disjoint connected graphs $H_1,H_2,\dots ,H_n$ and $m\leq n$ admit a $W$-type labelling $g$. We do a transformation $\varphi_j$ to two labellings $f$ and $g$ of $T$ and $\textbf{H}^c$ respectively, so we get $T'~(\cong T)$ and $\textbf{H}'=(H'_1,H'_2,\dots ,H'_n)$~$(H'_i\cong H)$ admit the labelling $\varphi_j$ in common. Suppose that $\varphi_j(y_i)=\varphi_j(x_i)$ for some $y_i\in V(H'_i)$ and $x_i\in V(T')$ with $i\in [1,n]$. We join $y_i\in V(H'_i)$ with some $x_{i_j}\in V(T')$ by an edge $y_ix_{i_j}$ with $i\in [1,n]$, and such that the resulting graph, denoted as $T'\ominus ^n_{k=1}H'_k$, is connected and admits a $W$-type total labelling. The above process is written as $T(\varphi_j,\ominus)^n_{k=1}H_k$, see examples shown in Fig.\ref{fig:edge-join-0} and Fig.\ref{fig:edge-join-1}. Let $F$ be a set containing trees and forests admitting (flawed) $W$-type total labellings, and let $C_{set}$ be a set of graph labellings equivalent to set-ordered graceful labellings. So, we get a \emph{$(\varphi,\ominus)$-graphic lattice}
\begin{equation}\label{eqa:edge-2-operations-graphic-lattice}
{
\begin{split}
\textbf{\textrm{L}}(\textbf{H}^c(\varphi,\ominus)\textbf{\textrm{C}}_{set})=\{T(\varphi_j,\ominus)^n_{k=1}H_k:~\varphi_j\in C_{set},~T\in F\}
\end{split}}
\end{equation} with the graphic lattice base $\textbf{H}^c=(H_1,H_2,\dots ,H_n)$, where $C_{set}$ is a set of graph labellings/labellings.

\begin{figure}[h]
\centering
\includegraphics[width=16.2cm]{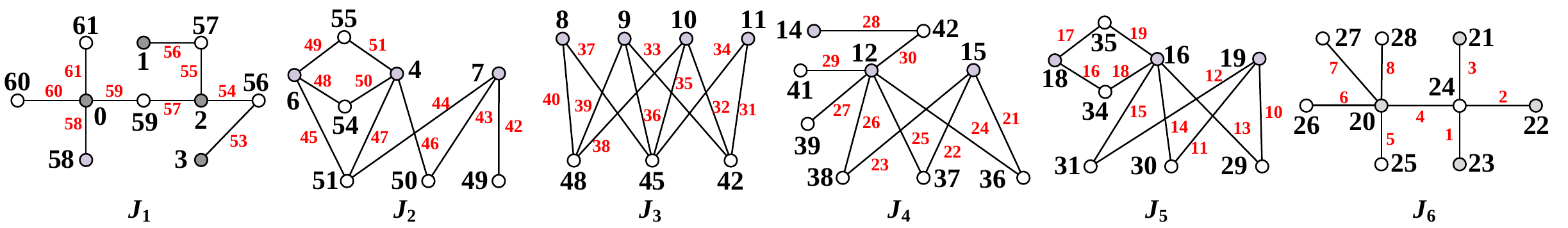}\\
\caption{\label{fig:edge-join-0} {\small Each connected graph $J_i$ admits a set-ordered labelling $h_i$ obtained by $\theta(f_i)$ with $i\in [1,6]$, each $f_i$ is shown in Fig.\ref{fig:set-ordered-grow}.}}
\end{figure}

\begin{figure}[h]
\centering
\includegraphics[width=16cm]{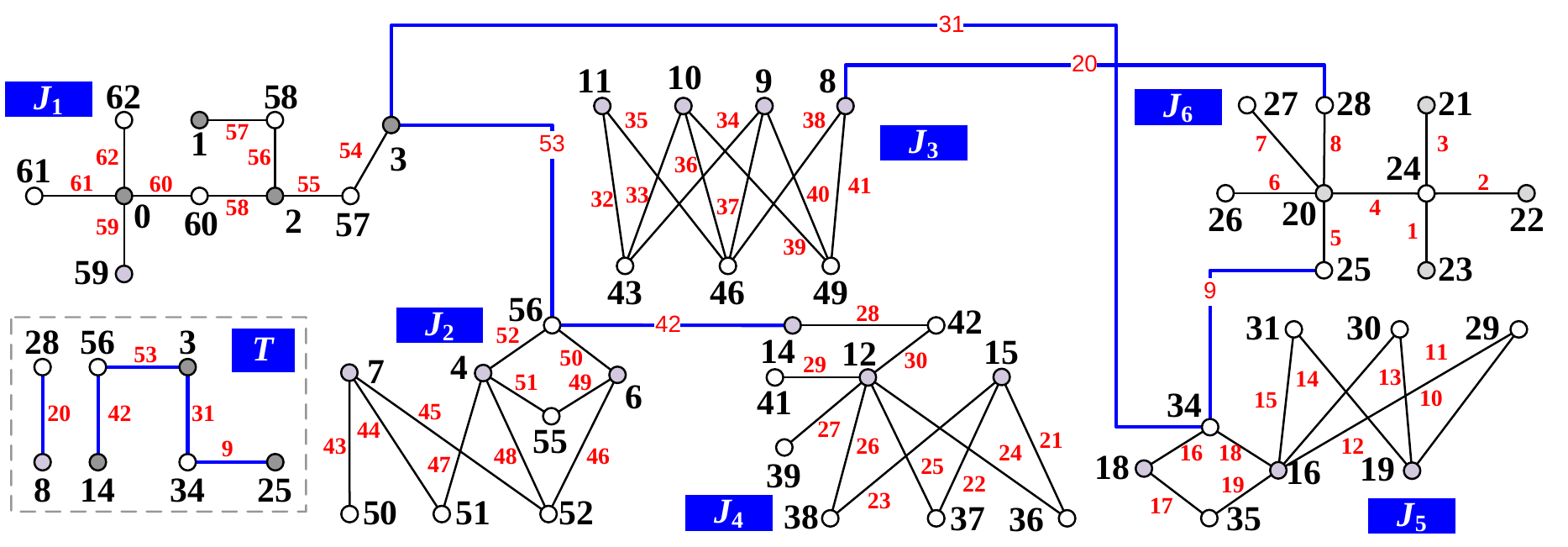}\\
\caption{\label{fig:edge-join-1} {\small A connected graph $G=T(\varphi_j,\ominus)^6_{k=1}H_k$ made by doing the operation $(\theta,\ominus)$ to the graphic lattice base $\textbf{H}^c=(H_1,H_2,\dots ,H_6)$ shown in Fig.\ref{fig:set-ordered-grow} and six graphs $J_1,J_2,\dots ,J_6$ shown in Fig.\ref{fig:edge-join-0}.}}
\end{figure}

In Fig.\ref{fig:edge-join-1}, a connected graph $G=T(\varphi_j,\ominus)^6_{k=1}H_k$ admits a set-ordered graceful labelling. Observing Fig.\ref{fig:edge-join-1} carefully, it is not hard to see that there are two or more connected graphs, like $G$, made by doing the operation $(\theta,\ominus)$ to the graphic lattice base $\textbf{H}^c=(H_1,H_2,\dots ,H_6)$ shown in Fig.\ref{fig:set-ordered-grow} and six graphs $J_1,J_2,\dots ,J_6$ shown in Fig.\ref{fig:edge-join-0} and admitting set-ordered graceful labellings.

\begin{thm} \label{thm:edge-2-operations-graphic-lattices}
$^*$ Suppose that each set-ordered $W$-type labelling of $C_{set}$ is equivalent to a set-ordered graceful labelling, and the graphic lattice base $\textbf{H}^c=(H_1,H_2,\dots ,H_n)$ admits a flawed set-ordered $W$-type labelling. Then each set $\{T(\varphi_j,\ominus)^n_{k=1}H_k\}$ contains at least a connected graph admitting a $W$-type labelling.
\end{thm}

\subsection{Graph homomorphism lattices}

Graph homomorphism lattices are like graphic lattices. Let $\textrm{H}_{\textrm{om}}(H,W)$ be the set of all $W$-type totally-colored graph homomorphisms $G\rightarrow H$. For a fixed $W_k$-type graph homomorphism, suppose that there are mutually different total colorings $g_{k,1},g_{k,2},\dots, g_{k,m_k}$ of the graph $H$ to form $W_k$-type graph homomorphisms $G\rightarrow H_{k,i}$ with $i\in [1,m_k]$, where $H_{k,i}$ is a copy of $H$ and colored by a total coloring $g_{k,i}$. Thereby, we have the sets $\textrm{H}_{\textrm{om}}(H_{k,i},W_k)$ with $i\in [1,m_k]$, and $\textrm{H}_{\textrm{om}}(H_{k},W_k)=\bigcup^{m_k}_{i=1}\textrm{H}_{\textrm{om}}(H_{k,i},W_k)$. We get a \emph{$W_k$-type totally-colored graph homomorphism lattice} as follows
\begin{equation}\label{eqa:Wk-type-graph homomorphism-lattice}
\textbf{\textrm{L}}(\textbf{\textrm{W}}_k,\textbf{\textrm{H}}^k_{\textrm{om}})=\Biggr \{\bigcup^{m_k}_{i=1}a_i(G\rightarrow H_{k,i}):a_k\in \{0,1\}; H_{k,i}\in \textrm{H}_{\textrm{om}}(H_{k},W_k)\Biggr \}
\end{equation}
with $\sum^{m_k}_{i=1}a_i=1$ and the base $\textbf{\textrm{H}}^k_{\textrm{om}}=(H_{k,i})^{m_k}_{i=1}$.

For example, as a $W_k$-type totally-colored graph homomorphism is a set-ordered gracefully graph homomorphism, $G$ admits a set-ordered graceful total coloring $f$ and $H$ admits a set-ordered gracefully total coloring $g_k$ in a set-ordered gracefully graph homomorphism $G\rightarrow H_k$. Thereby, a $W_k$-type totally-colored graph homomorphism lattice may be feasible and effective in application. In real computation, finding all of mutually different set-ordered gracefully total colorings of the graph $H$ is a difficult math problem, since there is no polynomial algorithm for this problem.

Notice that $\textrm{H}_{\textrm{om}}(H,W)=\bigcup^M_{k=1}\bigcup^{m_k}_{i=1}\textrm{H}_{\textrm{om}}(H_{k,i},W_k)$, where $M$ is the number of all $W$-type totally-colored graph homomorphisms, immediately, we get a \emph{$W$-type totally-colored graph homomorphism lattice}
\begin{equation}\label{eqa:total-type-graph-homomorphism-lattice}
{
\begin{split}
\textbf{\textrm{L}}(\textbf{\textrm{W}},\textbf{\textrm{H}}_{\textrm{om}})=\bigcup ^M_{k=1}\textbf{\textrm{L}}(\textbf{\textrm{W}}_k,\textbf{\textrm{H}}^k_{\textrm{om}})
\end{split}}
\end{equation}
with the base $\textbf{\textrm{H}}_{\textrm{om}}=((H_{k,i})^{m_k}_{i=1})^M_{k=1}$.

\subsection{Graphic lattice homomorphisms}

Let $\textrm{\textbf{G}}=(G_k)^m_{k=1}$ and $\textrm{\textbf{H}}=(H_k)^m_{k=1}$ be two \emph{bases}, and let $\theta_k:V(G_k)\rightarrow V(H_k)$ be a $W_k$-type totally-colored graph homomorphism with $k\in [1,m]$, and let $(\bullet)$ be a graph operation on graphs. Suppose that $F$ and $J$ are two sets of graphs, such that each graph $G\in F$ corresponds a graph $H\in J$, and there is a $W_k$-type totally-colored graph homomorphism $\theta_{G,H}:V(G)\rightarrow V(H)$. We have the following graphic lattices:
\begin{equation}\label{eqa:general-operation}
{
\begin{split}
\textbf{\textrm{L}}(\textbf{\textrm{F}}(\bullet)\textbf{\textrm{G}})=\left \{(\bullet)^m_{i=1}a_iG_i:~a_i\in Z^0;G_i\in \textrm{\textbf{G}}\right \},~\textbf{\textrm{L}}(\textbf{\textrm{J}}(\bullet)\textbf{\textrm{H}})=\left \{(\bullet)^m_{j=1}b_jH_j:~b_j\in Z^0;H_j\in \textrm{\textbf{H}}\right \}
\end{split}}
\end{equation} with $\sum^m_{i=1}a_i\geq 1$ and $\sum^m_{j=1}b_j\geq 1$. Let $\pi=(\bigcup^m_{k=1} \theta_k)\cup (\bigcup _{G\in F,H\in J}\theta_{G,H})$. We have a \emph{$W$-type graphic lattice homomorphism}
\begin{equation}\label{eqa:general-graphic lattice-homomorphism}
\pi:\textbf{\textrm{L}}(\textbf{\textrm{F}}(\bullet)\textbf{\textrm{G}})\rightarrow \textbf{\textrm{L}}(\textbf{\textrm{J}}(\bullet)\textbf{\textrm{H}}).
\end{equation}

In particular cases, we have: (1) The operation $(\bullet)=\ominus$ is an operation by joining some vertices $x_{k,i}$ of $G_i$ with some vertices $y_{k,j}$ of $G_j$ together by new edges $x_{k,i}y_{k,j}$ with $k\in [1,a_k]$ and $a_k\geq 1$, the resultant graph is denoted as $G_i\ominus G_j$, called \emph{edge-joined graph}. (2) The operation $(\bullet)=\odot$ is an operation by coinciding a vertex $u_{k,i}$ of $G_i$ with some vertex $v_{k,j}$ of $G_j$ into one vertex $u_{k,i}\odot v_{k,j}$ for $k\in [1,b_k]$ with integer $b_k\geq 1$, the resultant graph is denoted as $G_i\odot G_j$.

Thereby, we get an edge-joined graph $H_i\ominus H_j$ since $\theta_k(x_{k,i})\in V(H_i)$ and $\theta_k(y_{k,j})\in V(H_j)$, $\theta_k(x_{k,i}y_{k,j})\in E(H_i\ominus H_j)$ and a $W_k$-type totally-colored graph homomorphism $\theta_k:V(G_i\ominus G_j)\rightarrow V(H_i\ominus H_j)$. Similarly, we have another $W_k$-type totally-colored graph homomorphism $\phi_k:V(G_i\odot G_j)\rightarrow V(H_i\odot H_j)$. In totally, we have two $W_k$-type totally-colored graph homomorphisms
\begin{equation}\label{eqa:c3xxxxx}
{
\begin{split}
\theta:\ominus^m_{i=1} V(G_i)\rightarrow \ominus^m_{i=1} V(H_i),\quad \phi:\odot^m_{i=1} V(G_i)\rightarrow \odot^m_{i=1} V(H_i).
\end{split}}
\end{equation}

We have two graphic lattices based on the operation ``$\ominus$'':
\begin{equation}\label{eqa:graphic lattice-homomorphism-11}
\textbf{\textrm{L}}(\ominus\textbf{\textrm{G}})=\left \{\ominus ^m_{k=1}a_kG_k:~a_k\in Z^0;G_k\in \textrm{\textbf{G}}\right \},~\textbf{\textrm{L}}(\ominus\textbf{\textrm{H}})=\left \{\ominus ^m_{k=1}b_kH_k:~b_k\in Z^0;H_k\in \textrm{\textbf{H}}\right \}
\end{equation}
with $\sum^m_{k=1}a_k\geq 1$ and $\sum^m_{k=1}b_k\geq 1$.

The above works enable us to get a homomorphism $\theta:\textbf{\textrm{L}}(\ominus\textbf{\textrm{G}})\rightarrow \textbf{\textrm{L}}(\ominus\textbf{\textrm{H}})$, called \emph{$W$-type graphic lattice homomorphism}. Similarly, we have another $W$-type graphic lattice homomorphism $\pi':\textbf{\textrm{L}}(\odot\textbf{\textrm{G}})\rightarrow \textbf{\textrm{L}}(\odot\textbf{\textrm{H}})$ by the following two graphic lattices based on the operation ``$\odot$''
\begin{equation}\label{eqa:graphic lattice-homomorphism-33}
\textbf{\textrm{L}}(\odot\textbf{\textrm{G}})=\left \{\odot ^m_{k=1}c_kG_k:~c_k\in Z^0;G_k\in \textrm{\textbf{G}}\right \},~\textbf{\textrm{L}}(\odot\textbf{\textrm{H}})=\left \{\odot ^m_{k=1}d_kH_k:~d_k\in Z^0;H_k\in \textrm{\textbf{H}}\right \}
\end{equation}
with $\sum^m_{k=1}c_k\geq 1$, and $\sum^m_{k=1}d_k\geq 1$.

Notice that there are mixed operations of the operation ``$\ominus$'' and the operation ``$\odot$'', so we have more complex $W$-type graphic lattice homomorphisms. If two \emph{bases} $\textbf{\textrm{G}}_{\textrm{group}}=\{F_f(G);\oplus\}$ and $\textbf{\textrm{H}}_{\textrm{group}}=\{F_h(H)$; $\oplus\}$ are two every-zero graphic groups, so we have an every-zero graphic group homomorphism $\varphi:\textbf{\textrm{G}}_{\textrm{group}}\rightarrow \textbf{\textrm{H}}_{\textrm{group}}$ and two \emph{every-zero graphic group homomorphisms}:
$$\textbf{\textrm{L}}(\ominus\textbf{\textrm{G}}_{\textrm{group}})\rightarrow \textbf{\textrm{L}}(\ominus\textbf{\textrm{H}}_{\textrm{group}}),\quad \textbf{\textrm{L}}(\odot\textbf{\textrm{G}}_{\textrm{group}})\rightarrow \textbf{\textrm{L}}(\odot\textbf{\textrm{H}}_{\textrm{group}}).$$

\subsection{Dynamic graph lattices}
At time step $t$, let $F_{p,q}(t)$ be a \emph{dynamic set} of graphs $G(t)$ with $r~(\leq p)$ vertices and $s~(\leq q)$ edges and let $\textbf{\textrm{H}}(t)=(H_k(t))^m_{k=1}$ be a \emph{dynamic base}, where each $H_k(t)$ admits a $W_k(t)$-type coloring $f_{k,t}$ and $H_k(t_r)\cong H_k(t_s)$ for $t_r\neq t_s$. By the vertex-coinciding operation ``$\odot$'', we have a \emph{dynamic graph lattice}
\begin{equation}\label{eqa:dynamic-graph-lattice}
\textbf{\textrm{L}}(\textbf{\textrm{H}}\odot F_{p,q})(t)=\{G(t)\odot ^m_{i=1}a_iH_i(t):~a_i(t)\in Z^0,~H_i(t)\in \textbf{\textrm{H}}(t),~G(t)\in F_{p,q}(t)\}.
\end{equation}
with $\sum ^m_{i=1}a_i(t)\geq 1$.

\begin{rem}\label{rem:dynamic-graph-lattice}
For $t_r\neq t_s$, it is not hard to understand $\textbf{\textrm{L}}(\textbf{\textrm{H}}\odot F_{p,q})(t_r)\neq \textbf{\textrm{L}}(\textbf{\textrm{H}}\odot F_{p,q})(t_s)$, in general, since it is related with $G(t_r)\neq G(t_s)$ or $f_{k,t_r}\neq f_{k,t_s}$. We can use a dynamic graph lattice $\textbf{\textrm{L}}(\textbf{\textrm{H}}\odot F_{p,q})(t)$ to encrypt a dynamic network $N(t)$ by the method ``$N(t)\odot ^m_{i=1}a_iH_i(t)$''. Also, we use a graph operation ``$(\diamond)$'' to a \emph{general dynamic graph lattice}
\begin{equation}\label{eqa:general-dynamic-graph-lattice}
\textbf{\textrm{L}}(\textbf{\textrm{H}}(\diamond)F_{p,q})(t)=\{G(t)(\diamond)^m_{i=1}a_iH_i(t):~a_i(t)\in Z^0,~H_i(t)\in \textbf{\textrm{H}}(t),~G(t)\in F_{p,q}(t)\}.
\end{equation}
with $\sum ^m_{i=1}a_i(t)\geq 1$, where each $H_k(t)$ of the base $\textbf{\textrm{H}}(t)=(H_k(t))^m_{k=1}$ admits a $W_k(t)$-type coloring $f_{k,t}$, and it is allowed that $f_{k,t_r}=f_{k,t_s}$, or $H_k(t_r)\not \cong H_k(t_s)$ for some $t_r\neq t_s$.\qqed
\end{rem}

\subsection{Network lattices made by communities}
Suppose that each community-vector $C_k(t)$ of a network lattice base $\textbf{C}_{o}(t)=(C_1(t)$, $C_2(t)$, $\dots $, $C_n(t))$ made by disjoint connected networks $C_1(t)$, $C_2(t)$, $\dots $, $C_n(t)$ at time step $t$, where each community-vector $C_k(t)$ obeys a power-law $P_k$ with $k\in [1,n]$, so we say $\textbf{C}_{o}(t)$ obeys a $(P_k)^n_1$-degree distribution. Let $G_{set}(t)$ be a set of connected graphs. At time step $t$, we have a \emph{$(\triangleleft \ominus)$-graphic lattice}
\begin{equation}\label{eqa:two-operations-graphic-lattice}
{
\begin{split}
\textbf{\textrm{L}}(\textbf{\textrm{C}}_{o}(t)(\triangleleft \ominus)\textbf{\textrm{G}}_{set}(t))=\{G(t)(\triangleleft \ominus)^n_{k=1}a_kC_k(t):~a_k\in Z^0,~G(t)\in G_{set}(t)\},
\end{split}}
\end{equation} with $\sum^n_{k=1} a_k\geq 1$ and based on the \emph{network lattice base} $\textbf{C}_{o}(t)$ with a group $P$ of power-law degree distributions, where the operation ``$(\triangleleft \ominus)$'' is a process of substituting two ends $u$ and $v$ of each edge $uv$ of $G(t)$ by two communities $C_k(t)$ and $C_j(t)$, and joining some vertices of $C_k(t)$ with some vertices of $C_j(t)$ by new edges.

In \cite{Wang-Fei-Ma-Yao-ITNEC-2020}, Wang \emph{et al.} introduced the multiple probabilistic community network model $N^c(t)$ with $n$ communities $N^c_{i}(t)$ for $i\in [1,n]$ and each community $N^c_{i}(t)$ has a probability $p_i$ in an algorithm: Add new vertices $u_i$ or new graphs $G_i$ into a community $N^c_{i}(t-1)$, use a probability $p_i$ to select object-vertices $u_{i,j}\not \in V(N^c_{i}(t-1))$ or object-edges $u_{i,j}v_{i,j}\not \in E(N^c_{i}(t-1))$ from other communities $N^c_{j}(t-1)$ with $i\neq j$, and put these $u_{i,j}$ and $u_{i,j}v_{i,j}$ into a set $F_i$, and then do an operator $O^v_i$ consisted of $u_{i,j}$ and a graph operation, or an operator $O^e_i$ consisted of $G_i$ and a graph operation to an element of $F_i$. After each of $F_i$ is implemented an operator, the result is just $N^c(t)$, called \emph{multiple probabilistic community network model} (Ref. \cite{Wang-Yao-AIP-Advances-2020}).

%\section{Star-type graphic lattices}
%%\input{3-section/Star-type-lattices-3}

\section{Star-graphic lattices}

Star-graphic lattices will be made by the leaf-coinciding operation, the leaf-splitting operation and the star-type ice-flower systems made by four particular proper total colorings defined in Definition \ref{defn:combinatoric-definition-total-coloring}.

\subsection{Leaf-coinciding and leaf-splitting operations}

We will define two operations, called the colored leaf-coinciding operation ``$G_1\overline{\ominus} G_2$'', and the colored leaf-splitting operation ``$xy\prec$'', respectively. These two operations will help us to realize many useful and important star-graphic lattices.

\subsubsection{Colored leaf-coinciding operation}

We defined the so-called colored leaf-coinciding operation ``$G_1\overline{\ominus} G_2$'' on two disjoint colored graphs $G_1$ and $G_2$, each $G_i$ admits a $W$-type proper total coloring $f_i$ with $i=1,2$; and $G_1$ has an edge $u_1v_1$ with $\textrm{deg}_{G_1}(u_1)\geq 2$ and $\textrm{deg}_{G_1}(v_1)=1$; and $G_2$ has an edge $u_2v_2$ with $\textrm{deg}_{G_2}(u_2)=1$ and $\textrm{deg}_{G_2}(v_2)\geq 2$. Suppose that $f_1(u_1v_1)=f_2(u_2v_2)$, $f_1(u_1)=f_2(u_2)$ and $f_1(v_1)=f_2(v_2)$. We leaf-coincide the edge $u_1v_1$ with the edge $u_2v_2$ into one edge $xy=u_1v_1\overline{\ominus} u_2v_2$ with $x=u_1\odot u_2$ and $y=v_1\odot v_2$, the resulting graph is denoted as $G_1\overline{\ominus} G_2$ with a $W$-type proper total coloring $f=f_1\overline{\ominus} f_2$, where $f(w)=f_1(w)$ for each element $w\in [V(G_1)\cup E(G_1)]\setminus \{u_1,v_1,u_1v_1\}$, $f(z)=f_2(z)$ for $z\in [V(G_2)\cup E(G_2)]\setminus \{u_2,v_2,u_2v_2\}$, and $f(xy)=f_1(u_1v_1)=f_2(u_2v_2)$, $f(x)=f_1(u_1)=f_2(u_2)$ and $f(y)=f_1(v_1)=f_2(v_2)$.

\subsubsection{Colored leaf-splitting operation}
In the graph ``$H=G_1\overline{\ominus} G_2$'' admitting a $W$-type proper total coloring $f$, we leaf-split the edge $xy$ into two edges $u_1v_1$ and $u_2v_2$, so $H$ is leaf-split into $H(xy\prec)$ with two disjoint colored graphs $G_1$ and $G_2$, such that $G_i$ admits a $W$-type proper total coloring $f_i$ with $i=1,2$, and each $f_i$ is defined as that in the colored leaf-coinciding operation above. In a connected graph $G$ admitting a $W$-type coloring $f$, do a leaf-split operation to an edge $xy$ of $G$: First, remove the edge $xy$ from $G$, the resultant graph is denoted as $G-xy$, and then add two new vertices $x'$ and $y'$ to $G-xy$, and join $x'$ with the vertex $y$ of $G-xy$ by an edge $x'y$, and join $y'$ with the vertex $x$ of $G-xy$ by an edge $xy'$, and we color $f(x'y)=f(xy')=f(xy)$, $f(x')=f(y)$ and $f(y')=f(x)$. Hence, we get the leaf-splitting graph $G(xy\prec)$.

\subsubsection{Ice-flower systems}
\textbf{Uncolored ice-flower system.} An \emph{ice-flower system} consists of stars $K_{1,m_i}$ and the leaf-coinciding operation. A \emph{star} $K_{1,m_i}$ with $m_i\geq 2$ is a complete bipartite connected graph with its vertex set $V(K_{1,m_i})=\{x_i,y_{i,j}:j\in [1,m_i]\}$ and its edge set $V(K_{1,m_i})=\{x_iy_{i,j}:j\in [1,m_i]\}$, so $K_{1,m_i}$ has $m_i$ leaves $y_{i,1},y_{i,2},\dots ,y_{i,m_i}$ and a unique non-leaf vertex $x_i$ with degree $\textrm{deg}_{K_{1,m_i}}(v)=m_i$. We call $n$ disjoint stars $K_{1,m_1},K_{1,m_2},\dots, K_{1,m_n}$ with $n\geq 1$ and $m_i\geq 2$ as an \emph{ice-flower system}, denoted as
\begin{equation}\label{eqa:ice-flower-base}
\textbf{\textrm{K}}=(K_{1,m_1},K_{1,m_2},\dots, K_{1,m_n})=(K_{1,m_j})^n_{j=1}
\end{equation}

The leaf-coinciding operation ``$K_{1,m_i}\overline{\ominus} K_{1,m_j}$'' on two stars $K_{1,m_i}$ and $K_{1,m_j}$ of an ice-flower system $\textbf{\textrm{K}}$ defined in (\ref{eqa:ice-flower-base}) is defined as doing a leaf-coinciding operation to an edge $x_iy_{i,a}$ of $K_{1,m_i}$ and an edge $x_jy_{j,b}$ of $K_{1,m_j}$, such that two edges $x_iy_{i,a}$ and $x_jy_{j,b}$ are coincided one edge $x_ix_j=x_iy_{i,a}\overline{\ominus} x_jy_{j,b}$, $x_i=x_i\odot y_{j,b}$ and $x_j=x_j\odot y_{i,a}$. Speaking simply, removing vertices $y_{i,a}$ and $y_{j,b}$ from $K_{1,m_i}$ and $K_{1,m_j}$ respectively, and adding a new edge, denoted as $x_ix_j$, joining $x_i$ with $x_j$ together by an edge. Thereby, an uncolored ice-flower system $\textbf{\textrm{K}}$ holds the leaf-coinciding operation ``$K_{1,m_i}\overline{\ominus} K_{1,m_j}$'' for any pair of stars $K_{1,m_i}$ and $K_{1,m_j}$, we call $\textbf{\textrm{K}}$ a \emph{strongly uncolored ice-flower system}.

\textbf{Colored ice-flower system.} If each star $K_{1,m_i}$ of an ice-flower system $\textbf{\textrm{K}}$ defined in (\ref{eqa:ice-flower-base}) admits a $W_i$-type coloring $g_i$, we call $\textbf{\textrm{K}}$ a \emph{colored ice-flower system}, and rewrite it by \begin{equation}\label{eqa:colored-ice-flower-base}
\textbf{\textrm{K}}^c=(K^c_{1,m_1},K^c_{1,m_2},\dots, K^c_{1,m_n})=(K^c_{1,m_j})^n_{j=1}
\end{equation} for distinguishing. If each pair of stars $K^c_{1,m_i}$ and $K^c_{1,m_j}$ holds the leaf-coinciding operation ``$K^c_{1,m_i}\overline{\ominus} K^c_{1,m_j}$'', we call $\textbf{\textrm{K}}^c$ a \emph{strongly colored ice-flower system}. Notice that a colored ice-flower system may be not strongly, in general.

\subsection{Graceful-difference star-graphic lattices}

By Definition \ref{defn:combinatoric-definition-total-coloring}, we will make the following particular ice-flower systems: An \emph{edge-magic ice-flower system} $I_{ce}(E_{1,n}M_k)^{2n+3}_{k=1}$; a \emph{graceful-difference ice-flower system} $I_{ce}(G_{1,n}D_k)^{2n+3}_{k=1}$; the \emph{edge-difference ice-flower system} $I_{ce}(E_{1,n}D_k)^{2n+3}_{k=1}$; two \emph{felicitous-difference ice-flower systems} $I_{ce}(F_{1,n}D_k)^{2n}_{k=1}$ and $I_{ce}(SF_{1,n}D_k)^{n}_{k=1}$. Some examples about ice-flower systems are shown in Figs.\ref{fig:Dsaturated-graceful-difference-even}, \ref{fig:Dsaturated-graceful-difference-odd}, \ref{fig:Dsaturated-edge-difference-even}, \ref{fig:Dsaturated-edge-difference-odd}, \ref{fig:Dsaturated-felicitous-difference-even}, \ref{fig:Dsaturated-felicitous-difference-odd}, \ref{fig:Dsaturated-edge-magic-even} and \ref{fig:Dsaturated-edge-magic-odd}.

In the following discussion, a star $K_{1,n}$ has its own vertex set $V(K_{1,n})=\{x_0,x_i:i\in [1,n]\}$ and edge set $E(K_{1,n})=\{x_0x_i:i\in [1,n]\}$. Let $V\left (K^{(k)}_{1,n}\right )=\big \{x^k_0$, $x^k_1$, $x^k_2$, $\dots ,x^k_{n}\big \}$ and $E\left (K^{(k)}_{1,n}\right )=\left \{x^k_0x^k_j:j\in[1,n]\right \}$ be the vertex set and edge set of $k$th copy $K^{(k)}_{1,n}$ of $K_{1,n}$ with integer $k\geq 1$, respectively.

\subsubsection{Graceful-difference ice-flower systems}

A proper total coloring $\theta_k$ of a star $K_{1,n}$ is defined as: $\theta_k(x_0)=m\in [1,3n]$, $\theta_k(x_j)=m_j$ and $\theta_k(x_0x_j)=L-|m-m_j|$ with $j\in[1,n]$ and $m\in [1,3n]$, so $\big ||\theta_k(x_0)-\theta_k(x_j)|-\theta_k(x_0x_j)\big |=L$ for each edge $x_0x_j$ of $K_{1,n}$, such that $\theta_k$ is just a graceful-difference proper total coloring of $K_{1,n}$, denoted this colored star as $K^{(k)}_{1,n}=LG_{1,n}D_k$. For each fixed $m$, there are $a(m)$ groups of integers $m_1,m_2,\dots ,m_n$ from $[1,3n]$ holding $\theta_k$ to be a graceful-difference proper total coloring of $K_{1,n}$, then we get $n_{gdt}$ colored stars $LG_{1,n}D_k$ in total, where $n_{gdt}=\sum^{3n}_{m=1}a(m)$, and these colored stars form an \emph{$L$-magic graceful-difference ice-flower system} $I_{ce}(LG_{1,n}D_k)^{n_{gdt}}_{k=1}$.

In a particular graceful-difference ice-flower system $I_{ce}(G_{1,n}D_k)^{2n+3}_{k=1}$ based on the graceful-difference proper total coloring defined in Definition \ref{defn:combinatoric-definition-total-coloring}, and each $G_{1,n}D_k$ is a star $K^{(k)}_{1,n}$ admitting a graceful-difference proper total coloring $f_k$ with $k\in [1,2n+3]$. We define each graceful-difference proper total coloring $f_k$ as follows:

Case GD-1. $n=2m$. We set $f_k$ in the following two parts: (GD-1-1) $f_k(x^k_0)=k$ with $k\in [1,2m]$, $f_k(x^k_j)=4m+4-j$ with $j\in [1,2m]$, and $f_k(x^k_0x^k_j)=f_k(x^k_j)-f_k(x^k_0)=4m+4-j-i$ with $1\leq i,j\leq 2m$. So, $\big ||f_k(x^k_j)-f_k(x^k_0)|-f_k(x^k_0x^k_j)\big |=0$. If $f_k(x^k_0)=f_k(x^k_0x^k_{j'})$ happens for some $j'\in [1,2m]$ and $k\in [1,2m]$, then we recolor the edge $x^k_0x^k_{j'}$ and the vertex $x^k_{j'}$ with $f_k(x^k_0x^k_{j'})=f_k(x^k_0x^k_{2m})-1=2m+1-k$ and $f_k(x^k_{j'})=f_k(x^k_0)+f_k(x^k_0x^k_{j'})-1=2m+1$, respectively.

(GD-1-2) Set $f_k(x^k_0)=k$ with $k\in [2m+1,4m+3]$, $f_k(x^k_j)=j$ with $j\in [1,2m]$, and $f_k(x^k_0x^k_j)=f_k(x^k_0)-f_k(x^k_j)=k-j$. So, $\big ||f_k(x^k_j)-f_k(x^k_0)|-f_k(x^k_0x^k_j)\big |=0$. If we meet $f_k(x^k_0)=f_k(x^k_0x^k_{j'})$ for some $j'\in [1,2m]$ and $k\in [2m+1,4m+3]$, we reset $f_k(x^k_0x^k_{j'})=f_k(x^k_0x^k_{2m})-1=k-(2m+1)$ and $f_k(x^k_{j'})=f_k(x^k_0)-f_k(x^k_0x^k_{j'})+1=2m+1$.

Case GD-2. $n=2m+1$. We make $f_k$ in the following two parts: (GD-2-1) $f_k(x^k_0)=k$ with $k\in [1,2m+1]$, $f_k(x^k_j)=4m+5-j$ with $j\in [1,2m+1]$, and $f_k(x^k_0x^k_j)=4m+5-j-i$ with $1\leq i,j\leq 2m+1$. If we meet $2f_k(x^k_0)=f_k(x^k_{j'})$ for some $j'\in [1,2m+1]$ and $k\in [1,2m+1]$, then we recolor the edge $x^k_0x^k_{j'}$ and the vertex $x^k_{j'}$ by $f_k(x^k_0x^k_{j'})=f_k(x^k_0x^k_{2m})-1=2m+4-k$, $f_k(x^k_{j'})=f_k(x^k_0)+f_k(x^k_0x^k_{j'})=2m+4$, respectively.

(GD-2-2) Set $f_k(x^k_0)=k$ with $k\in [2m+1,4m+5]$, $f_k(x^k_j)=j$ with $j\in [1,2m+1]$, and $f_k(x^k_0x^k_j)=f_k(x^k_0)-f_k(x^k_j)=k-j$. So, $\big ||f_k(x^k_j)-f_k(x^k_0)|-f_k(x^k_0x^k_j)\big |=0$. If $f_k(x^k_0)=f_k(x^k_0x^k_{j'})$ for some $j'\in [1,2m+1]$ and $k\in [2m+1,4m+5]$, we reset $f_k(x^k_0x^k_{j'})=f_k(x^k_0x^k_{2m})-1=k-(2m+2)$ and $f_k(x^k_{j'})=f_k(x^k_0)-f_k(x^k_0x^k_{j'})=2m+2$. Here, if the case $k-(2m+2)\leq 0$ happens, so we set $f_k(x^k_{j'})=4m+5$ and $f_k(x^k_0x^k_{j'})=4m+5-f_k(x^k_0)=4m+5-k$, and $4m+5-k\neq k-j$, otherwise $2[k-(2m+2)]=1+j>0$, a contradiction.

\begin{figure}[h]
\centering
\includegraphics[width=15.6cm]{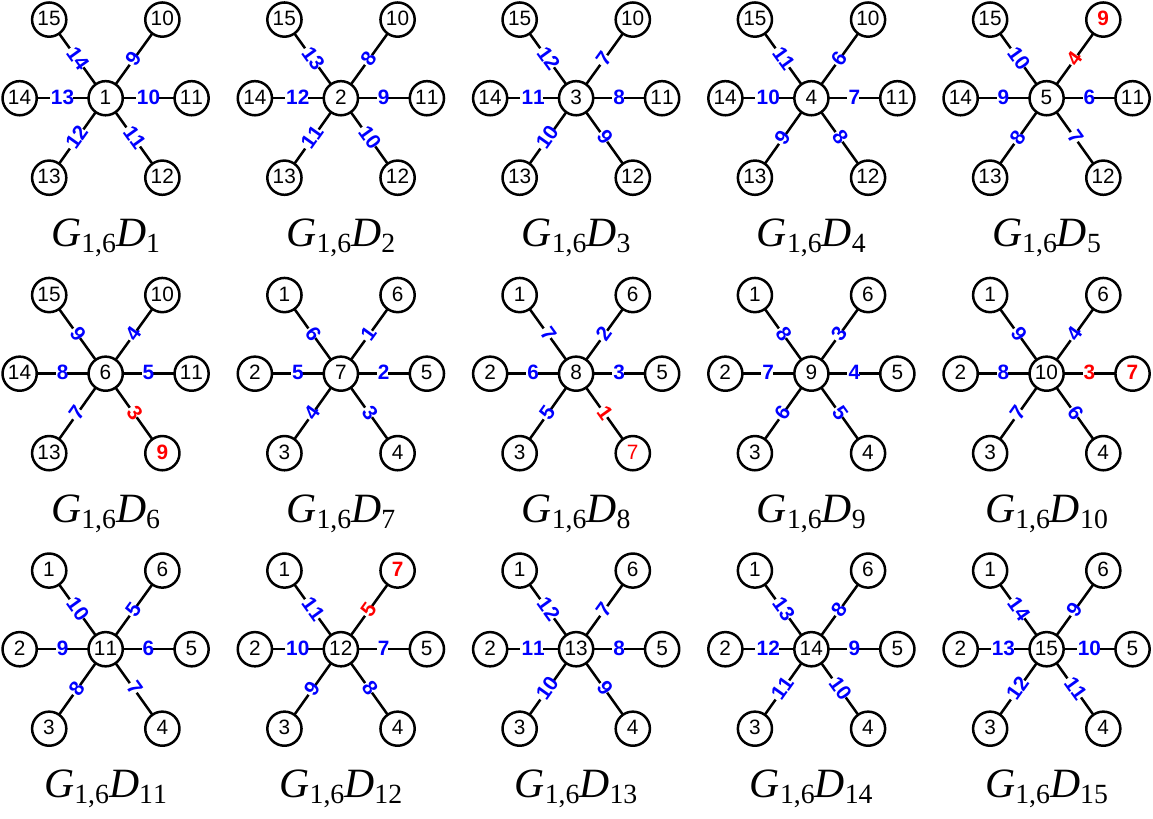}
\caption{\label{fig:Dsaturated-graceful-difference-even}{\small A graceful-difference ice-flower system $I_{ce}(G_{1,6}D_k)^{15}_{k=1}$.}}
\end{figure}

\begin{figure}[h]
\centering
\includegraphics[width=16.4cm]{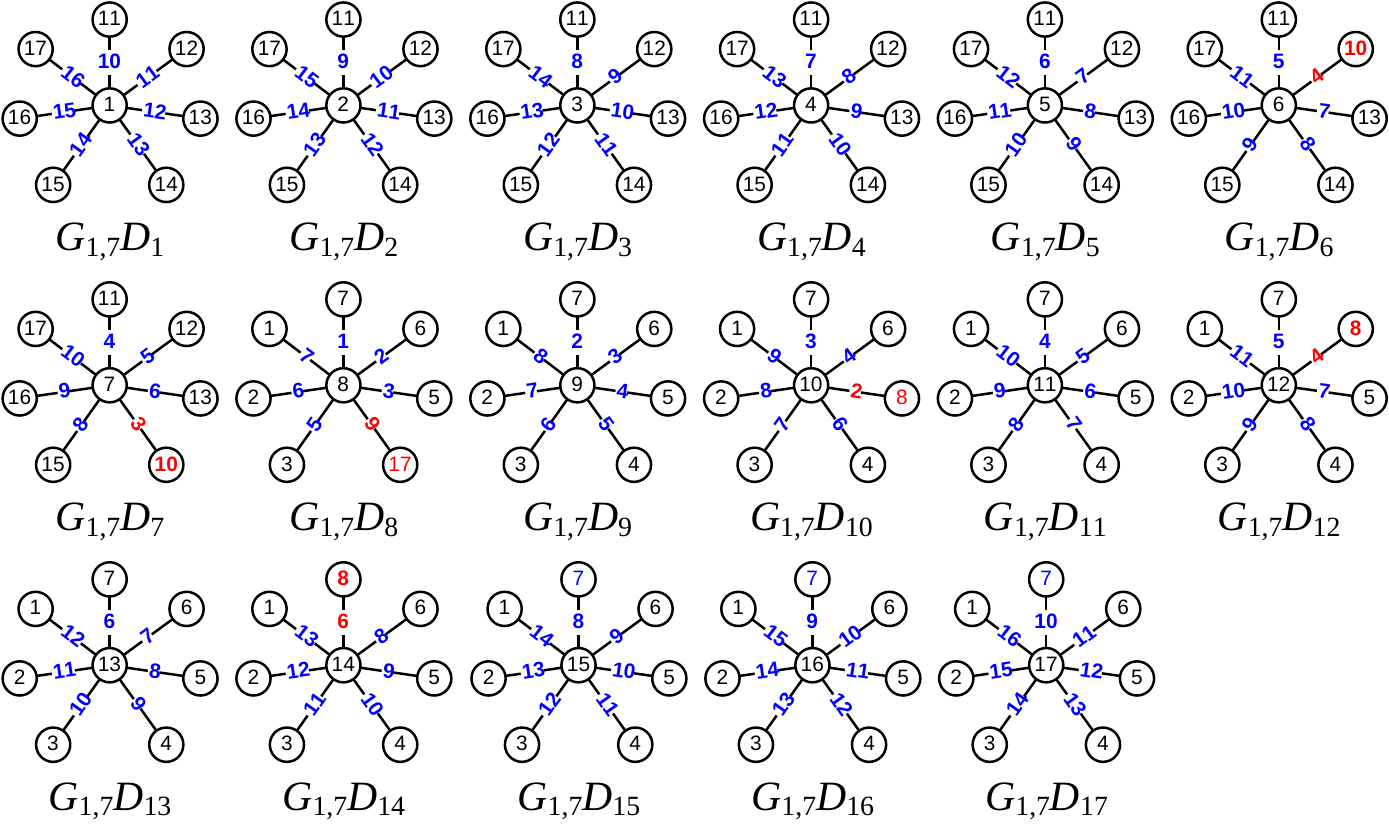}
\caption{\label{fig:Dsaturated-graceful-difference-odd}{\small A graceful-difference ice-flower system $I_{ce}(G_{1,7}D_k)^{17}_{k=1}$.}}
\end{figure}

We define a particular graceful-difference ice-flower system $I_{ce}(SG_{1,n}D_k)^{n}_{k=1})^{3n}_{k=2n+1}$ as: (pgd-1) Each star $SG_{1,n}D_k$~$\left (\cong K^{(k)}_{1,n}\right )$ admits a graceful-difference proper total coloring $\gamma_k(x^k_0)=k\in [1,n]$, $\gamma_k(x^k_j)=3n+1-j$ and $\gamma_k(x^k_0x^k_j)=3n+1-j-k$ with $j\in [1,n]$; (pgd-2) $\gamma_k(x^k_0)=k\in [2n+1,3n]$, $\gamma_k(x^k_j)=j$ and $\gamma_k(x^k_0x^k_j)=k-j$ with $j\in [1,n]$. Thereby, we have $\big ||\gamma_k(x^k_j)-\gamma_k(x^k_0)|-\gamma_k(x^k_0x^k_j)\big |=0$, and $\max\{\gamma_k(w):w\in V(SG_{1,n}D_k)\cup E(SG_{1,n}D_k)\}=3n$. For $n$ stars $SG_{1,n}D_k$ with $k\in [1,n]$, next, we do the vertex-coinciding operation to $n$ vertices $x^1_j,x^2_j,\dots ,x^n_j$ of these stars into one $y_j=x^1_j\odot x^2_j\odot \dots \odot x^n_j$ with $j\in [1,n]$, the resultant graph is just a complete bipartite graph $K_{n,n}$, immediately, we have

\begin{lem} \label{thm:graceful-difference-lemma}
$\chi''_{gdt}(K_{n,n})\leq 3n$ and $\chi''_{gdt}(G)\leq 3\max \{|X|,|Y|\}$ for each bipartite graph $G$ with vertex bipartition $(X,Y)$.
\end{lem}

\begin{thm}\label{thm:mixed-difference-coloring-trees}
\cite{Wang-Su-Yao-mixed-difference-2019} If $T$ is a tree with maximum degree $\Delta$, then $\chi''_{gdt}(T)\leq 2\Delta+3$.
\end{thm}

\begin{figure}[h]
\centering
\includegraphics[width=15cm]{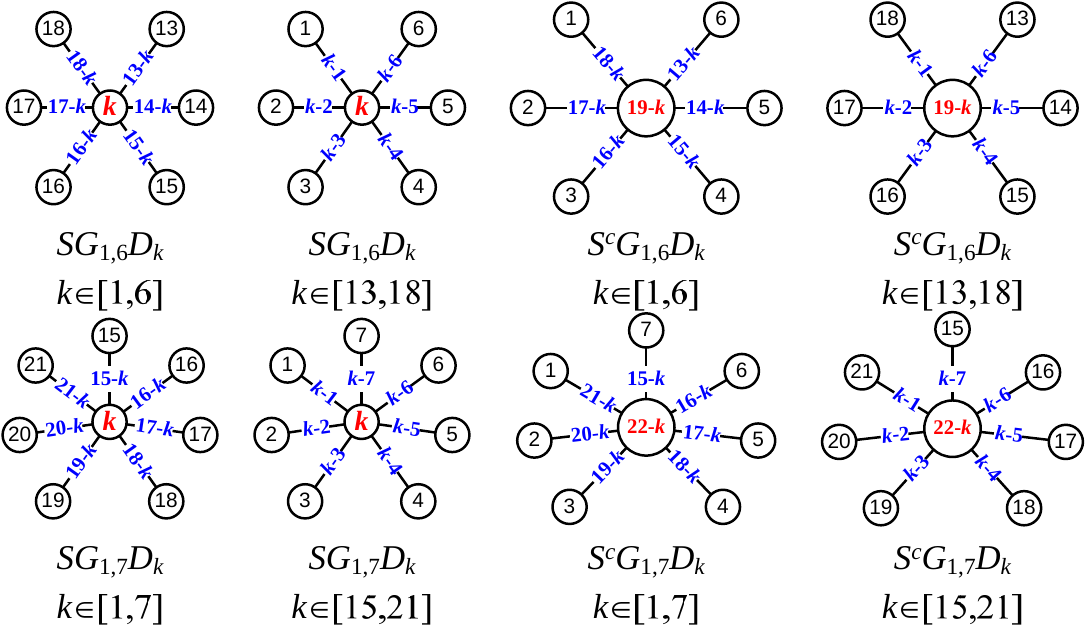}
\caption{\label{fig:graceful-differ-3n}{\small Four particular graceful-difference ice-flower systems.}}
\end{figure}

In Fig.\ref{fig:graceful-differ-3n}, we can see two particular graceful-difference ice-flower systems $I_{ce}(SG_{1,6}D_k)^{6}_{k=1})^{18}_{k=2n+1}$ and $I_{ce}(SG_{1,7}D_k)^{7}_{k=1})^{21}_{k=2n+1}$, as well as their dual graceful-difference ice-flower systems
$$I_{ce}(S^cG_{1,6}D_k)^{6}_{k=1})^{18}_{k=2n+1},\quad I_{ce}(S^cG_{1,7}D_k)^{7}_{k=1})^{21}_{k=2n+1}.$$

\subsubsection{Graceful-difference star-graphic lattices}

Each graceful-difference ice-flower system $I_{ce}(G_{1,n}D_k)^{2n+3}_{k=1}$ distributes us a \emph{graceful-difference star-graphic lattice} as follows:
\begin{equation}\label{eqa:graceful-difference-stars-lattice}
\textrm{\textbf{L}}(\overline{\ominus} \textbf{I}_{ce}(GD)) =\left \{\overline{\ominus}^{2n+3}_{i=1}a_iG_{1,n}D_i: a_i\in Z^0, G_{1,n}D_i\in I_{ce}(G_{1,n}D_k)^{2n}_{k=1}\right \}
\end{equation} with $\sum^{2n+3}_{i=1} a_i\geq 1$ and the base is $\textbf{I}_{ce}(GD)=I_{ce}(G_{1,n}D_k)^{2n+3}_{k=1}$.

By the $L$-magic graceful-difference ice-flower system $I_{ce}(LG_{1,n}D_k)^{n_{gdt}}_{k=1}$, we have a \emph{$L$-magic graceful-difference star-graphic lattice} as follows
\begin{equation}\label{eqa:general-graceful-difference-stars-lattice}
\textrm{\textbf{L}}(\overline{\ominus} \textbf{I}_{ce}(LGD)) =\left \{\overline{\ominus}^{n_{gdt}}_{i=1}a_iLG_{1,n}D_i: a_i\in Z^0, LG_{1,n}D_i\in I_{ce}(LG_{1,n}D_k)^{n_{gdt}}_{k=1}\right \}
\end{equation} with $\sum^{n_{gdt}}_{i=1} a_i\geq 1$ and the base is $\textbf{I}_{ce}(LGD)=I_{ce}(LG_{1,n}D_k)^{n_{gdt}}_{k=1}$.

Notice that each graph in one of the graceful-difference star-graphic lattice and the $L$-magic graceful-difference star-graphic lattice admits a graceful-difference proper total coloring, and each one of the graceful-difference star-graphic lattice and the $L$-magic graceful-difference star-graphic lattice contains infinite graphs admitting a graceful-difference proper total colorings.

\subsection{Edge-difference star-graphic lattices}

\subsubsection{Edge-difference ice-flower systems}

In general, there is a total coloring $\pi_k$ of a star $K_{1,n}$ defined in the following way: $\pi_k(x_0)=s\in [1,3n]$, $\pi_k(x_j)=L-s_j~(>s)$ and $\pi_k(x_0x_j)=s+s_j$ with $j\in[1,n]$ and $s\in [1,3n]$, so $\pi_k(x_0x_j)+|\pi_k(x_j)-\pi_k(x_0)|=L$ for each edge $x_0x_j$ of $K_{1,n}$, such that $\pi_k$ is just an edge-difference proper total coloring of $K_{1,n}$, denoted this colored star as $K^{(k)}_{1,n}=LE_{1,n}D_k$. For each fixed $s$, there are $e(s)$ groups of integers $s_1,s_2,\dots s_n$ in $[1,3n]$ holding $\pi_k$ to be an edge-difference proper total coloring of $K_{1,n}$, then we get $n_{edt}$ colored stars $LE_{1,n}D_k$, where $n_{edt}=\sum^{3n}_{s=1}e(s)$, and collect these stars in a set $I_{ce}(LE_{1,n}D_k)^{n_{edt}}_{k=1}$, called \emph{$L$-magic edge-difference ice-flower system}. Moreover, let $\pi^c_k$ be the dual of $\pi_k$ defined as $\pi_k(x)+\pi^c_k(x)=\max \pi_k+\min \pi_k$ for $x\in V(LE_{1,n}D_k)$, and $\pi^c_k(uv)=\pi_k(uv)$ for $uv\in E(LE_{1,n}D_k)$. Then
\begin{equation}\label{eqa:edge-difference-ice-flower-dual}
\pi^c_k(uv)+|\pi^c_k(u)-\pi^c_k(v)|=\pi_k(uv)+|\pi_k(u)-\pi_k(v)|=L
\end{equation} for every edge $uv\in E(LE_{1,n}D_k)$. We get the \emph{dual $L$-magic edge-difference ice-flower system} $I_{ce}(\pi^c_k(LE_{1,n}D_k))^{n_{edt}}_{k=1}$ of the $L$-magic edge-difference ice-flower system $I_{ce}(LE_{1,n}D_k)^{n_{edt}}_{k=1}$

We show a particular edge-difference ice-flower system $I_{ce}(E_{1,n}D_k)^{2n+3}_{k=1}$, where each $E_{1,n}D_k$ is a copy $K^{(k)}_{1,n}$ of $K_{1,n}$ and admits an edge-difference proper total coloring $h_k$ with $k\in [1,2n+3]$. Each edge-difference proper total coloring $h_k$ is defined as follows:

Case ED-1. $n=2m$. We define $h_k$ in the way: (ED-1-1) $h_k(x^k_0)=k$ with $k\in [1,2m+1]$, $h_k(x^k_j)=4m+4-j$ with $j\in [1,2m+1]$, and $h_k(x^k_0x^k_j)=k+j$ with $1\leq i,j\leq 2m+1$. So, $h_k(x^k_0x^k_j)+|h_k(x^k_j)-h_k(x^k_0)|=4m+4$, see Definition \ref{defn:combinatoric-definition-total-coloring}. If $h_k(x^k_{j'})=h_k(x^k_0x^k_{j'})$ for some $j'\in [1,2m+1]$, we recolor the edge $x^k_0x^k_{j'}$ with $h_k(x^k_0x^k_{j'})=4m+3$ and the vertex $x^k_{j'}$ with $h_k(x^k_{j'})=h_k(x^k_0)-1=k-1$, respectively.

(ED-1-2) Set $h_k(x^k_0)=k$ with $k\in [2m+2,4m+3]$, $h_k(x^k_j)=j$ with $j\in [1,2m]$, and $h_k(x^k_0x^k_j)=4m+4-k+j$. Thereby, $h_k(x^k_0x^k_j)+|h_k(x^k_j)-h_k(x^k_0)|=4m+4$. If $h_k(x^k_0)=h_k(x^k_0x^k_{j'})$ for some $j'\in [1,2m]$ and $k\in [2m+1,4m+3]$, we recolor the edge $x^k_0x^k_{j'}$ with $h_k(x^k_0x^k_{j'})=4m+3$ and the vertex $x^k_{j'}$ with $h_k(x^k_{j'})=h_k(x^k_0)+1=k+1$.

Case ED-2. $n=2m+1$. We define $h_k$ as follows: (ED-2-1) $h_k(x^k_0)=k$ with $k\in [1,2m+1]$, $h_k(x^k_j)=4m+6-j$ with $j\in [1,2m+1]$, and $h_k(x^k_0x^k_j)=k+j$ with $1\leq i,j\leq 2m+1$. Immediately, $h_k(x^k_0x^k_j)+|h_k(x^k_j)-h_k(x^k_0)|=4m+6$. If $h_k(x^k_0x^k_{j'})=h_k(x^k_{j'})$ happens for some $j'\in [1,2m+1]$ and $k\in [1,2m+1]$, then we reset the edge $x^k_0x^k_{j'}$ with $h_k(x^k_0x^k_{j'})=4m+5$, the vertex $x^k_{j'}$ with $h_k(x^k_{j'})=h_k(x^k_0)+1=k+1$.

(ED-2-2) Set $h_k(x^k_0)=k$ with $k\in [2m+2,4m+5]$, $h_k(x^k_j)=j$ with $j\in [1,2m+1]$, and $h_k(x^k_0x^k_j)=4m+6-k+j$. So, $h_k(x^k_0x^k_j)+|h_k(x^k_j)-h_k(x^k_0)|=4m+6$. If we meet $h_k(x^k_0)=h_k(x^k_0x^k_{j'})$ for some $j'\in [1,2m+1]$ and $k\in [2m+1,4m+5]$, we recolor the edge $x^k_0x^k_{j'}$ with $h_k(x^k_0x^k_{j'})=4m+5$ and the vertex $x^k_{j'}$ with $h_k(x^k_{j'})=h_k(x^k_0)+1=k+1$.

\begin{figure}[h]
\centering
\includegraphics[width=15.6cm]{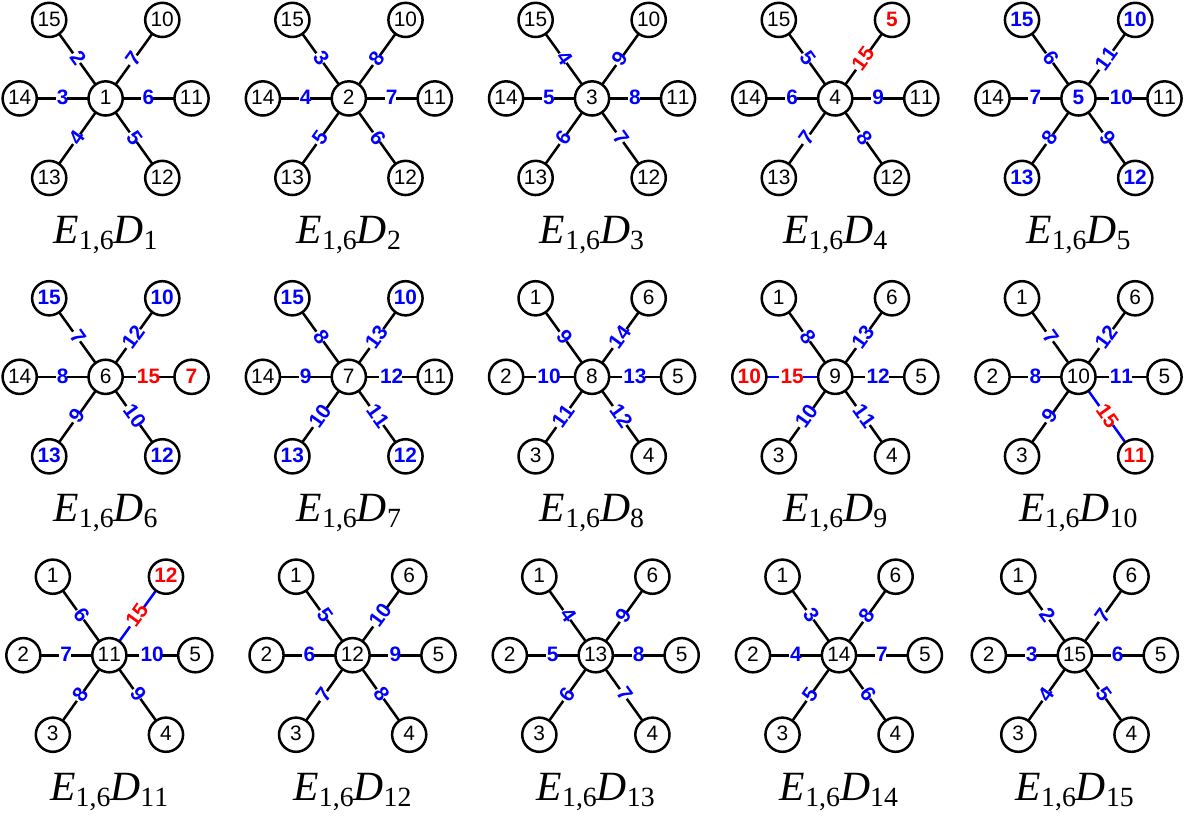}
\caption{\label{fig:Dsaturated-edge-difference-even}{\small An edge-difference ice-flower system $I_{ce}(E_{1,6}D_k)^{15}_{k=1}$.}}
\end{figure}

\begin{figure}[h]
\centering
\includegraphics[width=16.2cm]{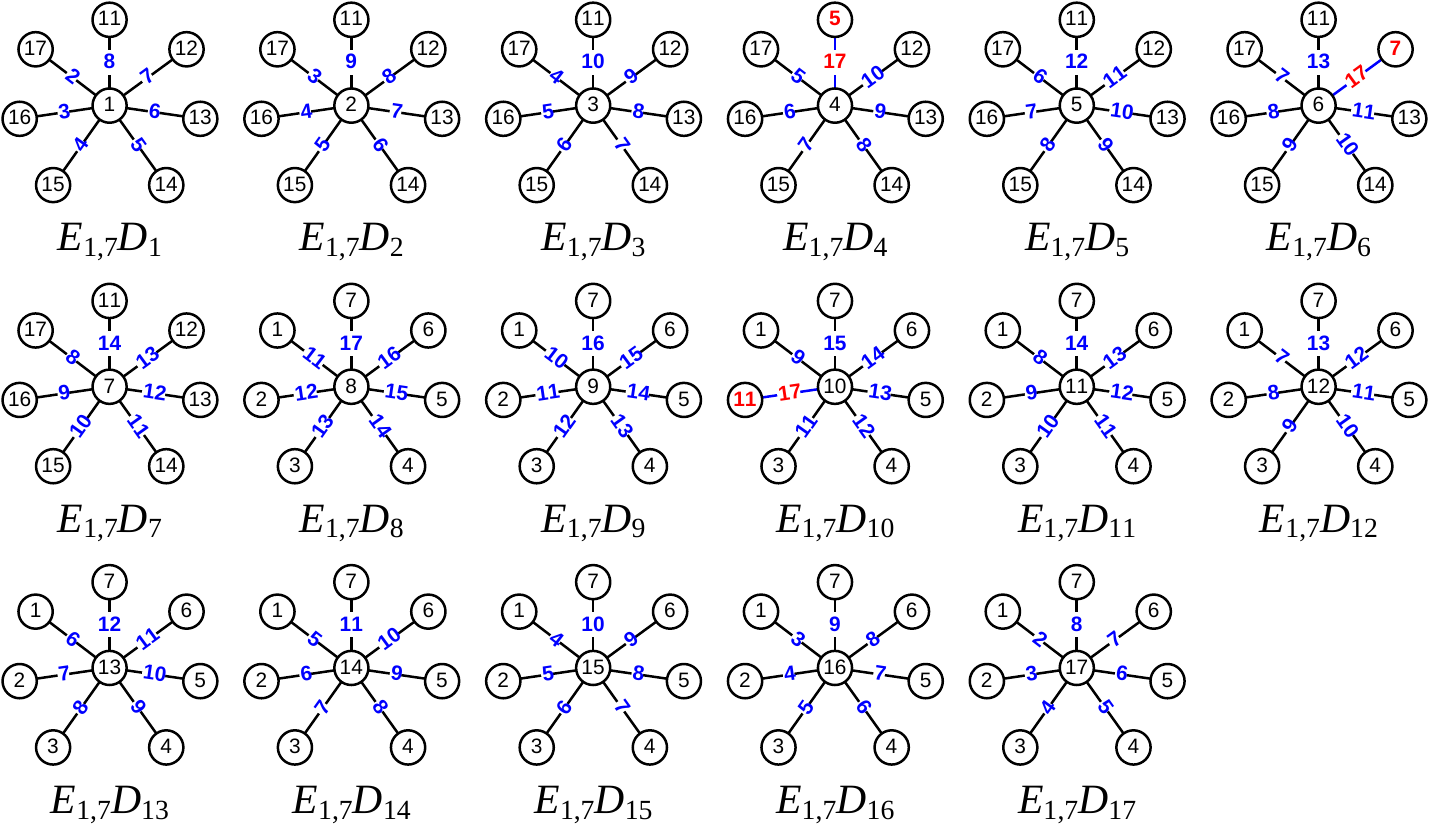}
\caption{\label{fig:Dsaturated-edge-difference-odd}{\small An edge-difference ice-flower system $I_{ce}(E_{1,7}D_k)^{17}_{k=1}$.}}
\end{figure}

\vskip 0.4cm

In \cite{Wang-Yao-edge-difference-2019} Wang \emph{et al.}, by the particular edge-difference ice-flower system $I_{ce}(E_{1,n}D_k)^{2n+3}_{k=1}$, have shown the following results:
\begin{thm} \label{thm:edge-difference-many-results}
\cite{Wang-Yao-edge-difference-2019} Let $K_m$ be a complete graph of $m$ vertices.
\begin{asparaenum}[\textbf{\textrm{ED}}-1. ]
\item Each connected graph $G$ of $m$ vertices satisfies $\chi''_{edt}(G)\leq \chi''_{edt}(K_m)$.
\item There are infinite graphs $G$ admitting perfect edge-difference proper total coloring, that is, $\chi''_{edt}(G)=\chi''(G)$.
\item There are infinite non-tree like graphs $G$ holds $\chi''_{edt}(G)\leq 2\Delta(G)+3$.
\item There are infinite $\Delta$-regular graphs $H$ hold $\chi''_{edt}(H)\leq 2\Delta+3$.
\item There is $\chi''(T)\leq \chi''_{edt}(T)\leq 2\Delta+3$ for each tree $T$.
\item Suppose that doing a series of vertex-splitting operations to a connected graph $G$ gets a tree $T$. If an edge-difference proper total coloring $f$ of $T$ induces an edge-difference proper total coloring $g$ of $G$, then $\chi''_{edt}(G)\leq \chi''_{edt}(T)$.
\item Let $(X,Y)$ be the vertex set bipartition of a bipartite graph $G$, then $\chi''_{edt}(G)\leq 2\Delta(G)+3$ if $|X|=|Y|$.
\item Suppose that a connected graph $H$ admits an edge-difference proper total coloring $f$, such that $f(uv)+|f(u)-f(v)|=k_1>0$ for each edge $uv\in E(H)$. Then, for any given strictly increasing number sequence $\{k_1,k_2,\dots ,k_m\}=\{k_i\}^m_1$, that is $k_i<k_{i+1}$, $H$ admits a series of real-valued edge-difference proper total colorings $g_i$ with $i\in[1,m]$, such that $g_i(uv)+|g_i(u)-g_i(v)|=k_i$ for each edge $uv\in E(H)$.
\end{asparaenum}
\end{thm}

\subsubsection{Edge-difference star-graphic lattices}

Each edge-difference ice-flower system $I_{ce}(E_{1,n}D_k)^{2n+3}_{k=1}$ distributes us an \emph{edge-difference star-graphic lattice} as follows:
\begin{equation}\label{eqa:edge-difference-stars-lattice}
\textrm{\textbf{L}}(\overline{\ominus} \textbf{I}_{ce}(ED)) =\left \{\overline{\ominus}^{2n+3}_{i=1}a_iE_{1,n}D_i: a_i\in Z^0, E_{1,n}D_i\in I_{ce}(E_{1,n}D_k)^{2n+3}_{k=1}\right \}
\end{equation} with $\sum^{2n+3}_{i=1} a_i\geq 1$ and the base is $\textbf{I}_{ce}(ED)=I_{ce}(E_{1,n}D_k)^{2n+3}_{k=1}$.

By the $L$-magic edge-difference ice-flower system $I_{ce}(LE_{1,n}D_k)^{n_{edt}}_{k=1}$, we get an \emph{edge-difference star-graphic lattice}:
\begin{equation}\label{eqa:general-edge-difference-stars-lattice}
\textrm{\textbf{L}}(\overline{\ominus} \textbf{I}_{ce}(LED)) =\left \{\overline{\ominus}^{n_{edt}}_{i=1}a_iLE_{1,n}D_i: a_i\in Z^0, LE_{1,n}D_i\in I_{ce}(LE_{1,n}D_k)^{n_{edt}}_{k=1}\right \}
\end{equation} with $\sum^{n_{edt}}_{i=1} a_i\geq 1$ and the base is $\textbf{I}_{ce}(LED)=I_{ce}(LE_{1,n}D_k)^{n_{edt}}_{k=1}$.
The dual $L$-magic edge-difference ice-flower system $I_{ce}(\pi^c_k(LE_{1,n}D_k))^{n_{edt}}_{k=1}$ induces a dual edge-difference star-graphic lattice:
\begin{equation}\label{eqa:general-edge-difference-stars-lattice}
\textrm{\textbf{L}}(\overline{\ominus} \textbf{I}^c_{ce}(LED)) =\left \{\overline{\ominus}^{n_{edt}}_{i=1}a_i\pi^c_k(LE_{1,n}D_i): a_i\in Z^0, \pi^c_k(LE_{1,n}D_i)\in I_{ce}(\pi^c_k(LE_{1,n}D_k))^{n_{edt}}_{k=1}\right \}
\end{equation} with $\sum^{n_{edt}}_{i=1} a_i\geq 1$ and the dual base is $\textbf{I}^c_{ce}(LED)=I_{ce}(\pi^c_k(LE_{1,n}D_k))^{n_{edt}}_{k=1}$.

Another particular edge-difference ice-flower system $I_{ce}(BE_{1,n}D_k)^{3n}_{k=1}$ is defined as: (ped-1) $h_k(x_0)=k\in [1,n]$, $h_k(x_j)=3n-j$ and $h_k(x_0x_j)=k+j$ with $j\in[1,n]$ and $k\in [1,n]$, so $h_k(x_0x_j)+|h_k(x_j)-h_k(x_0)|=3n$ for each edge $x_0x_j$ of $K_{1,n}$; (ped-2) $h_k(x_0)=k\in [n+1,3n]$, $h_k(x_j)=j$ and $h_k(x_0x_j)=3n+1-k+j$ with $j\in[1,n]$ and $k\in [n+1,3n]$, so $h_k(x_0x_j)+|h_k(x_j)-h_k(x_0)|=3n$ for each edge $x_0x_j$ of $K_{1,n}$, and this colored star is denoted as $K^{(k)}_{1,n}=BE_{1,n}D_k$ with $k\in [1,3n]$. See two particular edge-difference ice-flower systems $I_{ce}(BE_{1,s}D_k)^{3n}_{k=1}$ with $s=6,7$ and their dual edge-difference ice-flower systems $I_{ce}(B^cE_{1,s}D_k)^{3n}_{k=1}$ with $s=6,7$ shown in Fig.\ref{fig:edge-difference-3n}.

\begin{figure}[h]
\centering
\includegraphics[width=15cm]{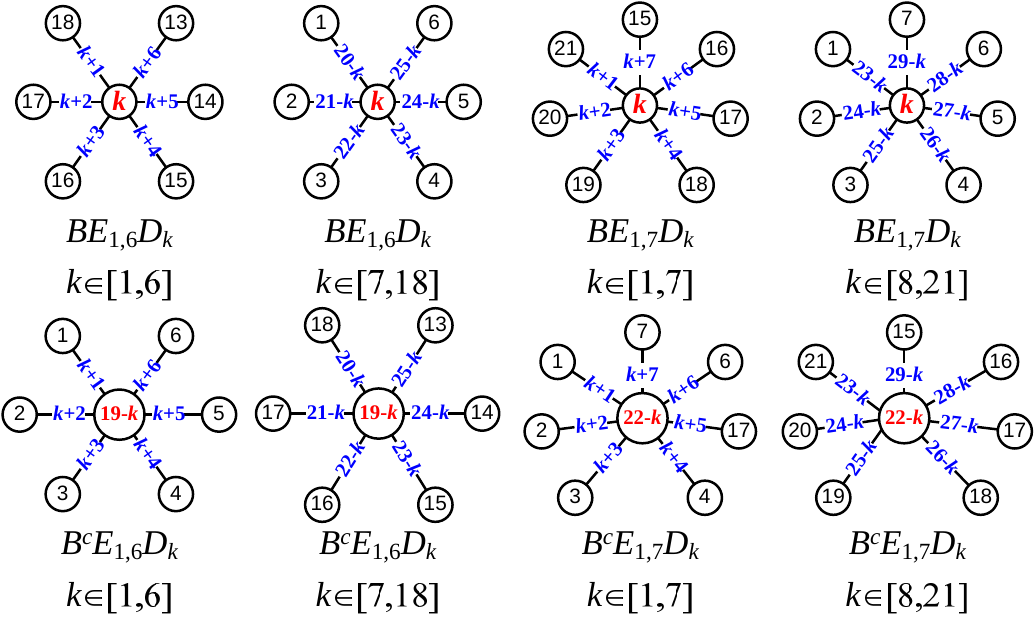}
\caption{\label{fig:edge-difference-3n}{\small Four edge-difference ice-flower systems $I_{ce}(BE_{1,s}D_k)^{18}_{k=1}$ and $I_{ce}(B^cE_{1,s}D_k)^{21}_{k=1}$ with $s=6,7$.}}
\end{figure}

The edge-difference ice-flower system $I_{ce}(BE_{1,n}D_k)^{3n}_{k=1}$ can help us to get the following results:
\begin{thm} \label{thm:edge-difference-ice-flower-lattice}
(1) Every $\Delta$-saturated tree $H\in \textrm{\textbf{L}}(\overline{\ominus} \textbf{I}_{ce}(ED))$ obeys $\chi''_{edt}(H)\leq 1+2\Delta(H)$.

(2) Each bipartite complete graph $K_{n,n}$ of $2n$ vertices holds $\chi''_{edt}(K_{n,n})=3n$ true.

(3) Each bipartite graph $G$ with bipartition $(X,Y)$ holds $\chi''_{edt}(G)\leq 3\max \{|X|,|Y|\}$ true.
\end{thm}

\subsection{Felicitous-difference star-graphic lattices}

\subsubsection{Felicitous-difference ice-flower systems}

In general, there is a proper total coloring $\zeta_k$ of a star $K_{1,n}$ defined as: $\zeta_k(x_0)=p$, $\zeta_k(x_j)=p_j$ and $\zeta_k(x_0x_j)=L+p+p_j$ with $j\in[1,n]$ and $p\in [1,3n]$, so $|\zeta_k(x_0)+\zeta_k(x_j)-\zeta_k(x_0x_j)|=L$ for each edge $x_0x_j$ of $K_{1,n}$. For each fixed $p\in [1,3n]$, there are $b(p)$ groups of integers $p_1,p_2,\dots ,p_n$ of $[1,3n]$ holding the above coloring $\zeta_k$ to be a felicitous-difference proper total coloring of $K_{1,n}$, then we get $n_{fdt}$ colored stars $LF_{1,n}D_s$ in total, where $n_{fdt}=\sum^{3n}_{k=1}b(p)$, and put them into a set $I_{ce}(LF_{1,n}D_k)^{n_{fdt}}_{k=1}$, called a \emph{$L$-magic felicitous-difference ice-flower system}. Moreover, we have a \emph{felicitous-difference star-graphic lattice} as follows:
\begin{equation}\label{eqa:big-largest-felicitous-difference-lattice}
\textrm{\textbf{L}}(\overline{\ominus} \textbf{I}_{ce}(LFD)) =\left \{\overline{\ominus}^{n_{fdt}}_{j=1}a_jLF_{1,n}D_j: a_j\in Z^0, LF_{1,n}D_j\in I_{ce}(LF_{1,n}D_s)^{n_{fdt}}_{s=1}\right \}
\end{equation} with $\sum^{n_{fdt}}_{j=1} a_j\geq 1$ and the base is $\textbf{I}_{ce}(LFD)=I_{ce}(LF_{1,n}D_k)^{n_{fdt}}_{k=1}$.

We have two particular felicitous-difference ice-flower systems $I_{ce}(F_{1,n}D_k)^{2n}_{k=1}$ and $I_{ce}(SF_{1,n}D_k)^{n}_{k=1}$ defined as follows:

1. A particular felicitous-difference ice-flower system $I_{ce}(F_{1,n}D_k)^{2n}_{k=1}$ is made by the felicitous-difference proper total coloring, where each $F_{1,n}D_k$ is a copy $K^{(k)}_{1,n}$ of $K_{1,n}$ and admits a felicitous-difference proper total coloring $g_k$ with $k\in [1,2n]$. Each felicitous-difference proper total coloring $g_k$ is defined in the following:

Case FD-1. $n=2m$. By Definition \ref{defn:combinatoric-definition-total-coloring}, we define $g_k$ in the way: (tg-1-1) $g_k(x^k_0)=k$ with $k\in [1,2m]$, $g_k(x^k_j)=4m+1-j$ with $j\in [1,2m]$, and $g_k(x^k_0x^k_j)=g_k(x^k_0)+g_k(x^k_j)=4m+1-j+k\leq 6m=3n$ with $1\leq i,j\leq 2m$; (tg-1-2) $g_k(x^k_0)=k$ with $k\in [2m+1,4m]$, $g_k(x^k_j)=j$ with $j\in [1,2m]$, and $g_k(x^k_0x^k_j)=g_k(x^k_0)+g_k(x^k_j)\leq 6m=3n$.

Case FD-2. $n=2m+1$. A proper total coloring $g_k$ is defined as: (tg-2-1) $g_k(x^k_0)=k$ with $k\in [1,2m+1]$, $g_k(x^k_j)=4m+3-j$ with $j\in [1,2m+1]$, and $g_k(x^k_0x^k_j)=4m+3-j+k\leq 6m+3=3n$ with $1\leq i,j\leq 2m+1$. (tg-2-2) $g_k(x^k_0)=k$ with $k\in [2m+1,4m+2]$, $g_k(x^k_j)=j$ with $j\in [1,2m+1]$, and $g_k(x^k_0x^k_j)=k+j\leq 6m+3=3n$.

2. A \emph{smallest felicitous-difference ice-flower system} $I_{ce}(SF_{1,n}D_k)^{n}_{k=1}$ is defined as: Each $SF_{1,n}D_k$ is a copy $K^{(k)}_{1,n}$ of $K_{1,n}$ and admits a felicitous-difference total coloring $h_k$ with $k\in [1,n]$, where $h_k$ is defined as: $h_k(x^k_0)=k$ with $k\in [1,n]$, $h_k(x^k_s)=s$ with $s\in [1,n]$ and $s\neq k$, and $h_k(x^k_k)=n+1$, as well as $h_k(x^k_0x^k_s)=h_k(x^k_0)+h_k(x^k_s)\leq 2n+1$ with $s\in [1,n]$. See two felicitous-difference ice-flower systems shown in Fig.\ref{fig:smallest-star-system-big}, these two ice-flower systems are \emph{strong}.

\begin{figure}[h]
\centering
\includegraphics[width=15.6cm]{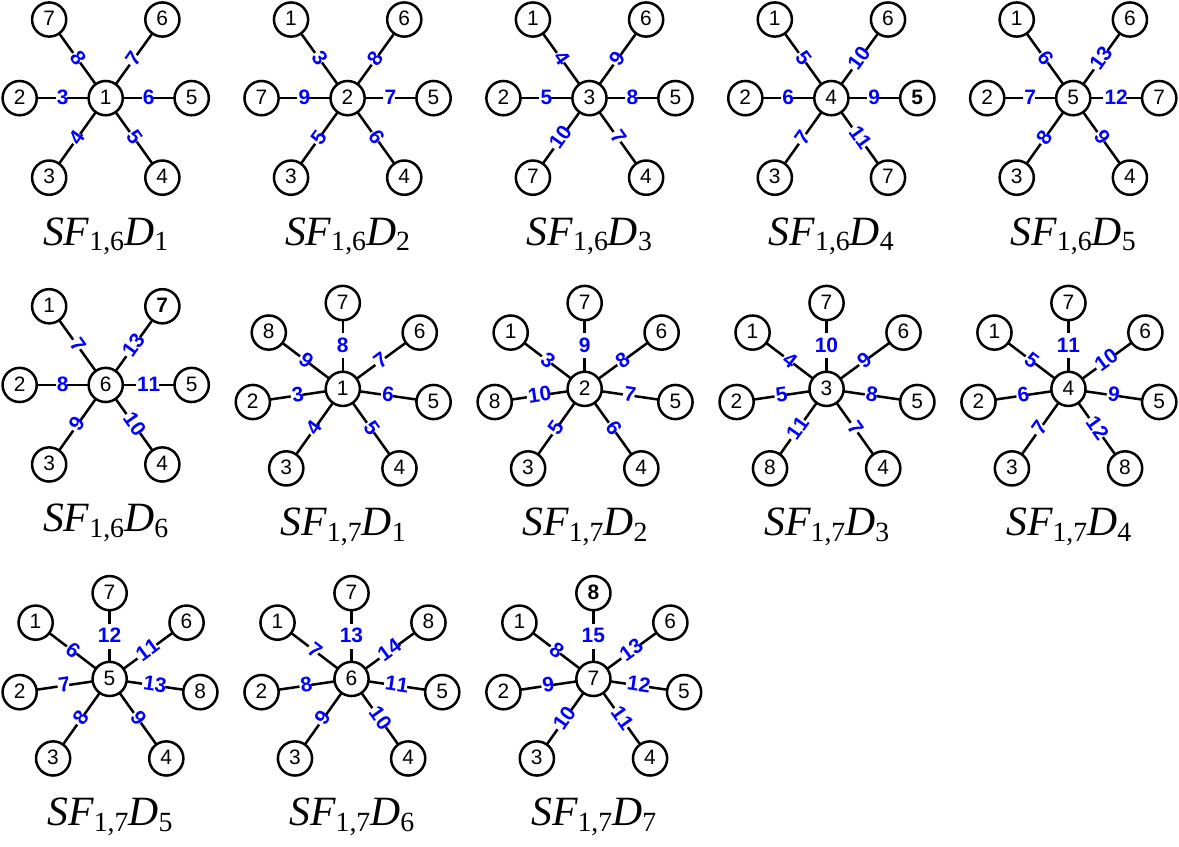}\\
\caption{\label{fig:smallest-star-system-big} {\small Two felicitous-difference ice-flower systems $I_{ce}(SF_{1,6}D_k)^{6}_{k=1}$ and $I_{ce}(SF_{1,7}D_k)^{7}_{k=1}$.}}
\end{figure}

\subsubsection{Felicitous-difference star-graphic lattices}

Since a colored leaf-coinciding operation $F_{1,n}D_j\overline{\ominus} F_{1,n}D_k$ between two colored stars $F_{1,n}D_k$ and $F_{1,n}D_j$ produces a graph with diameter three, so each group of colored stars $K_{1,n_1}$, $K_{1,n_2}$, $\dots$, $K_{1,n_m}$ is linearly independent under the colored leaf-coinciding operation. By the colored leaf-coinciding operation and the felicitous-difference ice-flower systems $I_{ce}(F_{1,n}D_k)^{2n}_{k=1}$ and $I_{ce}(SF_{1,n}D_k)^{n}_{k=1}$, each graph contained in the following graphic lattice
\begin{equation}\label{eqa:stars-lattice}
\textrm{\textbf{L}}(\overline{\ominus} \textbf{I}_{ce}(FD)) =\left \{\overline{\ominus}^{2n}_{i=1}a_iF_{1,n}D_i: a_i\in Z^0, F_{1,n}D_i\in I_{ce}(F_{1,n}D_k)^{2n}_{k=1}\right \}
\end{equation} is $\Delta$-saturated, where $\sum ^{2n}_{i=1}a_i\geq 1$ and the base is $\textbf{I}_{ce}(FD)=I_{ce}(F_{1,n}D_k)^{2n}_{k=1}$. We call $\textrm{\textbf{L}}(\overline{\ominus} \textbf{I}_{ce}(FD))$ a \emph{felicitous-difference star-graphic lattice}. Similarly, by the smallest felicitous-difference ice-flower system $I_{ce}(SF_{1,n}D_k)^{n}_{k=1}$, we have another \emph{felicitous-difference star-graphic lattice} defined as follows:
\begin{equation}\label{eqa:smallest-stars-lattice}
\textrm{\textbf{L}}(\overline{\ominus} \textbf{I}_{ce}(SFD)) =\left \{\overline{\ominus}^n_{j=1}a_jSF_{1,n}D_j: a_j\in Z^0, SF_{1,n}D_j\in I_{ce}(SF_{1,n}D_k)^{n}_{k=1}\right \}
\end{equation} with $\sum^n_{j=1} a_j\geq 1$ and the base is $\textbf{I}_{ce}(SFD)=I_{ce}(SF_{1,n}D_k)^{n}_{k=1}$.

As an application of the felicitous-difference ice-flower systems, a $\Delta$-saturated graph $G$ shown in Fig.\ref{fig:saturated-big} (a) is obtained by doing a series of colored leaf-coinciding operations on a smallest felicitous-difference ice-flower system $I_{ce}(S_{1,6}G_k)^{6}_{k=1}$ shown in Fig.\ref{fig:smallest-star-system-big}, so $G$ belongs to the felicitous-difference star-graphic lattice $\textrm{\textbf{L}}(\overline{\ominus} \textbf{I}_{ce}(SFD))$; Fig.\ref{fig:saturated-big} (b) is obtained by doing a series of colored leaf-coinciding operations on the $\Delta$-saturated graph (a). Conversely, the $\Delta$-saturated graph (a) is obtained by doing a series of colored leaf-splitting operations on the $\Delta$-saturated graph (b).

\begin{figure}[h]
\centering
\includegraphics[width=16cm]{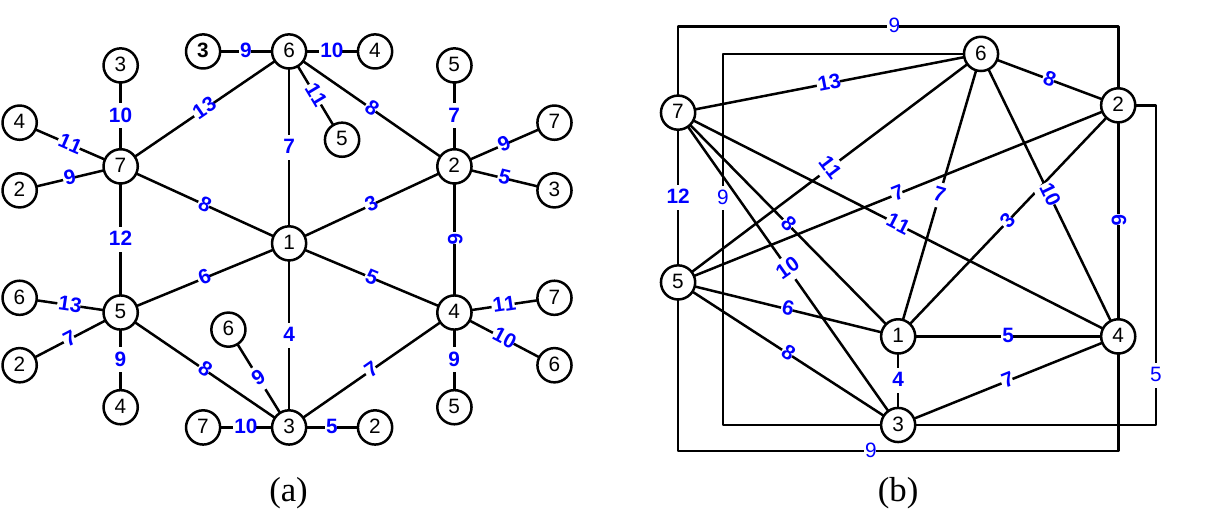}\\
\caption{\label{fig:saturated-big} {\small Two $\Delta$-saturated graphs admitting felicitous-difference proper total colorings.}}
\end{figure}

Thereby, the felicitous-difference star-graphic lattice $\textrm{\textbf{L}}(\overline{\ominus} \textbf{I}_{ce}(SFD))=L^s_G\cup L^s_H$, each graph $G\in L^s_G$ admits a felicitous-difference graph homomorphism to some $H\in L^s_H$, where each graph $H$ of $ L^s_H$ is $\Delta$-regular, i.e., $\textrm{deg}_H(x)=\Delta$ for $x\in V(H)$. Similarly, we have that felicitous-difference star-graphic lattice $\textrm{\textbf{L}}(\overline{\ominus} \textbf{I}_{ce}(FD))=L_G\cup L_H$, such that each graph $G\in L_G$ admits a felicitous-difference graph homomorphism to some graph $H\in L_H$, we write this case as $L_G\rightarrow L_H$, called a set-graph homomorphism from $L_G$ to $L_H$.

By the properties of two systems $I_{ce}(F_{1,n}D_k)^{2n}_{k=1}$ and $I_{ce}(SF_{1,n}D_k)^{n}_{k=1}$, we have

\begin{thm} \label{thm:property-star-type-lattice}
Each bipartite graph $G\in \textrm{\textbf{L}}(\overline{\ominus} \textbf{I}_{ce}(FD))$ holds $\chi''_{fdt}(G)\leq 3\Delta(G)$, and every $\Delta$-saturated tree $H\in \textrm{\textbf{L}}(\overline{\ominus} \textbf{I}_{ce}(SFD))$ holds $\chi''_{fdt}(H)\leq 1+2\Delta(H)$ true.
\end{thm}

Let $f$ be a felicitous-difference proper total coloring of a graph $G$ in one of two felicitous-difference star-graphic lattices $\textrm{\textbf{L}}(\overline{\ominus} \textbf{I}_{ce}(FD))$ and $\textrm{\textbf{L}}(\overline{\ominus} \textbf{I}_{ce}(SFD))$, and let $C_e(u)=\{f(uw): w\in N(u)\}$ and $C_v[u]=\{f(u)\}\cup \{f(w): w\in N(u)\}$ for each vertex $u\in V(G)$. Then we have $C_e(x)\neq C_e(y)$ and $C_v[x]\neq C_v[y]$ for each edge $xy\in E(G)$. We call $f$ an \emph{adjacent-vertex distinguishing felicitous-difference proper total coloring}. We apply two felicitous-difference star-graphic lattices $\textrm{\textbf{L}}(\overline{\ominus} \textbf{I}_{ce}(FD))$ and $\textrm{\textbf{L}}(\overline{\ominus} \textbf{I}_{ce}(SFD))$ to show the following results:

\begin{thm} \label{thm:particular-graphs-felicitous-difference}
Let $P_m$ be a path of $m$ vertices, $C_n$ be a cycle of $n$ vertices, and $T_{cat}$ be a caterpillar.

(1) $\chi''_{fdt}(C_{3m})=5$, and $\chi''_{fdt}(C_{n})=6$ for $n\neq 3m$.

(2) $\chi''_{fdt}(P_2)=3$, $\chi''_{fdt}(P_3)=4$, and $\chi''_{fdt}(P_m)=5$ for $m\geq 4$.

(3) $D(T_{cat})$ is the diameter of a $\Delta$-saturated caterpillar $T_{cat}$, then $\chi''_{fdt}(T_{cat})=\Delta(T_{cat})+2$ for $D(T_{cat})=2$, $\chi''_{fdt}(T_{cat})=\Delta(T_{cat})+3$ for $D(T_{cat})=3$, and $\chi''_{fdt}(T_{cat})=\Delta(T_{cat})+4$ for $D(T_{cat})\geq 4$.
\end{thm}

\begin{thm} \label{thm:saturated-felicitous-difference}
1. For each complete bipartite graph $K_{n,n}$, we have $\chi''_{fdt}(K_{n,n})=3n$.

2. Each complete bipartite graph $K_{m,n}$ with $m\leq n$ holds $\chi''_{fdt}(K_{m,n})=2m+n$.

3. Each vertex $u$ of a $\Delta$-saturated tree $H$ has its degree $\textrm{deg}_H(u)=1$ or $\textrm{deg}_H(u)=\Delta(H)$. Then $\chi''_{fdt}(H)\leq 1+2\Delta(H)$ as diameter $D(H)$ is not less than 3.
\end{thm}

\begin{lem} \label{thm:complete-graph-felicitous-difference}
Each complete graph $K_{n}$ with $n\geq 3$ holds $\chi''_{fdt}(K_{n})=2n-1$ and admits a pair of perfect all-dual felicitous-difference proper total colorings.
\end{lem}

\begin{lem} \label{thm:subgraph-vs-graph-felicitous-difference}
For each subgraph $H$ of a graph $G$, we have $\chi''_{fdt}(H)\leq \chi''_{fdt}(G)$.
\end{lem}

Since each graph is a subgraph of some complete graph, immediately, we have

\begin{thm} \label{thm:every-graph-tgadpt-coloring}
Every graph $G$ containing a subgraph $K_m$ admits a pair of perfect all-dual felicitous-difference proper total colorings and holds $$2m-1\leq \chi''_{fdt}(G)\leq 2|V(G)|-1.$$
\end{thm}

\begin{problem}\label{qeu:felicitous-difference-colorings}
For more researching star-graphic lattices, we present the following questions:
\begin{asparaenum}[\textrm{FDQ}-1. ]
\item \textbf{Characterize} the structures of two graphic lattices $\textrm{\textbf{L}}(\overline{\ominus} \textbf{I}_{ces}(FD))$ and $\textrm{\textbf{L}}(\overline{\ominus} \textbf{I}_{ces}(SFD))$.
\item \textbf{Find} all $k$ with $|f_k(x)+f_k(y)-f_k(xy)|=k$ for $\max \{f_k(w):w\in V(F)\cup E(G)\}=\chi''_{fdt}(G)$, see examples shown in Fig.\ref{fig:more-magic-numbers-big}.
\begin{figure}[h]
\centering
\includegraphics[width=16.2cm]{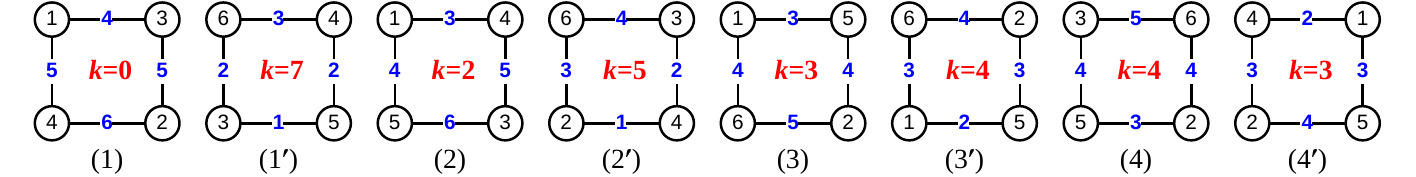}
\caption{\label{fig:more-magic-numbers-big}{\small A cycle $C_4$ admits four pairs of all-dual felicitous-difference proper total colorings ($k$) and ($k'$) with $k\in [1,4]$.}}
\end{figure}
\item \textbf{Does} every planar graph $H$ belong to the felicitous-difference star-graphic lattice $\textrm{\textbf{L}}(\overline{\ominus} \textbf{I}_{ces}(SFD))$, in other word, $\chi''_{fdt}(H)\leq 1+2\Delta(H)$?
\item Since each graph $H\in \textrm{\textbf{L}}(\overline{\ominus} \textbf{I}_{ces}(SFD))$ holding $\chi''_{fdt}(H)\leq 1+2\Delta(H)$, so \textbf{find} other subset $S\subset \textrm{\textbf{L}}(\overline{\ominus} \textbf{F}^c_{star\Delta})$ such that each graph $L\in S$ holding $\chi''_{fdt}(L)\leq 1+2\Delta(L)$.
\item Since a caterpillar $T$ corresponds a topological vector $V_{ec}(T)$, \textbf{can} we characterize a traditional lattice by some graphic lattices?
\item \textbf{Plant} some results of a traditional lattice $\textrm{\textbf{L}}(\textbf{B})$ to the felicitous-difference star-graphic lattices.
\item \textbf{Optimal felicitous-difference ice-flower system.} Find a $L$-magic felicitous-difference ice-flower system $I_{ce}(LF_{1,n}D_k)^{n_{fdt}}_{k=1}$, such that each graph $G$ is colored well by $G=\overline{\ominus}^{n_{fdt}}_{j=1}a_jLF_{1,n}D_j$ with color set $[1,\chi''_{fdt}(G)]$, where $LF_{1,n}D_j\in I_{ce}(LF_{1,n}D_k)^{n_{fdt}}_{k=1}$, $\sum^{n_{fdt}}_{j=1}a_j\geq 1$ and $a_j\in Z^0$. In other word, this felicitous-difference ice-flower system $I_{ce}(LF_{1,n}D_k)^{n_{fdt}}_{k=1}$ is \emph{optimal}.\qqed
\end{asparaenum}
\end{problem}

\subsubsection{Dual felicitous-difference ice-flower systems}

By two perfect all-dual ice-flower systems $I^d_{ce}(F^d_{1,n}D_k)^{2n}_{k=1}$ and $I^d_{ce}(S^dF_{1,n}D_k)^{n}_{k=1}$, we have two \emph{perfect all-dual felicitous-difference star-graphic lattices}:
\begin{equation}\label{eqa:t-all-dual-lattice}
\textrm{\textbf{L}}(\overline{\ominus} \textbf{I}^d_{ce}(FD)) =\left \{\overline{\ominus}^{2n}_{i=1}a_iF^d_{1,n}D_i: a_i\in Z^0, F^d_{1,n}D_i\in I^d_{ce}(F^d_{1,n}D_k)^{2n}_{k=1}\right \}
\end{equation}
with $\sum ^{2n}_{i=1}a_i\geq 1$ and the base is $\textbf{I}^d_{ce}(FD)=I^d_{ce}(F^d_{1,n}D_k)^{2n}_{k=1}$, and moreover
\begin{equation}\label{eqa:s-all-dual-lattice}
\textrm{\textbf{L}}(\overline{\ominus} \textbf{I}^d_{ce}(SFD)) =\left \{\overline{\ominus}^n_{j=1}a_jS^dF_{1,n}D_j: a_j\in Z^0, S^dF_{1,n}D_j\in I^d_{ce}(S^dF_{1,n}D_k)^{n}_{k=1}\right \}
\end{equation} with $\sum^n_{j=1} a_j\geq 1$ and the base is $\textbf{I}^d_{ce}(SFD)=I^d_{ce}(S^dF_{1,n}D_k)^{n}_{k=1}$.

In two felicitous-difference ice-flower systems $I_{ce}(F_{1,n}D_k)^{2n}_{k=1}$ and $I_{ce}(SF_{1,n}D_k)^{n}_{k=1}$, notice that each colored star $F_{1,n}D_k$ admitting a felicitous-difference proper total coloring $g_k$ and its dual $F^d_{1,n}D_k$ admitting a perfect all-dual felicitous-difference proper total coloring $g^d_k$ of $g_k$ holding $g^d_k(w)=3n+1-g_k(w)$ for each element $w\in V(F_{1,n}D_k)\cup E(F_{1,n}D_k)$ and $|g_k(u)+g_k(v)-g_k(uv)|=0=|g^d_k(u)+g^d_k(v)-g^d_k(uv)|$ for $uv\in E(F_{1,n}D_k)$. We write $F^d_{1,n}D_k=g^d_k(F_{1,n}D_k)$ and $F_{1,n}D_k=g_k(F^d_{1,n}D_k)$, thus, we have $$\overline{\ominus}^{2n}_{i=1}a_iF^d_{1,n}D_i=\overline{\ominus}^{2n}_{i=1}a_ig^d_i(F_{1,n}D_i),~\overline{\ominus}^{2n}_{i=1}a_iF_{1,n}D_i
=\overline{\ominus}^{2n}_{i=1}a_ig_i(F^d_{1,n}D_i)$$

For each colored star $SF_{1,n}D_k$ admitting a felicitous-difference proper total coloring $h_k$ and its dual $S^dF_{1,n}D_k$ admitting a perfect all-dual felicitous-difference proper total coloring $h^d_k$ of $h_k$, we have $h^d_k(w)=2n+1-h_k(w)$ for each element $w\in V(SF_{1,n}D_k)\cup E(SF_{1,n}D_k)$ and $|h_k(u)+h_k(v)-h_k(uv)|=0=|h^d_k(u)+h^d_k(v)-h^d_k(uv)|$ for $uv\in E(SF_{1,n}D_k)$. Moreover, we can write $S^dF_{1,n}D_k=h^d_k(SF_{1,n}D_k)$ and $SF_{1,n}D_k=h_k(S^dF_{1,n}D_k)$, as well as
$$\overline{\ominus}^{2n}_{i=1}a_iS^dF_{1,n}D_i=\overline{\ominus}^{2n}_{i=1}a_ih^d_i(SF_{1,n}D_i),~\overline{\ominus}^{2n}_{i=1}a_iSF_{1,n}D_i
=\overline{\ominus}^{2n}_{i=1}a_ih_i(S^dF_{1,n}D_i)$$
Thereby, two perfect all-dual felicitous-difference star-graphic lattices can be expressed as:
\begin{equation}\label{eqa:other-express-dual-lattice}
{
\begin{split}
&\textrm{\textbf{L}}(\overline{\ominus} \textbf{I}^d_{ce}(FD)) =\left \{\overline{\ominus}^{2n}_{i=1}a_ig^d_i(F_{1,n}D_i): a_i\in Z^0, F_{1,n}D_i\in I_{ce}(F_{1,n}D_k)^{2n}_{k=1}\right \}\\
&\textrm{\textbf{L}}(\overline{\ominus} \textbf{I}^d_{ce}(SFD)) =\left \{\overline{\ominus}^{n}_{j=1}a_jg^d_i(SF_{1,n}D_j): a_i\in Z^0, SF_{1,n}D_j\in I_{ce}(SF_{1,n}D_k)^{n}_{k=1}\right \}
\end{split}}
\end{equation}with $\sum^{2n}_{i=1}a_i\geq 1$ and $\sum^{n}_{j=1}a_j\geq 1$.

We define a coloring $\eta_k$ of a star $K_{1,n}$ by setting: (fd-1) $\eta_k(x_0)=k\in [1,n]$, $\eta_k(x_j)=2n+1-j$ and $\eta_k(x_0x_j)=2n+1-j+k\leq 3n$ with $j\in[1,n]$ and $k\in [1,n]$, so $|\eta_k(x_0)+\eta_k(x_j)-\eta_k(x_0x_j)|=0$ for each edge $x_0x_j$ of $K_{1,n}$; (fd-2) $\eta_k(x_0)=k$, $\eta_k(x_j)=j$ and $\eta_k(x_0x_j)=k+i\leq 3n$ with $j\in[1,n]$ and $k\in [n+1,2n]$, each edge $x_0x_j$ of $K_{1,n}$ satisfies $|\eta_k(x_0)+\eta_k(x_j)-\eta_k(x_0x_j)|=0$.
Clearly, $\eta_k$ is just a felicitous-difference proper total coloring of $K_{1,n}$, denoted this colored star as $K^{(s)}_{1,n}=F_{1,n}D_s$ with $s\in [1,2n]$. We get a \emph{felicitous-difference ice-flower system} $I_{ce}(F_{1,n}D_s)^{2n}_{s=1}$. See $I_{ce}(F_{1,6}D_k)^{12}_{k=1}$ and $I_{ce}(F_{1,7}D_k)^{14}_{k=1}$ shown in Fig.\ref{fig:Dsaturated-felicitous-difference-even} and Fig.\ref{fig:Dsaturated-felicitous-difference-odd}. Thereby, we have a \emph{felicitous-difference star-graphic lattice} as follows:
\begin{equation}\label{eqa:largest-felicitous-difference-lattice}
\textrm{\textbf{L}}(\overline{\ominus} \textbf{I}_{ce}(FD)) =\left \{\overline{\ominus}^{2n}_{j=1}a_jF_{1,n}D_j: a_j\in Z^0, F_{1,n}D_j\in I_{ce}(F_{1,n}D_s)^{2n}_{s=1}\right \}
\end{equation} with $\sum^{2n}_{j=1} a_j\geq 1$ and the base is $\textbf{I}_{ce}(FD)=I_{ce}(F_{1,n}D_s)^{2n}_{s=1}$.

\begin{figure}[h]
\centering
\includegraphics[width=16.2cm]{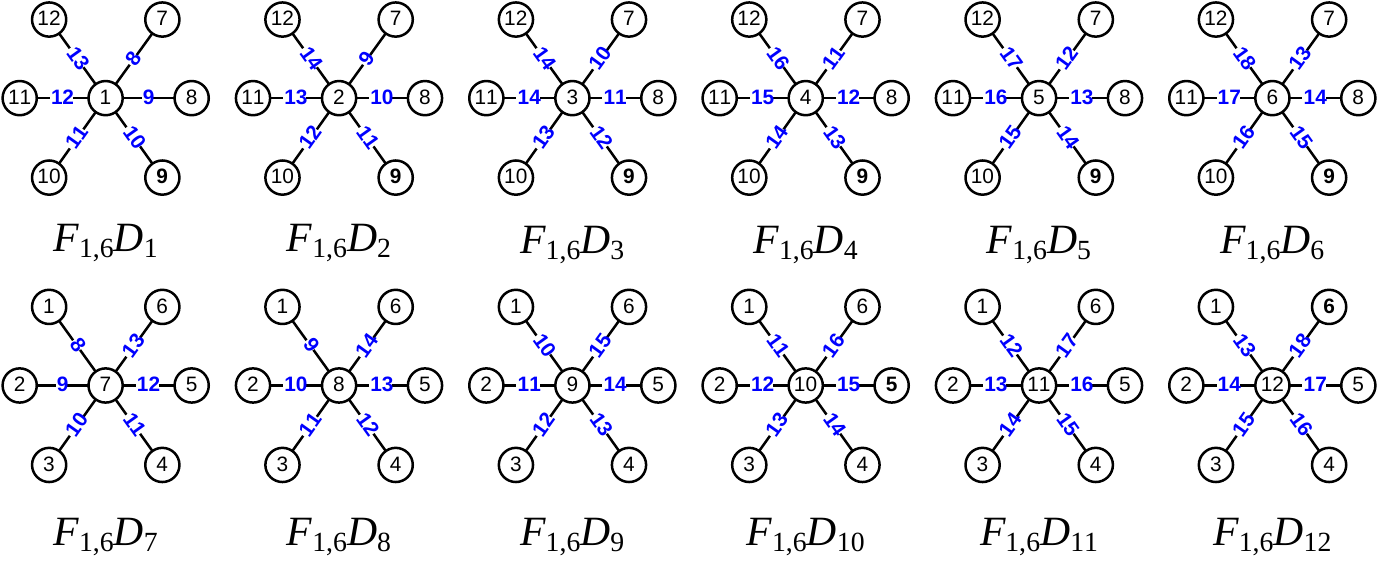}
\caption{\label{fig:Dsaturated-felicitous-difference-even}{\small A felicitous-difference ice-flower system $I_{ce}(F_{1,6}D_k)^{12}_{k=1}$.}}
\end{figure}

\begin{figure}[h]
\centering
\includegraphics[width=15.6cm]{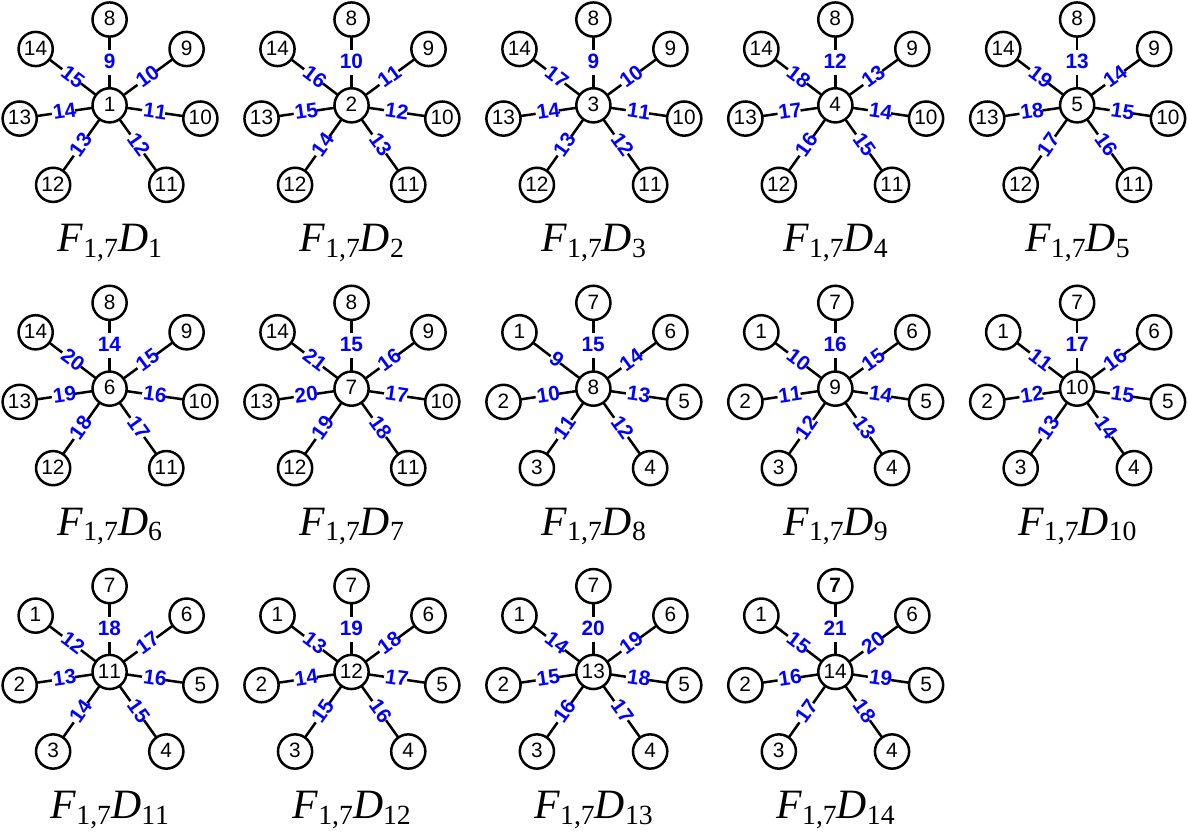}
\caption{\label{fig:Dsaturated-felicitous-difference-odd}{\small A felicitous-difference ice-flower system $I_{ce}(F_{1,7}D_k)^{14}_{k=1}$.}}
\end{figure}

As known, by the help of the felicitous-difference ice-flower system $I_{ce}(F_{1,n}D_s)^{2n}_{s=1}$ we can show $\chi''_{fdt}(K_{n,n})=3n$, which induces $\chi''_{fdt}(G)\leq 3\max \{|X|,|Y|\}$ for each bipartite graph $G$ with bipartition $(X,Y)$, we have two ice-flower systems as follows:

Notice that $\max \eta_k=3n$ and $\min \eta_k=1$, so the dual $\eta^c_k$ holds $\eta_k(w)+\eta^c_k(w)=3n+1$, that is, two coloring $\eta_k$ and $\eta^c_k$ are a pair of perfect all-dual felicitous-difference proper total colorings, and the dual lattice of the felicitous-difference star-graphic lattice is
\begin{equation}\label{eqa:dual-largest-felicitous-difference-lattice}
\textrm{\textbf{L}}(\overline{\ominus} \textbf{I}^c_{ce}(FD)) =\left \{\overline{\ominus}^{2n}_{j=1}a_j\eta^c_k(F_{1,n}D_j): a_j\in Z^0, F_{1,n}D_j\in I_{ce}(F_{1,n}D_s)^{2n}_{s=1}\right \}
\end{equation} with $\sum^{2n}_{j=1} a_j\geq 1$ and the base is $\textbf{I}^c_{ce}(FD)=I_{ce}(\eta^c_k(F_{1,n}D_j))^{2n}_{s=1}$.

\subsection{Edge-magic star-graphic lattices}

\subsubsection{Edge-magic ice-flower systems}

We define a general edge-magic ice-flower system $I_{ce}(LE_{1,n}M_k)^{n_{emt}}_{k=1}$ in the following way: $\varphi_k(x_0)=r\in [1,\beta]$ with $\beta= 3n + 3$ for even $n$ and $\beta = 3n + 4$ for odd $n$, $\varphi_k(x_j)=L-r_j$ and $\varphi_k(x_0x_j)=L'-r+-r_j$ with $j\in[1,n]$ and $r\in [1,\beta]$, so $\varphi_k(x_0)+\varphi_k(x_j)+\varphi_k(x_0x_j)=L+L'$ for each edge $x_0x_j$ of $K_{1,n}$, such that $\varphi_k$ is just an edge-magic proper total coloring of $K_{1,n}$, denoted this colored star as $K^{(k)}_{1,n}=LE_{1,n}M_k$. For each fixed $r\in [1,\beta]$, there are $c(r)$ groups of integers $r_1,r_2,\dots ,r_n$ of $[1,\beta]$ holding $\varphi_k$ to be an edge-magic proper total coloring of $K_{1,n}$, so we have $n_{emt}$ different colored stars $LE_{1,n}M_k$ in total, where $n_{emt}=\sum^{\beta}_{t=1}c(r)$, and then we get an edge-magic ice-flower system $I_{ce}(LE_{1,n}M_k)^{n_{emt}}_{k=1}$, and furthermore we have an \emph{edge-magic star-graphic lattice}
\begin{equation}\label{eqa:general-edge-magic-stars-lattice}
\textrm{\textbf{L}}(\overline{\ominus} \textbf{I}_{ce}(LEM)) =\left \{\overline{\ominus}^{n_{emt}}_{i=1}a_iLE_{1,n}M_i: a_i\in Z^0, E^L_{1,n}M_i\in I_{ce}(LE_{1,n}M_k)^{n_{emt}}_{k=1}\right \}
\end{equation} with $\sum^{n_{emt}}_{i=1} a_i\geq 1$ and the base is $\textbf{I}_{ce}(LEM)=I_{ce}(LE_{1,n}M_k)^{n_{emt}}_{k=1}$.

In particular cases, an \emph{edge-magic ice-flower system} $I_{ce}(LE_{1,n}M_k)^{n_{emt}}_{k=1}$ is defined as: First of all, a star $K_{1,n}$ admits a proper total coloring $\varphi_k$ defined by setting: (em-1) $\varphi_k(x_0)=k$, $\varphi_k(x_j)=4n-j$ and $\varphi_k(x_0x_j)=2n-k+j$ with $j\in[1,n]$ and $k\in [1,2n]$, so $\varphi_k(x_0)+\varphi_k(x_j)+\varphi_k(x_0x_j)=6n$ for each edge $x_0x_j$ of $K_{1,n}$, such that $\varphi_k$ is just an edge-magic proper total coloring of $K_{1,n}$, denoted this colored star as $K^{(k)}_{1,n}=LE_{1,n}M_k$; (em-2) $\varphi_k(x_0)=4n-k\in [1,2n-1]$, $\varphi_k(x_j)=j$ and $\varphi_k(x_0x_j)=2n+k-j$ with $j\in[1,n]$ and $k\in [1,2n-1]$.

We show $\varphi_k$ bo be an edge-magic proper total coloring. In case (em-1), there are $j\in[1,n]$ and $k\in [1,2n]$, if $\varphi_k(x_0)=k=\varphi_k(x_j)=4n-j$, then $k=4n-j\geq 3n$, a contradiction; if $\varphi_k(x_0x_j)=2n-k+j=\varphi_k(x_j)=4n-j$, thus, $2n\geq 2j=2n+k$, an obvious mistake. In case (em-2), there are $j\in[1,n]$ and $k\in [1,2n-1]$, if $\varphi_k(x_0)=4n-k=\varphi_k(x_j)=j$, then $4n=k+j\leq 3n-1$, a contradiction; if $\varphi_k(x_0x_j)=2n+k-j=\varphi_k(x_j)=j$, that is, $2n+k=2j\leq 2n$, it is impossible.

Thereby, we get a particular edge-magic ice-flower system $I_{ce}(E^L_{1,n}M_k)^{4n-1}_{k=1}$. See examples of the edge-magic ice-flower systems $I_{ce}(LE_{1,n}M_k)^{4n-1}_{k=1}$ shown in Fig.\ref{fig:largest-edge-magic}. The edge-magic ice-flower system $I_{ce}(E^L_{1,n}M_k)^{4n-1}_{k=1}$ induces the following result:

\begin{thm}\label{thm:edge-magic-complete-bipart}
Each complete bipartite graph $K_{n,n}$ holds $\chi''_{emt}(K_{n,n})\leq 4n-1$ true.
\end{thm}

\begin{figure}[h]
\centering
\includegraphics[width=15.6cm]{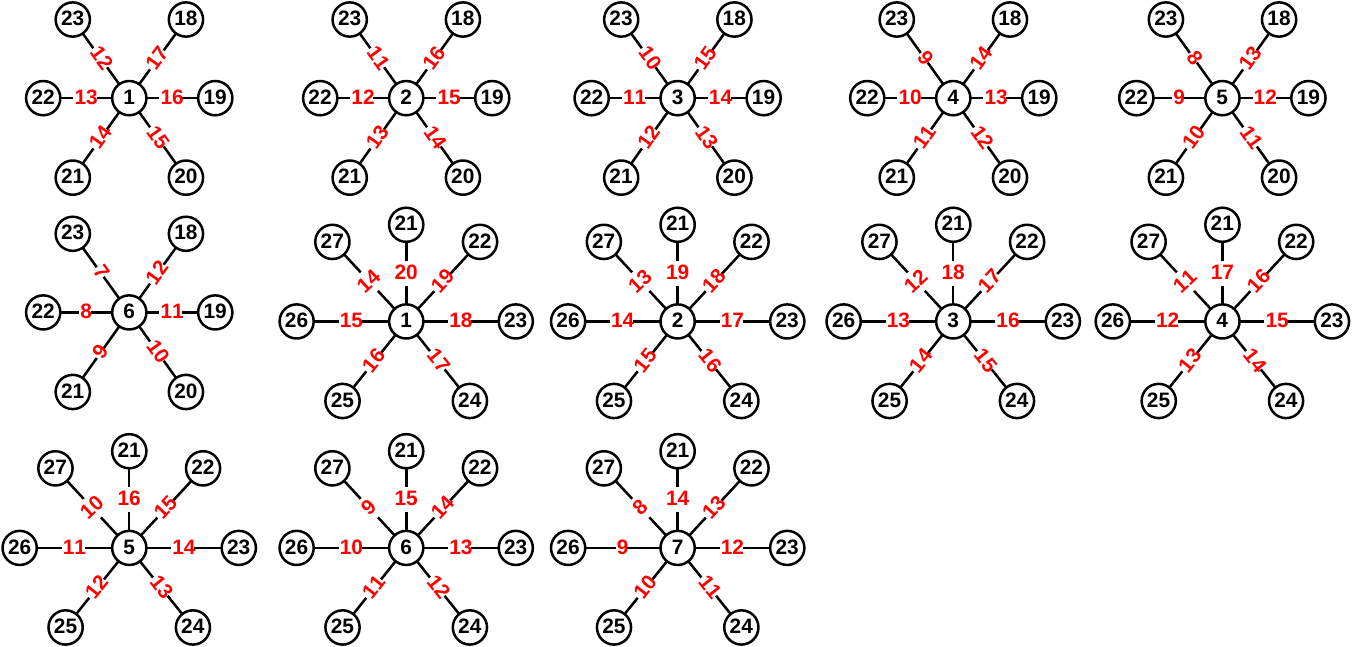}
\caption{\label{fig:largest-edge-magic}{\small Two parts of two edge-magic ice-flower systems $I_{ce}(E^L_{1,n}M_k)^{23}_{k=1}$ and $I_{ce}(E^L_{1,n}M_k)^{27}_{k=1}$.}}
\end{figure}

We introduce another particular edge-magic ice-flower system $I_{ce}(E_{1,n}M_k)^{2n+3}_{k=1}$ made by the edge-magic total coloring, where each $E_{1,n}M_k$ is a copy $K^{(k)}_{1,n}$ of $K_{1,n}$ and admits an edge-magic total coloring $e_k$ with $k\in [1,n+2]$. Now, we show each edge-magic total coloring $e_k$ below:

Case EM-1. $n=2m$. We define $e_k$ by two parts: (EM-1-1) $e_k(x^k_0)=k$ with $k\in [1,2m+1]$, $e_k(x^k_j)=4m+4-j$ with $j\in [1,2m+1]$, and $e_k(x^k_0x^k_j)=2m+2-k+j$ with $1\leq i,j\leq 2m+1$. So, $e_k(x^k_0)+e_k(x^k_0x^k_j)+e_k(x^k_j)=6m+6$ (see Definition \ref{defn:combinatoric-definition-total-coloring}). If $e_k(x^k_{j'})=e_k(x^k_0x^k_{j'})$ happens for some $j'\in [1,2m+1]$ and $k\in [1,2m+1]$, we recolor the edge $x^k_0x^k_{j'}$ with $e_k(x^k_0x^k_{j'})=4m+3$ and the vertex $x^k_{j'}$ with $e_k(x^k_{j'})=2m+3-i$.

(EM-1-2) $e_k(x^k_0)=k$ with $k\in [2m+2,4m+3]$, $e_k(x^k_j)=j$ with $j\in [1,2m]$, and $e_k(x^k_0x^k_j)=6m+6-k-j$. Thereby, $e_k(x^k_0)+e_k(x^k_0x^k_j)+e_k(x^k_j)=6m+6$. If we meet $e_k(x^k_0)=e_k(x^k_0x^k_{j'})$ for some $j'\in [1,2m]$ and $k\in [2m+1,4m+3]$, then we recolor the edge $x^k_0x^k_{j'}$ with $e_k(x^k_0x^k_{j'})=1$ and the vertex $x^k_{j'}$ with $e_k(x^k_{j'})=6m+5-k$.

Case EM-2. $n=2m+1$. We define $e_k$ in the following two parts: (EM-2-1) $e_k(x^k_0)=k$ with $k\in [1,2m+2]$, $e_k(x^k_j)=4m+6-j$ with $j\in [1,2m+1]$, and $e_k(x^k_0x^k_j)=2m+3-k-j$ with $1\leq i,j\leq 2m+1$. Immediately, $e_k(x^k_0)+e_k(x^k_0x^k_j)+e_k(x^k_j)=6m+9$. If $e_k(x^k_0x^k_{j'})=e_k(x^k_0)$ occurs for some $j'\in [1,2m+1]$ and $k\in [1,2m+1]$, then we recolor the edge $x^k_0x^k_{j'}$ with $e_k(x^k_0x^k_{j'})=4m+4$ and $e_k(x^k_{j'})=2m+5-k$ when $k=m+2$, otherwise $e_k(x^k_0x^k_{j'})=4m+5$ and the vertex $x^k_{j'}$ with $e_k(x^k_{j'})=2m+4-k$.

(EM-2-2) $e_k(x^k_0)=k$ with $k\in [2m+2,4m+5]$, $e_k(x^k_j)=j$ with $j\in [1,2m+1]$, and $e_k(x^k_0x^k_j)=6m+9-k-j$. So, $e_k(x^k_0)+e_k(x^k_0x^k_j)+e_k(x^k_j)=6m+9$. If $e_k(x^k_0)=e_k(x^k_0x^k_{j'})$ happens for some $j'\in [1,2m+1]$ and $k\in [2m+1,4m+5]$, then we recolor the edge $x^k_0x^k_{j'}$ with $e_k(x^k_0x^k_{j'})=1$ and the vertex $x^k_{j'}$ with $e_k(x^k_{j'})=6m+8-k\neq e_k(x^k_0)$, otherwise $e_k(x^k_0x^k_{j'})=2$ and $e_k(x^k_{j'})=6m+7-k$. If we meet $e_k(x^k_{j'})=e_k(x^k_0x^k_{j'})$ for some $j'\in [1,2m+1]$ and $k\in [2m+1,4m+5]$, then we recolor the edge $x^k_0x^k_{j'}$ and the vertex $x^k_{j'}$ by $e_k(x^k_0x^k_{j'})=1$ and $e_k(x^k_{j'})=6m+8-k$, respectively.

\begin{figure}[h]
\centering
\includegraphics[width=15.6cm]{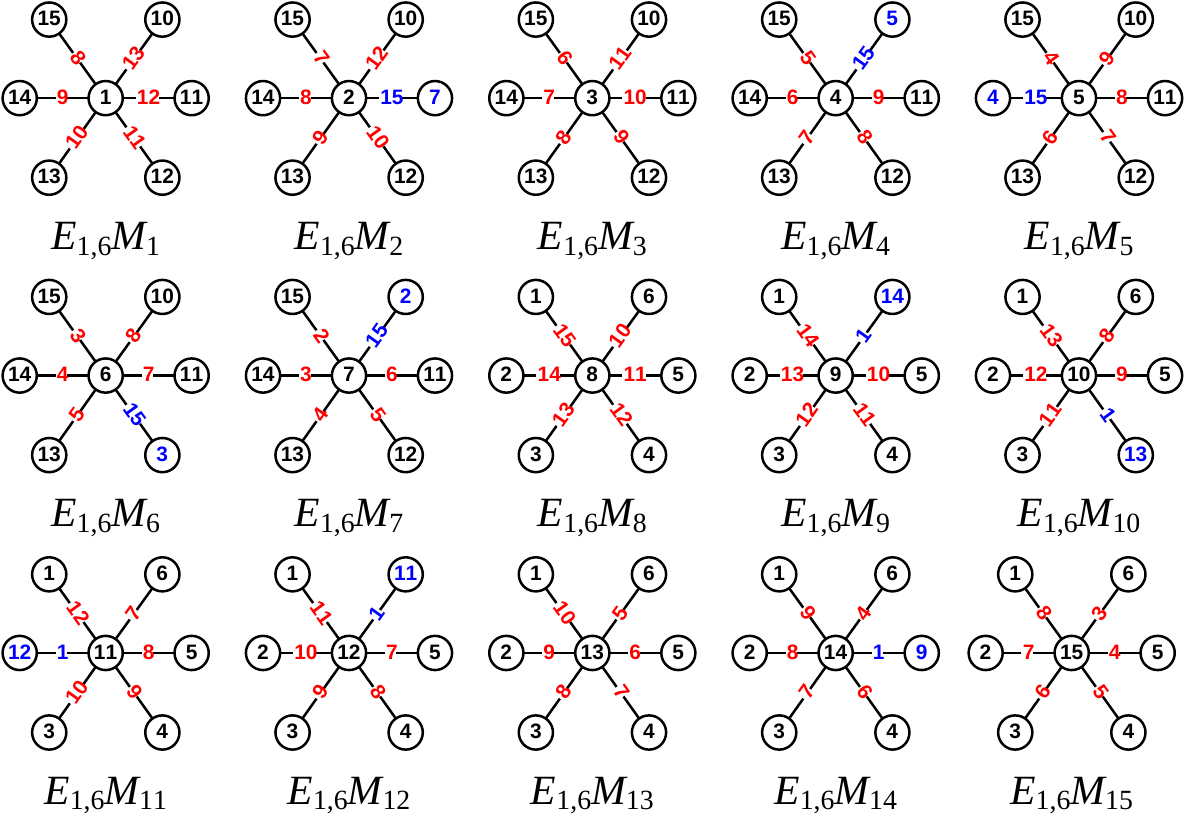}
\caption{\label{fig:Dsaturated-edge-magic-even}{\small An edge-magic ice-flower system $I_{ce}(E_{1,6}M_k)^{15}_{k=1}$.}}
\end{figure}

\begin{figure}[h]
\centering
\includegraphics[width=16.4cm]{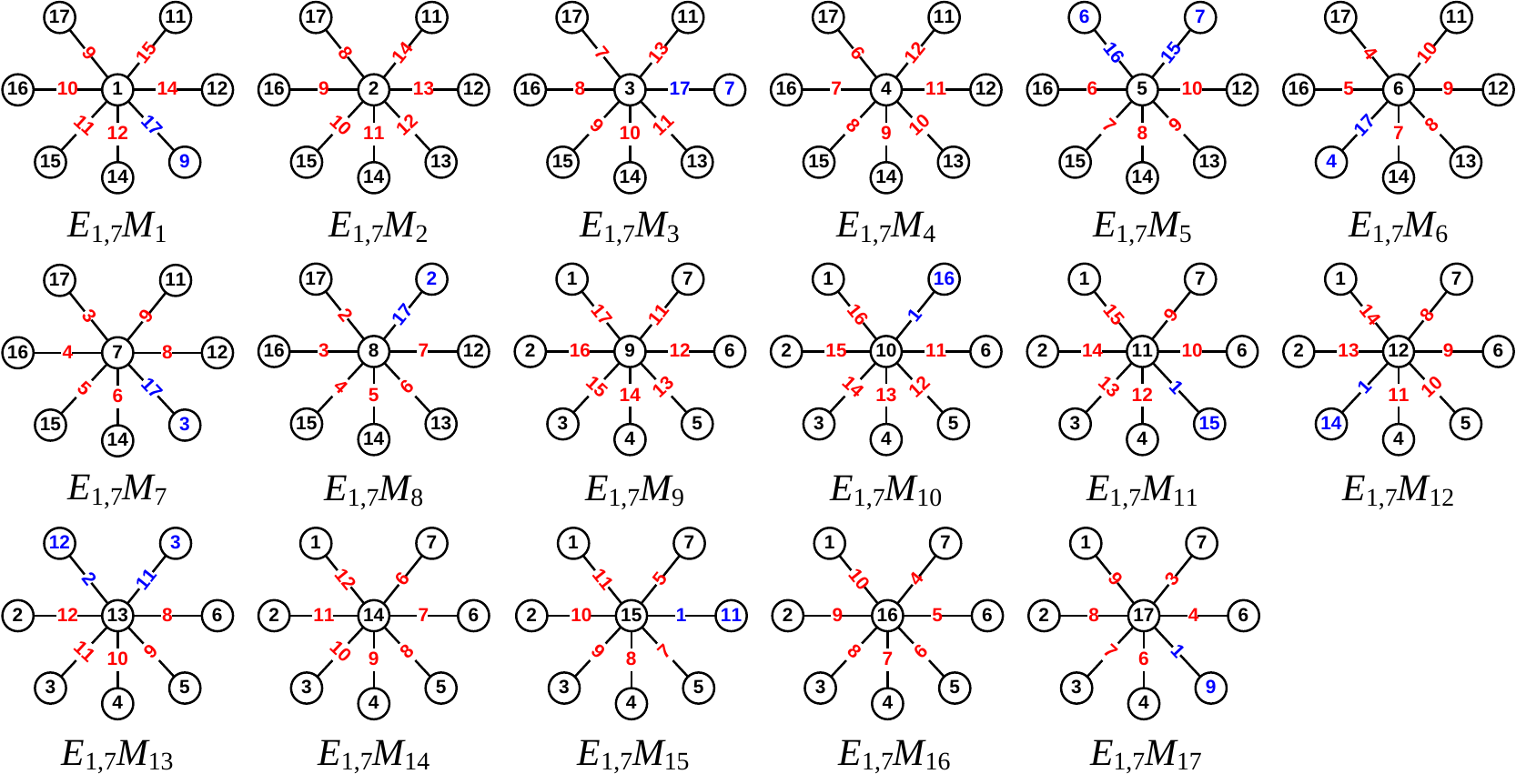}
\caption{\label{fig:Dsaturated-edge-magic-odd}{\small An edge-magic ice-flower system $I_{ce}(E_{1,7}M_k)^{17}_{k=1}$.}}
\end{figure}

The edge-magic ice-flower system $I_{ce}(E_{1,n}M_k)^{2n+3}_{k=1}$ enables us to get the following results:

\begin{lem} \label{thm:tree-edge-magic-ice-flower}
\cite{Wang-Su-Yao-Total-2019} Each tree $T$ holds $\chi''_{emt}(t)\leq 2\Delta(T)+3$ true.
\end{lem}
\begin{proof} Notice that $K^{(k)}_{1,n}$ has its own vertex set $V(K^{(k)}_{1,n})=\{x_{k,0},x_{k,j}:j\in [1,n]\}$ in the edge-magic ice-flower system $I_{ce}(E^L_{1,n}M_k)^{4n-1}_{k=1}$. We vertex-coincide $x_{1,j},x_{2,j},\dots ,x_{n,j}$ for each $j\in [1,n]$ into one $$y_j=x_{1,j}\odot x_{2,j}\odot \dots \odot x_{n,j},$$ the resultant graph is just a complete bipartite graph $K_{n,n}$ with bipartition $(X,Y)$, where $X=\{x_{i,0}:i\in [1,n]\}$ and $Y=\{y_{i}:i\in [1,n]\}$. Clearly, $\chi''_{emt}(K_{n,n})\leq 4n-1$.
\end{proof}

Each bipartite graph $G$ is a subgraph of $K_{M,M}$ with $M=\max \{|X|,|Y|\}$, we have

\begin{thm} \label{thm:property-edge-magic-ice-flower}
Each bipartite graph $G$ holds $\chi''_{emt}(G)\leq 4\max \{|X|,|Y|\}-1$ true.
\end{thm}

\subsubsection{Edge-magic star-graphic lattices}

An edge-magic ice-flower system $I_{ce}(E_{1,n}M_k)^{2n+3}_{k=1}$ enables us to get an \emph{$(EM)$-magic edge-magic star-graphic lattice} as folllows:
\begin{equation}\label{eqa:2n-3-edge-magic-stars-lattice}
\textrm{\textbf{L}}(\overline{\ominus} \textbf{I}_{ce}(EM)) =\left \{\overline{\ominus}^{2n+3}_{i=1}a_iE_{1,n}M_i: a_i\in Z^0, E_{1,n}M_i\in I_{ce}(E_{1,n}M_k)^{2n+3}_{k=1}\right \}
\end{equation} with $\sum^{2n+3}_{i=1} a_i\geq 1$ and the base is $\textbf{I}_{ce}(EM)=I_{ce}(E_{1,n}M_k)^{2n+3}_{k=1}$.

By the edge-magic ice-flower system $I_{ce}(E^L_{1,n}M_k)^{4n-1}_{k=1}$ we have an \emph{$(E^LM)$-magic edge-magic star-graphic lattice} as
\begin{equation}\label{eqa:4n-1-edge-magic-stars-lattice}
\textrm{\textbf{L}}(\overline{\ominus} \textbf{I}_{ce}(E^LM)) =\left \{\overline{\ominus}^{4n-1}_{i=1}a_iE^L_{1,n}M_i: a_i\in Z^0, E^L_{1,n}M_i\in I_{ce}(E^L_{1,n}M_k)^{4n-1}_{k=1}\right \}
\end{equation} with $\sum^{4n-1}_{i=1} a_i\geq 1$ and the base is $\textbf{I}^L_{ce}(EM)=I_{ce}(E^L_{1,n}M_k)^{4n-1}_{k=1}$.

\subsubsection{All-dual edge-magic star-graphic lattices}
Because of $\max \{\varphi_k(w):w\in V(G)\cup E(G)\}=4n-1$ and $\min \{\varphi_k(w):w\in V(G)\cup E(G)\}=1$ in the edge-magic ice-flower system $I_{ce}(E^L_{1,n}M_k)^{4n-1}_{k=1}$, the dual $\varphi^c_k$ of $\varphi_k$ is defined as: $\varphi^c_k(w)=4n-\varphi_k(w)$ for each element $w\in V(G)\cup E(G)$, so
$$\varphi^c_k(u)+\varphi^c_k(uv)+\varphi^c_k(v)=12n-[\varphi_k(u)+\varphi_k(uv)+\varphi_k(v)]=6n.$$
We claim that $\varphi^c_k$ and $\varphi_k$ are a pair of perfect all-dual edge-magic proper total colorings since $\varphi^c_k(w)+\varphi_k(w)=4n$ for each element $w\in V(G)\cup E(G)$. So, $\varphi^c_k(E^L_{1,n}M_k)$ is just the perfect all-dual star of $E^L_{1,n}M_k$ of the edge-magic ice-flower system $I_{ce}(E^L_{1,n}M_k)^{4n-1}_{k=1}$, we get the all-dual edge-magic ice-flower system $I^c_{ce}(\varphi^c_k(E^L_{1,n}M_k))^{4n-1}_{k=1}$ of the edge-magic ice-flower system $I_{ce}(E^L_{1,n}M_k)^{4n-1}_{k=1}$, and an edge-magic graphic lattice
\begin{equation}\label{eqa:all-dual-edge-magic-stars-lattice}
\textrm{\textbf{L}}(\overline{\ominus} \textbf{I}^c_{ce}(E^LM)) =\left \{\overline{\ominus}^{4n-1}_{i=1}a_i\varphi^c_k(E^L_{1,n}M_i): a_i\in Z^0, \varphi^c_k(E^L_{1,n}M_i)\in I^c_{ce}(\varphi^c_k(E^L_{1,n}M_k))^{4n-1}_{k=1}\right \}
\end{equation} with $\sum^{4n-1}_{i=1} a_i\geq 1$ and the base is $\textbf{I}^c_{ce}(E^LM)=I^c_{ce}(\varphi^c_k(E^L_{1,n}M_k))^{4n-1}_{k=1}$.

\subsubsection{Optimal edge-magic ice-flower systems}

Let $G$ be a graph of $n$ vertices and admit an edge-magic proper total coloring $f$ such that $\max\{f(w):w\in V(G)\cup E(G)\}=\chi''_{emt}(G)$. We do a series of leaf-splitting operations to $G$, such that the resultant graph is the union of colored stars $E_{1,t}M_k$ with $t\in \{d_i(G)\}$, where $\{d_i(G)\}$ is a degree subsequence of the degree sequence of $G$ with $d_i(G)\geq 2$. So, we get $G=\overline{\ominus}^{n_{G}}_{j=1}a_jE_{1,t}M_j$ with $\Sigma^{n_{G}}_{j=1}a_j\geq 1$, and we say $I_{ce}(EM,G)=\overline{\ominus}^{n_{G}}_{j=1}E_{1,t}M_j$ to be a \emph{$G$-edge-magic ice-flower system}. Thereby, we have a \emph{complete edge-magic ice-flower system} $I_{ce}(EM,n)=\bigcup_{G,|G|=n} I_{ce}(EM,G)$, and a \emph{complete edge-magic ice-flower graphic lattice}
\begin{equation}\label{eqa:complete-edge-magic-stars-lattice}
\textrm{\textbf{L}}(\overline{\ominus} \textbf{I}_{ce}(EM,n)) =\left \{\overline{\ominus}^{m(n)}_{i=1}a_iE_{1,t}M_i: a_i\in Z^0, E_{1,t}M_i\in I_{ce}(EM,n)\right \}
\end{equation} with $\sum^{m(n)}_{i=1} a_i\geq 1$ and the base is $I_{ce}(EM,n)$, where $m(n)$ is the cardinality of the complete edge-magic ice-flower system. In other word, this complete edge-magic ice-flower system $I_{ce}(EM,n)$ is \emph{optimal}.

\begin{thm} \label{thm:optimal-edge-magic-ice-flower}
There exists an optimal edge-magic ice-flower system $I_{ce}(EM,n)$ such that each graph $G$ of $n$ vertices is isomorphic to a graph $H\in \textrm{\textbf{L}}(\overline{\ominus} \textbf{I}_{ce}(EM,n)) $.
\end{thm}

We make a largest edge-magic ice-flower system $I_{ce}(E^M_{1,n-1}M_k)^{n}_{k=1}$ as follows: Suppose that each copy $K^{(k)}_{1,n-1}$ of $K_{1,n-1}$ has its own vertex set $V(K^{(k)}_{1,n-1})=\{x^k_0$, $x^k_1$, $x^k_2$, $\dots ,x^k_{n-1}\}$ and edge set $E(K^{(k)}_{1,n})=\{x^k_0x^k_j:j\in[1,n-1]\}$. We define a coloring $\theta_k$ as: $\theta_k(x^k_0)=k$ for $k\in [1,n]$, $\theta_k(x^k_j)=j$ for $j\in [1,k-1]\cup [k+1,n-1]$, $\theta_k(x^k_0x^k_j)=3n-\theta_k(x^k_0)-\theta_k(x^k_j)$ with $j\in [1,k-1]\cup [k+1,n-1]$ and $k\in [1,n]$. It is not hard to verify that $\theta_k$ is an edge-magic proper total coloring, and we get a largest edge-magic ice-flower system $I_{ce}(E^M_{1,n-1}M_k)^{n}_{k=1}$ based on the edge-magic proper total coloring $\theta_k$. Thereby, we obtain
\begin{thm} \label{thm:largest-edge-magic-ice-flower}
Each complete graph $K_n$ can be expressed as $K_n=\overline{\ominus}^{n}_{k=1}E^M_{1,n-1}M_k$, such that $\chi''_{emt}(K_n)=3(n-1)$ for $n\geq 3$, and moreover $\chi''_{emt}(G)\leq 3(n-1)$ for each graph $G$ of $n$ vertices.
\end{thm}

\subsection{4-ice-flower lattices}

We have the graceful-difference ice-flower system $I_{ce}(G_{1,n}D_k)^{2n+3}_{k=1}$, the edge-difference ice-flower system $I_{ce}(E_{1,n}D_k)^{2n+3}_{k=1}$, the edge-magic ice-flower system $I_{ce}(E_{1,n}M_k)^{2n+3}_{k=1}$, and two felicitous-difference ice-flower systems $I_{ce}(F_{1,n}D_k)^{2n}_{k=1}$ and $I_{ce}(SF_{1,n}D_k)^{n}_{k=1}$. Suppose that there are two stars $K^G_{1,n}$ admitting a $W_G$-type coloring $f_G$ and $K^H_{1,n}$ admitting a $W_H$-type coloring $f_H$, an edge $uv\in E(K^G_{1,n})$ with leaf $v$ and another edge $xy\in E(K^H_{1,n})$ with leaf $y$ in the above five ice-flower systems.

1. If $K^G_{1,n}$ and $K^H_{1,n}$ belong to the same ice-flower system, then do a leaf-coinciding operation to them, we get $K^G_{1,n}\overline{\ominus} K^H_{1,n}$ when $f_G(u)=f_H(x)$, $f_G(x)=f_H(y)$ and $f_G(uv)=f_H(xy)$.

2. Do a vertex-coinciding operation to $K^G_{1,n}$ and $K^H_{1,n}$, we get $K^G_{1,n}\odot K^H_{1,n}$ by vertex-coinciding $u$ with $x$ into one $w=u\odot x$ as $f_G(u)=f_H(x)$ and $f_G(uv)\neq f_H(xy)$, and color $w$ with the color $f_G(u)$, $K^G_{1,n}\odot K^H_{1,n}$ admits a $W$-type coloring $g$, and delete some leaves $z_i$ if $g(wz_i)=g(wz_j)$ such that $g$ is really a proper total coring of the resulting graph.

By $(\overline{\ominus}^{2n}_{i=1}a_iF_{1,n}D_i)$ in (\ref{eqa:stars-lattice}), $ (\overline{\ominus}^n_{i=1}b_iSF_{1,n}D_i)$ in (\ref{eqa:smallest-stars-lattice}), $(\overline{\ominus}^{2n+3}_{k=1}a_iE_{1,n}M_i)$ in (\ref{eqa:2n-3-edge-magic-stars-lattice}),
$(\overline{\ominus}^{2n+3}_{k=1}d_iE_{1,n}D_i)$ in (\ref{eqa:edge-difference-stars-lattice}) and $(\overline{\ominus}^{2n+3}_{k=1}e_iG_{1,n}D_i)$ in (\ref{eqa:graceful-difference-stars-lattice}) with $a_i,b_i,c_i,d_i,e_i\in Z^0$ and $\sum \varepsilon_i\geq 1$ for $\varepsilon=a,b,c,d,e$, we have a \emph{4-ice-flower graph set}

\begin{equation}\label{eqa:4-ice-flower-lattice}
{
\begin{split}
[\overline{\ominus} 4_{1,n}\odot]_{a_ib_ic_id_ie_i}=&\left (\overline{\ominus}^{2n}_{i=1}a_iF_{1,n}D_i\right )\odot \left (\overline{\ominus}^n_{i=1}b_iSF_{1,n}D_i\right )\odot \left (\overline{\ominus}^{2n+3}_{i=1}c_iE_{1,n}M_i\right )\\
&\odot \left (\overline{\ominus}^{2n+3}_{i=1}d_iE_{1,n}D_i\right )\odot \left (\overline{\ominus}^{2n+3}_{i=1}e_iG_{1,n}D_i \right)
\end{split}}
\end{equation}
and a \emph{4-ice-flower graphic lattice} as follows
\begin{equation}\label{eqa:graceful-difference-stars-lattice}
\textrm{\textbf{L}}(\overline{\ominus} \odot\textbf{4I}_{ce}) = \bigcup [\overline{\ominus} 4_{1,n}\odot]_{a_ib_ic_id_ie_i}
\end{equation} See an example for the 4-ice-flower graphic lattice $\textrm{\textbf{L}}(\overline{\ominus} \odot\textbf{4I}_{ce})$ shown in Fig.\ref{fig:4-ice-flower-lattice}.

\begin{figure}[h]
\centering
\includegraphics[width=16.2cm]{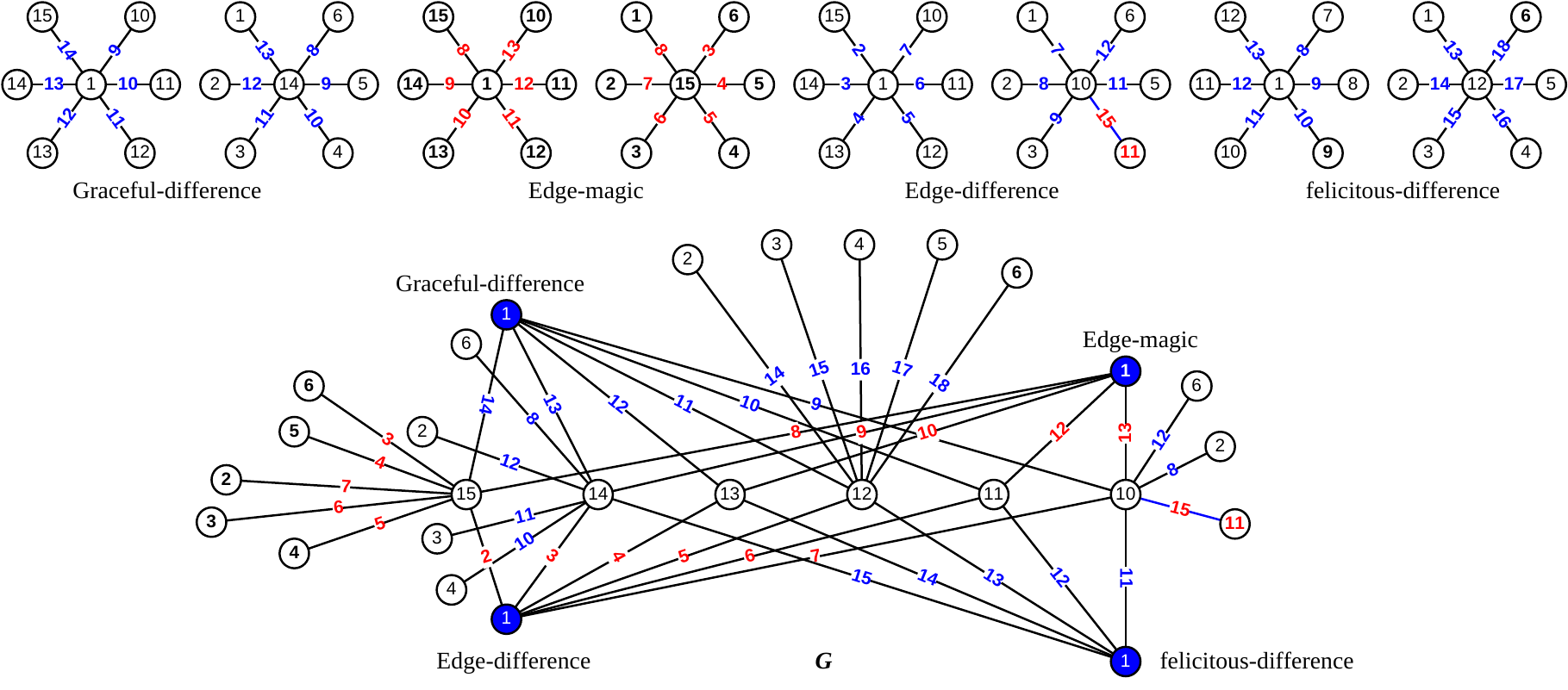}
\caption{\label{fig:4-ice-flower-lattice}{\small A graph $G$ made by four ice-flower systems.}}
\end{figure}

\begin{defn} \label{defn:4-ice-flower-total-coloring}
$^*$ Let $f:V(G)\cup E(G)\rightarrow [1,M]$ be a proper total coloring of a simple graph $G$. If $E(G)=\bigcup^4_{i=1}E_i$ with $E_i\cap E_j=\emptyset$ and $E_i\neq \emptyset$ for $1\leq i,j\leq 4$, such that $f(u)+f(uv)+f(v)=k_1$ for each edge $uv\in E_1$, $f(xy)+|f(x)-f(y)|=k_2$ for each edge $xy\in E_2$, $|f(s)+f(t)-f(st)|=k_3$ for each edge $st\in E_3$, and $\big ||f(a)-f(b)|-f(ab)\big |=k_4$ for each edge $ab\in E_4$. We call $f$ a \emph{4-ice-flower proper total coloring} of $G$, and the smallest number of $\max\{f(w):w\in V(G)\cup E(G)\}$ for which $G$ admits a 4-ice-flower proper total coloring is denoted as $\chi''_{\textrm{4ice}}(G)$, called \emph{4-ice-flower total chromatic number}.\qqed
\end{defn}

Similarly, we can define a \emph{$k$-ice-flower proper total coloring} with $k\in [1,4]$, in general.
\begin{thm} \label{thm:4-ice-flower-coloring}
Every tree $T$ with maximum degree $\Delta \geq 4$ and diameter $D(T)\geq 5$ admits a 4-ice-flower proper total coloring.
\end{thm}

\subsection{General star-graphic lattices}

1. Uncolored stars. In general, let $F_{star\Delta}$ be the set of uncolored stars $K_{1,k}$ with $k\leq \Delta$, the uncolored leaf-coinciding operation ``$K_{1,i}\overline{\ominus} K_{1,j}$'' enables us to get an \emph{uncolored star-graphic lattice}
\begin{equation}\label{eqa:uncolored-stars-lattice}
\textrm{\textbf{L}}(\overline{\ominus} \textbf{F}_{star\Delta}) =\left \{\overline{\ominus}^{\Delta}_{k=2}a_k K_{1,k}: a_k\in Z^0, K_{1,k}\in F_{star\Delta}\right \}
\end{equation} containing all uncolored graphs with maximum degrees no more than $\Delta$, and $\sum^{\Delta}_{k=2} a_k\geq 1$.

\begin{thm} \label{thm:General-star-type-graphic-lattices}
Each graph is contained in the uncolored star-graphic lattice $\textrm{\textbf{L}}(\overline{\ominus} \textbf{F}_{star\Delta})$.
\end{thm}

As known, many problems of graph theory can be expressed or illustrated by uncolored star-graphic lattices, such as hamiltonian graphs discussed in \cite{Yao-Wang-Liu-Wang-Ma-Su-Sun-2020star-graphic}.

\begin{exa}\label{exa:2020star-graphic}
Each star $K_{1,m_i}$ of an ice-flower system $\textbf{\textrm{K}}=(K_{1,m_1},K_{1,m_2},\dots, K_{1,m_n})$ has its vertex set $V(K_{1,m_i})=\{x_i,y_{i,j}:j\in [1,m_i]\}$ and its edge set $V(K_{1,m_i})=\{x_iy_{i,j}:j\in [1,m_i]\}$ with $m_i\geq 2$, so $K_{1,m_i}$ has its own leaf set $L(K_{1,m_i})=\{y_{i,1},y_{i,2},\dots ,y_{i,m_i}\}$ and its unique non-leaf vertex $x_i$ with $\textrm{deg}_{K_{1,m_i}}(x_i)=m_i$. As known, a sequence $\textbf{\textrm{d}}=(m_1$, $ m_2, \dots , m_n)$ with $1\leq m_{i}\leq m_{i+1}$ to be the degree sequence of a certain graph $G$ of $n$ vertices if and only if $\sum^n_{i=1}m_i$ is even and
\begin{equation}\label{eqa:c3xxxxx}
\sum^k_{i=1}m_i\leq k(k-1)+\sum ^n_{i=k+1}\min\{k,m_i\}
\end{equation}
shown by Erd\"{o}s and Gallai in 1960 \cite{Bondy-2008}.

We leaf-coincide two leaf-edges $x_iy_{i,m_{i}}$ and $x_{i+1}y_{i+1,1}$ for each pair of $K_{1,m_i}$ and $K_{1,m_{i+1}}$ for $i\in [1,n-1]$ into one edge $x_ix_{i+1}$ joining $K_{1,m_i}$ and $K_{1,m_{i+1}}$ together, such that $x_i \odot y_{i+1,1}$ and $x_{i+1}\odot y_{i,m_{i}}$. Then we get a caterpillar $T$ for its ride $P=x_1x_2\dots x_n$ with $x_i\in V(K_{1,m_i})$ having vertex set $V(K_{1,m_i})=\{x_i,y_{i,j}:j\in [1,m_i]\}$, next we leaf-coincide two leaf-edges $x_1y_{1,1}$ and $x_{n}y_{n,m_{n}}$ into one edge $x_1x_{n}$ with $x_1\odot y_{n,m_{n}}$ and $x_{n}\odot y_{1,1}$, the resultant graph, denoted as $T^*$, is like a \emph{haired-cycle}.
We write $T^*=\overline{\ominus}^n_{j=1}K_{1,m_j}$, and then do some leaf-coinciding operations on some pairs of leaf-edges of $T^*$ to get a connected graph $G$, such that $G$ has no leaf. Thereby, $G$ has a cycle $C=x_1x_2\dots x_nx_1$ containing each vertex of $G$, that is, $G$ is \emph{hamiltonian}. Especially, we denote $G$ by
\begin{equation}\label{eqa:hamiltonian}
G=\overline{\ominus} T^*=\overline{\ominus}[\overline{\ominus}^n_{j=1}K_{1,m_j}]=\overline{\ominus}^2|^n_{j=1}K_{1,m_j}.
\end{equation} Notice that there are two or more hamiltonian graphs like $G$ by leaf-coinciding some pairs of leaf-edges of $T^*$. On the other hand, each permutation $k_{1}k_{2}\cdots k_{n}$ of $m_{1}m_{2}\cdots m_{n}$ distributes us a set of hamiltonian graphs $\overline{\ominus}^2|^n_{j=1}K_{1,k_j}$. As known, there are $M_p(=n!)$ permutations, we have $M_p$ ice-flower systems $\textbf{\textrm{K}}_k=(K_{1,k_1},K_{1,k_2},\dots, K_{1,k_n})$ with $k\in [1,M_p]$, and each ice-flower system $\textbf{\textrm{K}}_k$ induces a set of hamiltonian graphs $\overline{\ominus}^2|^n_{j=1}K_{1,k_j}$, we write this set as $\overline{\ominus}^2\textbf{\textrm{K}}_k$ and $P_{ermu}(\textbf{\textrm{K}})$ to be the set of $M_p$ ice-flower systems $\textbf{\textrm{K}}_k$. Then the following set
\begin{equation}\label{eqa:hamiltonian-star-lattices}
\textbf{\textrm{L}}(\overline{\ominus}^2 P_{\textrm{ermu}}(\textbf{\textrm{K}}))=\left \{\overline{\ominus}^2\big |^{M_p}_{k=1}a_k\textbf{\textrm{K}}_k,~\textbf{\textrm{K}}_k\in P_{\textrm{ermu}}(\textbf{\textrm{K}})\right \}
\end{equation} with $\sum^{M_p}_{k=1}a_k=1$ a \emph{hamiltonian star-graphic lattice}. We can construct graphs $G=\overline{\ominus}^{m+1}\big |^{M_p}_{k=1}a_k\textbf{\textrm{K}}_k$ with $\textbf{\textrm{K}}_k\in P_{\textrm{ermu}}(\textbf{\textrm{K}})$ containing $m$ edge-disjoint Hamilton cycles, also, such technique can be used to deal with regular graphs. An Euler graph $H$ with degree sequence $\textbf{\textrm{d}}=(2m_1$, $ 2m_2, \dots , 2m_n)$ can be expressed as $H=\overline{\ominus} |^{n}_{j=1}K_{1,2m_j}$, and so on.\qqed
\end{exa}

2. Colored stars. Suppose that $F^c_{star\Delta}$ is the set of colored stars $C^K_{1,k}$ ($\cong K_{1,k}$) with $k\leq \Delta$, in general. We define a \emph{general star-graphic lattice}
\begin{equation}\label{eqa:general-stars-lattice}
\textrm{\textbf{L}}(\overline{\ominus} \textbf{F}^c_{star\Delta}) =\left \{\overline{\ominus}^{\Delta}_{k=2}a_k C^K_{1,k}: a_k\in Z^0, C^K_{1,k}\in F^c_{star\Delta}\right \}
\end{equation} with $\sum^{\Delta}_{k=2} a_k\geq 1$, this graphic lattice contains all colored graphs with maximum degrees no more than $ \Delta$. Clearly, $\textrm{\textbf{L}}(\overline{\ominus} \textbf{I}_{ces}(FD))\subset \textrm{\textbf{L}}(\overline{\ominus} \textbf{F}^c_{star\Delta})$, $\textrm{\textbf{L}}(\overline{\ominus} \textbf{I}_{ces}(SFD))\subset \textrm{\textbf{L}}(\overline{\ominus} \textbf{F}^c_{star\Delta})$ and $\textrm{\textbf{L}}(\odot \textbf{I}_{ces}(FD))\subset \textrm{\textbf{L}}(\overline{\ominus} \textbf{F}^c_{star\Delta})$.

Let $S_{\textrm{X}}\subset \textrm{\textbf{L}}(\overline{\ominus} \textbf{F}^c_{star\Delta})$ be the set of all trees admitting $W$-type colorings. Each tree $T\in S_{\textrm{X}}$ corresponds a set of connected graphs, such that each graph $G$ of this set can be leaf-split into a tree $H$ of $S_{\textrm{X}}$, we say $G$ corresponds to $H$, and vice versa. Then we have

\begin{thm} \label{thm:graphic-lattices-GTC}
A connected graph admits a $W$-type coloring if and only if its corresponding a tree admitting a $W$-type coloring too.
\end{thm}

\textbf{Spanning star-graphic lattices.} A connected graph $T$ is a tree if and only if $n_1(T)=2+\Sigma_{d\geq 3}(d-2)n_d(T)$ (Ref. \cite{Yao-Zhang-Yao-2007, Yao-Zhang-Wang-Sinica-2010}). An ice-flower system $\textbf{\textrm{K}}^c$ defined as: Each $K^c_{1,m_j}$ of $\textbf{\textrm{K}}^c$ admits a proper vertex coloring $g_j$ such that $g_j(x)\neq g_j(y)$ for any pair of vertices $x,y$ of $K^c_{1,m_j}$, and each leaf-coinciding graph $T=\overline{\ominus}|^n_{j=1}K^c_{1,m_j}$ is connected based on the leaf-coinciding operation ``$K^c_{1,m_i}\overline{\ominus} K^c_{1,m_j}$'', such that

(1) $n_1(T)=2+\Sigma_{d\geq 3}(d-2)n_d(T)$ holds true;

(2) $T$ admits a proper vertex coloring $f=\overline{\ominus}|^n_{j=1}g_j$ with $f(u)\neq f(w)$ for any pair of vertices $u,w$ of $T$.

We get a set $\textbf{\textrm{L}}(\overline{\ominus} (m,g)\textbf{\textrm{K}}^c)$ containing the above leaf-coinciding graphs $T=\overline{\ominus}|^n_{j=1}K^c_{1,m_j}$ if $|V(T)|=m$. Since each graph $T\in\textbf{\textrm{L}}(\overline{\ominus} (m,g)\textbf{\textrm{K}}^c)$ is a tree, and Cayley's formula $\tau(K_m)=m^{m-2}$ in graph theory (Ref. \cite{Bondy-2008}) tells us the number of elements of $\textbf{\textrm{L}}(\overline{\ominus} (n)\textbf{\textrm{K}}^c)$ to be equal to $m^{m-2}$. We call this set
$$\textbf{\textrm{L}}(\overline{\ominus} (m,g)\textbf{\textrm{K}}^c)=\{\overline{\ominus}|^n_{j=1}K^c_{1,m_j},~K^c_{1,m_j}\in \textbf{\textrm{K}}^c\}$$ a \emph{spanning star-graphic lattice}.

\textbf{A 4-color ice-flower system.} Each star $K_{1,d}$ with $d\in [2,M]$ admits a proper vertex-coloring $f_i$ with $i\in [1,4]$ defined as $f_i(x_0)=i$, $f_i(x_j)\in [1,4]\setminus\{i\}$, and $f_i(x_s)\neq f_i(x_t)$ for some $s,t\in [1,d]$, where $V(K_{1,d})=\{x_0,x_j:j\in [1,d]\}$. For each pair of $d$ and $i$, $K_{1,d}$ admits $n(d,i)$ proper vertex-colorings like $f_i$ defined above. Such colored star $K_{1,d}$ is denoted as $P_dS_{i,k}$, we have a set $(P_{d}S_{i,k})^{n(a,i)}_{k=1}$ with $i\in [1,4]$ and $d\in [2,M]$, and moreover we obtain a \emph{4-color ice-flower system} $\textbf{\textrm{I}}_{ce}(PS,M)=I_{ce}(P_{d}S_{i,k})^{n(a,i)}_{k=1})^{4}_{i=1})^{M}_{d=2}$, which induces a \emph{$4$-color star-graphic lattice}
\begin{equation}\label{eqa:4-color-star-system-lattices}
\textbf{\textrm{L}}(\Delta\overline{\ominus} \textbf{\textrm{I}}_{ce}(PS,M))=\left \{\Delta\overline{\ominus}^{A}_{(d,i,k)} a_{d,i,k}P_{d}S_{i,j}:~a_{d,i,k}\in Z^0,~P_{d}S_{i,j}\in \textbf{\textrm{I}}_{ce}(PS,M)\right \}
\end{equation} with $\sum ^{A}_{(d,i,k)} a_{d,i,k}\geq 3$, and the base is $\textbf{\textrm{I}}_{ce}(PS,M)=I_{ce}(P_{d}S_{i,k})^{n(a,i)}_{k=1})^{4}_{i=1})^{M}_{d=2}$, where $A=|\textbf{\textrm{I}}_{ce}(PS,M)|$, and the operation ``$\Delta\overline{\ominus}$'' is doing a series of leaf-coinciding operations to colored stars $a_{d,i,k}P_{d}S_{i,j}$ such that the resultant graph to be a planar with each inner face being a triangle.

See two examples shown in Fig.\ref{fig:process-leaf-splitting} and Fig.\ref{fig:4-color-star-system-example} about 4-color ice-flower systems.

\begin{figure}[h]
\centering
\includegraphics[width=16cm]{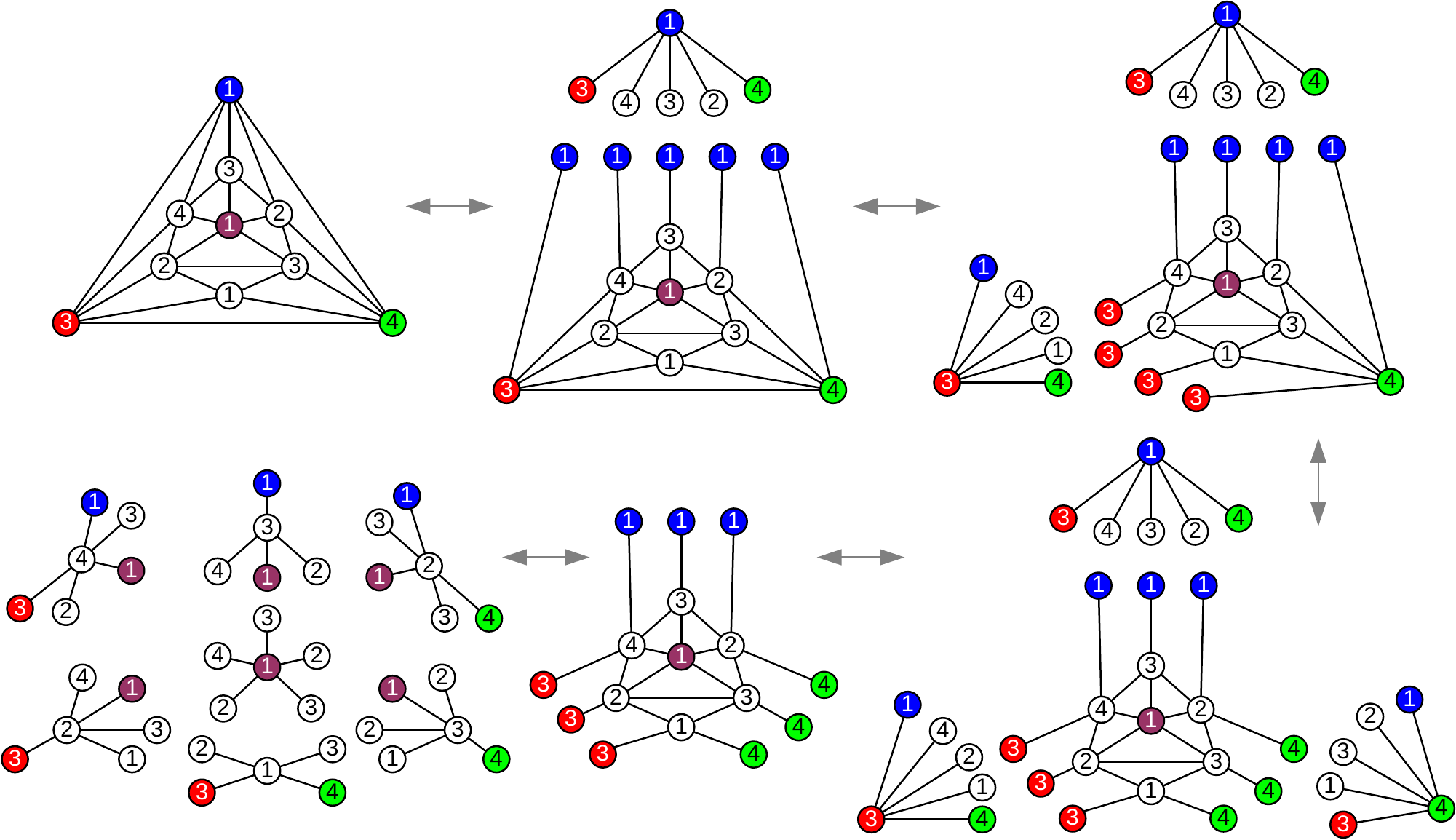}
\caption{\label{fig:process-leaf-splitting}{\small A process of doing leaf-splitting operations.}}
\end{figure}

\begin{figure}[h]
\centering
\includegraphics[width=16cm]{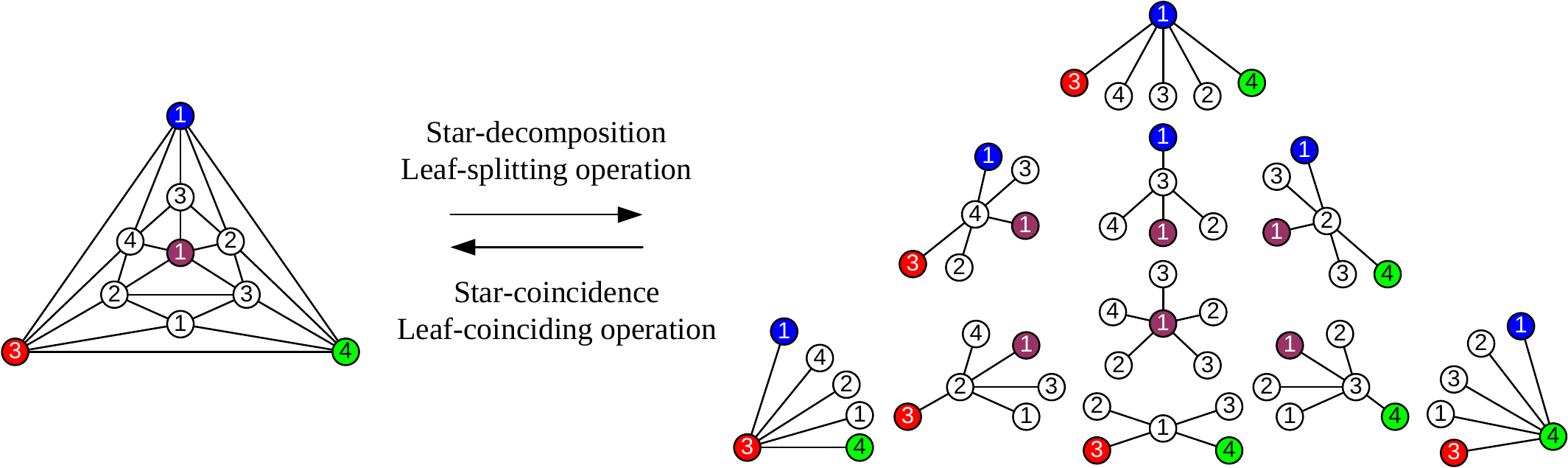}
\caption{\label{fig:4-color-star-system-example}{\small An example for understanding the $4$-coloring ice-flower system.}}
\end{figure}

\begin{thm} \label{thm:4-color-star-system-lattices}
Each graph of the $4$-coloring star-graphic lattice $\textbf{\textrm{L}}(\Delta\overline{\ominus} \textbf{\textrm{I}}_{ce}(PS,M))$ is 4-colored well and planar with each inner face being a triangle and has maximum degree $\leq M$.
\end{thm}

\begin{problem}\label{qeu:GTC}
The $4$-coloring star-graphic lattice helps us to ask for the following questions:
\begin{asparaenum}[\textrm{GTC}-1. ]
\item \textbf{Find} various sublattices of the general star-graphic lattice $\textrm{\textbf{L}}(\overline{\ominus} \textbf{F}^c_{star\Delta})$ by graph colorings/labellings.
\item Let $\textrm{\textbf{L}}(\overline{\ominus} \textbf{G}_{race}(T))$ be the set of graceful trees of $\textrm{\textbf{L}}(\overline{\ominus} \textbf{F}^c_{star\Delta})$. \textbf{Does} $\textrm{\textbf{L}}(\overline{\ominus} \textbf{G}_{race}(T))$ contains every tree with maximum degree $\Delta$?
\item A \emph{gracefully total coloring} $f$ of a tree $T$ is a proper total coloring $f:V(T)\cup E(T)\rightarrow [1,M]$ such that such that $f(x)=f(y)$ for some pair of vertices $x,y\in V(T)$, $f(uv)=|f(u)-f(v)|$ for each edge $uv\in E(T)$, and $f(E(T))=[1,|V(T)|-1]$. A star-graphic sublattice $\textrm{\textbf{L}}(\overline{\ominus} \textbf{T}_{grace})\subset \textrm{\textbf{L}}(\overline{\ominus} \textbf{F}^c_{star\Delta})$ is formed by all trees admitting gracefully total colorings in $\textrm{\textbf{L}}(\overline{\ominus} \textbf{F}^c_{star\Delta})$. \textbf{Does} each tree $H$ with maximum degree $\Delta$ correspond a tree $H'\in \textrm{\textbf{L}}(\overline{\ominus} \textbf{T}_{grace})$ such that $H\cong H'$?
\item \textbf{Is} every planar graph isomorphic to a group of 4-colored planar graphs of the $4$-coloring star-graphic lattice $\textbf{\textrm{L}}(\overline{\ominus} \textbf{\textrm{I}}_{ce}(PS,M))$?
\item \textbf{Find} connection between a planar graphic lattice $\textbf{\textrm{L}}(\textbf{\textrm{T}}^r\bigtriangleup F_{\textrm{inner}\bigtriangleup})$ defined in (\ref{eqa:planar-lattice}) and a $4$-coloring star-graphic lattice $\textbf{\textrm{L}}(\overline{\ominus} \textbf{\textrm{I}}_{ce}(PS,M))$ defined in (\ref{eqa:4-color-star-system-lattices}).
\item \textbf{Tree and planar graph authentication.} As known, each tree $T$ admits a proper $k$-coloring $f$ with $k\geq 2$, we do a vertex-coinciding operation to some vertices $x,y$ with $f(x)=f(y)$, such that $w=x\odot y$ and $f(w)=f(x)=f(y)$, and the resultant graph $T^*$ obtained by doing a series of vertex-coinciding operations to those vertices of $T$ colored with the same color is just a $k$-colorable graph with some particular properties. For instance, $T^*$ is a $k$-colorable Euler's graph holding $\chi(T^*)=k$, or a $k$-colorable planar graph with each inner face to be a triangle, or a $k$-colorable Hamilton graph, \emph{etc.} \textbf{Characterize} a $4$-colorable tree $T$ (as a public key) such that $T^*$ (as a private key) is a $4$-colorable planar graph, or a $4$-colorable planar graph with each inner face to be a triangle, that is, $T$ admits $4$-colorable graph homomorphism to $T^*$.

\item \textbf{Tree topological authentication.} By Theorem \ref{thm:trees-vs-connected-graphs}: Each connected graph $G$ corresponds a tree $T$ based on the vertex-splitting operation and the leaf-splitting operation, denoted as $(\wedge, \prec )(G)=T$. Conversely, a tree $T$ can produce a connected graph $G$ by means of the vertex-coinciding operation and the leaf-coinciding operation, or a mixed operation of them, for the convenience of statement, we write this process as $(\odot, \overline{\ominus})(T)=G$. Suppose that a tree $T$ (as a public key) admits a $W$-type coloring $f_T$ and $G$ (as a private key) admits a $W$-type coloring $h_G$ too, if $(\odot, \overline{\ominus})(f_T)=h_G$ in the process $(\odot, \overline{\ominus})(T)=G$, we say $(\odot, \overline{\ominus})(T)=G$ to be a \emph{topological authentication}. Oppositely, $(\wedge, \prec )(G)=T$ is called a topological authentication if $(\wedge, \prec )(h_G)=f_T$. For a given public tree $T$, \textbf{find} a $W$-type coloring $f_T$ of $T$, and \textbf{determine} a connected graph $G$ admitting a $W$-type coloring $h_G$ such that $(\odot, \overline{\ominus})(T)=G$ and $(\odot, \overline{\ominus})(f_T)=h_G$ true.\qqed
\end{asparaenum}
\end{problem}

\subsection{Star-type $H$-graphic lattices}

Based on various ice-flower systems $I_{ce}(F_{1,n}D_k)^{2n}_{k=1}$ and $I_{ce}(\mu)^{2n+3}_{k=1}$ with $\mu\in \{E_{1,n}M_k,E_{1,n}D_k$, $G_{1,n}D_k\}$ and the fully vertex-replacing operation, we present the following \emph{$H$-star-graphic lattices} under the fully vertex-replacing operation:
\begin{equation}\label{eqa:fully-replacing-star-saturated-system-1}
\textbf{\textrm{L}}(\textbf{\textrm{T}}\triangleleft F_{p,q})=\left \{H\triangleleft |^{2n+3}_{k=1} a_k\mu:~a_k\in Z^0,~H\in F_{p,q}\right \}
\end{equation}
with $\mu\in \{E_{1,n}M_k,E_{1,n}D_k$, $G_{1,n}D_k\}$ and $\sum^{2n+3}_{k=1} a_k\geq 1$; and
\begin{equation}\label{eqa:fully-replacing-star-saturated-system-2}
\textbf{\textrm{L}}(\textbf{\textrm{T}}\triangleleft F_{p,q})=\left \{H\triangleleft |^{2n}_{j=1} a_jF_{1,n}D_j:~a_j\in Z^0,~H\in F_{p,q}\right \}
\end{equation} where $\sum^{2n}_{j=1} a_j\geq 1$. The $H$-star-graphic lattices enable us to obtain a result:
\begin{thm} \label{thm:two-saturated-graphs}
Any $\Delta$-saturated graph $G$ grows or induces to another ($\Delta$-saturated) graph $G'$ such that they admit the same $W$-type proper total colorings. (see examples shown in Fig.\ref{fig:1-edge-diff-replacing-2} and Fig.\ref{fig:2-edge-diff-re-grow})
\end{thm}

\begin{figure}[h]
\centering
\includegraphics[width=16cm]{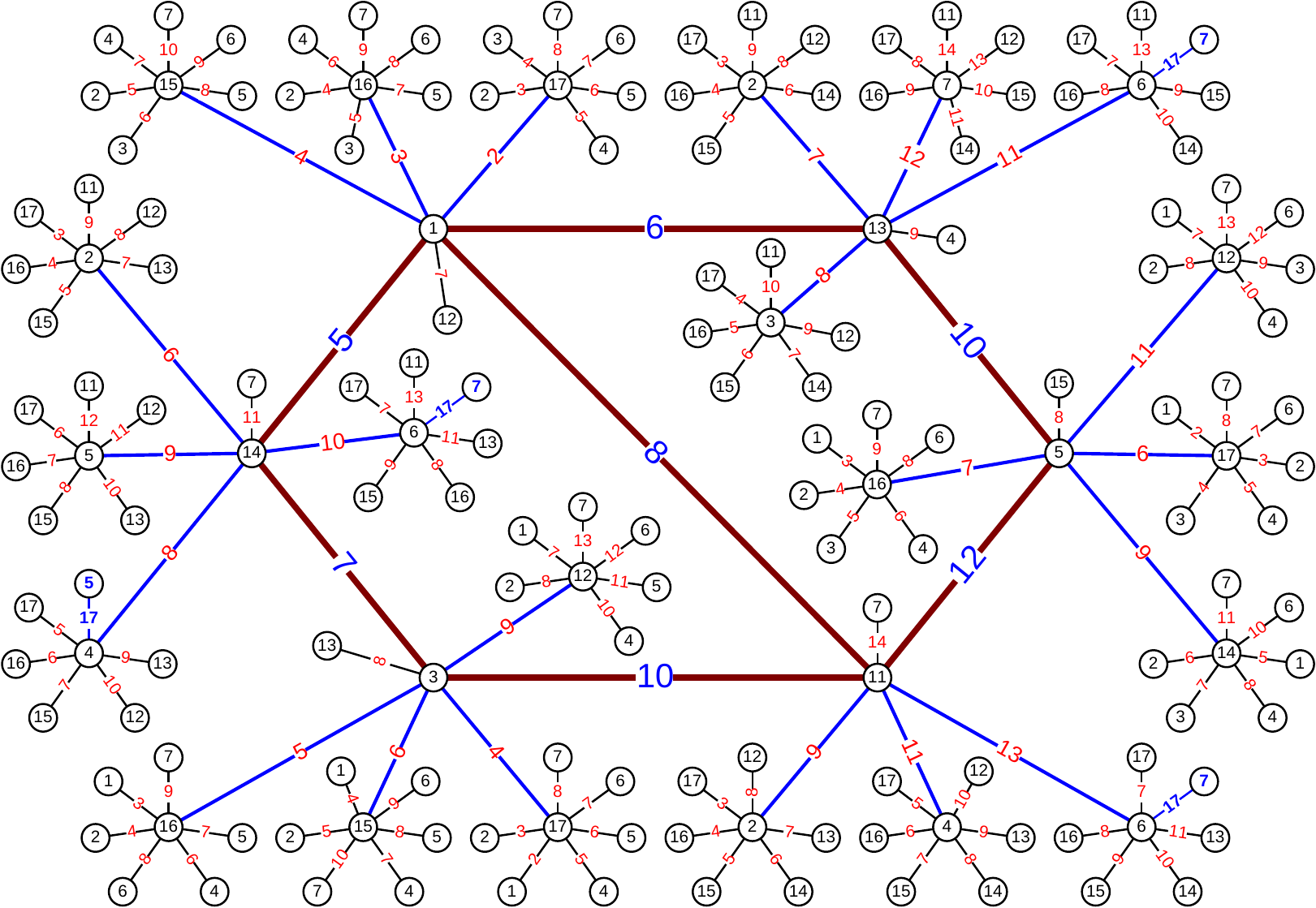}
\caption{\label{fig:1-edge-diff-replacing-2}{\small A $\Delta$-saturated graph $H^*$ admits an edge-difference proper total coloring $g$ holding $g(xy)+|g(x)-g(y)|=18$ for each edge $xy\in E(H^*)$.}}
\end{figure}

\begin{figure}[h]
\centering
\includegraphics[width=16.2cm]{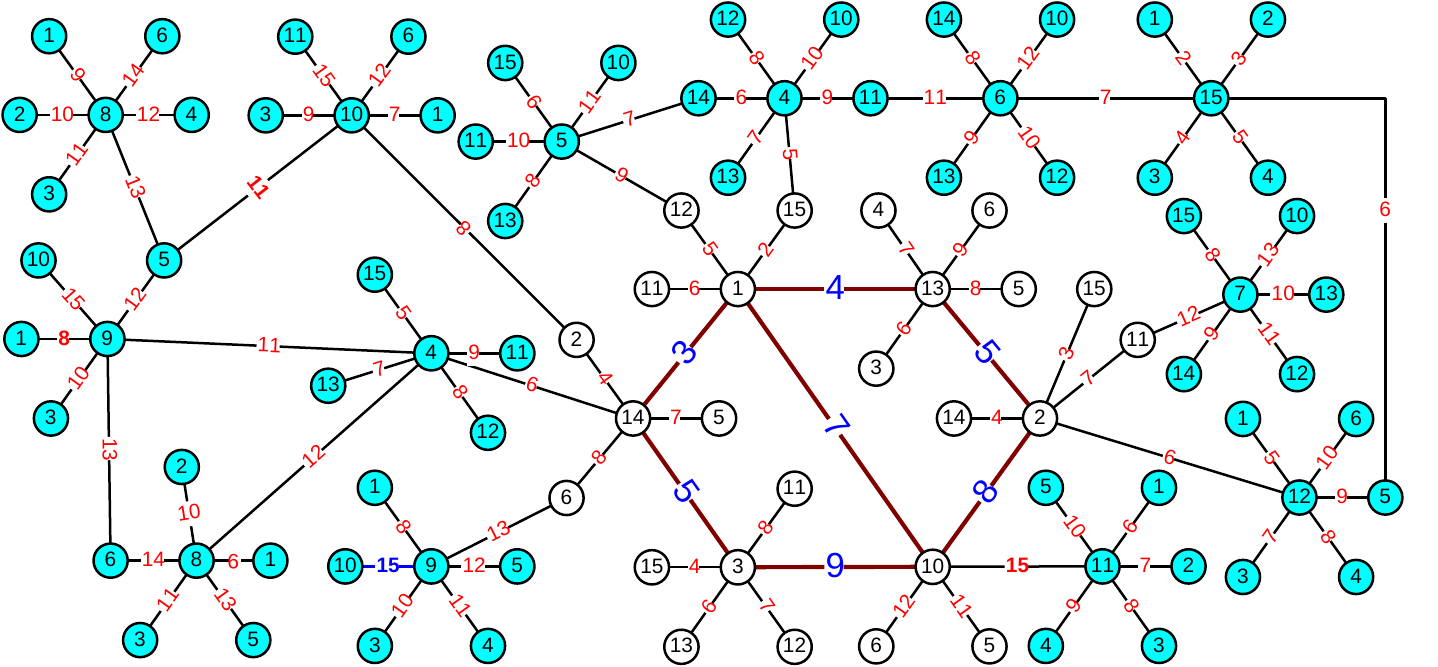}
\caption{\label{fig:2-edge-diff-re-grow}{\small A graph $G$ obtained from the growing $H'$ shown in Fig.\ref{fig:2-edge-diff-replacing} by the ice-flower system shown in Fig.\ref{fig:Dsaturated-edge-difference-even}, and $G$ admits an edge-difference proper total coloring.}}
\end{figure}

%\section{Various colorings for graphic lattices, Theorems for various $X$-type proper total colorings}
%%\input{4-section/Theorems-4}

\section{Colorings and theorems for graphic lattices}

If showing different graphic bases for a given graphic lattice, or computing or estimating the exact cardinality of a (colored) graphic lattice, we may meet the Graphic Isomorphism Problem, which is NP-hard as known. We will introduce techniques for dealing with some problems from various graphic lattices.

\subsection{Isomorphism, graph homomorphism}

\begin{thm}\label{thm:vertex-splitting-isomorphic}
Let both $G$ and $H$ be graphs with $n$ vertices. If there is a bijection $f: V (G)\rightarrow V (H)$ such that $G\wedge u\cong H\wedge f(u)$ for each vertex $u \in V (G)$ with $\textrm{deg}_G(u)\geq 2$, then $G\cong H$.
\end{thm}
\begin{proof} Let $N_G(u)=\{x_1,x_2,\dots ,x_{d_u}\}$ be the neighbor set of a vertex $u$ of $G$, where $d_u=\textrm{deg}_G(u)\geq 2$. We vertex-split $u$ into two vertices $u'$ and $u''$ such that $N_G(u)$ is cut into two disjoint subsets $N_G(u')=\{x_1,x_2,\dots ,x_{k}\}$ and $N_G(u'')=\{x_{k+1},x_{k+2},\dots ,x_{d_u}\}$ with $1<k<d_u$. By the help of a bijection $f: V (G)\rightarrow V (H)$, we vertex-split $f(u)$ into two vertices $f(u')$ and $f(u'')$, the consequent $N_G(u)$ is cut into two subsets two disjoint subsets $N_H(f(u'))=\{y_1,y_2,\dots ,y_{k}\}$ and $N_H(f(u''))=\{y_{k+1},y_{k+2},\dots ,y_{d_u}\}$ with $1<k<d_u$ and $y_j=f(x_j)$ for $j\in [1,d_u]$, such that $N_H(u)=N_H(f(u'))\cup N_H(f(u''))$.

According to the hypotheses of the theorem, $G\wedge u\cong H\wedge f(u)$, as we do the vertex-coinciding operation to $G\wedge u$ and $H\wedge f(u)$ respectively by vertex-coinciding $u'$ with $u''$ into one vertex $u$, and vertex-coinciding $f(u')$ with $f(u'')$ into one vertex $f(u)$, then we get $G\cong H$.
\end{proof}

\begin{defn} \label{defn:isomorphic-coloring-equals}
Let $G$ be a totally colored graph with a $W_G$-type proper total coloring $f_G$, and let $H$ be a totally colored graph with a $W_H$-type proper total coloring $g_H$. We say $G=H$ if there is a bijection $\varphi: V (G)\rightarrow V (H)$ such that (i) $G\wedge u\cong H\wedge \varphi(u)$ for each vertex $u \in V (G)$ with $\textrm{deg}_G(u)\geq 2$; and (ii) for each $w \in V(G)\cup E(G)$, there exists $w' \in V(H)\cup E(H)$ holding $g_H(w')=f_G(w)$ when $w'=\varphi(w)$.\qqed
\end{defn}

\begin{thm} \label{thm:trees-vs-connected-graphs}
Each connected $(p,q)$-graph $G$ corresponds a set of trees of $q+1$ vertices under the vertex-splitting operation, or another set of trees of $2q-p+1$ vertices under the leaf-splitting operation, as well as a set of trees by a mixed operation of vertex-splitting operation and leaf-splitting operation.
\end{thm}
\begin{proof} Let $G$ be a connected $(p,q)$-graph. Take a cycle $C=x_1x_2\cdots x_mx_1$ with $m\geq 3$.

First, we vertex-split $x_1$ of $C$ into two $x',x''$ to get a connected graph $H=G\wedge x_1$, clearly, $H$ contains no the cycle $C$ and $|V(H)|=|V(G)|+1$ and $|E(H)|=|E(G)|$. Notice there are $q-p+1$ vertex-splitting operations on $G$, so a tree obtained by doing $q-p+1$ vertex-splitting operations has just $q+1$ vertices.

Second, we do a leaf-splitting operation to an edge $x_ix_{i+1}$, the resultant graph $H=G(x_ix_{i+1}\prec )$ is connected and holds $|E(H)|=|E(G)|+1$ and $|V(H)|=|V(G)|+2$. After $q-(p-1)$ times process as the above, we get a tree $T$ of $p+2[q-p+1]$ vertices.
\end{proof}

\begin{rem}\label{rem:isomorphic-coloring-equals}
For $G_1=(G\wedge u)-u'$ and $H_1=(H\wedge f(u))-f(u')$ in the proof of Theorem \ref{thm:vertex-splitting-isomorphic}, if $G_1 \cong H_1$, we do not claim $G\cong H$, since it is the Kelly-Ulam's Reconstruction Conjecture, 1942 \cite{Bondy-2008}. Similarly, we can present Colored Reconstruction Conjecture: $(G-x)(=)(H-\theta(x))$ for each vertex $x\in V(G)$ under a bijection $\theta:V(G)\rightarrow V(H)$, then $G=H$.

As a graph $G$ admits a $W_i$-type proper total coloring $f$ and a $W_j$-type proper total coloring $g$, if there is a linear transformation $\theta$ such that $g(w)=\theta(f(w))$ for each $w \in V(G)\cup E(G)$, then we say that two colorings $f$ and $g$ are equivalent to each other.

Based on Theorem \ref{thm:vertex-splitting-isomorphic} and equivalent colorings, we can construct many interesting graphic lattices. Especially, we can pay attention on various colorings of trees based on Theorem \ref{thm:trees-vs-connected-graphs}.\qqed
\end{rem}

In \cite{Bing-Yao-Hongyu-Wang-graph-homomorphisms-2020}, the authors introduce infinite graph homomorphisms as follows:

\begin{thm}\label{thm:sequence-graph-homomorphisms}
There are infinite graphs $G^*_n$ forming a sequence $\{G^*_n\}$, such that $G^*_n\rightarrow G^*_{n-1}$ is really a graph homomorphism for $n\geq 1$.
\end{thm}
\begin{proof} First, we present an algorithm as follows: $G_0$ is a triangle $\Delta x_1x_2x_3$, we use a coloring $h$ to color the vertices of $G_0$ as $h(x_i)=0$ with $i\in [1,3]$.

Step 1: Add a new $y$ vertex for each edge $x_ix_j$ of the triangle $\Delta x_1x_2x_3$ with $i\neq j$, and join $y$ with two vertices $x_i$ and $x_j$ of the edge $x_ix_j$ by two new edges $yx_i$ and $yx_j$, the resulting graph is denoted by $G_1$, and color $y$ with $h(y)=1$.

Step 2: Add a new $w$ vertex for each edge $uv$ of $G_1$ if $h(u)=1$ and $h(v)=0$ (or $h(v)=1$ and $h(u)=0$), and join $y$ respectively with two vertices $u$ and $v$ by two new edges $wu$ and $wv$, the resulting graph is denoted by $G_2$, and color $w$ with $h(w)=2$.

Step $n$: Add a new $\gamma$ vertex for each edge $\alpha\beta$ of $G_{n-1}$ if $h(\alpha)=n-1$ and $h(\beta)=n-2$ (or $h(\alpha)=n-2$ and $h(\beta)=n-1$), and join $\gamma$ respectively with two vertices $\alpha$ and $\beta$ by two new edges $\gamma\alpha$ and $\gamma\beta$, the resulting graph is denoted by $G_n$, and color $\gamma$ with $h(\gamma)=n$.

Second, we construct a graph $G^*_n=G_n\cup K_1$ with $n\geq 0$, where $K_1$ is a complete graph of one vertex $z_0$. For each $n\geq 1$, there is a mapping $\theta_n:V(G^*_n)\rightarrow V(G^*_{n-1})$ in this way: $V(G^*_n\setminus V^n_{(2)})=V(G^*_{n-1}\setminus {V(K_1)})$, each $x\in V^n_{(2)}$ holds $\theta_n(x)=z_0$, where $V^n_{(2)}$ is the set of vertices of degree two in $G^*_n$. So $G^*_n\rightarrow G^*_{n-1}$ is really a graph homomorphism. We write this case by $\{G^*_n\}\rightarrow G^*_0$, called as a \emph{graph homomorphism sequence}.
\end{proof}

The notation $\{G^*_n\}\rightarrow G^*_0$ can be written as
\begin{equation}\label{eqa:inverse-limit}
\lim_{\infty \rightarrow 0}\{G^*\}^{\infty}_{0}=G^*_0
\end{equation} called an \emph{inverse limitation}. There are many graph homomorphism sequence $\{G^*_n\}$ holding $G^*_n\rightarrow G^*_{n-1}$, i.e., $\{G^*_n\}\rightarrow G^*_0$ in network science. For example, we can substitute the triangle $G_0$ in the proof of Theorem \ref{thm:sequence-graph-homomorphisms} by any connected graph.

\subsection{Colorings for graphic lattices}

In general, we have the following definition of splitting $\epsilon$-colorings:
\begin{defn}\label{defn:splitting-(odd)graceful-coloring}
\cite{Yao-Zhao-Zhang-Mu-Sun-Zhang-Yang-Ma-Su-Wang-Wang-Sun-arXiv2019} A connected $(p,q)$-graph $G$ admits a coloring $f:S \rightarrow [a,b]$, where $S\subseteq V(G)\cup E(G)$, and there exists $f(u)=f(v)$ for some distinct vertices $u,v\in V(G)$, and the edge label set $f(E(G))$ holds an $\epsilon$-condition, so we call $f$ a \emph{splitting $\epsilon$-coloring} of $G$. \qqed
\end{defn}

For the splitting $\epsilon$-colorings of graphs, we have the following examples:

\begin{defn}\label{defn:splitting-(odd)graceful-coloring}
\cite{Yao-Mu-Sun-Zhang-Yang-Wang-Wang-Su-Ma-Sun-2019} Suppose that a connected $(p,q)$-graph $G$ admits a coloring $f:V(G) \rightarrow [0,q]$ (resp. $[0,2q-1]$), such that $f(u)=f(v)$ for some pairs of vertices $u,v\in V(G)$, and the edge label set $f(E(G))=\{f(uv)=|f(u)-f(v)|: ~uv\in E(G)\}=[1,~q]$ (resp. $[1,2q-1]^o$), then we call $f$ a \emph{splitting gracefully total coloring} (resp. splitting odd-gracefully total coloring). \qqed
\end{defn}

\begin{defn}\label{defn:6C-labelling}
\cite{Yao-Sun-Zhang-Mu-Sun-Wang-Su-Zhang-Yang-Yang-2018arXiv} A total labelling $f:V(G)\cup E(G)\rightarrow [1,p+q]$ for a bipartite $(p,q)$-graph $G$ is a bijection and holds:

(i) (e-magic) $f(uv)+|f(u)-f(v)|=k$;

(ii) (ee-difference) each edge $uv$ matches with another edge $xy$ holding $f(uv)=|f(x)-f(y)|$ (or $f(uv)=2(p+q)-|f(x)-f(y)|$);

(iii) (ee-balanced) let $s(uv)=|f(u)-f(v)|-f(uv)$ for $uv\in E(G)$, then there exists a constant $k'$ such that each edge $uv$ matches with another edge $u'v'$ holding $s(uv)+s(u'v')=k'$ (or $2(p+q)+s(uv)+s(u'v')=k'$) true;

(iv) (EV-ordered) $\min f(V(G))>\max f(E(G))$ (or $\max f(V(G))<\min f(E(G))$, or $f(V(G))\subseteq f(E(G))$, or $f(E(G))\subseteq f(V(G))$, or $f(V(G))$ is an odd-set and $f(E(G))$ is an even-set);

(v) (ve-matching) there exists a constant $k''$ such that each edge $uv$ matches with one vertex $w$ such that $f(uv)+f(w)=k''$, and each vertex $z$ matches with one edge $xy$ such that $f(z)+f(xy)=k''$, except the \emph{singularity} $f(x_0)=\lfloor \frac{p+q+1}{2}\rfloor $;

(vi) (set-ordered) $\max f(X)<\min f(Y)$ (or $\min f(X)>\max f(Y)$) for the bipartition $(X,Y)$ of $V(G)$.

We refer to $f$ as a \emph{6C-labelling} of $G$.\qqed
\end{defn}

In order to meet the needs of graphic lattices, we present a generalization of flawed coloring/labelling in the following definition:

\begin{defn}\label{defn:flawed-w-type-labelling}
$^*$ Suppose that $H_1,H_2,\dots, H_m$ and $T$ are disjoint graphs, and $H=\bigcup^m_{i=1}H_i$. A $W$-type coloring means: a $W$-type coloring, or a splitting $W$-type labelling.

(1) If there exists a graph operation ``$(\diamond)$'' on $T$ and $H$ such that the resultant graph $T(\diamond)H$ is a connected graph admitting a $W$-type coloring $f$, then $f$ is called a \emph{flawed $W$-type coloring} of $H$, and $f$ is called a \emph{$W$-type joining coloring} of $T$.

(2) If there is a graph operation ``$(\ast)$'' on $H$ such that the resultant graph $(\ast)H$ is a connected graph admitting a $W$-type coloring $g$, we call $g$ a \emph{flawed $W$-type coloring} of $H$.\qqed
\end{defn}

Here, we generalize ``$T(\diamond)H$'' to a set of connected graphs ``$T(\diamond)^m_{k=1}a_kH_k$'' with $a_k\in Z^0$ and $H_k\in \textbf{\textrm{H}}_f=(H_1,H_2,\dots, H_m)$, such that each connected graph $T(\diamond)^m_{k=1}a_kH_k$ admits a $W$-type coloring, where the base $\textbf{\textrm{H}}_f$ admits a flawed $W$-type coloring $f$, and $T$ is a forest or a tree. Immediately, the following set
\begin{equation}\label{eqa:coloring-labelling-graphic-lattice}
\textbf{\textrm{L}}(F_{\textrm{orest}}(\diamond)\textbf{\textrm{H}}_f)=\{T(\diamond)^m_{k=1}a_kH_k,~a_k\in Z^0,~H_k\in \textbf{\textrm{H}}_f,T\in F_{\textrm{orest}}\}
\end{equation} is called a \emph{$W$-type coloring $(\diamond)$-graphic lattice} with $\sum^m_{k=1}a_k\geq 1$. As $a_1=a_2=\cdots =a_m=1$, we call $\textbf{\textrm{L}}(F_{\textrm{orest}}(\diamond)\textbf{\textrm{H}}_f)$ a \emph{standard $W$-type coloring $(\diamond)$-graphic lattice}, and rewrite it as $\textbf{\textrm{L}}_{\textbf{\textrm{stand}}}(F_{\textrm{orest}}(\diamond)\textbf{\textrm{H}}_f)$.

Similarly, the graph $(\ast)H$ in Definition \ref{defn:flawed-w-type-labelling} enables us to define a \emph{$W$-type coloring $(\ast)$-graphic lattice}
\begin{equation}\label{eqa:coloring-labelling-graphic-lattice-11}
\textbf{\textrm{L}}((\ast)\textbf{\textrm{H}}_f)=\{(\ast)^m_{k=1}a_kH_k,~a_k\in Z^0,~H_k\in \textbf{\textrm{H}}_f\}
\end{equation} with $\sum^m_{k=1}a_k\geq 1$, and a \emph{standard $W$-type coloring $(\ast)$-graphic lattice} $\textbf{\textrm{L}}_{\textbf{\textrm{stand}}}((\ast)\textbf{\textrm{H}}_f)$ when $a_1=a_2=\cdots =a_m=1$.

An interesting and important study is to build up connections between different $W$-type coloring $(\diamond)$-graphic lattices $\textbf{\textrm{L}}(F_{\textrm{orest}}(\diamond)\textbf{\textrm{H}}_f)$, or $W$-type coloring $(\ast)$-graphic lattices $\textbf{\textrm{L}}((\ast)\textbf{\textrm{H}}_f)$.

\begin{defn}\label{defn:L-multiple-type-coloring-labelling}
$^*$ Suppose that a connected graph $G$ admits a coloring $f$. If there is a spanning subgraph $T$ of $G$, such that $f$ is just a $W$-type coloring of $T$, we call $f$ an \emph{inner $W$-type coloring} of $T$, and say $G$ admits a \emph{coloring including a $W$-type coloring}. Moreover, if there are $L$ spanning subgraphs $H_i,H_2,\dots ,H_L$ with $E(H_i)\neq E(H_j)$ for distinct $i,j\in [1,L]$ and $L\geq 2$ such that $f_i~(=f)$ is an \emph{inner $W_i$-type coloring} of $H_i$ with $i\in [1,L]$, we call $f$ a \emph{coloring including $L$-multiple colorings} of $G$, and furthermore $f$ is an \emph{$(W_i)^L_1$-type coloring} of $G$ if $E(G)=\bigcup^L_{i=1} E(H_i)$.\qqed
\end{defn}

It is allowed that $W_i=W_j$ for $i\neq j$ in Definition \ref{defn:L-multiple-type-coloring-labelling}. There are simple results on the coloring including $L$-multiple colorings as follows:

\begin{thm} \label{thm:multiple-coloring-labellings}
$^*$ According to Definition \ref{defn:L-multiple-type-coloring-labelling}, we have:
\begin{asparaenum}[(1) ]
\item Any complete graph $K_n$ contains a spanning tree $T$ admitting a graceful labelling $f$, then $K_n$ contains another spanning tree $T^c$ admitting a graceful labelling $f^c$, such that $f^c$ is the dual labelling of $f$ and $E(T)\neq E(T^c)$.
\item Suppose that a connected graph $G$ contains a caterpillar $T$, such that deletion of all leaves of $T$ results in a path $P=x_1x_2\cdots x_m$, and $x_1$ is adjacent with a leaf $u$. If the degree $\textrm{deg}_G(u)=|V(G)|-1$, then $G$ admits an $(X_i)^L_1$-type coloring with $L\geq 2$.
\item There are infinite graphs admitting $(X_i)^L_1$-type labellings with $L\geq 2$.
\end{asparaenum}
\end{thm}
\begin{proof} Consider the result (1). Let $V(K_n)=\{x_i:i\in [1,n]\}$. We define a labelling $g(x_i)=i-1$ with $i\in [1,n]$, so $g(x_1x_j)=|g(x_1)-g(x_j)|=j-1$ for $j\in [2,n]$, that is, $g$ is a graceful labelling of a subgraph $K_{1,n-1}$ of $K_n$, where $E(K_{1,n-1})=\{x_1x_j:j\in [2,n]\}$. Notice that $K_n$ has another subgraph $K'_{1,n-1}$ having its edge set $E(K'_{1,n-1})=\{x_nx_s:s\in [1,n-1]\}$ and admitting another labelling $h$ defined by setting $h(x_s)=s-1$ with $s\in [1,n]$, and $h(x_nx_t)=|h(x_n)-h(x_t)|=n-1-(t-1)=n-t$ with $t\in [1,n-1]$, so $h$ is a graceful labelling of the subgraph $K'_{1,n-1}$ since $h(E(K'_{1,n-1}))=[1,n-1]$, and $h$ is the dual of $g$. Thereby, we claim that $g$ is an $(X_i)^2_1$-graceful labelling of $K_n$, since $g=h$.

(2) Let $L(x_k)=\{y_{k,i}:i\in [1,d_k]\}$ be the set of leaves adjacent with each vertex $x_k$ of the path $P=x_1x_2\cdots x_m$ of a caterpillar $T$ of $G$, $k\in [1,m]$. We define a graceful labelling $\gamma$ of the caterpillar $T$ in the way: $\gamma(x_1)=0$, $\gamma(y_{2,i})=i$ with $i\in [1,d_2]$; $\gamma(x_3)=1+d_2$, $\gamma(y_{4,i})=i+(1+d_2)$ with $i\in [1,d_4]$; go on in this way, without loss of generality, $m=2p$, so $\gamma(x_{2p-1})=p-1+\sum ^{p-1}_{j=1}d_{2j}$, $\gamma(y_{2p,i})=i+p+\sum ^{p-1}_{j=1}d_{2j}$ with $i\in [1,d_{2p}]$. Next, $\gamma(x_{2p})=1+p+\sum ^{p}_{j=1}d_{2j}$, $\gamma(y_{2p-1,i})=i+\gamma(x_{2p})$ with $i\in [1,d_{2p-2}]$; $\gamma(x_{2p-2})=1+d_{2p-2}+\gamma(x_{2p})$, $\gamma(y_{2p-3,i})=i+\gamma(x_{2p-2})$ with $i\in [1,d_{2p-4}]$; go on in this way, $\gamma(x_{2})=1+(p-1)+\gamma(x_{2p})+\sum ^{p}_{j=2}d_{2j-1}$, $\gamma(y_{1,i})=i+\gamma(x_{2})$ with $i\in [1,d_{1}]$. Notice $\gamma(y_{1,d_{1}})=|V(T)|-1$.

Assume that $\textrm{deg}_G(y_{1,d_{1}})=|V(T)|-1$, so $G$ contains $K_{1,|V(T)|-1}$ with the center $y_{1,d_{1}}$. We define a graceful labelling $\alpha$ as: $\alpha(y_{1,d_{1}})=|V(T)|-1$, $\alpha(w)=\gamma(w)$ for $w\in V(G)\setminus \{y_{1,d_{1}}\}$. Hence, we defined a total coloring $\beta$ of $G$ by setting $\beta(w)=\alpha(w)$ for $w\in V(K_{1,|V(T)|-1})\cup E(K_{1,|V(T)|-1})$, $\beta(w)=\gamma(w)$ for $w\in V(T)\cup E(T)$, and $\beta(w)\in [1,|V(T)|-1]$ for $w\in E(G)\setminus [E(K_{1,|V(T)|-1})\cup E(T)]$. Clearly, $\beta$ is an $(X_i)^2_1$-graceful labelling of $G$.

The result (3) stands by the results (1) and (2).
\end{proof}

\begin{problem}\label{qeu:Multiple-coloring}
We have questions about the coloring including $L$-multiple colorings in the following:
\begin{asparaenum}[\textbf{\textrm{Mul}}-1. ]
\item It is natural to \textbf{guess}: ``\emph{Every connected graph $G$ admits a coloring including a graceful labelling by Graceful Tree Conjecture}''.
\item For a complete bipartite graph $K_{1,n}$ with vertices $x_0,x_i,\dots ,x_n$ and edge set $E(K_{1,n})=\{x_0x_i:i\in [1,n]\}$, we define a graceful labelling $f$ of $K_{1,n}$ as: $f(x_0)=n$, $f(x_j)=j$ with $j\in [1,n]$, so we have $f(x_0x_j)=|f(x_0)-f(x_j)|=n-j$, and $f(E(K_{1,n}))=[1,n]$. However, we can see $f(x_0)+f(x_0x_j)+f(x_j)=2n$, in other word, $f$ is an edge-magic total labelling of $K_{1,n}$ too. \textbf{Is} this going to happen to other graphs ($\neq K_{1,n}$) else?
\end{asparaenum}
\end{problem}

\subsection{Connections between colorings/labellings}

\begin{thm} \label{thm:connections-several-labellings}
\cite{Yao-Liu-Yao-2017} Let $T$ be a tree on $p$ vertices, and let $(X,Y)$ be its
bipartition of vertex set $V(T)$. For integers $k\geq 1$ and $d\geq 1$, the following assertions are mutually equivalent:

$(1)$ $T$ admits a set-ordered graceful labelling $f$ with $\max f(X)<\min f(Y)$.

$(2)$ $T$ admits a super felicitous labelling $\alpha$ with $\max \alpha(X)<\min \alpha(Y)$.

$(3)$ $T$ admits a $(k,d)$-graceful labelling $\beta$ with
$\beta(x)<\beta(y)-k+d$ for all $x\in X$ and $y\in Y$.

$(4)$ $T$ admits a super edge-magic total labelling $\gamma$ with $\max \gamma(X)<\min \gamma(Y)$ and a magic constant $|X|+2p+1$.

$(5)$ $T$ admits a super $(|X|+p+3,2)$-edge antimagic total labelling $\theta$ with $\max \theta(X)<\min \theta(Y)$.

$(6)$ $T$ admits an odd-elegant labelling $\eta$ with $\eta(x)+\eta(y)\leq 2p-3$ for every edge $xy\in E(T)$.

$(7)$ $T$ admits a $(k,d)$-arithmetic labelling $\psi$ with $\max \psi(x)<\min \psi(y)-k+d\cdot |X|$ for all $x\in X$ and $y\in Y$.

$(8)$ $T$ admits a harmonious labelling $\varphi$ with $\max \varphi(X)<\min \varphi(Y\setminus \{y_0\})$ and $\varphi(y_0)=0$.
\end{thm}

We have some results similarly with that in Theorem \ref{thm:connections-several-labellings} about flawed labellings as follows:

\begin{thm} \label{thm:connection-flawed-labellings}
\cite{Yao-Sun-Zhang-Mu-Sun-Wang-Su-Zhang-Yang-Yang-2018arXiv, Yao-Zhang-Sun-Mu-Sun-Wang-Wang-Ma-Su-Yang-Yang-Zhang-2018arXiv, Yao-Zhao-Zhang-Mu-Sun-Zhang-Yang-Ma-Su-Wang-Wang-Sun-arXiv2019} Suppose that $T=\bigcup ^m_{i=1}T_i$ is a forest made by disjoint trees $T_1,T_2,\dots ,T_m$, and $(X,Y)$ be the vertex
bipartition of $T$. For integers $k\geq 1$ and $d\geq 1$, the following assertions are mutually equivalent:
\begin{asparaenum}[F-1. ]
\item $T$ admits a flawed set-ordered graceful labelling $f$ with $\max f(X)<\min f(Y)$;
\item $T$ admits a flawed set-ordered odd-graceful labelling $f$ with $\max f(X)<\min f(Y)$;

\item $T$ admits a flawed set-ordered elegant labelling $f$ with $\max f(X)<\min f(Y)$;

\item $T$ admits a flawed odd-elegant labelling $\eta$ with $\eta(x)+\eta(y)\leq 2p-3$ for every edge $xy\in E(T)$.

\item $T$ admits a flawed super felicitous labelling $\alpha$ with $\max \alpha(X)<\min \alpha(Y)$.

\item $T$ admits a flawed super edge-magic total labelling $\gamma$ with $\max \gamma(X)<\min \gamma(Y)$ and a magic constant $|X|+2p+1$.
\item $T$ admits a flawed super $(|X|+p+3,2)$-edge antimagic total labelling $\theta$ with $\max \theta(X)<\min \theta(Y)$.

\item $T$ admits a flawed harmonious labelling $\varphi$ with $\max \varphi(X)<\min \varphi(Y\setminus \{y_0\})$, $\varphi(y_0)=0$.
\end{asparaenum}
\end{thm}
We present some equivalent definitions with parameters $k,d$ for flawed $(k,d)$-labellings.

\begin{thm} \label{thm:flawed-(k,d)-labellings}
\cite{Yao-Sun-Zhang-Mu-Sun-Wang-Su-Zhang-Yang-Yang-2018arXiv, Yao-Zhang-Sun-Mu-Sun-Wang-Wang-Ma-Su-Yang-Yang-Zhang-2018arXiv, Yao-Zhao-Zhang-Mu-Sun-Zhang-Yang-Ma-Su-Wang-Wang-Sun-arXiv2019} Let $T=\bigcup ^m_{i=1}T_i$ be a forest having disjoint trees $T_1,T_2,\dots ,T_m$, and its bipartition $(X,Y)$ of $V(T)$. For some values of two integers $k\geq 1$ and $d\geq 1$, the following assertions are mutually equivalent:
\begin{asparaenum}[\textrm{KD}-1. ]
\item $T$ admits a flawed set-ordered graceful labelling $f$ with $\max f(X)<\min f(Y)$.

\item $T$ admits a flawed $(k,d)$-graceful labelling $\beta$ with $\max \beta(x)<\min \beta(y)-k+d$ for all $x\in X$ and $y\in Y$.

\item $T$ admits a flawed $(k,d)$-arithmetic labelling $\psi$ with $\max \psi(x)<\min \psi(y)-k+d\cdot |X|$ for all $x\in X$ and $y\in Y$.
\item $T$ admits a flawed $(k,d)$-harmonious labelling $\varphi$ with $\max \varphi(X)<\min \varphi(Y\setminus \{y_0\})$, $\varphi(y_0)=0$.
\end{asparaenum}
\end{thm}

\begin{thm} \label{thm:flawed-graph-labellings}
\cite{Yao-Sun-Zhang-Mu-Sun-Wang-Su-Zhang-Yang-Yang-2018arXiv, Yao-Zhang-Sun-Mu-Sun-Wang-Wang-Ma-Su-Yang-Yang-Zhang-2018arXiv, Yao-Zhao-Zhang-Mu-Sun-Zhang-Yang-Ma-Su-Wang-Wang-Sun-arXiv2019} Let $H=E^*+G$ be a connected graph, where $E^*$ is a set of some edges and $G=\bigcup^m_{i=1}G_i$ is a disconnected graph with disjoint connected graphs $G_1,G_2,\dots, G_m$. About graph labellings, $G$ admits a \emph{flawed $\alpha$-labelling} if $H$ admits one of the following $\alpha$-labellings:
\begin{asparaenum}[\textbf{\textrm{Lab}}-1. ]
\item $\alpha$ is a graceful labelling, or a set-ordered graceful labelling, or graceful-intersection total set-labelling, or a graceful group-labelling.
\item $\alpha$ is an odd-graceful labelling, or a set-ordered odd-graceful labelling, or an edge-odd-graceful total labelling, or an odd-graceful-intersection total set-labelling, or an odd-graceful group-labelling, or a perfect odd-graceful labelling.
\item $\alpha$ is an elegant labelling, or an odd-elegant labelling.
\item $\alpha$ is an edge-magic total labelling, or a super edge-magic total labelling, or super set-ordered edge-magic total labelling, or an edge-magic total graceful labelling.
\item $\alpha$ is a $(k,d)$-edge antimagic total labelling, or a $(k, d)$-arithmetic.
\item $\alpha$ is a relaxed edge-magic total labelling.
\item $\alpha$ is an odd-edge-magic matching labelling, or an ee-difference odd-edge-magic matching labelling.
\item $\alpha$ is a 6C-labelling, or an odd-6C-labelling.
\item $\alpha$ is an ee-difference graceful-magic matching labelling.
\item $\alpha$ is a difference-sum labelling, or a felicitous-sum labelling.
\item $\alpha$ is a multiple edge-meaning vertex labelling.
\item $\alpha$ is a perfect $\varepsilon$-labelling.
\item $\alpha$ is an image-labelling, or a $(k,d)$-harmonious image-labelling.
\item $\alpha$ is a twin $(k,d)$-labelling, or a twin Fibonacci-type graph-labelling, or a twin odd-graceful labelling.
\end{asparaenum}
\end{thm}

\begin{thm} \label{thm:graph-colorings}
$^*$ Let $H=E^*+G$ be a connected graph, where $E^*$ is a set of some edges and $G=\bigcup^m_{i=1}G_i$ is a disconnected graph with disjoint connected graphs $G_1,G_2,\dots, G_m$. About graph colorings, $G$ admits a \emph{flawed $\beta$-coloring} if $H$ admits one of the following $\beta$-labellings:
\begin{asparaenum}[\textbf{\textrm{Col}}-1. ]
\item $\beta$ is a splitting gracefully total coloring.
\item $\beta$ is a splitting odd-gracefully total coloring.
\item $\beta$ is a splitting elegant coloring.
\item $\beta$ is a splitting odd-elegant total coloring.
\item $\beta$ is a splitting edge-magic total coloring.
\item $\beta$ is an (a perfect) edge-magic proper total coloring.
\item $\beta$ is an (a perfect) edge-difference proper total coloring.
\item $\beta$ is a (perfect) graceful-difference proper total coloring.
\item $\beta$ is a (perfect) felicitous-difference proper total coloring.
\end{asparaenum}
\end{thm}

For understanding various flawed colorings, we show an example in Fig.\ref{fig:set-ordered-vs-set-ordered}, where a connected graph $(H_1\cup H_2)+E^*_a$ admits a splitting set-ordered gracefully total coloring $f_a$, so $f_a$ is a flawed set-ordered gracefully total coloring of the disconnected graph $(H_1\cup H_2)$; a connected graph $(T_1\cup T_2)+E^*_b$ admits a set-ordered gracefully total coloring $g_b$, which is a flawed set-ordered gracefully total coloring of the disconnected graph $(T_1\cup T_2)$; a connected graph $(H_1\cup H_2)+E^*_c$ admits a splitting set-ordered gracefully total coloring $f_c$, which is a flawed set-ordered gracefully total coloring of the disconnected graph $(H_1\cup H_2)$; and a connected graph $(T_1\cup T_2)+E^*_d$ admits a set-ordered gracefully total coloring $g_d$, which is a flawed set-ordered gracefully total coloring of the disconnected graph $(T_1\cup T_2)$.

\begin{figure}[h]
\centering
\includegraphics[width=16cm]{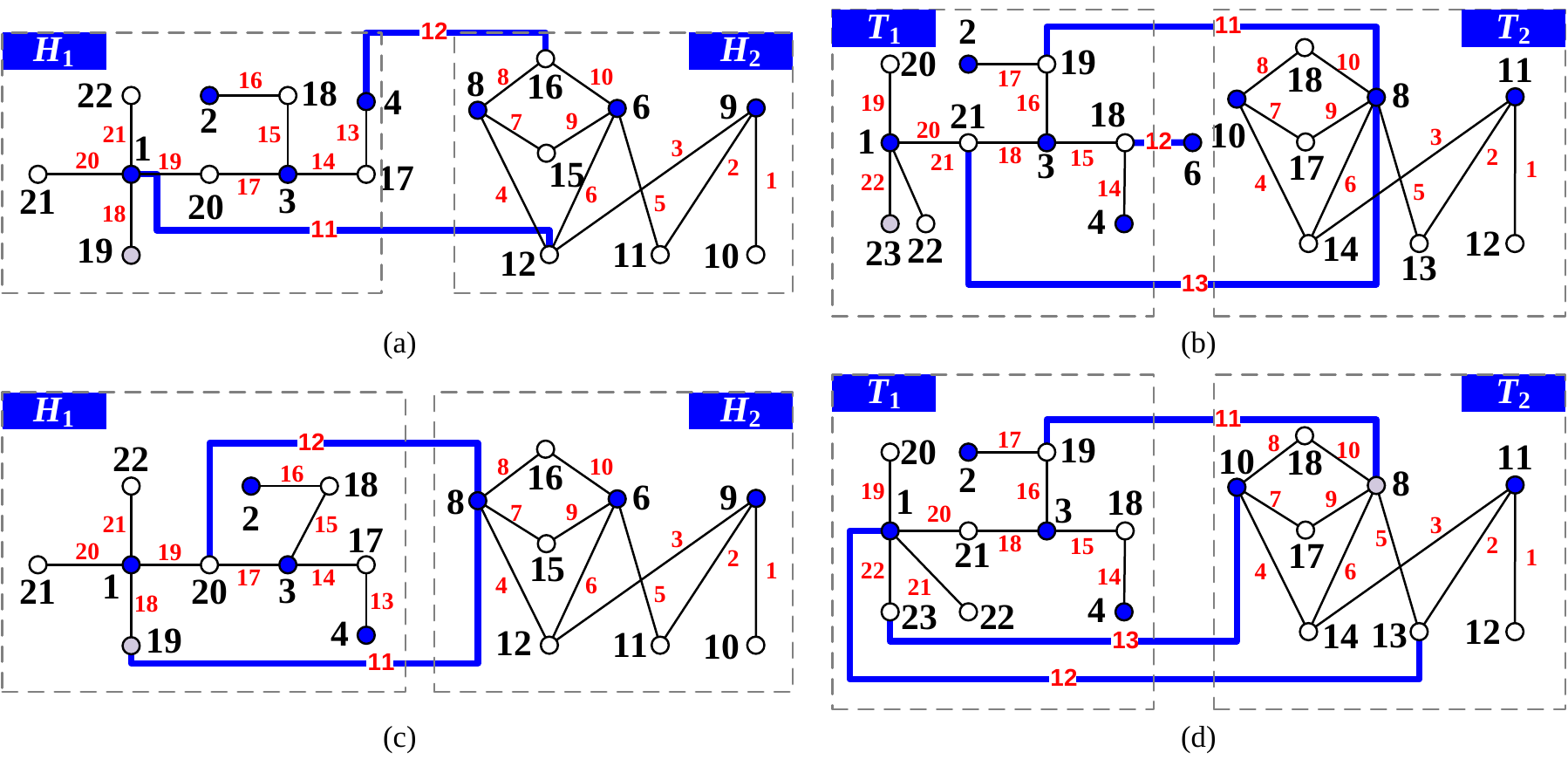}\\
\caption{\label{fig:set-ordered-vs-set-ordered} {\small (a) $(H_1\cup H_2)+E^*_a$, where the join set $E^*_a=\{(3,15),(5,16)\}$; (b) $(T_1\cup T_2)+E^*_b$, where the join set $E^*_b=\{(3,14),(6,18),(8,21),(4,18)\}$; (c) $(H_1\cup H_2)+E^*_c$, where the join set $E^*_c=\{(1,12),(8,20)\}$; (b) $(T_1\cup T_2)+E^*_d$, where the join set $E^*_d=\{(9,18),(1,13),(10,23),(0,14)\}$.}}
\end{figure}

For obtaining graphs $G\odot ^m_{k=1}a_kH_k$ admitting $W$-type colorings/labellings by means of a base $\textbf{\textrm{H}}=(H_k)^m_{k=1}$ and the vertex-coinciding operation ``$\odot$'', observe Fig.\ref{fig:set-ordered-vs-set-ordered} carefully, we can see some phenomenons:
\begin{asparaenum}[\textbf{\textrm{Dist}}-1]
\item $H_i\neq T_i$ for $i=1,2$, although $H_i\cong T_i$ in the view of topological structure.
\item Two edge sets $E^*_a$ and $E^*_c$ of joining $H_1$ and $H_2$ together are different to each other, so are to $(T_1\cup T_2)+E^*_b$ and $(T_1\cup T_2)+E^*_d$.
\item Each of four graphs $H_i$ and $T_i$ with $i=1,2$ is bipartite and admits a set-ordered coloring that can be induced by set-ordered graceful labellings. There are many edge sets like $E^*_k$ with $k=a,b,c,d$ to join $H_1$ and $H_2$ together, or $T_1$ and $T_2$ together.
\item In the view of vertex-coinciding operation, we have four graphs $L_k$ with edge sets $E(L_k)=E^*_k$ for $k=a,b,c,d$, such that four graphs $L_s\odot (H_1\cup H_2)$ with $s=a,c$ and $L_j\odot (T_1\cup T_2)$ with $j=b,d$ are connected. However, $L_a \odot (H_1\cup H_2)\not\cong L_c\odot (H_1\cup H_2)$ although $f_a(E(H_1))=[13,21]=f_c(E(H_1))$, $f_a(E(H_2))=[1,10]=f_c(E(H_2))$ and $f_a(E(L_a))=\{11,12\}=f_c(E(L_c))$; and $L_b \odot (T_1\cup T_2)\not\cong L_d\odot (T_1\cup T_2)$ in spite of $f_b(E(T_1))=[14,21]=f_d(E(T_1))$, $f_b(E(T_2))=[1,10]=f_d(E(T_2))$ and $f_b(E(L_b))=[11,13]=f_d(E(L_d))$. There are many graphs like $L_k$ with $k=a,b,c,d$ to join $H_1$ and $H_2$ together, or $T_1$ and $T_2$ together.
\end{asparaenum}

\begin{lem} \label{thm:H-connecte-two-graphs}
Suppose that each $(p_i,q_i)$-graph $G_i$ are bipartite and connected, and admits a proper total coloring $f_i$ to be a set-ordered gracefully total coloring with $i=1,2$. For integer $m\geq 1$, there is a graph $H$ having $m$ edges to join $G_1$ and $G_2$ together based on the vertex-coinciding operation ``$\odot$'', such that resultant graph $H\odot(G_1\cup G_2)$ is connected and admits a proper total coloring $h$ to be a set-ordered gracefully total coloring, also, $h$ is a flawed set-ordered gracefully total coloring of the disconnected graph $G_1\cup G_2$, and a set-ordered graceful joining coloring of $H$.
\end{lem}
\begin{proof} For $i=1,2$, let $(X_i,Y_i)$ be the bipartition of vertex set of each $(p_i,q_i)$-graph $G_i$ and $X_i=\{x_{i,j}:j\in [1,s_i]$ and $Y_i=\{y_{i,j}:j\in [1,t_i]\}$ with $s_i+t_i=p_i$, so each set-ordered gracefully total coloring $f_i$ holds $\max f_i(X_i)<\min f_i(Y_i)$ and $f_i(E(G_i))=[1,q_i]$ by the hypothesis of then theorem. Without loss of generality, we have
$$f_i(x_{i,1})\leq f_i(x_{i,2})\leq \cdots \leq f_i(x_{i,s_i})<f_i(y_{i,1})\leq f_i(y_{i,2})\leq \cdots \leq f_i(y_{i,t_i})$$
for $x_{i,j}\in X_i$ and $y_{i,j}\in Y_i$ with $i=1,2$. Notice that $f_i(x)=f_i(y)$ for some distinct vertices $x,y\in V(G_i)$ with $i=1,2$. Clearly,
$\max f_i(X_i)=f_i(x_{i,s_i})$ and $\max f_i(Y_i)=f_i(y_{i,t_i})$ with $i=1,2$. We take arbitrarily an integer $m\geq 1$ and define a coloring $g$ as follows:

Step 1. $g(x_{2,j})=m+1+f_1(x_{1,s_1})+f_2(w)$ for $w\in X_2\cup Y_2=V(G_2)$. Thereby, $g(x_{2,j}y_{2,i})=f_2(x_{2,j}y_{2,i})$ for $x_{2,j}y_{2,i}\in E(G_2)$, which induces $g(E(G_2))=[1,q_2]$, and
$$g(V(G_2))=\{m+f_2(w):w\in V(G_2)\}\subseteq [m+1+f_1(x_{1,s_1})+1,m+1+f_1(x_{1,s_1})+q_2].$$ Notice that $\max g(V(G_2))=m+1+f_1(x_{1,s_1})+f_2(y_{2,t_2})=m+1+f_1(x_{1,s_1})+q_2$

Step 2. $g(x_{1,j})=f_1(x_{1,j})$ for $x_{1,j}\in X_1$, $g(y_{1,1}x_{1,s_1})=m+q_2+f_1(y_{1,1}x_{1,s_1})=m+q_2+1$, so $g(y_{1,i})=m+q_2+f_1(y_{1,i})$ for $y_{1,i}\in Y_1$, and $\max g(V(G_1))=m+q_2+f_1(y_{1,t_1})=m+q_1+q_2$. Moreover,
$g(x_{1,j}y_{1,j})=g(y_{1,i})-g(x_{1,j})=m+q_2+f_1(x_{1,j}y_{1,i})$ for $x_{1,j}y_{1,j}\in E(G_1)$, that is, $g(E(G_1))=[m+q_2+1,m+q_2+q_1]$.

Thereby, this coloring holds $$g(E(G_2))\cup g(E(G_1))=[1,q_2]\cup [m+q_2+1,m+q_2+q_1].$$

We need to build up a graph $H$ having $m$ edges only, that is, $g(E(H))=[q_2+1,m+q_2]$.

Notice that $g(y_{2,t_2})-g(x_{1,s_1})=m+1+f_1(x_{1,s_1})+f_2(y_{2,t_2})-f_1(x_{1,s_1})=m+1+q_2$. We have new vertices $z_1,z_2,\dots, z_m$ can be colored with $g(z_j)=g(x_{1,s_1})+j=f_1(x_{1,s_1})+j$ with $j\in [1,m]$, and can be added to $G_1\cup G_2$ to realize the goal of this lemma.

Notice that $g(y_{2,t_2})-g(z_j)=m+1+f_1(x_{1,s_1})+f_2(y_{2,t_2})-[f_1(x_{1,s_1})+j]=m+1+q_2-j\in [q_2+1,m+q_2]$ with $j\in [1,m]$, that is $\{g(y_{2,t_2})-g(z_j):j\in [1,m]\}=[q_2+1,m+q_2]$. We have the following ways:

Way-1. If joining $z_j$ with some vertices $y_{2,i}$ of $Y_2$, then we get $g(y_{2,i})-g(z_j)=m+1+f_1(x_{1,s_1})+f_2(y_{2,i})-f_1(x_{1,s_1})-j=m+1+f_2(y_{2,i})-j$, and solve the following inequalities
\begin{equation}\label{eqa:Case-1}
q_2+1\leq m+1+f_2(y_{2,i})-j\leq m+q_2,
\end{equation}that is, $1\leq q_2-f_2(y_{2,i})+j\leq m$, there is at least one solution $q_2-f_2(y_{2,i})=0$ for the inequalities (\ref{eqa:Case-1}).

Way-2. If joining $z_j$ with some vertices $y_{1,i}$ of $Y_1$, then we get $g(y_{1,i})-g(z_j)=m+q_2+f_1(y_{1,i})-f_1(x_{1,s_1})-j$ with $j\in [1,m]$. Solve inequalities
\begin{equation}\label{eqa:Case-2}
q_2+1\leq m+q_2+f_1(y_{1,i})-f_1(x_{1,s_1})-j\leq m+q_2,
\end{equation}
so $1\leq m+f_1(y_{1,i})-f_1(x_{1,s_1})-j\leq m$, there is at least one solution $f_1(y_{1,1})-f_1(x_{1,s_1})=1$ for the inequalities (\ref{eqa:Case-2}).

Nonetheless, we have other ways as follows:

Way-3. If $g(y_{2,j})-g(x_{1,i})\in [q_2+1,m+q_2]$, we can join $x_{1,i}$ with $y_{2,j}$ together with an edge, and join $z_j$ with $y_{2,t_2}$ by an edge, where $g(y_{2,t_2})-g(z_j)=m+1+q_2-j\neq g(y_{2,j})-g(x_{1,i})$, such that $H\odot (G_1\cup G_2)$ is connected and admits a splitting set-ordered gracefully total coloring.

Way-4. If $g(y_{1,i})-g(x_{2,j})\in [q_2+1,m+q_2]$, we can join $y_{1,i}$ with $x_{2,j}$ together with an edge, and join $z_j$ with $y_{2,t_2}$ for $g(y_{2,t_2})-g(z_j)=m+1+q_2-j\neq g(y_{1,i})-g(x_{2,j})$, such that $H\odot (G_1\cup G_2)$ is connected and admits a splitting set-ordered gracefully total coloring.

Way-5. If we hope that $g(y_{1,i})-g(x_{2,j})\in [q_2+1,m+q_2]$, then $g(y_{1,i})-g(x_{2,j})=m+q_2+f_1(y_{1,i})-[m+1+f_1(x_{1,s_1})+f_2(x_{2,j})]=q_2+f_1(y_{1,i})-1-f_1(x_{1,s_1})-f_2(x_{2,j})$, also
$$q_2+1\leq q_2+f_1(y_{1,i})-1-f_1(x_{1,s_1})-f_2(x_{2,j})\leq m+q_2.$$
We have some solutions $1\leq f_1(y_{1,i})-1-f_1(x_{1,s_1})-f_2(x_{2,j})\leq m$, so we can join $x_{2,j}$ with $y_{1,i}$ together with an edge.

Way-6. If $g(y_{2,j})-g(x_{1,i})\in [q_2+1,m+q_2]$, then we get $g(y_{2,j})-g(x_{1,i})=m+1+f_1(x_{1,s_1})+f_2(y_{2,j})-f_1(x_{1,i})$, and
$$q_2+1\leq m+1+f_1(x_{1,s_1})+f_2(y_{2,j})-f_1(x_{1,i})\leq m+q_2$$
as well as
$$1\leq q_2-f_1(x_{1,s_1})-f_2(y_{2,j})+f_1(x_{1,i})\leq m$$
We have some solutions $f_1(x_{1,s_1})-f_1(x_{1,i})\leq q_2-f_2(y_{2,j})$, such that the vertex $x_{1,i}$ can be joined with the vertex $y_{2,j}$ together by an edge.

The above deducing process has shown the lemma.
\end{proof}

\begin{figure}[h]
\centering
\includegraphics[width=16cm]{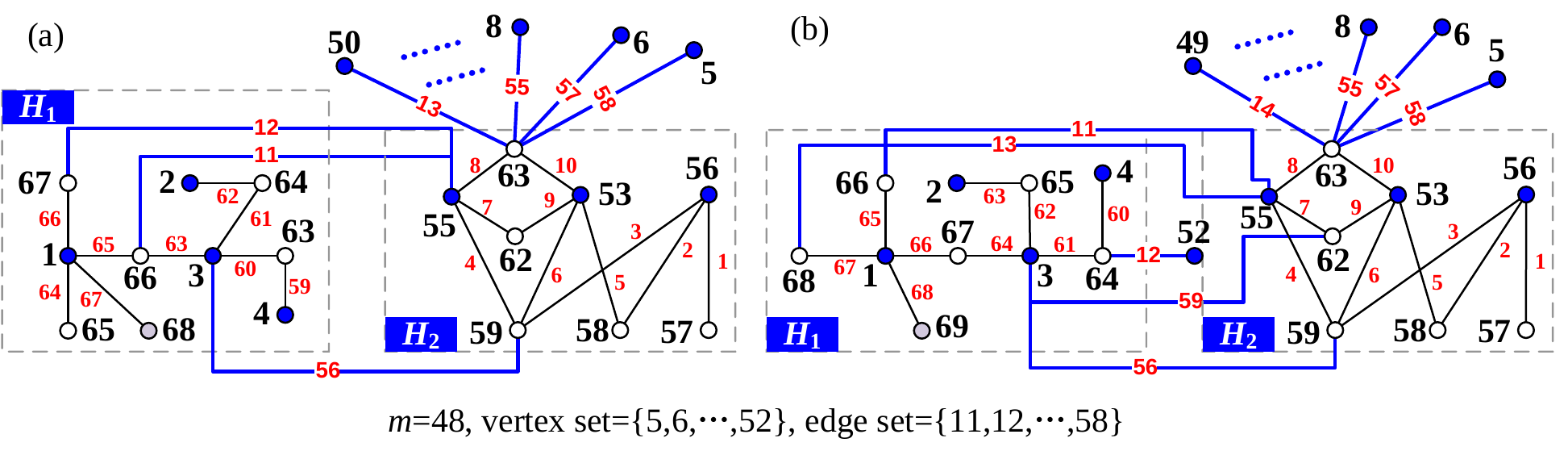}\\
\caption{\label{fig:set-ordered-add-m-edges} {\small A scheme for illustrating the proof of Lemma \ref{thm:H-connecte-two-graphs}.}}
\end{figure}

For understanding the proof of Lemma \ref{thm:H-connecte-two-graphs}, see examples shown in Fig.\ref{fig:set-ordered-vs-set-ordered} and Fig.\ref{fig:set-ordered-add-m-edges}. By Lemma \ref{thm:H-connecte-two-graphs} and mathematical induction, we can prove the following result:

\begin{thm} \label{thm:induction-more-graphs}
For a base $\textbf{\textrm{G}}=(G_k)^n_{k=1}$ made by disjoint graphs $G_1,G_2,\dots ,G_n$, where each $G_k$ is a bipartite and connected $(p_k,q_k)$-graph admitting a proper total coloring to be a set-ordered gracefully total coloring with $k\in [1,n]$. There is a graph $H$ such that the graph $H\odot^n_{k=1}G_k$ is connected and admits a set-ordered gracefully total coloring $f$, and then $f$ is a flawed set-ordered gracefully total coloring of the base $\textbf{\textrm{G}}=(G_k)^n_{k=1}$ too.
\end{thm}

\begin{rem}\label{rem:Lemma-H-connecte-two-graphs}
The graph $H$ in Lemma \ref{thm:H-connecte-two-graphs} is a hypergraph sometimes. In case $H\odot (\bigcup^n_{k=1}a_kG_k)$ for $G_k\in \textbf{\textrm{G}}=(G_k)^n_{k=1}$, we rewrite $H\odot (\bigcup^n_{k=1}a_kG_k)$ by $H\odot^n_{k=1}a_kG_k$, and let $F_{\textrm{set}}$ be the set of colored graphs. Furthermore the set
\begin{equation}\label{eqa:coloring-labelling-graphic-lattice}
\textbf{\textrm{L}}(F_{\textrm{set}}\odot \textbf{\textrm{G}})=\{H\odot^n_{k=1}a_kG_k,~a_k\in Z^0,~G_k\in \textbf{\textrm{G}}=(G_k)^n_{k=1},H\in F_{\textrm{set}}\}
\end{equation} is called a \emph{set-ordered gracefully total coloring graphic lattice} with $\sum^n_{k=1}a_k\geq 1$ based on Lemma \ref{thm:H-connecte-two-graphs}, and the base is $\textbf{\textrm{G}}=(G_k)^n_{k=1}$.

If the base $\textbf{\textrm{G}}=(G_k)^n_{k=1}$ holds some strong conditions (for instance, admitting a set-ordered graceful labellings), and a given tree $T$ of $n$ vertices (such as admitting a set-ordered graceful labelling), we can vertex-coincide a vertex $x_k$ of $T$ with a vertex of $G_k$ into one to obtain a connected graph $T\odot^n_{k=1}G_k$ admitting a $W$-type labelling, such as, $W$-type $\in \{$felicitous, super edge-magic total, super set-ordered $(k,d)$-edge-magic total, super total graceful, super set-ordered total graceful, super generalized total graceful$\}$ (Ref. \cite{Zhang-Yao-Wang-Wang-Yang-Yang-2013}, \cite{Wang-Yao-Yao-2014Information}, \cite{Wang-Yao-Yang-Yang-Chen-2013}, \cite{Wang-Xu-Yao-Ars-2018}). \qqed
\end{rem}

Let $G'$ be a copy of a graph $G$. Join a vertex $x$ of $G$ with its image vertex $x'$ of the copy $G'$ by an edge $xx'$, the resultant graph is denoted as $G\perp G'$, called a \emph{symmetric graph}. See some symmetric graphs shown in Fig.\ref{fig:non-set-ordered-1} and Fig.\ref{fig:non-set-ordered-2}.

\begin{figure}[h]
\centering
\includegraphics[width=16cm]{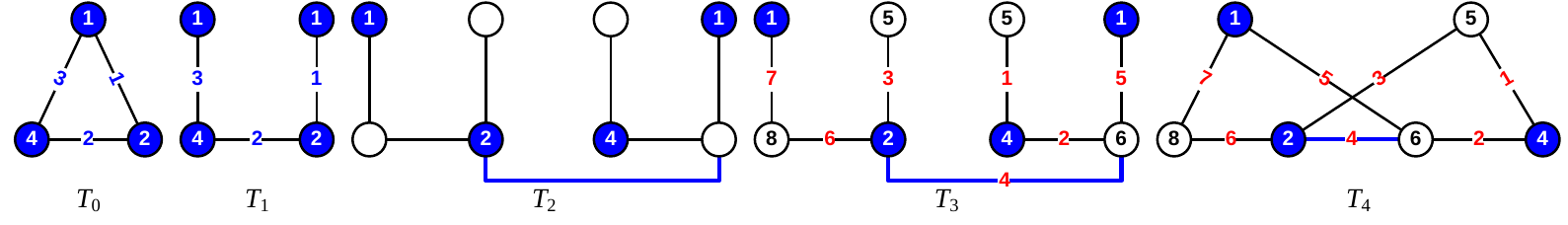}\\
\caption{\label{fig:non-set-ordered-1} {\small A process of constructing a symmetric graph $T_4=T_1\perp T'_1$ admitting a set-ordered graceful labelling/coloring from a non-set-ordered graceful graph $T_0$.}}
\end{figure}

\begin{lem} \label{thm:symmetric-graphs-graceful-coloring}
Let $G$ be a bipartite and connected graph admitting admits a proper total coloring to be a gracefully total coloring. Then the symmetric bipartite graph $G\perp G'$ admits a set-ordered gracefully total coloring.
\end{lem}
\begin{proof} Suppose that a bipartite and connected $(p,q)$-graph $G$ admits a gracefully total coloring $f: V(G)\cup E(G)\rightarrow [0,q]$ such that $f(x)=f(y)$ for some distinct vertices $x,y\in V(G)$, each edge $uv$ holds $f(uv)=|f(u)-f(v)|$, and the edge color set $f(E(G))=\{f(uv):uv\in E(G)\}=[1,q]$. Since $G$ is bipartite, so $V(G)=X\cup Y$ and $X\cap Y=\emptyset$.

Take a copy $G'$ of $G$, so we have its vertex set $V(G')=X'\cup Y'$ and $X'\cap Y'=\emptyset$, and then $G'$ admits the gracefully total coloring $f$. Join a vertex $x$ of $G$ with its image vertex $x'$ of the copy $G'$ by an edge $xx'$, the resultant graph is symmetric and denoted as $G\perp G'$. Notice that the symmetric graph $G\perp G'$ is bipartite too, and its vertex set $V(G\perp G')=(X\cup Y')\cup (X'\cup Y)$.

We define a total coloring $g$ of the symmetric graph $G\perp G'$ as follows: $g(w)=f(w)$ for each vertex $w\in X\cup Y'$, $g(z)=f(z)+q+1$ for each vertex $z\in X'\cup Y$, and $g(uv)=|g(u)-g(v)|$ for each edge $uv\in E(G\perp G')$.

For edges $xy\in E(G)$ with $x\in X$ and $y\in Y$, we have
$$g(xy)=|g(y)-g(x)|=g(y)-g(x)=f(y)+q+1-f(x),$$ and for edges $x'y'\in E(G')$ with $x'\in X'$ and $y'\in Y'$, we get $$g(x'y')=|g(y')-g(x')|=g(x')-g(y')=f(x')+q+1-f(y').$$ If $f(y)-f(x)>0$, then $f(x')-f(y')<0$, as well as $g(xy)=q+1+f(xy)$ and $g(x'y')=q+1-f(x'y')$; if $f(y)-f(x)<0$, then $f(x')-f(y')>0$, thus, $g(xy)=q+1-f(xy)$ and $g(x'y')=q+1+f(x'y')$, which means that the edge color set $g(E(G'))\cup g(E(G))=[1,q]\cup [q+2,2q+1]$.

In total, we have the edge color set of the symmetric graph $G\perp G'$ as:
$$g(E(G\perp G'))=g(E(G'))\cup g(xx')\cup g(E(G))=[1,q]\cup \{q+1\}\cup [q+2,2q+1]=[1,2q+1].$$ Since $\max g(X\cup Y')<\min g(X'\cup Y)$, so $g$ is a set-ordered gracefully total coloring of the symmetric graph $G\perp G'$. This lemma has been proven.
\end{proof}

\begin{figure}[h]
\centering
\includegraphics[width=16cm]{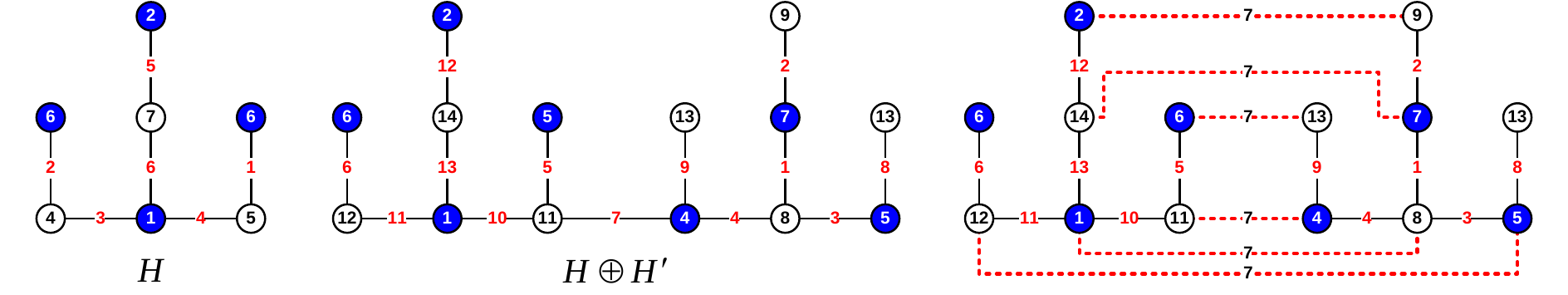}\\
\caption{\label{fig:non-set-ordered-2} {\small A scheme for illustrating the proof of Lemma \ref{thm:symmetric-graphs-graceful-coloring}.}}
\end{figure}

Observe Fig.\ref{fig:non-set-ordered-2} carefully, we can see that there are many ways to form a symmetric graph $H\perp H'$ admitting a set-ordered gracefully total coloring. So, for a base $\textbf{\textrm{T}}=(T_k)^n_{k=1}$ made by disjoint graphs $T_1,T_2,\dots ,T_n$, where each $T_k$ is a bipartite and connected $(p_k,q_k)$-graph admitting a non-set-ordered gracefully total coloring, we have a symmetric graph $T_k\perp T'_k$ admitting a set-ordered gracefully total coloring, and then we get a \emph{symmetric base} $\textbf{\textrm{T}}\perp \textbf{\textrm{T}}'=(T_k\perp T'_k)^n_{k=1}$ made by the base $\textbf{\textrm{T}}=(T_k)^n_{k=1}$ and its copy $\textbf{\textrm{T}}'=(T'_k)^n_{k=1}$. We call the following set
\begin{equation}\label{eqa:symmetric-graceful-coloring-graphic-lattice}
\textbf{\textrm{L}}(F\odot (\textbf{\textrm{T}}\perp \textbf{\textrm{T}}'))=\{H\odot^n_{k=1}a_k(T_k\perp T'_k),~a_k\in Z^0,~(T_k\perp T'_k)\in \textbf{\textrm{T}}\perp \textbf{\textrm{T}}',H\in F\}
\end{equation} a \emph{set-ordered gracefully total coloring symmetric graphic lattice} with $\sum^n_{k=1}a_k\geq 1$ based on Lemma \ref{thm:H-connecte-two-graphs} and Lemma \ref{thm:symmetric-graphs-graceful-coloring}, as well as the base $\textbf{\textrm{T}}\perp \textbf{\textrm{T}}'=(T_k\perp T'_k)^n_{k=1}$.

\begin{defn}\label{defn:5C-5C-labelling}
$^*$ If a proper total coloring $f:V(G)\cup E(G)\rightarrow [1,M]$ for a bipartite $(p,q)$-graph $G$ holds:

(i) (e-magic) $f(uv)+|f(u)-f(v)|=k$;

(ii) (ee-difference) each edge $uv$ matches with another edge $xy$ holding $f(uv)=|f(x)-f(y)|$ (or $f(uv)=2(p+q)-|f(x)-f(y)|$);

(iii) (ee-balanced) let $s(uv)=|f(u)-f(v)|-f(uv)$ for $uv\in E(G)$, then there exists a constant $k'$ such that each edge $uv$ matches with another edge $u'v'$ holding $s(uv)+s(u'v')=k'$ (or $2(p+q)+s(uv)+s(u'v')=k'$) true;

(iv) (set-ordered) $\max f(X)<\min f(Y)$ (or $\min f(X)>\max f(Y)$) for the bipartition $(X,Y)$ of $V(G)$.

(v) (edge-fulfilled) $f(E(T))=[1,q]$.

We call $f$ a \emph{$5$C-total coloring} of $G$.\qqed
\end{defn}

\begin{thm}\label{thm:bipartite-equivalent-total-coloring}
A bipartite and connected $(p,q)$-graph $T~(\neq K_{1,m})$ admits a proper total coloring $f:V(T)\cup E(T)$. The following assertions are equivalent to each other:
\begin{asparaenum}[(1)]
\item $T$ admits a set-ordered gracefully total coloring.
\item $T$ admits a set-ordered odd-gracefully total coloring.
\item $T$ admits a set-ordered edge-magic total coloring.
\item $T$ admits a set-ordered $5$C-total coloring.
\item $T$ admits a set-ordered felicitous total coloring.
\item $T$ admits a set-ordered odd-elegant total coloring.
\item $T$ admits a set-ordered harmonious total coloring.
\item $T$ admits a $(k,d)$-graceful coloring.
\end{asparaenum}
\end{thm}
\begin{proof} By the assertion (1), we suppose that the bipartite and connected $(p,q)$-graph $T$ admits a proper total coloring $f:V(T)\cup E(T)$ to be a set-ordered gracefully total coloring, $\max f(X)<\min f(Y)$, where $(X,Y)$ is the bipartition of $V(T)$. Notice that $f(x)=f(y)$ for some distinct vertices $x,y\in V(T)$, and $f(E(T))=\{f(uv)=|f(u)-f(v)|:uv\in E(T)\}=[1,q]$. Let $X=\{x_{i}:i\in [1,s]$ and $Y=\{y_{j}:j\in [1,t]\}$ with $s+t=p$. ``($k$)$\Rightarrow$ ($j$)'' means the assertion ($k$) deduces the assertion ($j$) in the following proof.

(1)$\Rightarrow$(2) We define a proper total coloring $g_{\textrm{odd}}$ of $T$ as: $g_{\textrm{odd}}(x_{i})=2f(x_{i})-1$ for $x_{i}\in X$, and $g_{\textrm{odd}}(y_{j})=2f(y_{j})-2$ for $y_{j}\in Y$, $g_{\textrm{odd}}(x_{i}y_{j})=|g_{\textrm{odd}}(y_{j})-g_{\textrm{odd}}(x_{i})|=|2f(y_{j})-1-2f(x_{i})|=2f(y_{j})-1-2f(x_{i})=2f(x_{i}y_{j})-1$, which induces $g(E(T))=[1,2q-1]^o$. $g_{\textrm{odd}}$ is a proper total coloring to be a \emph{set-ordered odd-gracefully total coloring}. Moreover, we have $f(x_{i})=\frac{1}{2}g_{\textrm{odd}}(x_{i})$ for $x_{i}\in X$, and $f(y_{j})=\frac{1}{2}[g_{\textrm{odd}}(y_{j})+1]$ for $y_{j}\in Y$, immediately, we have:``$T$ admits a set-ordered gracefully total coloring if and only if $T$ admits a set-ordered odd-gracefully total coloring''.

(1)$\Rightarrow$(3) We define a proper total coloring $g_{\textrm{mag}}$ of $T$ as: $g_{\textrm{mag}}(x_{i})=\max f(X)+\min f(X)-f(x_{i})$ for $x_{i}\in X$, and $g_{\textrm{mag}}(y_{j})=f(y_{j})$ for $y_{j}\in Y$, and $g_{\textrm{mag}}(x_{i}y_{j})=q+1-f(x_{i}y_{j})$ for $x_{i}y_{j}\in E(T)$. Then
$$g_{\textrm{mag}}(x_{i})+g_{\textrm{mag}}(x_{i}y_{j})+g_{\textrm{mag}}(y_{j})=q+1+\max f(X)+\min f(X)$$
for each edge $x_{i}y_{j}\in E(T)$, so we claim that $g_{\textrm{mag}}$ is a proper total coloring to be a \emph{set-ordered edge-magic total coloring}. Thereby, $T$ admits a set-ordered gracefully total coloring if and only if $T$ admits a set-ordered edge-magic total coloring.

(1)$\Rightarrow$(4) We define a proper total coloring $g_{\textrm{5C}}$ of $T$ as: $g_{\textrm{5C}}(w)=f(w)$ for each vertex $w\in V(T)$, $g_{\textrm{5C}}(x_{i}y_{j})=q+1-f(x_{i}y_{j})$ for each edge $x_{i}y_{j}\in E(T)$.

(i) (e-magic) $g_{\textrm{5C}}(x_{i}y_{j})+|g_{\textrm{5C}}(x_{i})-g_{\textrm{5C}}(y_{j})|=q+1-f(x_{i}y_{j})+|f(x_{i})-f(y_{j})|=q+1$ for each edge $x_{i}y_{j}\in E(T)$.

(ii) (ee-difference) $g_{\textrm{5C}}(x_{i}y_{j})=q+1-f(x_{i}y_{j})=f(x'_{i}y'_{j})$ for $x_{i}y_{j}\in E(T)$, since each edge $x_{i}y_{j}\in E(T)$ corresponds an edge $x'_{i}y'_{j}\in E(T)$ form $f(x_{i}y_{j})+f(x'_{i}y'_{j})=q+1$.

(iii) (ee-balanced) From $s(xy)=|f(x)-f(y)|-f(xy)$ for $uv\in E(G)$, we have $s_{\textrm{5C}}(xy)=|g_{\textrm{5C}}(x)-g_{\textrm{5C}}(y)|-g_{\textrm{5C}}(xy)=|f(x)-f(y)|-[q+1-f(xy)]=|f(x)-f(y)|-f(x'y')$, where $f(xy)+f(x'y')=q+1$. So, $s_{\textrm{5C}}(xy)+s_{\textrm{5C}}(x'y')=[|f(x)-f(y)|-f(x'y')]+|f(x')-f(y')|-f(xy)=0$.

(iv) (set-ordered) $\max g_{\textrm{5C}}(X)<\min g_{\textrm{5C}}(Y)$ by the property of the set-ordered gracefully total coloring $f$.

(v) (edge-fulfilled) $\max g_{\textrm{5C}}(E(T))=[1,q]$, since $f(E(T))=[1,q]$.

Thereby, we claim that $g_{\textrm{5C}}$ is a $5$C-total coloring of $T$.

(1)$\Rightarrow$(5) We define a proper total coloring $g_{\textrm{fel}}$ of $T$ as: $g_{\textrm{fel}}(y_j)=f(y_j)$ for $y_j\in Y$, $g_{\textrm{fel}}(x_{i})=\max f(X)+\min f(X)-f(x_{i})$ for $x_{i}\in X$, and
$${
\begin{split}
g_{\textrm{fel}}(x_{i}y_{j})&=g_{\textrm{fel}}(x_{i})+g_{\textrm{fel}}(y_{j})=\max f(X)+\min f(X)-f(x_{i})+f(y_j)~(\bmod~q)\\
&=\max f(X)+\min f(X)+f(x_{i}y_j)~(\bmod~q)
\end{split}}
$$ for each edge $x_{i}y_{j}\in E(T)$. So, $\max g_{\textrm{fel}}(X)<\min g_{\textrm{fel}}(Y)$, and $g_{\textrm{fel}}(E(T))=[0,q-1]$, that is, the coloring $g_{\textrm{fel}}$ is really a set-ordered felicitous total coloring of $T$.

(1)$\Rightarrow$(6) We define a proper total coloring $g_{\textrm{ele}}$ as: $g_{\textrm{ele}}(x_{i})=2[\max f(X)+\min f(X)-f(x_{i})]-1$ for $x_{i}\in X$, and $g_{\textrm{ele}}(y_{j})=2f(y_{j})-2$ for $y_{j}\in Y$, and for each edge $x_{i}y_{j}\in E(T)$, we have
$${
\begin{split}
g_{\textrm{ele}}(x_{i}y_{j})&=g_{\textrm{ele}}(x_{i}y_{j})+g_{\textrm{ele}}(x_{i}y_{j})~(\bmod~2q)\\
&=2[\max f(X)+\min f(X)-f(x_{i})]-1+[2f(y_{j})-2]~(\bmod~2q)\\
&=2[\max f(X)+\min f(X)+f(x_{i}y_j)]-1~(\bmod~2q).
\end{split}}
$$
We can see that $\max g_{\textrm{ele}}(X)<\min g_{\textrm{ele}}(Y)$, and $g_{\textrm{ele}}(E(T))=[1,2q-1]^o$, so $g_{\textrm{ele}}$ is really a set-ordered odd-elegant total coloring of $T$.

(1)$\Rightarrow$(7) We define a proper total coloring $g_{\textrm{har}}$ as: $g_{\textrm{har}}(x_{i})=\max f(X)+\min f(X)-f(x_{i})$ for $x_{i}\in X$, and $g_{\textrm{har}}(y_{j})=f(y_{j})$ for $y_{j}\in Y$, and for each edge $x_{i}y_{j}\in E(T)$, we have
$g_{\textrm{har}}(x_{i})+g_{\textrm{har}}(y_{j})=\max f(X)+\min f(X)-f(x_{i})+f(y_{j})=\max f(X)+\min f(X)+f(x_{i}y_{j})$, which induces a consecutive set $[\max f(X)+\min f(X)+1,\max f(X)+\min f(X)+q]$, so we define
$g_{\textrm{har}}(x_{i}y_{j})=g_{\textrm{har}}(x_{i})+g_{\textrm{har}}(y_{j})~~(\bmod~q)$. Clearly, $g_{\textrm{har}}$ is a set-ordered harmonious total coloring, since $\max g_{\textrm{har}}(X)<\min g_{\textrm{har}}(Y)$ and $g_{\textrm{har}}(E(T))=[0,q-1]$.

(1)$\Rightarrow$(8) We define a proper total coloring $g_{\textrm{kd}}$ as:$g_{\textrm{kd}}(x_{i})=f(x_{i})\cdot d$ for $x_{i}\in X$, and $g_{\textrm{kd}}(y_{j})=k+f(y_{j})\cdot d$ for $y_{j}\in Y$, and for each edge $x_{i}y_{j}\in E(T)$, we set $g_{\textrm{kd}}(x_{i}y_{j})=g_{\textrm{kd}}(y_{j})-g_{\textrm{kd}}(x_{i})=k+f(y_{j})\cdot d-f(x_{i})\cdot d=k+f(x_{i}y_{j})\cdot d$. Thereby, $g_{\textrm{kd}}(E(T))=\{k+d,k+2d,\dots ,k+qd\}$, which implies that $g_{\textrm{kd}}$ is a $(k,d)$-gracefully total coloring of $T$.

\vskip 0.2cm

Notice that each translation between $f$ and $g_{\varepsilon}$ with $\varepsilon\in \{$odd, mag, 5C, fel, har, kd$\}$ is linear, so it is easily to obtain the original coloring $f$ from $g_{\varepsilon}$, which means the equivalent proof. The theorem has been shown completely.
\end{proof}

\begin{figure}[h]
\centering
\includegraphics[width=16.2cm]{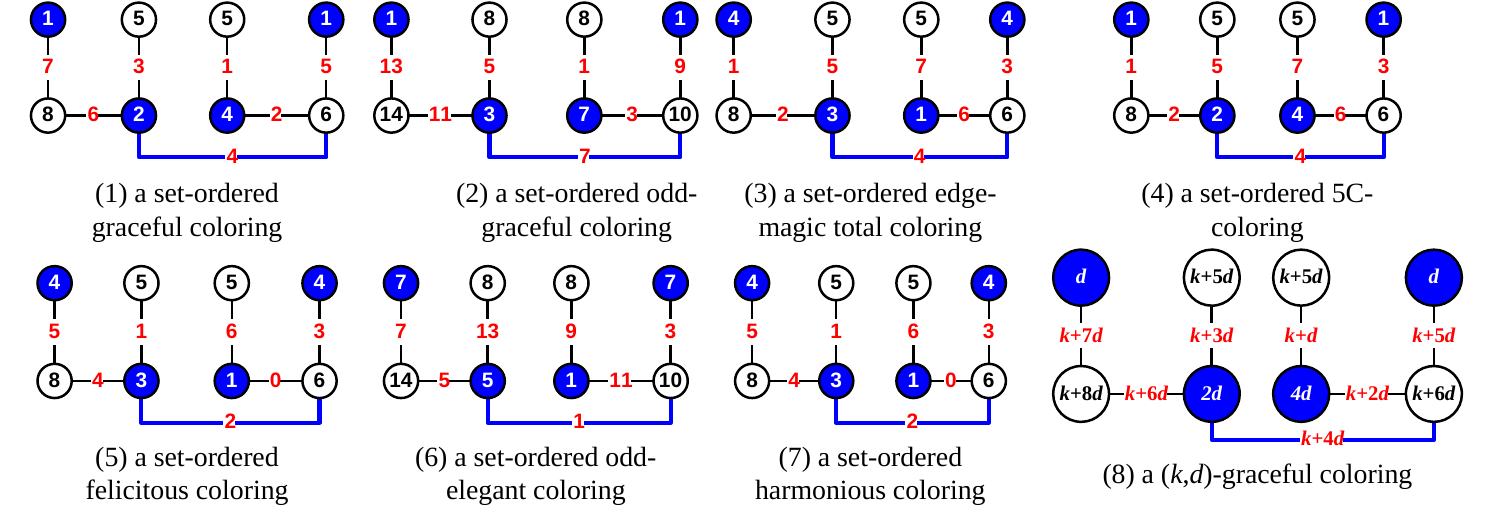}\\
\caption{\label{fig:equivalent-colorings} {\small Examples for understanding Theorem \ref{thm:bipartite-equivalent-total-coloring}.}}
\end{figure}

\begin{problem}\label{qeu:trees-various-colorings}
Since a colored connected $(p,q)$-graph $G$ can be vertex-split into some colored trees, or be leaf-split into colored trees, we have questions as follows:
\begin{asparaenum}[\textbf{\textrm{Tree}}-1. ]
\item \textbf{Construct} graphs or graphic lattices admitting set-ordered gracefully total colorings or set-ordered graceful labellings.
\item If a tree $T$ admits a graceful labelling, then \textbf{does} it admit a gracefully total coloring?
\item \textbf{Determine} each $W$-type coloring defined in Definition \ref{defn:new-graceful-strongly-colorings} for trees.
\item \textbf{Determine} trees admitting set-ordered gracefully total colorings, or determine trees refusing set-ordered gracefully total colorings.
\item About the parameter $v_{W}(G)=\min_f\{|f(V(G))|\}$ over all $W$-type coloring $f$ of $G$ for a fixed $W\in[1,27]$ based on Definition \ref{defn:new-graceful-strongly-colorings}, for each integer $m$ subject to $v_{W}(G)<m\leq p-1$, \textbf{does} there exist a $W$-type coloring $g$ holding $|g(V(G))|=m$?
\end{asparaenum}
\end{problem}

\subsection{Applications of equivalent colorings}

\subsubsection{Equivalent coloring-based lattices} Suppose that each graph $T^c_i$ of the base $\textbf{\textrm{T}}^c=(T^c_1,T^c_2,\dots, T^c_n)$ of a colored graphic lattice $\textbf{\textrm{L}}(\textbf{\textrm{T}}^c\odot F^c_{p,q})$ defined in (\ref{eqa:graphic-lattice-colored}) admits a $W$-type coloring $f_i$, such that $\textbf{\textrm{T}}^c$ admits a flawed $W$-type coloring $f$, then we rewrite $\textbf{\textrm{L}}(\textbf{\textrm{T}}^c\odot F^c_{p,q})$ as $\textbf{\textrm{L}}(\textbf{\textrm{T}}^c\odot (f) F^c_{p,q})$.
If each $W$-type coloring $f_i$ is equivalent to another $W'$-type coloring $g_i$, such that the flawed $W$-type coloring $f$ is equivalent to a $W'$-type coloring $g$, thus, we get a colored graphic lattice $\textbf{\textrm{L}}(\textbf{\textrm{T}}^c\odot (g) F^c_{p,q})$. Thereby, we say two lattices $\textbf{\textrm{L}}(\textbf{\textrm{T}}^c\odot (f) F^c_{p,q})$ (as a \emph{public-key set}) and $\textbf{\textrm{L}}(\textbf{\textrm{T}}^c\odot (g) F^c_{p,q})$ (as a \emph{private-key set}) are equivalent to each other. In the language of graph homomorphism, we have two homomorphically equivalent graphic lattice homomorphisms $\textbf{\textrm{L}}(\textbf{\textrm{T}}^c\odot (f) F^c_{p,q})\leftrightarrow \textbf{\textrm{L}}(\textbf{\textrm{T}}^c\odot (g) F^c_{p,q})$. Also, $\textbf{\textrm{L}}(\textbf{\textrm{T}}^c\odot (f) F^c_{p,q}) \leftrightarrow \textbf{\textrm{L}}(\textbf{\textrm{T}}^c\odot (g) F^c_{p,q})$, a pair of \emph{homomorphically equivalent graphic lattices}. Consequently, each colored graph of $\textbf{\textrm{L}}(\textbf{\textrm{T}}^c\odot (f) F^c_{p,q})$ is equivalent to some colored graph of $\textbf{\textrm{L}}(\textbf{\textrm{T}}^c\odot (g) F^c_{p,q})$ too.

Let $\varphi$ be a linear transformation between $f_i$ and $g_i$, that is, $g_i=\varphi(f_i)$ with $i\in [1,n]$, and $g=\varphi(f)$. So, $f_i=\varphi^{-1}(g_i)$ with $i\in [1,n]$, and $f=\varphi^{-1}(g)$. In the view of linear transformation, we have
\begin{equation}\label{eqa:equivalent-lattices}
\textbf{\textrm{L}}(\textbf{\textrm{T}}^c\odot (g) F^c_{p,q})=\varphi(\textbf{\textrm{L}}(\textbf{\textrm{T}}^c\odot (f) F^c_{p,q})),~\textbf{\textrm{L}}(\textbf{\textrm{T}}^c\odot (f) F^c_{p,q})=\varphi^{-1}(\textbf{\textrm{L}}(\textbf{\textrm{T}}^c\odot (g) F^c_{p,q}))
\end{equation}
with $G=\varphi(H)$ and $H=\varphi^{-1}(G)$ for $H\in \textbf{\textrm{L}}(\textbf{\textrm{T}}^c\odot (f) F^c_{p,q})$ and $G\in \textbf{\textrm{L}}(\textbf{\textrm{T}}^c\odot (g) F^c_{p,q})$.

\subsubsection{Encrypting graphs for topological authentications} Let $F_{\textrm{equ}}=\{H_k :k\in [1,m]\}$ be a set of disjoint connected graphs $H_1,H_2,\dots, H_m$ with $H_i\cong H_j$, where each $H_k$ admits a $W_k$-type coloring $h_k$ with $k\in [1,m]$, and moreover there a linear transformation $\theta_{i,j}$ holding $h_j=\theta_{i,j}(h_i)$ for any pair of distinct $h_i$ and $h_j$. For a given graph $G$ of $n$ vertices with $n\leq m$, we can apply $F_{\textrm{equ}}$ to encrypt $G$ wholly. In mathematical mapping, we color $G$ with the elements of $F_{\textrm{equ}}$ in the following way: Take a proper vertex coloring $\theta:V(G)\rightarrow F_{\textrm{equ}}$, and join some vertices $x_{i,s}$ of $H_i$ with some vertices $x_{j,t}$ of $H_j$ by edges if $u_iu_j\in E(G)$, where $V(G)=\{u_1,u_2,\dots u_n\}$. The new graph is denoted as $G\triangleleft F_{\textrm{equ}}$, clearly, there are many graphs of form $G\triangleleft F_{\textrm{equ}}$. Suppose that the vertex $x_{i,s}$ of $H_i$ is connected with $x_{i,t}$ by a path $P^i_{s,t}=x_{i,s} x_{i,s+1}\cdots x_{i,t-1}x_{i,t}$, correspondingly, $H_j$ has a path $P^j_{s,t}=x_{j,s} x_{j,s+1}\cdots x_{j,t-1}x_{j,t}$ to connect the vertex $x_{j,s}$ with $x_{j,t}$ in $H_j$. Then, $h_j(P^j_{s,t})=\theta_{i,j}(h_i(P^i_{s,t}))$, we color the edge $x_{j,s}x_{j,t}$ of $G\triangleleft F_{\textrm{equ}}$ by a function $f_{i,j}(h_i(P^i_{s,t}),h_j(P^j_{s,t}))$. Thereby, $G\triangleleft F_{\textrm{equ}}$ admits a total coloring $\varphi$ made by $\{h_k\}^m_1$ and $\{\{ \theta_{i,j}, f_{i,j}\}^m_1\}^m_1$.

Let $J$ be a connected bipartite graph with bipartition $(X,Y)$. By two equivalent lattices and the linear transformation $\varphi$ defined in (\ref{eqa:equivalent-lattices}), we use the elements of the lattice $\textbf{\textrm{L}}(\textbf{\textrm{T}}^c\odot (f) F^c_{p,q})$ to color the vertices of $X$, and apply the elements of the lattice $\textbf{\textrm{L}}(\textbf{\textrm{T}}^c\odot (g) F^c_{p,q})$ to color the vertices of $Y$, so the resultant graph is written as $SH=J\triangleleft (\textbf{\textrm{L}}(\textbf{\textrm{T}}^c\odot (f) F^c_{p,q}),\textbf{\textrm{L}}(\textbf{\textrm{T}}^c\odot (g) F^c_{p,q}))$, and it is a bipartite graph with the bipartition $(X_f,Y_g)$ such that its edge $G_xG_y$ with $G_x\in X_f$ and $G_y\in Y_g$ holds $G_y=\varphi(G_x)$ and $G_x=\varphi^{-1}(G_y)$ for $G_x\in \textbf{\textrm{L}}(\textbf{\textrm{T}}^c\odot (f) F^c_{p,q})$ (as a \emph{public-key set}) and $G_y\in \textbf{\textrm{L}}(\textbf{\textrm{T}}^c\odot (g) F^c_{p,q})$ (as a \emph{private-key set}). Since there are many ways to join $G_x$ with $G_y$ together by edges, and there are many graphs can be used to color two ends $x$ and $y$ of an edge $xy$ of $J$, so the number of the graphs of form $SH$ is greater than one.

\subsection{$(p,s)$-gracefully total numbers and $(p,s)$-gracefully total authentications}

As known, each bipartite complete graph $K_{m,n}$ does not admit a gracefully total coloring $g$ with $g(x)=g(y)$ for some distinct two vertices $x,y\in V(K_{m,n})$, meanwhile $K_{m,n}$ admits a graceful labelling $f$ with $f(u)\neq f(w)$ for any pair of vertices $u,w\in V(K_{m,n})$. We have two kinds of extremum graphs as follows:

1. If a connected $(p,q)$-graph $H^{+}$ admits a gracefully total coloring, and adding a new edge $e$ to $H^{+}$ makes a new graph $H^{+}+e$ such that $H^{+}+e$ does not admit a gracefully total coloring, we say $H^{+}$ a \emph{gracefully$^{+}$ critical graph}.

2. If a connected $(s,t)$-graph $H^{-}$ does not admit a gracefully total coloring, but removing an edge $e'$ from $H^{-}$ produces a new graph $H^{-}-e'$ admitting a gracefully total coloring, we say $H^{-}$ a \emph{gracefully$^{-}$ critical graph}.

A \emph{$(p,s)$-gracefully total number} $R_{grace}(p,s)$ is an extremum number, such that any red-blue edge-coloring of each complete graph $K_{m}$ of $m=R_{grace}(p,s)-1$ vertices does not induce a gracefully$^{+}$ critical graph $H^{+}$ of $p$ vertices and a gracefully$^{-}$ critical graph $H^{-}$ of $s$ vertices, such that each edge of $H^{+}$ is red and each edge of $H^{-}$ is blue.

If a connected graph $G$ contains a gracefully$^{+}$ critical graph $H^{+}$ of $p$ vertices and a gracefully$^{-}$ critical graph $H^{-}$ of $s$ vertices, and both critical graphs $H^{+}$ and $H^{-}$ are edge-disjoint in $G$, then we call $G$ a \emph{$(p,s)$-gracefully total authentication}, and $(H^{+},H^{-})$ a \emph{$(p,s)$-gracefully total matching}.

We have the following obvious facts:

\begin{prop} \label{thm:gracefully-total-numbers-graphs}

(1) Matching of gracefully total graphs hold: If $(H^{+},H^{-})$ is a gracefully total matching, so are $(H^{+}+e,H^{-}-e')$, $(H^{+},H^{+}+e)$ and $(H^{-},H^{-}-e')$ too.

(2) Gracefully total numbers hold: $R_{grace}(p,s)=R_{grace}(s,p)$ with $p\geq 4$ and $s\geq 4$.
\end{prop}

\begin{figure}[h]
\centering
\includegraphics[width=16cm]{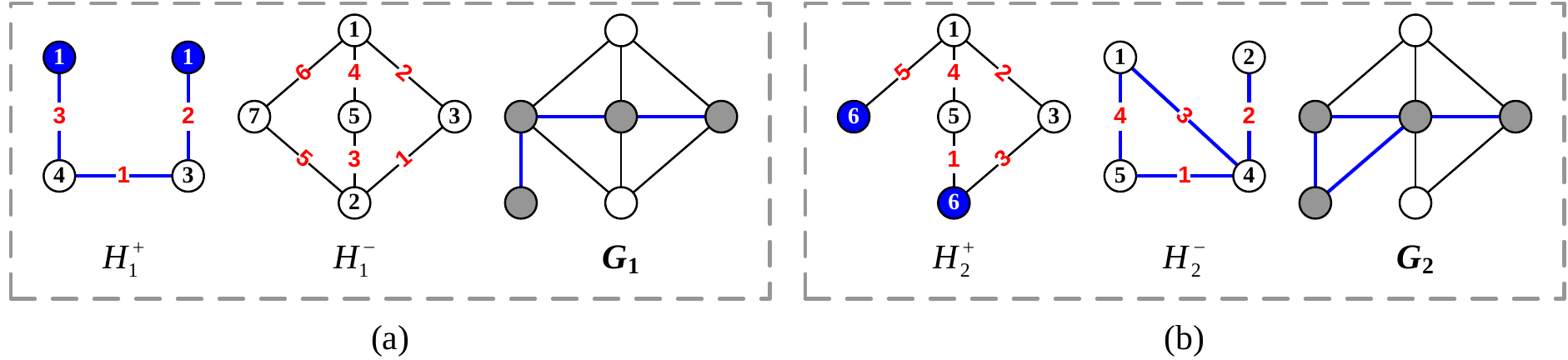}
\caption{\label{fig:extremum-11}{\small (a) $(H^{+}_1,H^{-}_1)$ is a $(4,5)$-gracefully total matching, and $G_1$ is a $(4,5)$-gracefully total authentication; (b) $(H^{+}_2,H^{-}_2)$ is a $(5,4)$-gracefully total matching, and $G_2$ is a $(5,4)$-gracefully total authentication.}}
\end{figure}

\begin{figure}[h]
\centering
\includegraphics[width=16cm]{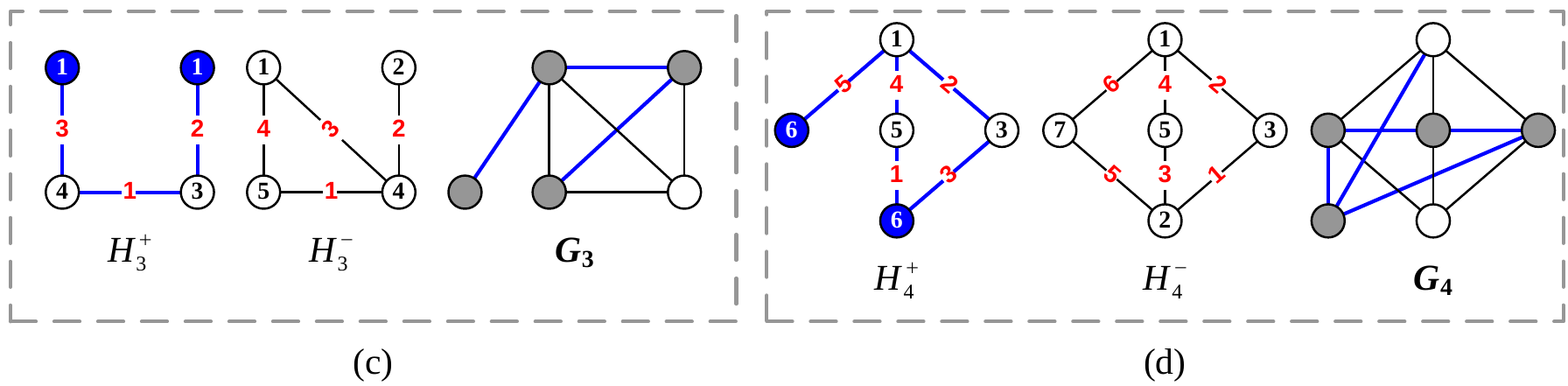}
\caption{\label{fig:extremum-22}{\small (c) $(H^{+}_3,H^{-}_3)$ is a $(5,4)$-gracefully total matching, and $G_3$ is a $(4,4)$-gracefully total authentication; (d) $(H^{+}_4,H^{-}_4)$ is a $(5,5)$-gracefully total matching, and $G_4$ is a $(5,5)$-gracefully total authentication.}}
\end{figure}

In Fig.\ref{fig:extremum-33}, we can see the following facts: (1) A red-blue edge-coloring of $K_5$ does not induce a gracefully$^{-}$ critical graph $H^{-}$ of four vertices, so the $(4,4)$-gracefully total number $R_{grace}(4,4)=6$. (2) a red-blue edge-coloring of $K_6$ does not induce a gracefully$^{-}$ critical graph $H^{-}$ of five vertices and a gracefully$^{+}$ critical graph $H^{+}$ of $5$ vertices, such that each edge of $H^{+}$ is red and each edge of $H^{-}$ is blue. So, the $(5,5)$-gracefully total number $R_{grace}(5,5)=7$. and (3) $R_{grace}(4,5)\geq 7$.

\begin{figure}[h]
\centering
\includegraphics[width=10cm]{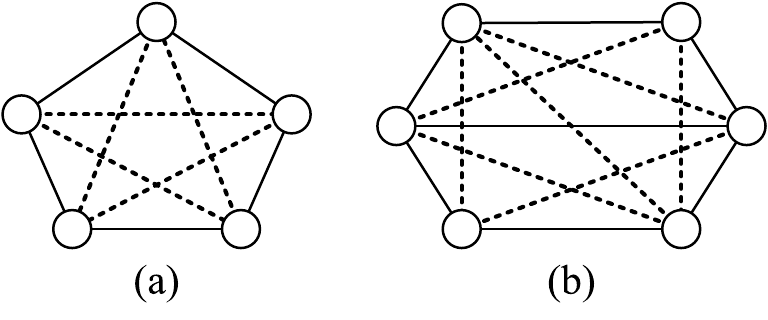}
\caption{\label{fig:extremum-33}{\small (a) A red-blue edge-coloring of $K_5$; (b) a red-blue edge-coloring of $K_6$.}}
\end{figure}

\begin{problem}\label{qeu:gracefully-total-numbers-graphs}
The gracefully total number $R_{grace}(p,s)$ has been designed by means of the idea of Ramsey number of graph theory. It seems to be not easy to \textbf{compute} the exact value of a gracefully total number $R_{grace}(p,s)$.

A connected graph $G$ containing a $(p,s)$-gracefully total matching $(H^{+},H^{-})$ is just a $(p,s)$-gracefully total authentication when $H^{+}$ is as a public key and $H^{-}$ is as a public key, and $G$ is smallest if $|V(G)|\leq |V(H)|$ and $|E(G)|\leq |E(H)|$ for any $(p,s)$-gracefully total authentication $H$. So, there are many smallest $(p,s)$-gracefully total authentications by Proposition \ref{thm:gracefully-total-numbers-graphs}, \textbf{find} all smallest $(p,s)$-gracefully total authentications.\qqed
\end{problem}

Let $K_{2,n}$ (as a gracefully$^+$ critical graph $H^{+}$) be a complete bipartite graph with its vertex set $V(K_{2,n})=\{x_1,x_2\}\cup \{y_1,y_2,\dots ,y_n\}$ and edge set $E(K_{2,n})=\{x_iy_j:~i\in[1,2],j\in [1,n]\}$. We define a graceful labelling $f$ for $K_{2,n}$ in the way: $f(x_i)=i$ for $i\in[1,2]$, and $f(y_j)=2j+1$ for $j\in [1,n]$; next we set $f(x_iy_j)=f(y_j)-f(x_i)=2j+1-i$, which deduces $f(E(K_{2,n}))=[1,2n]$. Clearly, $K_{2,n}$ does not admit a gracefully total coloring $g$ with $g(x)=g(y)$ for some distinct two vertices $x,y\in V(K_{2,n})$.

We remove an edge $x_2y_n$ from $K_{2,n}$ to obtain a connected bipartite graph $K_{2,n}-x_2y_n$ (as a gracefully$^+$ critical graph $H^{+}$), and define a total coloring $g$ as: $g(x_1)=f(x_1)=1$, $g(x_2)=2n$, $g(y_n)=2n$, and $g(y_j)=f(y_j)$ for $j\in [1,2n-1]$; set $g(x_iy_j)=g(y_j)-g(x_i)$. So, $g(E(K_{2,n}-x_2y_n))=[1,2n-1]$, and $g(x_2)=2n=g(y_n)$. We claim that $g$ is a gracefully total coloring of the graph $K_{2,n}-x_2y_n$.

We add a new vertex $w$ to $K_{2,n}$, and add new edges $wy_k$ with $k\in [2,n]$, $wx_1$ and $y_1y_k$ with $k\in [2,n]$ to $K_{2,n}+w$, the resultant graph is denoted as $G$ with $(n+3)$ vertices. It is not hard to see that these new edges induce just a connected bipartite graph $K_{2,n}-x_2y_n$. Thereby, $G$ contains a gracefully$^{+}$ critical $(p,q)$-graph $H^{+}=K_{2,n}-x_2y_n$ and a gracefully$^{-}$ critical $(s,t)$-graph $H^{-}=K_{2,n}$, and both critical graphs $K_{2,n}-x_2y_n$ and $K_{2,n}$ are edge-disjoint in $G$. Thereby, $G$ is a smallest $(n+2,n+2)$-gracefully total authentication.

\subsection{Constructing gracefully graphic lattices}

In \cite{Hongyu-Wang-Bing-Yao-submitted-2020}, the authors have shown the following results for building up gracefully graphic lattices admitting proper gracefully total colorings.

\begin{defn} \label{defn:proper-graceful-total-colorings}
Suppose that a connected $(p,q)$-graph $G$ admits a proper total coloring $f:V(G)\cup E(G)\rightarrow [1,M]$, and there are $f(x)=f(y)$ for some pairs of vertices $x,y\in V(G)$. If $f(uv)=|f(u)-f(v)|$ for each edge $uv\in E(G)$, $f(E(G))=[1,q]$ and, $f(V(G))\subseteq [1,q+1]$, we call $f$ a \emph{proper gracefully total coloring}.\qqed
\end{defn}

\begin{lem} \label{thm:adding-leaves-general-trees}
Let $G$ be a connected graph admitting a proper gracefully total coloring. Another connected graph obtained by adding leaves to $G$ admits a proper gracefully total coloring too.
\end{lem}

\begin{thm} \label{thm:very-trees-admitting-graceful-colorings}
Every tree $T$ with diameter $D(T)\geq 3$ admits a proper gracefully total coloring. Furthermore, let $L(T)$ be the set of leaves of a tree $T$, if the tree $T-L(T)$ obtained by removing all leaves from $T$ holds $|V(T-L(T))|\leq |L(T)|$, then the tree $T$ admits a proper gracefully total coloring.
\end{thm}

\begin{lem} \label{thm:adding-leaves-trees-admitting-set-ordered-graceful}
If a tree $T$ with its diameter $D(T)\geq 3$ admits a set-ordered graceful labelling $f$ with $f(V(T))=[1,|V(T)|]$, then the resulting tree obtained by vertex-coinciding each vertex $x_i$ of $T$ with the maximum degree vertex of some star $K_{1,m_i}$ admits a proper gracefully total coloring.
\end{lem}

Let $F_{\textrm{so-gra}}$ (resp. $F_{\textrm{so-odd}}$) be a set of non-star connected graphs admitting set-ordered graceful labellings (resp. set-ordered odd-graceful labellings), and let $\textbf{\textrm{K}}=(K_{1,a_k})^n_{k=1}$ be a star-base made by disjoint stars $K_{1,a_1},K_{1,a_2},\dots ,K_{1,a_{n}}$. Lemma \ref{thm:adding-leaves-trees-admitting-set-ordered-graceful} enables us to obtain a \emph{graceful-coloring star-graphic lattice} as follows:
\begin{equation}\label{eqa:star-graceful-coloring-graphic-lattices}
\textbf{\textrm{L}}(\textbf{\textrm{F}}_{\textrm{so-gra}}\odot \textbf{\textrm{K}})=\left \{T\odot^n_{k=1}a_kK_{1,a_k}, a_k\in Z^0,T\in F_{\textrm{so-gra}}\right \}
\end{equation} with $\sum^n_{k=1}a_k=|V(T)|$, where each vertex $x$ of $T$ is vertex-coincided with a vertex of some star $K_{1,a_k}$. Moreover, we have an \emph{odd-graceful-coloring star-graphic lattice}
\begin{equation}\label{eqa:star-odd-graceful-coloring-graphic-lattices}
\textbf{\textrm{L}}(\textbf{\textrm{F}}_{\textrm{so-odd}}\odot \textbf{\textrm{K}})=\left \{H\odot^n_{k=1}a_kK_{1,a_k}, a_k\in Z^0,H\in F_{\textrm{so-odd}}\right \}
\end{equation} with $\sum^n_{k=1}a_k=|V(H)|$, where each vertex $x$ of $H$ is vertex-coincided with a vertex of some star $K_{1,a_k}$.

If a star-base $\textbf{\textrm{K}}=(K_{1,a_k})^n_{k=1}$ holds $\{a_k\}^n_{k=1}$ to be a Fibonacci sequence, we call $\textbf{\textrm{L}}(\textbf{\textrm{F}}_{\textrm{so-gra}}\odot \textbf{\textrm{K}})$ a \emph{Fibonacci-star graceful-coloring graphic lattice}, and $\textbf{\textrm{L}}(\textbf{\textrm{F}}_{\textrm{so-odd}}\odot \textbf{\textrm{K}})$ a \emph{Fibonacci-star odd-graceful-coloring graphic lattice}.

Lemma \ref{thm:adding-leaves-trees-admitting-set-ordered-graceful} can be restated as: ``\emph{Adding leaves to a tree admitting a set-ordered graceful labelling produces a haired tree admitting a proper gracefully total coloring}.''

Notice that the trees $T\odot^n_{k=1}a_kK_{1,a_k}$ in $\textbf{\textrm{L}}(\textbf{\textrm{F}}_{\textrm{so-gra}}\odot \textbf{\textrm{K}})$ (resp. $H\odot^n_{k=1}a_kK_{1,a_k}$ in $\textbf{\textrm{L}}(\textbf{\textrm{F}}_{\textrm{so-odd}}\odot \textbf{\textrm{K}})$) forms a set, in fact, such that each tree in $\{T\odot^n_{k=1}a_kK_{1,a_k}\}$ (resp. $\{H\odot^n_{k=1}a_kK_{1,a_k}\}$) admits a proper gracefully total coloring (resp. a proper odd-gracefully total coloring).

\vskip 0.4cm

\noindent\textbf{VERTEX-INTEGRATING algorithm.}

Suppose that a connected and bipartite $(n,q)$-graph $T$ with bipartition $(X,Y)$ admits a set-ordered proper graceful coloring $f$ (resp. a set-ordered proper odd-graceful labelling) holding $f(X)<f(Y)$ for $X=\{x_1,x_2,\dots ,x_s\}$ and $Y=\{y_1,y_2,\dots ,y_t\}$ with $s+t=n=|V(T)|$. Each connected and bipartite graph $H_k$ with $k\in [1,n]$ admits a set-ordered proper gracefully total coloring (resp. a set-ordered odd-gracefully total coloring). $G=T\odot^n_{k=1}H_k$ is obtained by doing

S-1. For $k\in [1,s]$, we vertex-integrate each vertex $x_k$ of the graph $T$ with the vertex $x_{k,1}$ of $H_k$ into one vertex, denoted as $x_k$ still.

S-2. For $j\in [1,t]$, we vertex-integrate each vertex $y_j$ of the graph $T$ with the vertex $y_{s+j,b_{s+j}}$ of $H_{s+j}$ into one vertex, denoted as $y_j$ still.

Let $A=\sum^s_{k=1}e_k$, $B=\sum^t_{r=1}e_{s+r}$, where $e_k=|H_k|$ with $k\in [1,n]$. So, $G$ has $e_G=e_T+A+B=q+A+B$ edges in total.

\vskip 0.4cm

\begin{thm} \label{thm:adding-graphs-to-trees}
Suppose that $G=T\odot^n_{k=1}H_k$ is made by the VERTEX-INTEGRATING algorithm. If $f(x_k)<1+B+\sum^{s-k}_{r=1}e_{s-r+1}$, and $A\geq B$, the $G$ admits a proper gracefully total coloring $g$ (resp. a proper odd-gracefully total coloring) with $g(V(G))\subseteq [1,|E(G)|]$.
\end{thm}

Let $\textbf{\textrm{H}}=(H_1,H_2,\dots, H_n)=(H_k)^n_{k=1}$ be a base built by disjoint connected bipartite graph $H_1,H_2,\dots, H_n$, where each $H_k$ admits a set-ordered proper gracefully total coloring (resp. a set-ordered odd-gracefully total coloring). By Theorem \ref{thm:adding-graphs-to-trees}, we have a \emph{graceful-coloring graphic lattice}
\begin{equation}\label{eqa:graceful-coloring-graphic-lattices}
\textbf{\textrm{L}}(\textbf{\textrm{F}}_{\textrm{so-gra}}\odot \textbf{\textrm{H}})=\left \{T\odot^n_{k=1}a_kH_k,a_k\in Z^0,T\in F_{\textrm{so-gra}}\right \}
\end{equation}
with $\sum^n_{k=1}a_k=|V(T)|$, and furthermore we have an \emph{odd-graceful-coloring graphic lattice}
\begin{equation}\label{eqa:odd-graceful-coloring-graphic-lattices}
\textbf{\textrm{L}}(\textbf{\textrm{F}}_{\textrm{so-odd}}\odot \textbf{\textrm{H}})=\left \{G\odot^n_{i=1}a_kH_k,a_k\in Z^0,G\in F_{\textrm{so-odd}}\right \}.
\end{equation} with $\sum^n_{k=1}a_k=|V(G)|$. Each graph of $\textbf{\textrm{L}}(\textbf{\textrm{F}}_{\textrm{so-gra}}\odot \textbf{\textrm{H}})$ and $\textbf{\textrm{L}}(\textbf{\textrm{F}}_{\textrm{so-odd}}\odot \textbf{\textrm{H}})$ admits a proper gracefully total coloring, or a proper odd-gracefully total coloring.

\vskip 0.2cm

For understanding Theorem \ref{thm:adding-graphs-to-trees}, we show an example though Fig.\ref{fig:construct-graceful-space-1}, Fig.\ref{fig:construct-graceful-space-2} and Fig.\ref{fig:construct-graceful-compound}. In Fig.\ref{fig:construct-graceful-space-1}, we can see that a tree $T$ with bipartition $(X,Y)$ admits a set-ordered proper graceful coloring holding $f(X)<f(Y)$ for $X=\{x_1,x_2,x_3 ,x_4\}$ and $Y=\{y_1,y_2,y_3,y_4 ,y_5\}$ with $9=|V(T)|$, and each connected and bipartite graph $H_k$ with $k\in [1,n]$ admits a set-ordered proper gracefully total coloring in Theorem \ref{thm:adding-graphs-to-trees}. Fig.\ref{fig:construct-graceful-space-2} is for understanding the proof of Theorem \ref{thm:adding-graphs-to-trees}. The result is given in Fig.\ref{fig:construct-graceful-compound}.

\begin{figure}[h]
\centering
\includegraphics[width=16.2cm]{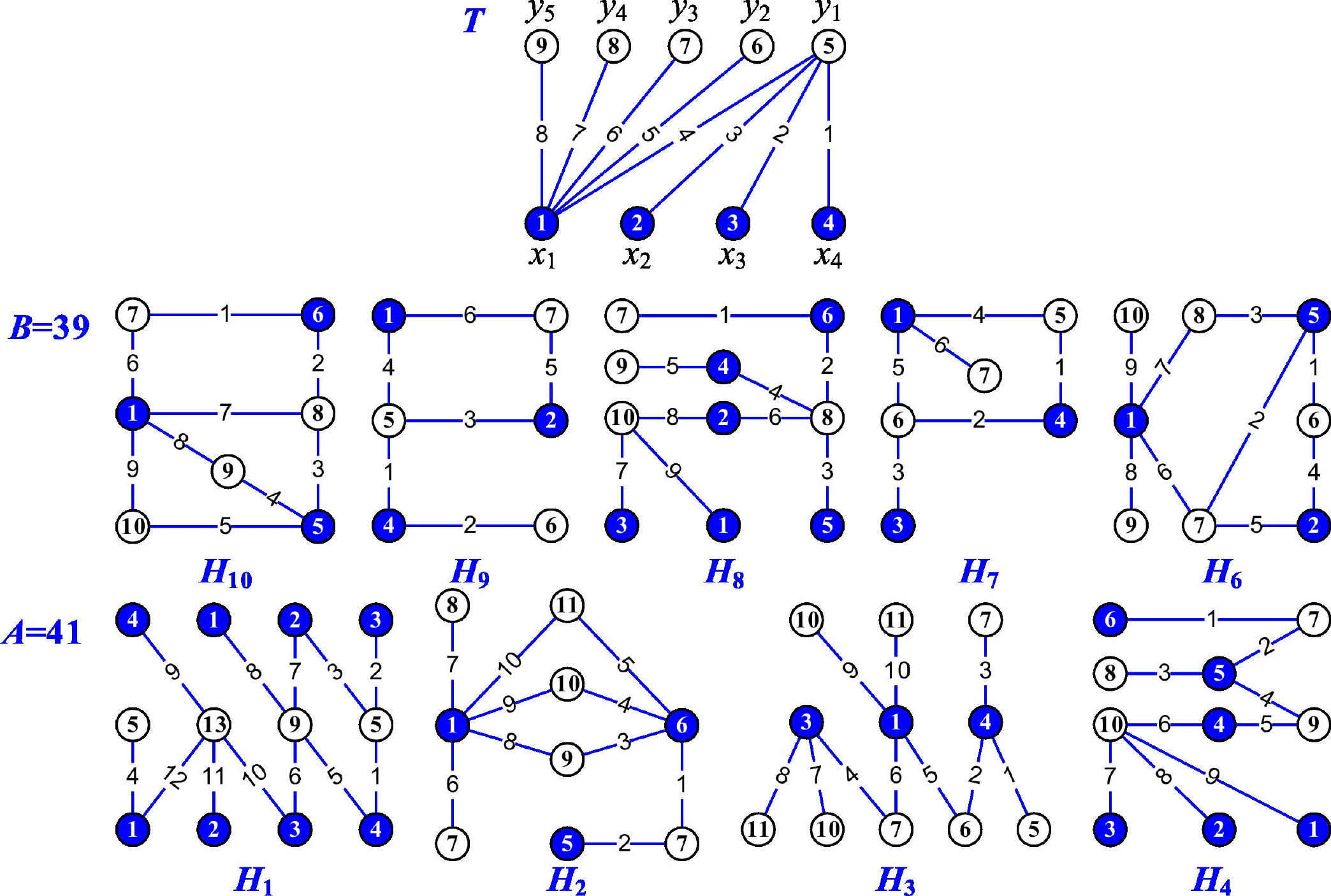}\\
\caption{\label{fig:construct-graceful-space-1} {\small A scheme for Illustrating Theorem \ref{thm:adding-graphs-to-trees}.}}
\end{figure}

\begin{figure}[h]
\centering
\includegraphics[width=16.2cm]{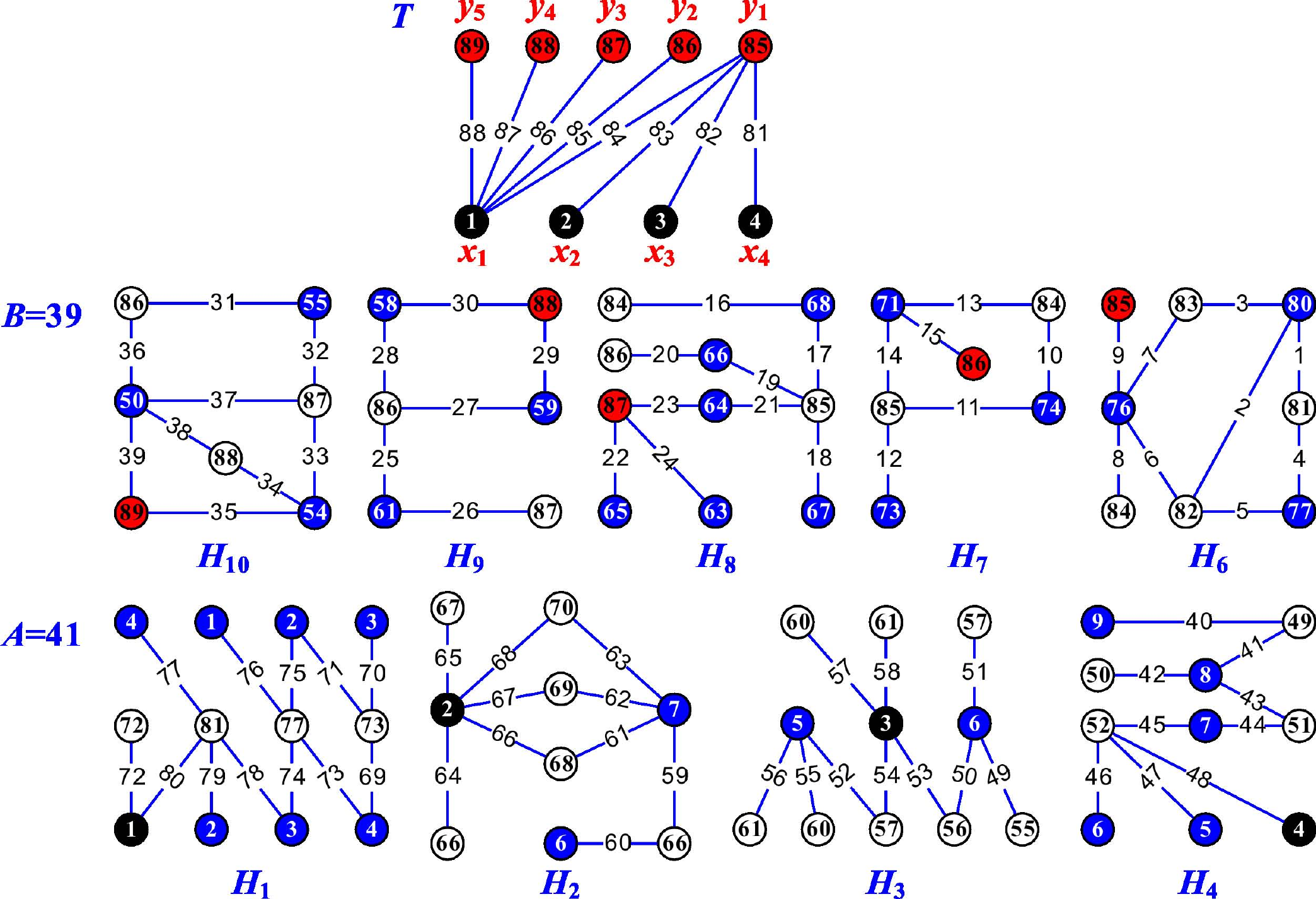}\\
\caption{\label{fig:construct-graceful-space-2} {\small A connected graph $G$ admits a proper gracefully total coloring as an example for Theorem \ref{thm:adding-graphs-to-trees}.}}
\end{figure}

\begin{figure}[h]
\centering
\includegraphics[width=16.2cm]{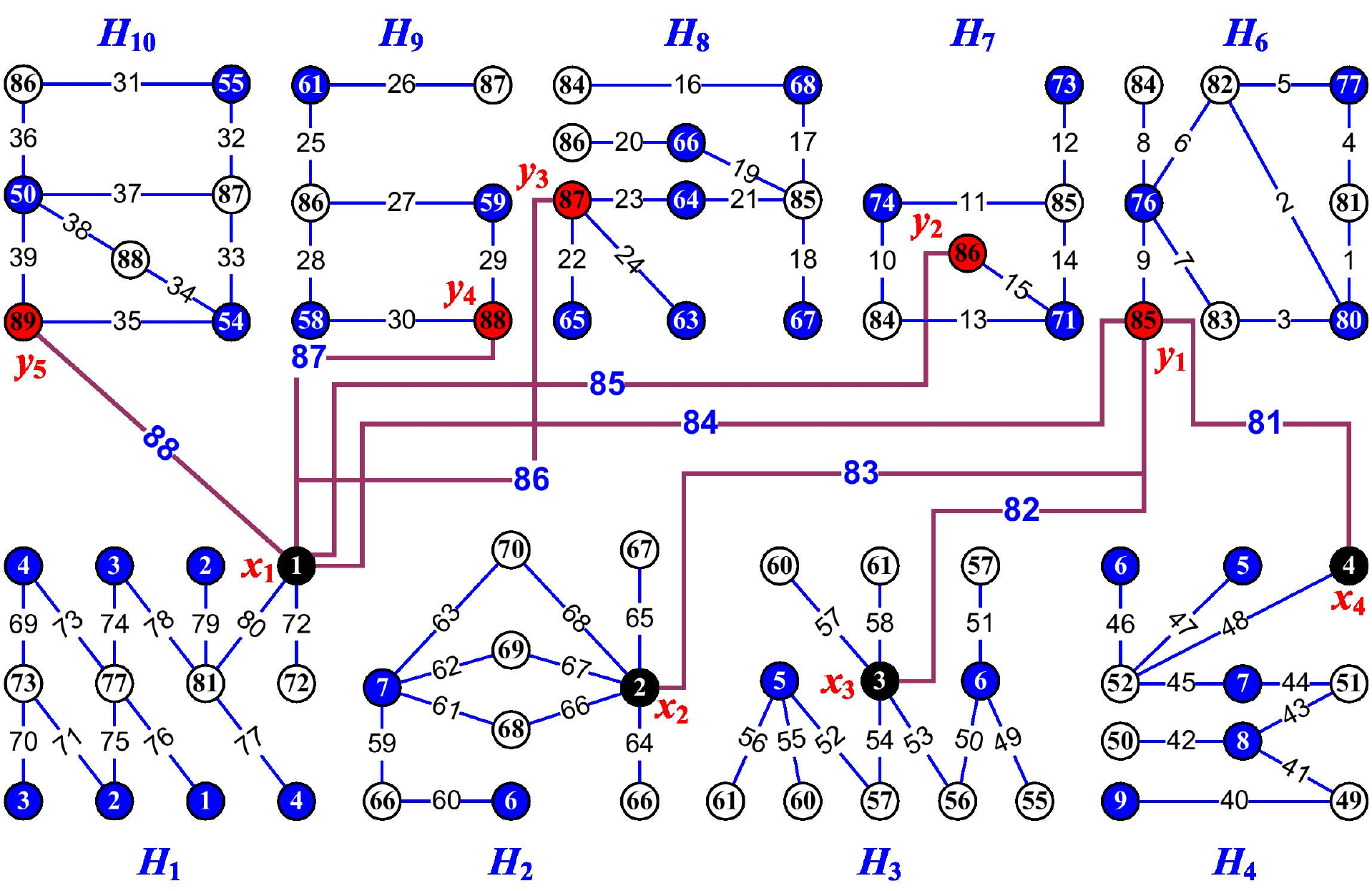}\\
\caption{\label{fig:construct-graceful-compound} {\small The compound of the connected graph $G$ shown in Fig.\ref{fig:construct-graceful-space-2}.}}
\end{figure}

\subsection{Constructing weak-gracefully graphic lattices}

\begin{defn} \label{defn:weak-graceful-total-colorings}
If a connected $(p,q)$-graph $G$ admits a total coloring $f:V(G)\cup E(G)\rightarrow [1,M]$, such that $f(uv)=|f(u)-f(v)|$ and $f(u)\neq f(v)$ with $uv\in E(G)$, and $f(E(G))=[1,q]$, as well as $f(V(G))\subseteq [1,q+1]$, we call $f$ a \emph{weak-gracefully total coloring}. It may happen $f(u)=2f(v)$ or $f(v)=2f(u)$ in a weak gracefully total coloring. Moreover, if $G$ is bipartite, and $\max f(X)<\min f(Y)$ for the bipartition $(X,Y)$ of vertex set of $G$, we call $f$ a \emph{set-ordered weak gracefully total coloring.}\qqed
\end{defn}

A base $\textbf{\textrm{H}}=(H_k)^n_{k=1}$ is made by $n$ disjoint connected graphs $H_1, H_2,\dots, H_n~(n\geq 2)$, we define particular kinds of graphs as follows:
\begin{asparaenum}[(a)]
\item If there exists an edge set $E^*$, such that a vertex $x_i$ of each $H_i$ is joined with a vertex $x_j$ of some $H_j$ by an edge $x_ix_j\in E^*$, the resultant graph is just connected, denoted as $E^*\oplus^n_{k=1}H_k$, called an \emph{edge-joined graph}.
\item A \emph{hand-in-hand graph} $G$ is made by coinciding a vertex $x_{k-1}$ of $H_{k-1}$ with a vertex $x_{k}$ of $H_{k}$ into one vertex $x_{k-1}\odot x_{k}$ for each $k\in [2,n]$, denoted $\textbf{\textrm{G}}=(H_{k-1}\odot H_k)^n_{k=2}$.
\item A \emph{single-series graph} $L$ is constructed by joining a vertex $x_{k-1}$ of $H_{k-1}$ with a vertex $x_{k}$ of $H_{k}$ by a new edge $x_{k-1}x_{k}$ for each $k\in [2,n]$, denoted $L=E^*\ominus ^n_{k=1}H_{k}$.
\item A connected bipartite graph $F$ has $n$ vertices $x_1,x_2,\dots x_n$, coinciding a vertex $u_i$ of each $H_i$ with the vertex $x_i$ of $F$ into one vertex $x_i\odot u_i$ for $i\in [1,n]$ produces a graph, called \emph{$F$-graph}, denoted as $F\odot ^n_{k=1}H_k$.
\end{asparaenum}

\begin{lem} \label{thm:graceful-total-coloring}
Given a base $\textbf{\textrm{H}}=(H_k)^n_{k=1}$ with $e_k=|E(H_k)|$ and $e_1\geq e_2\geq \cdots \geq e_n$, and each $H_k$ admits a set-ordered weak gracefully total coloring. Then there exists an edge set $E^*$, such that the edge-joined graph $E^*\oplus^n_{k=1}H_k$ admits a set-ordered weak gracefully total coloring too.
\end{lem}

\begin{thm} \label{thm:hand-in-hand-coinciding-lemma}
If each each $H_k$ of a base $\textbf{\textrm{H}}=(H_k)^n_{k=1}$ is a connected bipartite graph and admits a set-ordered weak gracefully total coloring, we have

(1) there exists a hand-in-hand graph $\textbf{\textrm{G}}=(H_{k-1}\odot H_k)^n_{k=2}$, such that $G$ is a connected bipartite graph admitting a set-ordered weak gracefully total coloring.

(2) there exists a single-series graph $H=E^*\ominus ^n_{k=1}H_{k}$, such that $H$ is a connected bipartite graph admitting a set-ordered weak gracefully total coloring.
\end{thm}

\begin{thm} \label{thm:coincide-graceful-total-coloring}
If each each $H_k$ of a base $\textbf{\textrm{H}}=(H_k)^n_{k=1}$ is a connected bipartite graph and admits a set-ordered weak gracefully total coloring, and another connected bipartite graph $F$ of $n$ vertices admits a set-ordered weak gracefully total coloring, then the $F$-graph $F\odot ^n_{k=1}H_k$ is a connected bipartite graph and admits a set-ordered weak gracefully total coloring.
\end{thm}

The examples for illustrating Lemma \ref{thm:graceful-total-coloring} are shown in Fig.\ref{fig:0-basic}, Fig.\ref{fig:hand-in-hand}, Fig.\ref{fig:pearl-string}, Fig.\ref{fig:hand-in-hand-1} and Fig.\ref{fig:add-edges}. Lemma \ref{thm:graceful-total-coloring} will produce a recursive connected graph $G_m$, where $G_1=E^*_{r,1}\oplus (H_1\cup H_2)$, $G_2=E^*_{r,2}\oplus (G_1\cup H_3)$ and $G_m=E^*_{r,m}\oplus (G_{m-1}\cup H_{m+1})$, $m\in [2,n-1]$. For all sets $\textbf{\textrm{E}}^*=\{(E^*_{r,k})^n_{k=1},r\in [1,M]\}$ of possible edges, we have a \emph{set-ordered weak gracefully total coloring graphic lattice} based on lattice base $\textbf{\textrm{H}}=(H_k)^n_{k=1}$
\begin{equation}\label{eqa:lattice-graph-join-operation}
\textbf{\textrm{L}}(\textbf{\textrm{E}}^*\oplus \textbf{\textrm{H}})=\left \{G_{n-1}=E^*_{r,n}\oplus (G_{n-1}\cup H_{n}), E^*_{r,n}\in \textbf{\textrm{E}}^*, H_{k}\in \textbf{\textrm{H}}\right \}
\end{equation} such that each graph of the lattice $\textbf{\textrm{L}}(\textbf{\textrm{E}}^*\oplus \textbf{\textrm{H}})$ defined in (\ref{eqa:lattice-graph-join-operation}) is a single-series graph and admits a set-ordered weak gracefully total coloring.

Let $(H_{i_1}, H_{i_2},\dots, H_{i_n})$ be a permutation of $(H_1, H_2,\dots, H_n)$, according to Theorem \ref{thm:hand-in-hand-coinciding-lemma}, each connected bipartite graph $(H_{i_{k-1}}\odot H_{i_{k}})^n_{k=2}$ admits a set-ordered weak gracefully total coloring, so let $P_{ermu}(\textbf{\textrm{H}})$ to be the set of all permutations of $(H_1, H_2,\dots, H_n)$, we call the following set
\begin{equation}\label{eqa:hand-in-hand-lattices-1}
\textbf{\textrm{L}}(\bowtie P_{ermu}(\textbf{\textrm{H}}))=\{(H_{i_{k-1}}\bowtie H_{i_{k}})^n_{k=2},~(H_{i_1}, H_{i_2},\dots, H_{i_n})\in P_{ermu}(\textbf{\textrm{H}})\}
\end{equation} a \emph{set-ordered weak gracefully total coloring hand-in-hand graphic lattice}. For each $k\in [1,n]$, we select $c_k$ connected bipartite graphs $H_{k}$ from the base $\textbf{\textrm{H}}=(H_k)^n_{k=1}$, and construct hand-in-hand graph $\bowtie^n_{k=1}c_kH_{k}$ by Theorem \ref{thm:hand-in-hand-coinciding-lemma}, we have a set-ordered weak gracefully total coloring hand-in-hand graphic lattice
\begin{equation}\label{eqa:hand-in-hand-lattices-2}
\textbf{\textrm{L}}(\bowtie \textbf{\textrm{H}})=\{\bowtie^n_{k=1}c_kH_{k}, c_k\in Z^0,H_{k}\in \textbf{\textrm{H}}\}.
\end{equation} We have a lattice
\begin{equation}\label{eqa:bigger-hand-in-hand}
\textbf{\textrm{L}}(\bowtie \{\textbf{\textrm{H}}\})=\textbf{\textrm{L}}(\bowtie P_{ermu}(\textbf{\textrm{H}}))\cup \textbf{\textrm{L}}(\bowtie \textbf{\textrm{H}}),
\end{equation} based on $\textbf{\textrm{H}}=(H_k)^n_{k=1}$, $\textbf{\textrm{L}}(\bowtie P_{ermu}(\textbf{\textrm{H}}))$ defined in (\ref{eqa:hand-in-hand-lattices-1}) and $\textbf{\textrm{L}}(\bowtie \textbf{\textrm{H}})$ defined in (\ref{eqa:hand-in-hand-lattices-2}). Moreover, Theorem \ref{thm:coincide-graceful-total-coloring} induces the following $F$-lattice:
\begin{equation}\label{eqa:lattice-graph-coincide-operation}
\textbf{\textrm{L}}(\textbf{\textrm{F}}^*\odot \textbf{\textrm{H}})=\left \{F\odot ^n_{k=1}H_k, H_{k}\in \textbf{\textrm{H}}, F\in \textbf{\textrm{F}}^*\right \}
\end{equation}

\begin{figure}[h]
\centering
\includegraphics[width=15cm]{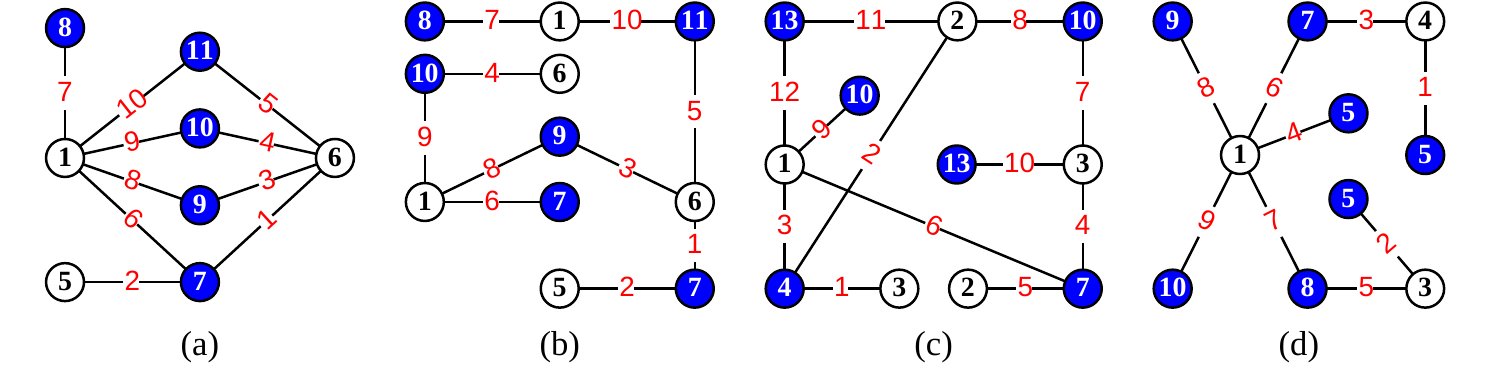}\\
{\small \caption{\label{fig:0-basic} Four graphs admitting set-ordered weak gracefully total colorings.}}
\end{figure}

\begin{figure}[h]
\centering
\includegraphics[width=15cm]{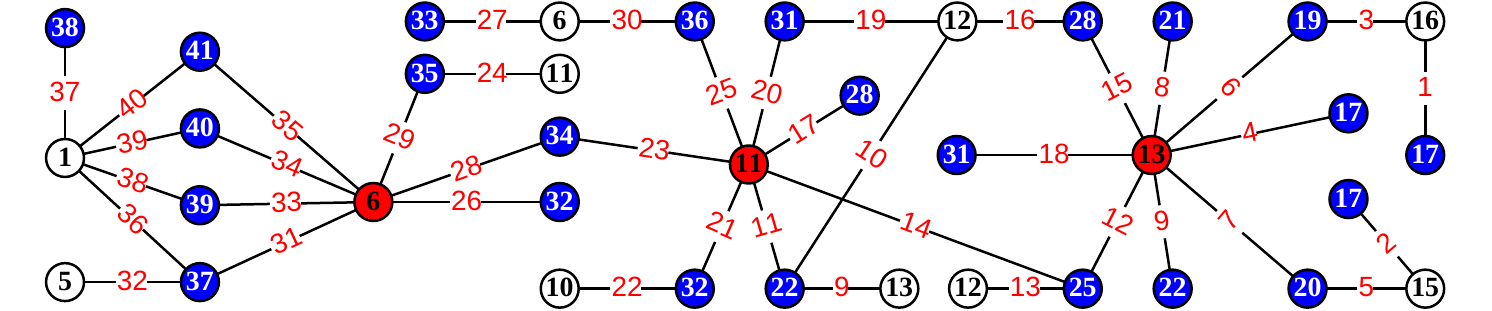}\\
{\small \caption{\label{fig:hand-in-hand} A hand-in-hand graph made by four graphs shown in Fig.\ref{fig:0-basic} admits a set-ordered weak gracefully total coloring.}}
\end{figure}

\begin{figure}[h]
\centering
\includegraphics[width=15cm]{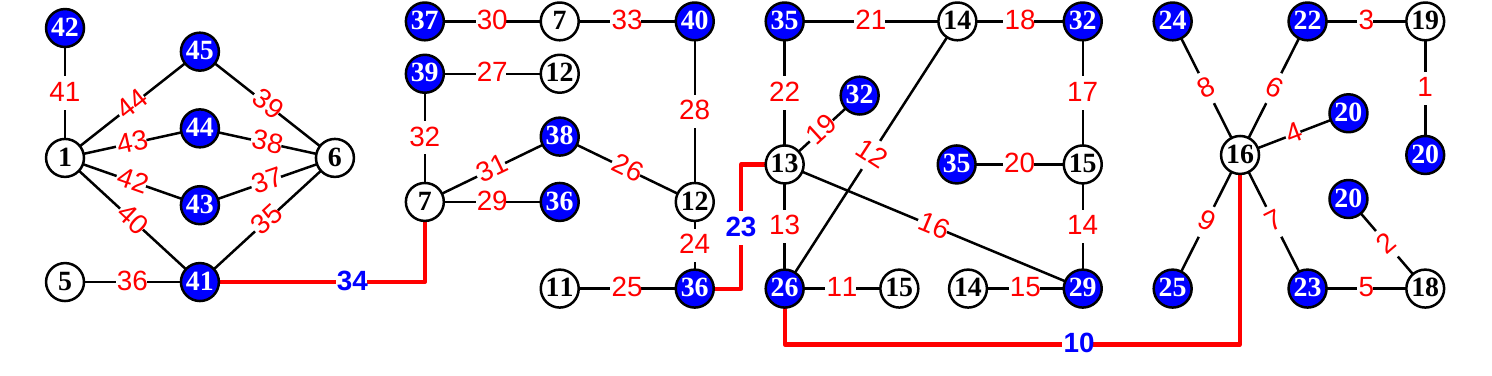}\\
{\small \caption{\label{fig:pearl-string} A single-series graph made by four graphs shown in Fig.\ref{fig:0-basic} admits a set-ordered weak gracefully total coloring.}}
\end{figure}

\begin{figure}[h]
\centering
\includegraphics[width=15cm]{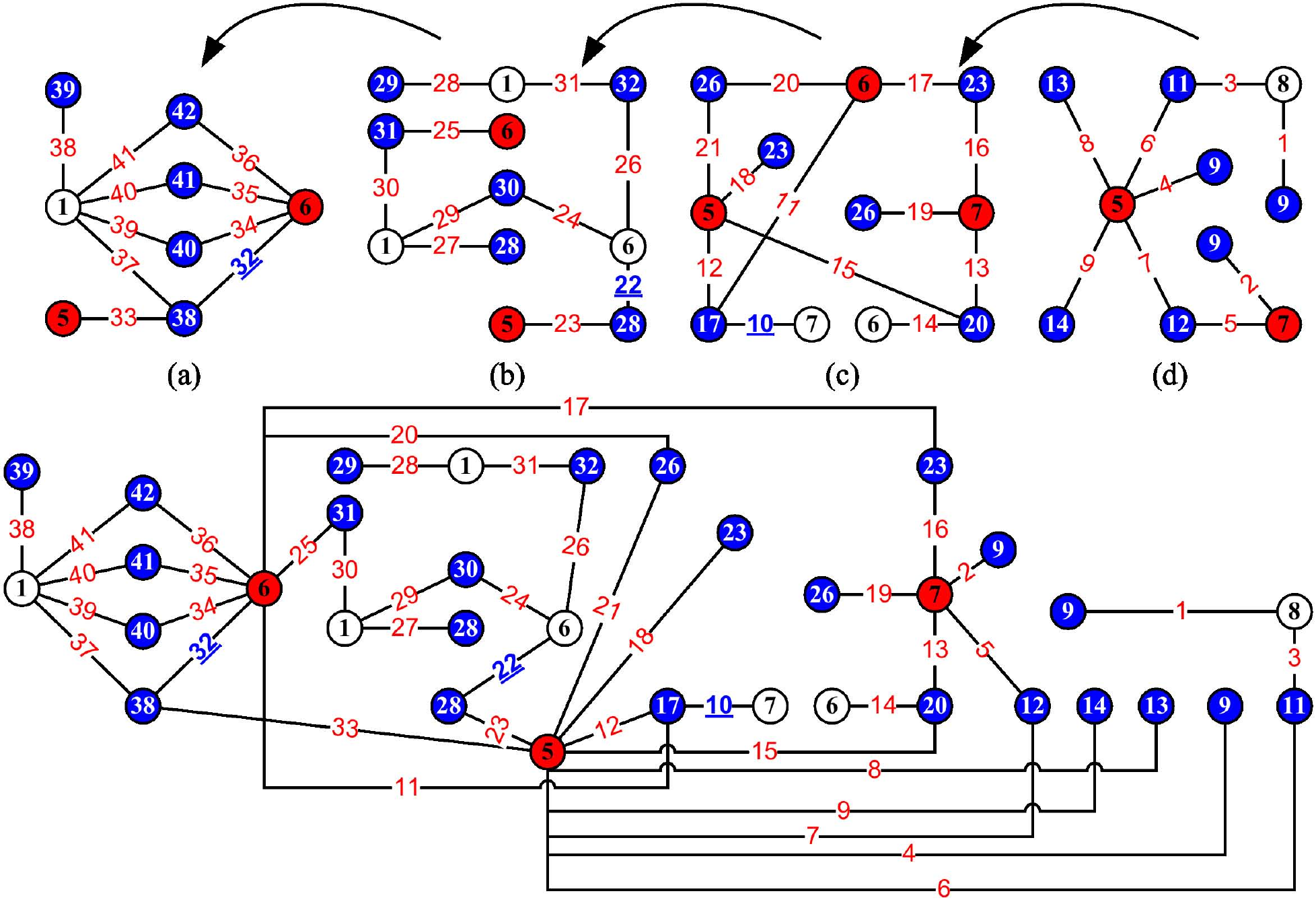}\\
{\small \caption{\label{fig:hand-in-hand-1} Another hand-in-hand graph made by four graphs shown in Fig.\ref{fig:0-basic} admits a set-ordered weak gracefully total coloring.}}
\end{figure}

\begin{figure}[h]
\centering
\includegraphics[width=15cm]{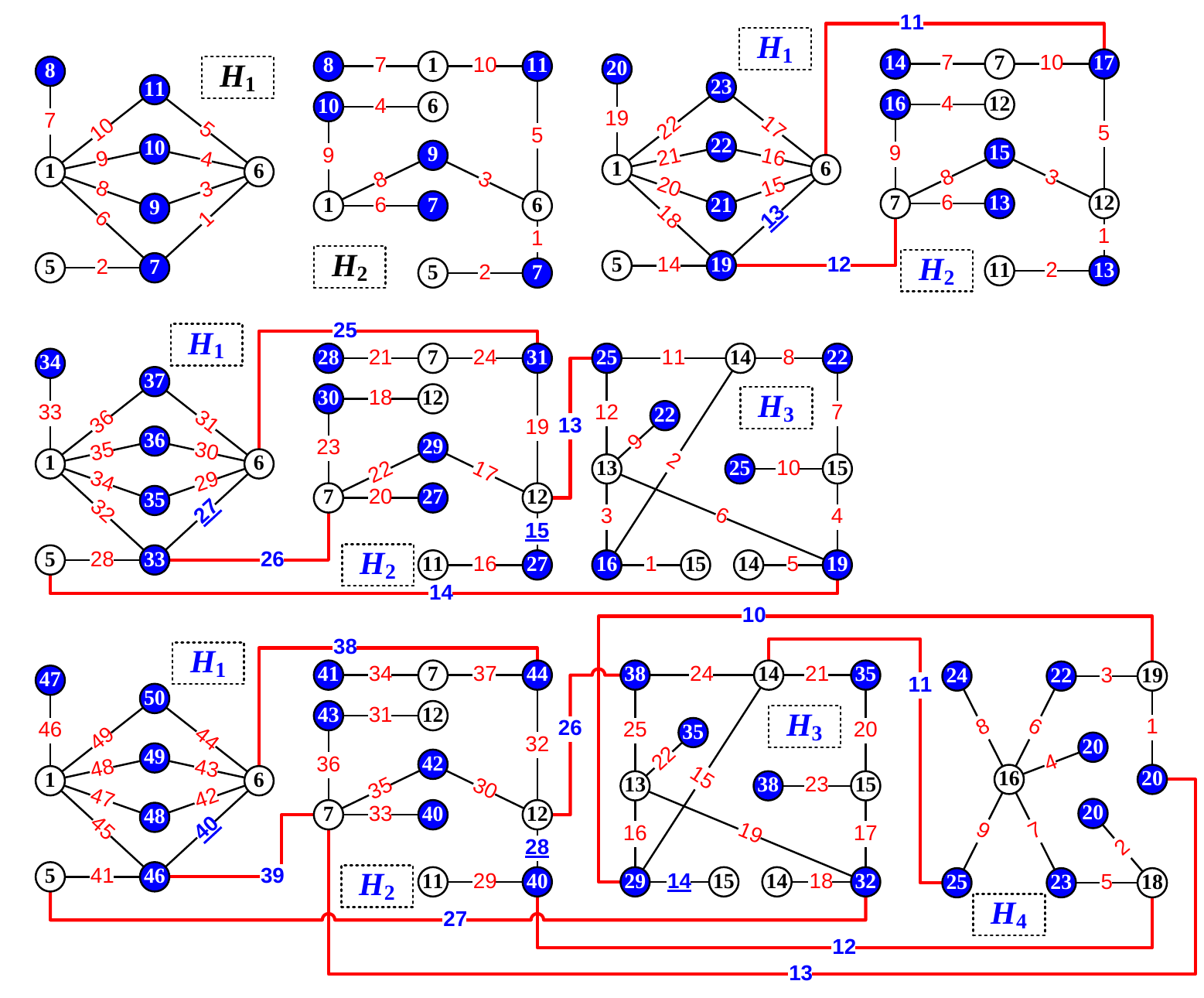}\\
{\small \caption{\label{fig:add-edges} A scheme for an idea of proving Lemma \ref{thm:graceful-total-coloring}.}}
\end{figure}

\begin{problem}\label{qeu:weak-gracefully-total-coloring}
Let $v_{\textrm{weakgtc}}(G)=\min_f\{|f(V(G))|\}$ over all weak gracefully total colorings of a connected graph $G$. Then any tree $T$ with diameter at least three golds $v_{\textrm{weakgtc}}(T)\geq \frac{1}{2}|V(T)|$.
\end{problem}

%\section{Graphic group lattices}
%%\input{5-section/group-lattices-5}

\section{Graphic group lattices}

Academician Xiaogang Wen of the United States, in his article ``New revolution in physics modern mathematics in condensed matter physics'', pointed out: ``But since the quantum revolution, especially since the second quantum revolution, we are more and more aware that our world is not continuous, but discrete. We should look at the world from the perspective of algebra.'' And moreover, the development of modern mathematics is exactly from continuous to discrete, from analysis to algebra. It also puts forward that \emph{discrete algebra} is more essential than \emph{continuous analysis}.

In \cite{Yao-Sun-Zhao-Li-Yan-2017}, \cite{Yao-Mu-Sun-Zhang-Wang-Su-Ma-IAEAC-2018} and \cite{Sun-Zhang-Zhao-Yao-2017} the authors proposed new-type groups (like Abelian additive finite groups), called every-zero graphic groups made by Topsnut-gpws of topological coding.

\subsection{Graphic groups}

We introduce a mixed every-zero graphic groups by the following algorithm:

\textbf{MIXED Graphic-group Algorithm}. Let $f: V(G)\cup E(G)\rightarrow [1,M]$ be a $W$-type proper total coloring of a graph $G$ such that two color sets $f(V(G))=\{f(x):x\in V(G)\}$ and $f(E(G))=\{f(uv):uv\in E(G)\}$ hold a collection of restrictions. We define a $W$-type proper total coloring $g_{s,k}$ by setting $g_{s,k}(x)=f(x)+s~(\bmod~p)$ for every vertex $x\in V(G)$, and $g_{s,k}(uv)=f(uv)+k~(\bmod~q)$ for each edge $uv\in E(G)$. Let $F_f(G)$ be the set of graphs $G_{s,k}$ admitting $W$-type proper total colorings $g_{s,k}$ defined above, and each graph $G_{s,k}\cong G$ in topological structure. We define an additive operation ``$\oplus$'' on the elements of $F_f(G)$ in the following way: Take arbitrarily an element $G_{a,b}\in F_f(G)$ as the \emph{zero}, and $G_{s,k}\oplus G_{i,j}$ is defined by the following computation
\begin{equation}\label{eqa:mixed-graphic-group}
[g_{s,k}(w)+g_{i,j}(w)-g_{a,b}(w)]~(\bmod~\varepsilon)=g_{\lambda,\mu}(w)
\end{equation}
for each element $w\in V(G)\cup E(G)$, where $\lambda=s+i-a~(\bmod~p)$ and $\mu=k+j-b~(\bmod~q)$. As $w=x\in V(G)$, the form (\ref{eqa:mixed-graphic-group}) is just equal to
\begin{equation}\label{eqa:mixed-graphic-group11}
[g_{s,k}(x)+g_{i,j}(x)-g_{a,b}(x)]~(\bmod~p)=g_{\lambda,\mu}(x)
 \end{equation}and as $w=uv\in E(G)$, the form (\ref{eqa:mixed-graphic-group}) is defined as follows:
 \begin{equation}\label{eqa:mixed-graphic-group22}
[g_{s,k}(uv)+g_{i,j}(uv)-g_{a,b}(uv)]~(\bmod~q)=g_{\lambda,\mu}(uv).
\end{equation}
Especially, as $s=i=a=\alpha$, we have $\bmod~\varepsilon=\bmod~q$ in (\ref{eqa:mixed-graphic-group}), and
\begin{equation}\label{eqa:mixed-graphic-group-edge}
[g_{\alpha,k}(uv)+g_{\alpha,j}(uv)-g_{\alpha,b}(uv)]~(\bmod~q)=g_{\alpha,\mu}(uv)
\end{equation} for $uv\in E(G)$; and when $k=j=b=\beta$, so $\bmod~\varepsilon=\bmod~p$ in (\ref{eqa:mixed-graphic-group}), we have
\begin{equation}\label{eqa:mixed-graphic-group-vertex}
[g_{s,\beta}(x)+g_{i,\beta}(x)-g_{a,\beta}(x)]~(\bmod~p)=g_{\lambda,\beta}(x),~x\in V(G).
\end{equation}

We have the following facts on the graph set $F_f(G)$:
\begin{asparaenum}[(1) ]
\item \emph{Zero.} Each graph $G_{a,b}\in F_f(G)$ can be determined as the \emph{zero} such that $G_{s,k}\oplus G_{a,b}=G_{s,k}$.

\item \emph{Uniqueness.} For $G_{s,k}\oplus G_{i,j}=G_{c,d}\in F_f(G)$ and $G_{s,k}\oplus G_{i,j}=G_{r,t}\in F_f(G)$, then $c=s+i-a~(\bmod~p)=r$ and $c=k+j-b~(\bmod~q)=t$ under the zero $G_{a,b}$.

\item \emph{Inverse.} Each graph $G_{s,k}\in F_f(G)$ has its own \emph{inverse} $G_{s',k'}\in F_f(G)$ such that $G_{s,k}\oplus G_{s',k'}=G_{a,b}$ determined by $[g_{s,k}(w)+g_{i,j}(w)]~(\bmod~\varepsilon)=2g_{a,b}(w)$ for each element $w\in V(G)\cup E(G)$.
\item \emph{Associative law.} Each triple $G_{s,k},G_{i,j},G_{c,d}\in F_f(G)$ holds $G_{s,k}\oplus [G_{i,j}\oplus G_{c,d}]=[G_{s,k}\oplus G_{i,j}]\oplus G_{c,d}$.
\item \emph{Commutative law.} $G_{s,k}\oplus G_{i,j}=G_{i,j}\oplus G_{s,k}$.
\end{asparaenum}

Thereby, we call $F_f(G)=\{G_{s,k}:s\in [0,p],k\in [0,q]\}$ an \emph{every-zero mixed graphic group} based on the additive operation ``$\oplus$'' defined in (\ref{eqa:mixed-graphic-group}), and write this group by $\textbf{\textrm{G}}=\{F_f(G);\oplus\}$. $\textbf{\textrm{G}}$ contains $pq$ graphs. There two particular \emph{every-zero graphic subgroups} $\{F_{v}(G);\oplus\}\subset \textbf{\textrm{G}}$ and $\{F_{e}(G);\oplus\}\subset \textbf{\textrm{G}}$, where $F_{v}(G)=\{G_{s,0}:s\in [0,p]\}$ and $F_{e}(G)=\{G_{0,k}:k\in [0,q]\}$. In fact, $\textbf{\textrm{G}}$ contains at least $(p+q)$ every-zero graphic subgroups.

For two every-zero graphic groups $\{F_f(G);\oplus\}$ and $\{F_h(H);\oplus\}$, suppose that there are graph homomorphisms $G_i\rightarrow H_i$ defined by $\theta_i:V(G_i)\rightarrow V(H_i)$ with $i\in [1,m]$. We define $\theta=\bigcup^m_{i=1}\theta_i$, and have an \emph{every-zero graphic group homomorphism} $\{F_f(G);\oplus\}\rightarrow \{F_h(H);\oplus\}$ from a set $F_f(G)$ to another set $F_h(H)$ (Ref. \cite{Bing-Yao-Hongyu-Wang-graph-homomorphisms-2020}).

\subsection{Graphic group lattices}

\textbf{A technique for encrypting networks.} Using an every-zero mixed graphic group $\textbf{\textrm{G}}$ to encrypt a connected graph $H$ by a mapping $\varphi:V(H)\cup E(H)\rightarrow \textbf{\textrm{G}}$ such that $\varphi(uv)=\varphi(u)\oplus _k \varphi(v)$ under a zero $G_{a,b}\in \textbf{\textrm{G}}$. So, we get another graph obtained by the set $\{\varphi(x)$, $\varphi(uv):x\in V(H)$, $ uv\in E(H)\}$, and join some vertices of the graph $G_u=\varphi(u)\in \textbf{\textrm{G}}$ with some vertices of the graph $G_{uv}=\varphi(uv)\in \textbf{\textrm{G}}$ by edges, and join some vertices of the graph $G_v=\varphi(v)\in \textbf{\textrm{G}}$ with some vertices of the graph $G_{uv}=\varphi(uv)\in \textbf{\textrm{G}}$ by edges, we write this graph as $H\triangleleft_{a,b} |^{(p,q)}_{(s,k)}a_{s,k}G_{s,k}$, where $G_{s,k}\in \textbf{\textrm{G}}$, each $a_{s,k}\in Z^0$ and $\sum^{(p,q)}_{(s,k)}a_{s,k}\geq 1$.

In general, $H\triangleleft_{a,b} |^{(p,q)}_{(s,k)}a_{s,k}G_{s,k}\not \cong H\triangleleft_{c,d} |^{(p,q)}_{(s,k)}a_{s,k}G_{s,k}$ for two different zeros $G_{a,b},G_{c,d}\in \textbf{\textrm{G}}$. Moreover, there are many different ways to join the graph $G_{uv}=\varphi(uv)$ with the graph $G_u=\varphi(u)$ (resp. $G_v=\varphi(v)$) together by edges. We call
\begin{equation}\label{eqa:graphic-group-lattice}
{
\begin{split}
\textbf{\textrm{L}}(F_{m,n}\triangleleft_{a,b} \textbf{\textrm{G}})=\left \{H\triangleleft_{a,b} |^{(p,q)}_{(s,k)}a_{s,k}G_{s,k}: a_{s,k}\in Z^0, G_{s,k}\in \textbf{\textrm{G}}, H\in F_{m,n}\right \}
\end{split}}
\end{equation} a \emph{graphic group lattice} based on a zero $G_{a,b}\in \textbf{\textrm{G}}$, where $\sum^{(p,q)}_{(s,k)}a_{s,k}\geq 1$, and $F_{m,n}$ is a set of graphs of vertex number $\leq m$ and edge number $\leq n$. Moreover, we call \begin{equation}\label{eqa:c3xxxxx}
\textbf{\textrm{L}}(F_{m,n}\triangleleft\textbf{\textrm{G}})=\bigcup_{G_{a,b}\in \textbf{\textrm{G}}}\textbf{\textrm{L}}(F_{m,n}\triangleleft_{a,b} \textbf{\textrm{G}}).
\end{equation}
a \emph{graphic $\textbf{\textrm{G}}$-group lattice}, since each element of the every-zero mixed graphic group $\textbf{\textrm{G}}$ can refereed as the zero of the additive operation ``$\oplus$''.

Since two graphs $G_{a,b},G_{c,d}\in \textbf{\textrm{G}}$ form two homomorphically equivalent graph homomorphisms $G_{a,b}\leftrightarrow G_{c,d}$, then we have a \emph{graphic group lattice homomorphism} $\textbf{\textrm{L}}(F_{m,n}\triangleleft_{a,b} \textbf{\textrm{G}})\leftrightarrow \textbf{\textrm{L}}(F_{m,n}\triangleleft_{c,d} \textbf{\textrm{G}})$ from a set to another set.

\subsection{Networks encrypted by graphic lattices}

We encrypt a network/graph $H\in F_{m,n}$ by constructing another graph $H\triangleleft_{a,b} |^{(p,q)}_{(s,k)}a_{s,k}G_{s,k}$ with $a_{s,k}\in Z^0$ and $G_{s,k}\in \textbf{\textrm{G}}$ defined in (\ref{eqa:graphic-group-lattice}), and we call this way to be one of \emph{graphic group colorings}.

\begin{thm} \label{thm:graphic-group-in-graphic-group}
Each graph $W$ of a graphic $\textbf{\textrm{G}}$-group lattice $\textbf{\textrm{L}}(F_{m,n}\triangleleft \textbf{\textrm{G}})$ forms an every-zero graphic group $\{F_h(W); \oplus\}$ too.
\end{thm}

\begin{thm} \label{thm:graphic-group-encrypting-every}
Any simple graph $H$ admits a graphic group coloring, i.e., there is a graphic group $\{F_f(G);\oplus\}$ such that $H$ is encrypted as $H\triangleleft_{a,b} |^{(p,q)}_{(s,k)}a_{s,k}G_{s,k}$ with $a_{s,k}\in Z^0$, $G_{s,k}\in \{F_f(G);\oplus\}$ and $\sum^{(p,q)}_{(s,k)}a_{s,k}\geq 1$.
\end{thm}

\subsection{Tree-like networks encrypted by graphic groups}

We apply \emph{graphic group colorings} to encrypt tree-like networks under restrictions as follows:

\begin{thm} \label{thm:graphic-group-total-coloring11}
Any tree $T$ with maximum degree $\Delta$ admits a \emph{graphic group total coloring} $\theta:V(T)\cup E(T)\rightarrow \{F_f(G);\oplus\}$ by any specified zero $G_k\in \{F_f(G);\oplus\}$ and $|\{F_f(G);\oplus\}|\geq \Delta$, such that $\theta(uv)\neq \theta(uw)$ for any pair of adjacent edges $uv,uw$ of $T$.
\end{thm}

\begin{thm} \label{thm:graphic-group-total-coloring22}
Any tree $T$ admits a \emph{graphic group total coloring} $\theta$ based on $\{F_f(G);\oplus\}$, by any specified zero $G_k\in \{F_f(G);\oplus\}$, such that $\theta(E(T))=[1,|T|-1]$.
\end{thm}

\begin{thm} \label{thm:graphic-group-total-coloring33}
The edges of any tree $T$ can be colored arbitrarily by $\varphi: E(T)\rightarrow \{F_f(G);\oplus\}$ under any specified zero $G_k\in \{F_f(G);\oplus\}$, and then $\varphi$ can be expended to $V(T)$, such that $\varphi(uv)=\varphi(u)\oplus \varphi(v)$ for each edge $uv\in E(T)$.
\end{thm}

\begin{thm}\label{thm:ve-graceful-evaluated-coloring}
\cite{Yao-Zhao-Zhang-Mu-Sun-Zhang-Yang-Ma-Su-Wang-Wang-Sun-arXiv2019} Let $\{F_p(G),\oplus\}$ be an every-zero $\varepsilon$-group, where $F_p(G)=\{G_1,G_2,\dots, G_p\}$, and $\varepsilon$-group is one of every-zero Topcode-matrix groups (Topcode-groups), every-zero number string groups, every-zero Topsnut-gpw groups and Hanzi-groups. If a tree $T$ of $p$ vertices admits a set-ordered graceful labelling, then $T$ admits a graceful graphic group labelling based on $\{F_p(G),\oplus\}$.
\end{thm}

There is an every-zero graphic group $Y_{ear}=\{F_{f}(G^*);\oplus\}$ shown in Fig.\ref{fig:nian-graceful-group}, where $F_{f}(G^*)=\{G_i:~i\in [1,14]\}$ with $G_i$ admitting a labelling $f_i(x)=f(x)+i-1~(\bmod~14)$ and $G_i\cong G^*$ for $i\in [1,14]$ and, each edge $uv\in E(G_i)$ is labelled by $f_i(uv)=|f_i(u)-f_i(v)|$. See three tree-like networks admitting graceful graphic group labellings shown in Fig.\ref{fig:network-by-nian-graceful-group}.

\begin{figure}[h]
\centering
\includegraphics[width=16.2cm]{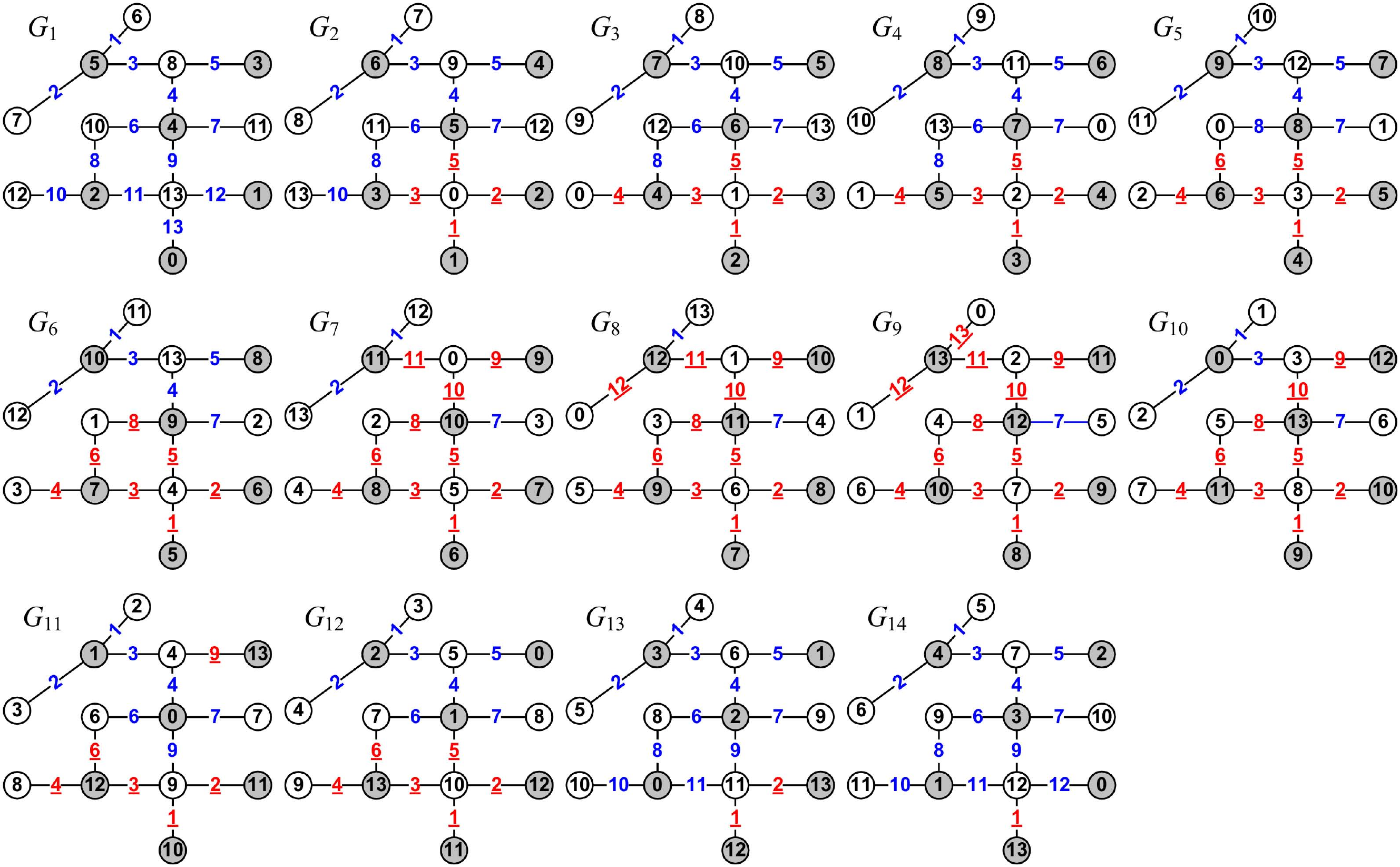}
\caption{\label{fig:nian-graceful-group}{\small The every-zero graphic group $Y_{ear}=\{F_{f}(G^*);\oplus\}$.}}
\end{figure}

\begin{figure}[h]
\centering
\includegraphics[width=16.2cm]{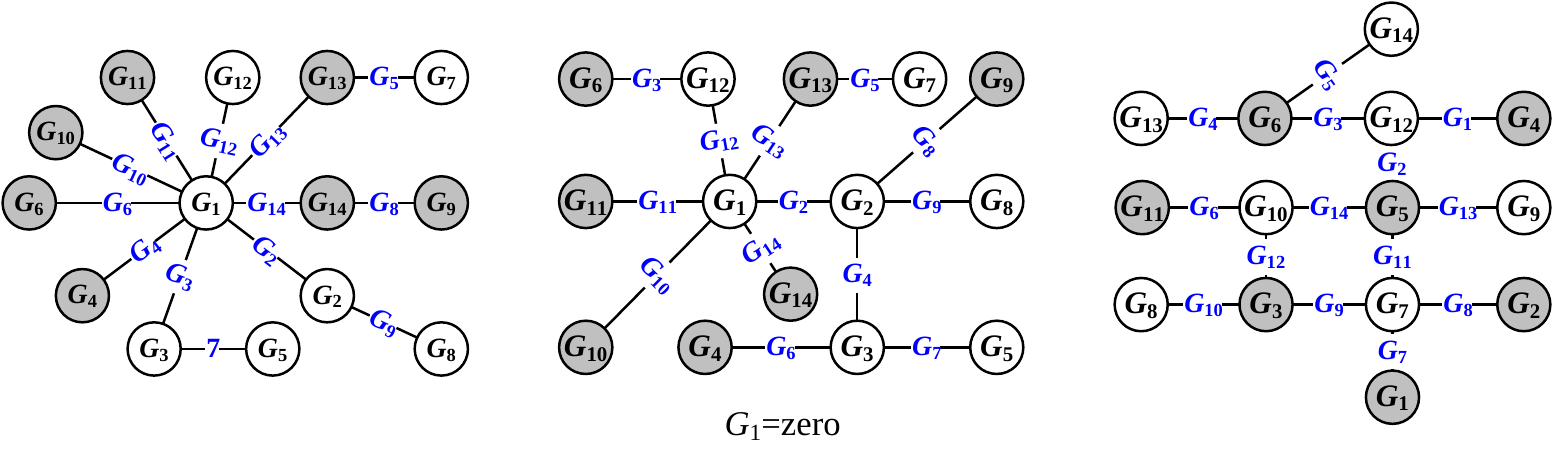}
\caption{\label{fig:network-by-nian-graceful-group}{\small Three tree-like networks are encrypted by the every-zero graphic group $Y_{ear}=\{F_{f}(G^*);\oplus\}$, and they both admit graceful graphic group labellings.}}
\end{figure}

\begin{problem}\label{qeu:graphic-groups}
We want the solutions, or part solutions for the following questions:
\begin{asparaenum}[\textrm{Ggco}-1. ]
\item \textbf{Does} any lobster admit an \emph{odd-graceful graphic group labelling} $\varphi$ by any specified zero $G_k\in \{F_f(G);\oplus\}$, such that $\varphi(u)\neq \varphi(v)$ for distinct vertices $u,v\in V(T)$, and the edge index set $\{k:~\varphi(xy)=G_k,xy\in E(T)\}=\{1,3,5$, $\dots $, $2|V(T)|-3\}$?
\item If a bipartite graph $T$ admits a set-ordered graceful labelling (\cite{Yao-Cheng-Yao-Zhao-2009, Zhou-Yao-Chen-Tao2012, Yao-Liu-Yao-2017}), \textbf{does} $T$ admit a \emph{graceful graphic group labelling} $\theta$ by any specified zero $G_k\in \{F_f(G);\oplus\}$, such that $\theta(u)\neq \theta(v)$ for distinct vertices $u,v\in V(T)$, the \emph{edge index set} $\{j:~\theta(xy)=G_j,xy\in E(T)\}=[1,|E(T)|]$?
\item \textbf{Find} \emph{$W$-type graphic group labellings}, such as, $W$-type is edge-magic total, elegant, felicitous, and so on.
\item Motivated from Graceful Tree Conjecture (Alexander Rosa, 1966), we \textbf{guess}: Every tree admits a graceful graphic group labelling.
\item \textbf{Find} \emph{$W$-type graphic group colorings} $\varphi$ with $\varphi(x)=\varphi(y)$ for some distinct vertices $x,y$ (resp. $\varphi(uv)=\varphi(wz)$ for some non-adjacent edges $uv,wz$), and the number of vertex pairs $(x,y)$ (resp. the number of edge pairs $(uv,wz)$~) is as less as possible.
\item \textbf{Consider} various \emph{graphic group proper total colorings} for the famous Total Coloring Conjecture (Behzad, 1965; Vadim G. Vizing, 1964).\qqed
\end{asparaenum}
\end{problem}

\subsection{Graphic groups as linearly independent colored graphic bases}

Graphic groups have been investigated in \cite{Yao-Sun-Zhao-Li-Yan-2017, Yao-Mu-Sun-Zhang-Wang-Su-Ma-IAEAC-2018, Sun-Zhang-Zhao-Yao-2017}. A set $F(\{G_i\}^n_{i=1};\oplus)$ of colored graphs $\{G_i\}^n_{i=1}$ admitting $W$-type colorings under the additive operation ``$\oplus$'' is called an \emph{every-zero graphic group} (also \emph{Abelian additive group}) based on the $W$-type coloring if:

(i) Every graph $G_k$ of $F(\{G_i\}^n_{i=1};\oplus)$ is as the ``zero'' such that $G_j\oplus G_k=G_j$ for any graph $G_j$ of $F(\{G_i\}^n_{i=1};\oplus)$;

(ii) for each zero $G_k$, $G_i\oplus G_j=G_s\in F(\{G_i\}^n_{i=1};\oplus)$;

(iii) $G_i\oplus (G_j\oplus G_s)=(G_i\oplus G_j)\oplus G_s$;

(iv) $G_j\oplus G_s=G_s\oplus G_j$.

Let $\textbf{\textrm{G}}^c=(G_1,G_2,\dots ,G_n)$ with $G_i\in F(\{G_i\}^n_{i=1};\oplus)$ be an linearly independent colored graphic base. We get a colored graphic lattice $\textbf{\textrm{L}}(\textbf{\textrm{G}}^c(\bullet)F^c_{p,q})$ under a graph operation ``$(\bullet)$''.

\begin{defn}\label{defn:odd-even-separable-6C-labelling}
\cite{Yao-Sun-Zhang-Mu-Sun-Wang-Su-Zhang-Yang-Yang-2018arXiv} A total labelling $f:V(G)\cup E(G)\rightarrow [1,p+q]$ for a bipartite $(p,q)$-graph $G$ is a bijection and holds:

(i) (e-magic) $f(uv)+|f(u)-f(v)|=k$;

(ii) (ee-difference) each edge $uv$ matches with another edge $xy$ holding $f(uv)=|f(x)-f(y)|$ (or $f(uv)=2(p+q)-|f(x)-f(y)|$);

(iii) (ee-balanced) let $s(uv)=|f(u)-f(v)|-f(uv)$ for $uv\in E(G)$, then there exists a constant $k'$ such that each edge $uv$ matches with another edge $u'v'$ holding $s(uv)+s(u'v')=k'$ (or $2(p+q)+s(uv)+s(u'v')=k'$) true;

(iv) (EV-ordered) $\min f(V(G))>\max f(E(G))$ (or $\max f(V(G))<\min f(E(G))$, or $f(V(G))\subseteq f(E(G))$, or $f(E(G))\subseteq f(V(G))$, or $f(V(G))$ is an odd-set and $f(E(G))$ is an even-set);

(v) (ve-matching) there exists a constant $k''$ such that each edge $uv$ matches with one vertex $w$ such that $f(uv)+f(w)=k''$, and each vertex $z$ matches with one edge $xy$ such that $f(z)+f(xy)=k''$, except the \emph{singularity} $f(x_0)=\lfloor \frac{p+q+1}{2}\rfloor $;

(vi) (set-ordered) $\max f(X)<\min f(Y)$ (or $\min f(X)>\max f(Y)$) for the bipartition $(X,Y)$ of $V(G)$.

(vii) (odd-even separable) $f(V(G))$ is an odd-set containing only odd numbers, as well as $f(E(G))$ is an even-set containing only even numbers.

We refer to $f$ as an \emph{odd-even separable 6C-labelling}.\qqed
\end{defn}

\begin{exa}\label{exa:group-6C-labelling}
A tree $G_1$ shown in Fig.\ref{fig:D6C-group} is a $(13,12)$-graph admitting an \emph{odd-even separable 6C-labelling} $f$ defined in Definition \ref{defn:odd-even-separable-6C-labelling}.

\begin{figure}[h]
\centering
\includegraphics[width=15.2cm]{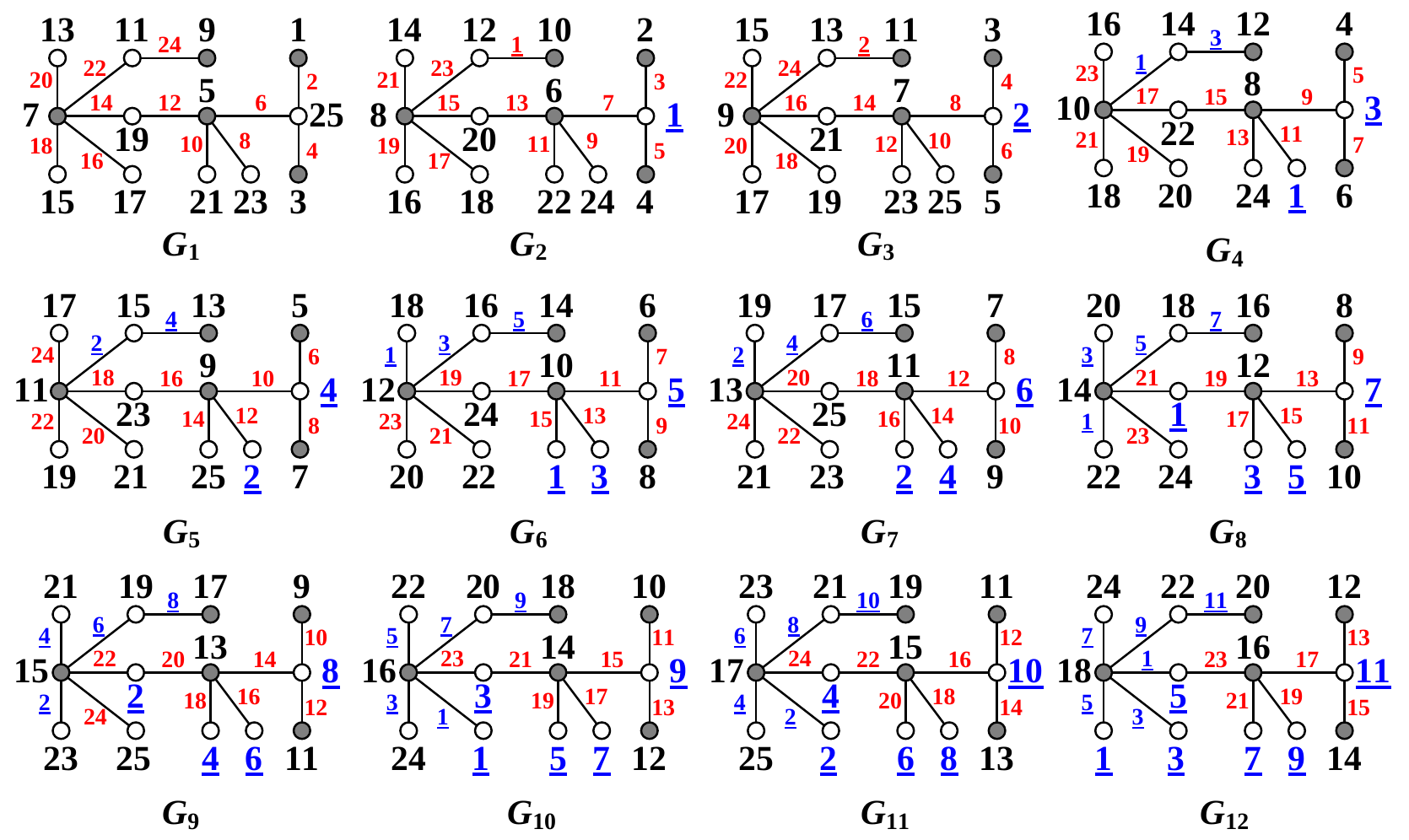}
\caption{\label{fig:D6C-group} {\small First part of an every-zero graphic group based on the 6C-labelling cited from \cite{Yao-Mu-Sun-Sun-Zhang-Wang-Su-Zhang-Yang-Zhao-Wang-Ma-Yao-Yang-Xie2019}.}}
\end{figure}

\begin{figure}[h]
\centering
\includegraphics[width=15.2cm]{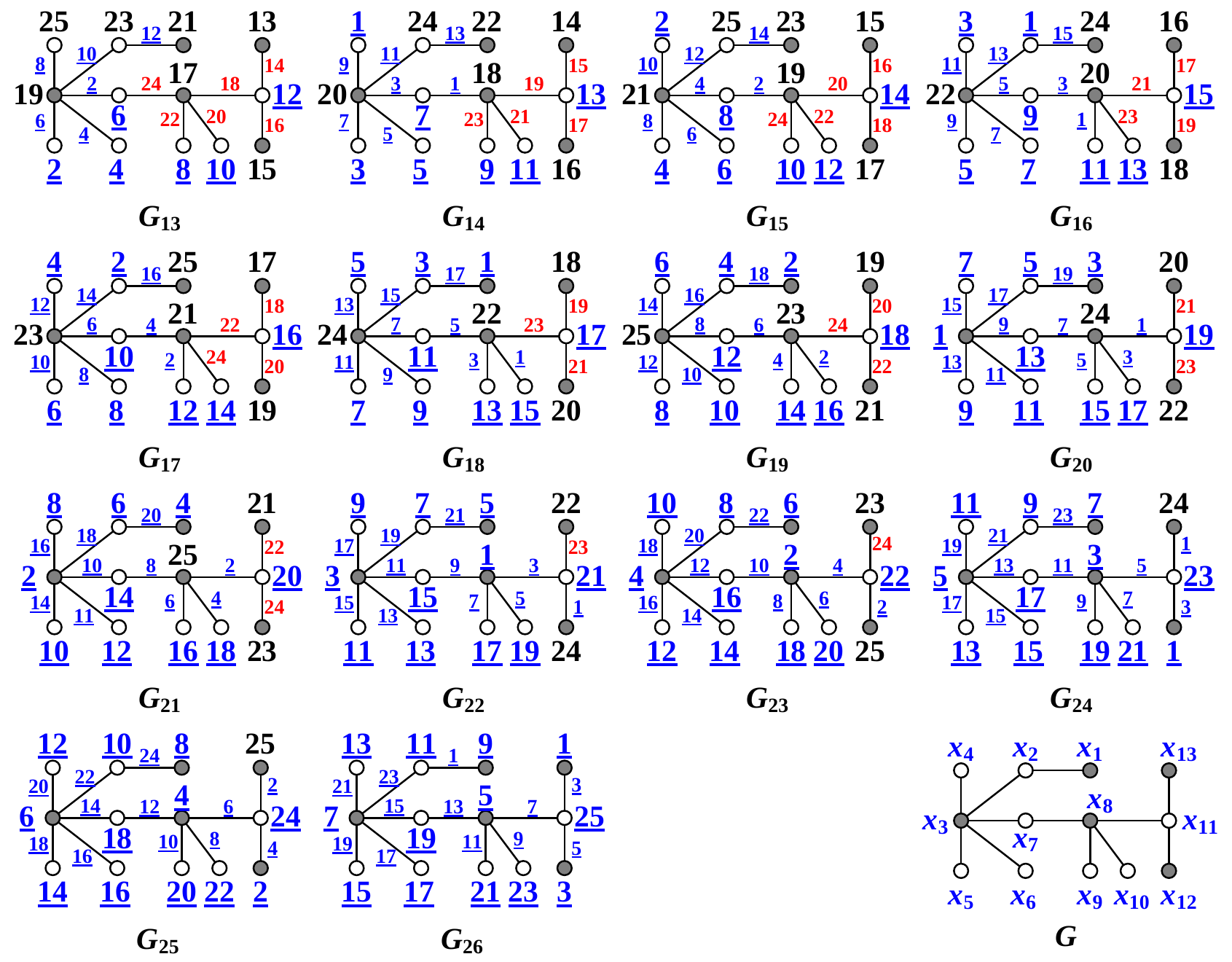}
\caption{\label{fig:D6C-group-11} {\small Second part of an every-zero graphic group based on the 6C-labelling cited from \cite{Yao-Mu-Sun-Sun-Zhang-Wang-Su-Zhang-Yang-Zhao-Wang-Ma-Yao-Yang-Xie2019}.}}
\end{figure}

We do the following jobs in this example:

1. We construct an \emph{every-zero graphic group} $G^{26}_{roup}=\{F^{odd}_f(G)\cup F^{even}_f(G),\oplus\}$ based on an odd-even separable 6C-labelling in the following: In Fig.\ref{fig:D6C-group} and Fig.\ref{fig:D6C-group-11}, a graph set $F^{odd}_f(G)=\{G_{1},G_{3},\dots ,G_{25}\}$ holds $G_{2i-1}\cong G$ and admits a labelling $f_{2i-1}$ with $i\in [1,13]$, and each $f_{2i-1}$ is defined by $f_{2i-1}(x)=f(x)+2(i-1)~(\bmod~25)$ for $x\in V(G)$ and $f_{2i-1}(xy)=f(xy)+2(i-1)~(\bmod~24)$ for $xy\in E(G)$. Another graph set $F^{even}_f(G)=\{G_{2},G_{4},\dots ,G_{26}\}$ contains $G_{2i}$ holding $G_{2i}\cong G$ and admitting a labelling $f_{2i}$ with $i\in [1,13]$, where each $f_{2i}$ is defined by $f_{2i}(u)=f(u)+(2i-1)~(\bmod~25)$ for $u\in V(G)$ and $f_{2i}(uv)=f(uv)+(2i-1)~(\bmod~24)$ for $uv\in E(G)$. It is easy to see
\begin{equation}\label{eqa:an-example-group}
[f_a(x)+f_b(x)-f_c(x)]~(\bmod~25)=f_{\lambda}(x)
\end{equation}
where $\lambda=a+b-c~(\bmod~25)$ for $x\in V(G)$, and $a,b,c\in [1,26]$. We have $G_a\oplus_c G_b=G_{\lambda}$ by (\ref{eqa:an-example-group}). So, $G^{26}_{roup}$ is really an every-zero graphic group based on the 6C-labelling. This every-zero graphic group $G^{26}_{roup}=\{F^{odd}_f(G)\cup F^{even}_f(G),\oplus\}$ obeys
\begin{equation}\label{eqa:c3xxxxx}
[f_{a'}(uv)+f_{b'}(uv)-f_{c'}(uv)]~(\bmod~26)=f_{\lambda'}(uv)
\end{equation}
where $\lambda'=a'+b'-c'~(\bmod~26)$ for $uv \in E(G)$, and $a',b',c'\in [1,26]$. We call $G^{26}_{roup}$ an \emph{edge-every-zero graphic group}, or a \emph{bimodule every-zero graphic group} based on the 6C-labelling.

Observe this every-zero graphic group $G^{26}_{roup}$ carefully, we can see the vertex color set $f_{2i-1}(V(G_{2i-1}))\cup f_{2j}(V(G_{2j}))=[1,25]$, and the vertex color set $f_{2i-1}(E(G_{2i-1}))\cup f_{2j}(E(G_{2j}))=[1,24]$.

2. Each element $G_{2i-1}\in F^{odd}_f(G)\subset G^{26}_{roup}$ distributes us a TB-paw $D_{2i-1}$ as
$${
\begin{split}
D_{2i-1}=&f_{2i-1}(x_1)f_{2i-1}(x_1x_2)f_{2i-1}(x_2)f_{2i-1}(x_2x_3)f_{2i-1}(x_3)f_{2i-1}(x_3x_4)f_{2i-1}(x_4)f_{2i-1}(x_3x_5)\\
&f_{2i-1}(x_5)f_{2i-1}(x_3x_6)f_{2i-1}(x_6)f_{2i-1}(x_3x_7)f_{2i-1}(x_7)f_{2i-1}(x_7x_3)f_{2i-1}(x_3)f_{2i-1}(x_7x_8)\\
&f_{2i-1}(x_8)f_{2i-1}(x_8x_7)f_{2i-1}(x_7)f_{2i-1}(x_8x_{11})f_{2i-1}(x_{11})f_{2i-1}(x_{11}x_{12})f_{2i-1}(x_{12})\\
&f_{2i-1}(x_{11}x_{13})f_{2i-1}(x_{13})
\end{split}}
$$ with $i\in [1,13]$. Thus, we get a string set $\{D_{2i-1}:~i\in [1,13]\}$. And, each element of the string set $\{D_{2i}:~i\in [1,13]\}$ is defined by
$${
\begin{split}
D_{2i}=&f_{2i}(x_1)f_{2i}(x_1x_2)f_{2i}(x_2)f_{2i}(x_2x_3)f_{2i}(x_3)f_{2i}(x_3x_4)f_{2i}(x_4)f_{2i}(x_3x_5)f_{2i}(x_5)\\
&f_{2i}(x_3x_6)f_{2i}(x_6)f_{2i}(x_3x_7)f_{2i}(x_7)f_{2i}(x_7x_3)f_{2i}(x_3)f_{2i}(x_7x_8)f_{2i}(x_8)f_{2i}(x_8x_7)\\
&f_{2i}(x_7)f_{2i}(x_8x_{11})f_{2i}(x_{11})f_{2i}(x_{11}x_{12})f_{2i}(x_{12})f_{2i}(x_{11}x_{13})f_{2i}(x_{13}).
\end{split}}
$$ based on $G_{2i}\in F^{even}_f(G)$ in Fig.\ref{fig:D6C-group} and Fig.\ref{fig:D6C-group-11}. By the operation ``$\oplus $'' defined in (\ref{eqa:an-example-group}) we can verify the string set $\{D_{2i-1}:~i\in [1,13]\}\cup \{D_{2i}:~i\in [1,13]\}$ to be an \emph{every-zero string group}.

3. We use Fig.\ref{fig:D6C-group-lattice} to introduce graphs $H\odot ^{26}_{k=1}a_kG_{k}$ with $a_k\in Z^0$ and $G_{k}\in G^{26}_{roup}$. Fig.\ref{fig:D6C-group-lattice} (a)-(d) are $H\odot ^{26}_{k=1}a_kG_{k}$ with $G_{k}\in G^{26}_{roup}$, where $G_i$ does not appear in (a)-(d) if $a_i=0$. In Fig.\ref{fig:D6C-group-lattice}(d), the zero is $G_{12}$ . Since $G_{12}\oplus_{12} G_8=G_{8}$ defined by (\ref{eqa:an-example-group}), based on $H\in F^c_{p,q}$ shown in (d-1), we join a vertex of $G_{12}$ with a vertex of $G_8$ by an edge, see Fig.\ref{fig:D6C-group-lattice}(d), go on in this way, we get $H\odot ^{26}_{k=1}a_kG_{k}$ shown in (d). We can encrypt the network (d-1) by the every-zero graphic group $G^{26}_{roup}$. In fact, we have a mapping $\theta: V(H)\cup E(H)\rightarrow G^{26}_{roup}$, such that $\theta(uv)=\theta(u)\oplus_{12}\theta(v)$ for each edge $uv\in E(H)$, (d-2) is just the resultant graph $W$ obtained by joining a vertex of the Topsnut-gpw $\theta(u)$ with a vertex of the Topsnut-gpw $\theta(uv)$, and joining a vertex of the Topsnut-gpw $\theta(v)$ with a vertex of the Topsnut-gpw $\theta(uv)$ for each edge $uv\in E(H)$.\qqed
\end{exa}
\begin{figure}[h]
\centering
\includegraphics[width=15.2cm]{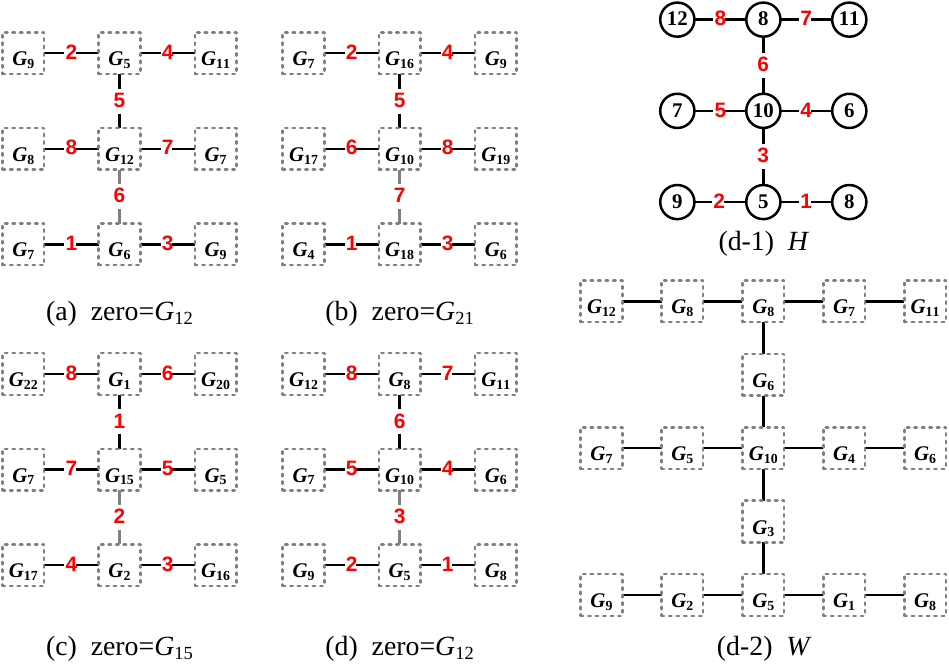}
\caption{\label{fig:D6C-group-lattice} {\small (a)-(d) are $H\odot ^{26}_{k=1}a_kG_{k}$; (d-1) $H\in F^c_{p,q}$; (d-2) $W$ is the result of encrypting $H$ by the every-zero graphic group $G^{26}_{roup}$.}}
\end{figure}

\subsection{Graphic sequence groups}

Suppose that a connected $(p,q)$-graph $G$ admits a $W$-type total coloring $f$, we define $W$-type total colorings $f_{s,k}$ by setting $f_{s,k}(x)=f(x)+s$ for every vertex $x\in V(G)$, and $f_{s,k}(uv)=f(uv)+k$ for each edge $uv\in E(G)$ as two integers $s,k$ belong to the set $Z$ of all integers. So, we have each connected $(p,q)$-graph $G_{s,k}\cong G$ admits the $W$-type total coloring $f_{s,k}$ defined above, immediately, we get an infinite graphic sequence $\{\{G_{s,k}\}^{+\infty}_{-\infty}\}^{+\infty}_{-\infty}$. We take $G_{a,b}$ as a \emph{zero}, and any two $G_{s,k}$ and $G_{i,j}$ in $\{\{G_{s,k}\}^{+\infty}_{-\infty}\}^{+\infty}_{-\infty}$ hold the following additive computations: For $uv\in E(G)$,
\begin{equation}\label{eqa:edge-graphic-group}
[f_{s,k}(uv)+f_{i,j}(uv)-f_{a,b}(uv)]~(\bmod~ q_W)=f_{\lambda,\mu}(uv).
\end{equation} with $\mu=k+j-b~(\bmod~ q_W)$; and for $x\in V(G)$,
\begin{equation}\label{eqa:vertex-graphic-group}
[f_{s,k}(x)+f_{i,j}(x)-f_{a,b}(x)]~(\bmod~ p_W)=f_{\lambda,\mu}(x)
 \end{equation} with $\lambda=s+i-a~(\bmod~ p_W)$.

Here, $p_W=|V(G)|$ and $q_W=|E(G)|$ if the $W$-type total coloring $f$ is a gracefully total coloring; and $p_W=q_W=2|E(G)|$ if the $W$-type total coloring $f$ is an odd-gracefully total coloring (see examples shown in Fig.\ref{fig:odd-graceful-group-vertex} and Fig.\ref{fig:odd-graceful-group-edge}).

\begin{figure}[h]
\centering
\includegraphics[width=16cm]{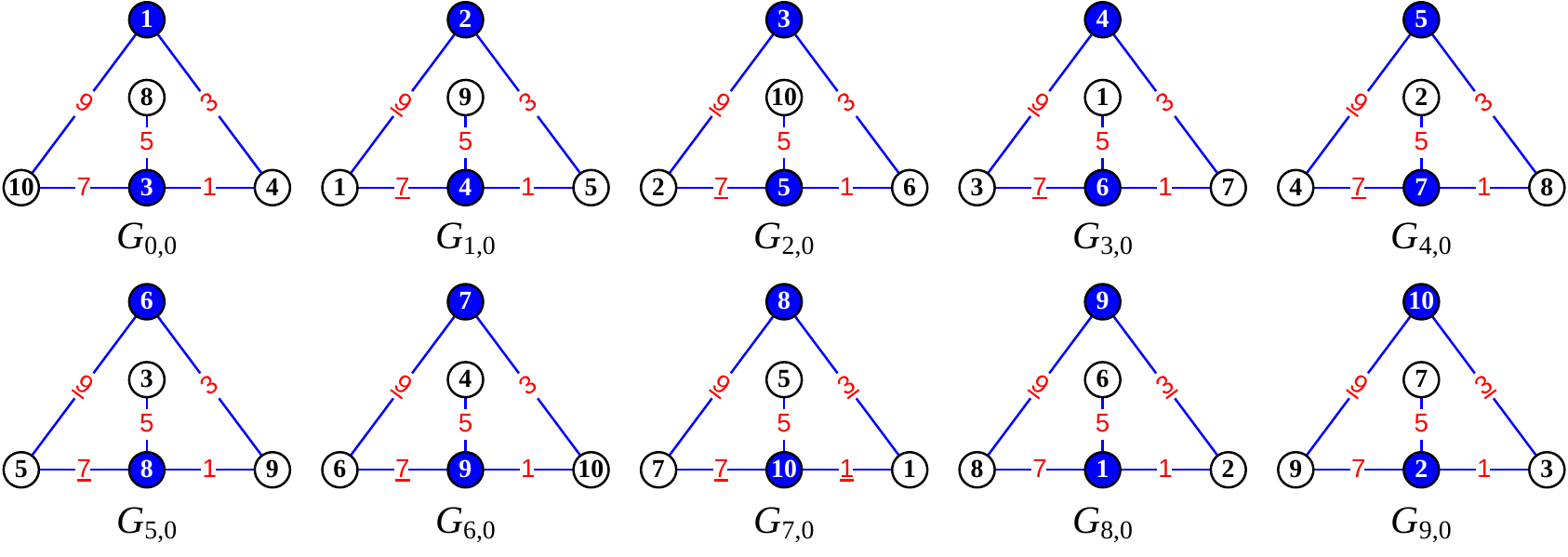}\\
\caption{\label{fig:odd-graceful-group-vertex} {\small An example for $p_W=2|E(G)|=10$ and holding (\ref{eqa:vertex-graphic-group}).}}
\end{figure}

\begin{figure}[h]
\centering
\includegraphics[width=16cm]{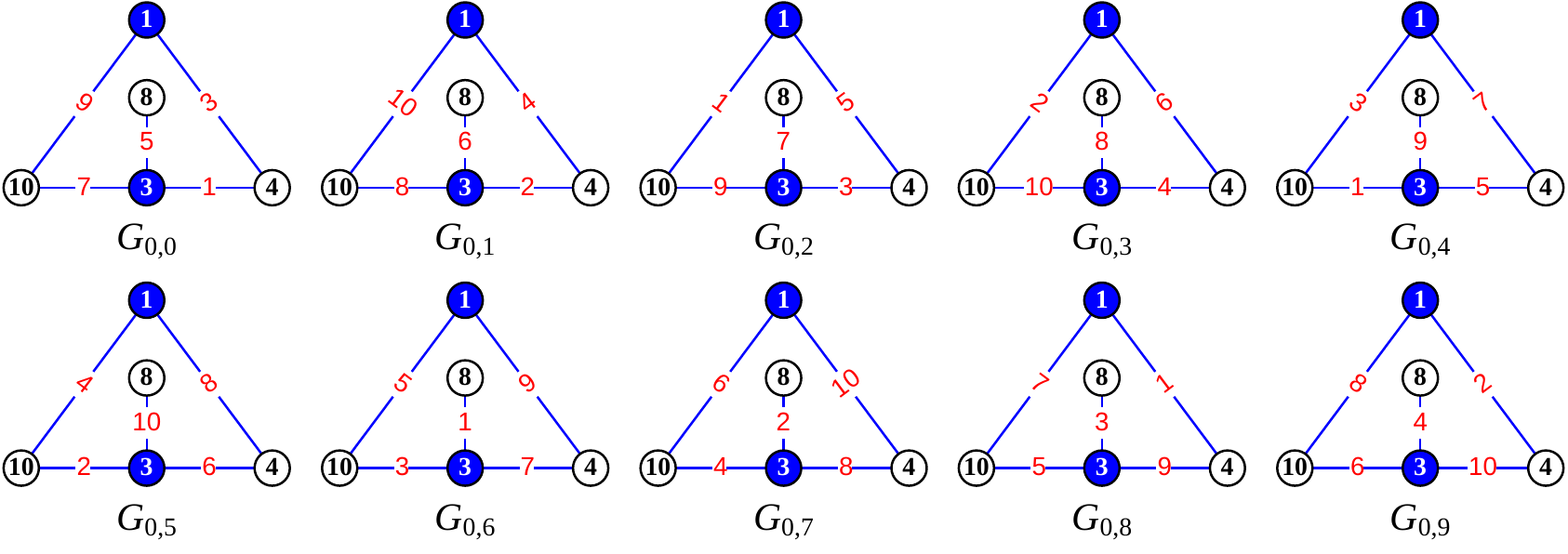}\\
\caption{\label{fig:odd-graceful-group-edge} {\small An example for $q_W=2|E(G)|=10$ and holding (\ref{eqa:edge-graphic-group}).}}
\end{figure}

Especially, for an edge subsequence $G_{s,k}, G_{s,k+1},\dots ,G_{s,k+q_W}$, and a vertex subsequence $G_{s,k}$, $G_{s+1,k}$, $\dots $, $G_{s+p_W,k}$, we have two sets $F(\{G_{s,k+j}\}^{q_W}_{j=1};\oplus;(G,f))$ and $F(\{G_{s+i,k}\}^{p_W}_{i=1};\oplus;(G,f))$. By the operation ``$\oplus$'' defined in (\ref{eqa:edge-graphic-group}), $F(\{G_{s,k+j}\}^{q_W}_{j=1};\oplus;(G,f))$ is an \emph{every-zero edge-graphic group}, since there are the following facts:

(i) \emph{Zero}. Every graph $G_{s,k+j}$ of $F(\{G_{s,k+j}\}^q_{j=1};\oplus)$ is as the ``zero'' such that $G_{s,k+r}\oplus G_{s,k+j}=G_{s,k+r}$ for any graph $G_{s,k+r}\in F(\{G_{s,k+j}\}^q_{j=1};\oplus)$.

(ii) \emph{Closure law}. For $r=i+j-j_0~(\bmod~ q_W)$, $G_{s,k+i}\oplus G_{s,k+j}=G_{s,k+r}\in F(\{G_{s,k+j}\}^q_{j=1};\oplus)$ under a zero $G_{s,k+j_0}$.

(iii) \emph{Inverse.} For $i+j=2j_0~(\bmod~ q_W)$, $G_{s,k+i}\oplus G_{s,k+j}=G_{s,k+j_0}$ under a zero $G_{s,k+j_0}$.

(iv) \emph{Associative law}. $G_{s,k+i}\oplus (G_{s,k+j}\oplus G_{s,k+r})=(G_{s,k+i}\oplus G_{s,k+j})\oplus G_{s,k+r}$.

(v) \emph{Commutative law}. $G_{s,k+i}\oplus G_{s,k+j}=G_{s,k+j}\oplus G_{s,k+i}$.

Similarly, $F(\{G_{s+i,k}\}^{p_W}_{i=1};\oplus;(G,f))$ is an \emph{every-zero vertex-graphic group} by the operation ``$\oplus$'' defined in (\ref{eqa:vertex-graphic-group}).

As considering some graphs arbitrarily selected from the sequence $\{\{G_{s,k}\}^{+\infty}_{-\infty}\}^{+\infty}_{-\infty}$, we have
\begin{equation}\label{eqa:mixed-infinite-graphic-group}
[f_{s,k}(w)+f_{i,j}(w)-f_{a,b}(w)]~(\bmod~ p_W, q_W)=f_{\lambda,\mu}(w).
\end{equation} with $\lambda=s+i-a~(\bmod~ p_W)$ and $\mu=k+j-b~(\bmod~ q_W)$ for each element $w\in V(G)\cup E(G)$. See examples shown in Fig.\ref{fig:infinite-graphic-group} and Fig.\ref{fig:odd-graceful-group-mixed}. Thereby, the set $F(\{\{G_{s,k}\}^{+\infty}_{-\infty}\}^{+\infty}_{-\infty};\oplus;(G,f))$ is an \emph{every-zero infinite graphic sequence group} under the additive operation ``$\oplus$'' based on two modules $p_W$ and $q_W$ and a connected $(p,q)$-graph $G$ admitting a $W$-type total coloring, since it possesses the properties of Zero, Closure law, Inverse, Associative law and Commutative law.

\begin{figure}[h]
\centering
\includegraphics[width=16cm]{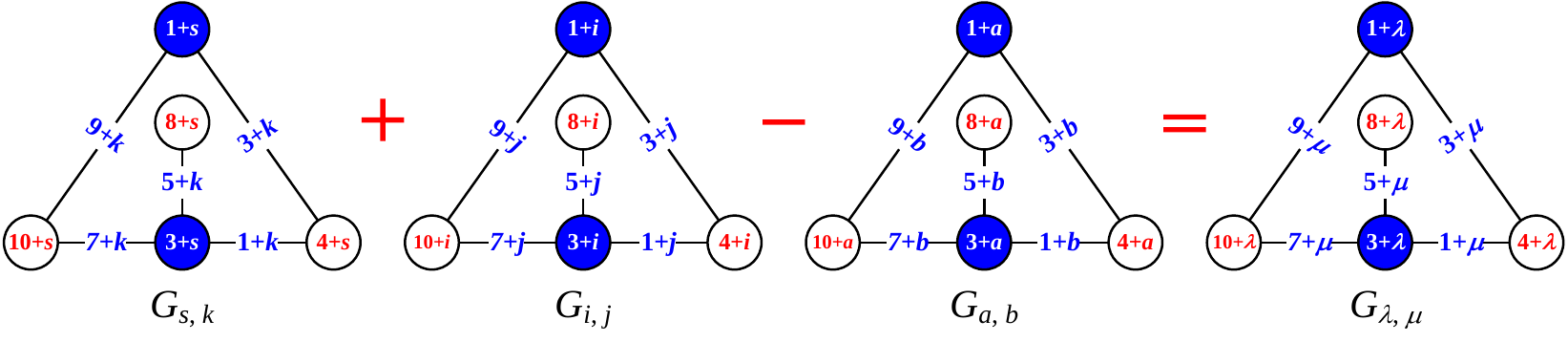}\\
\caption{\label{fig:infinite-graphic-group} {\small A graphic scheme for illustrating the formula (\ref{eqa:mixed-infinite-graphic-group}).}}
\end{figure}

\begin{figure}[h]
\centering
\includegraphics[width=16cm]{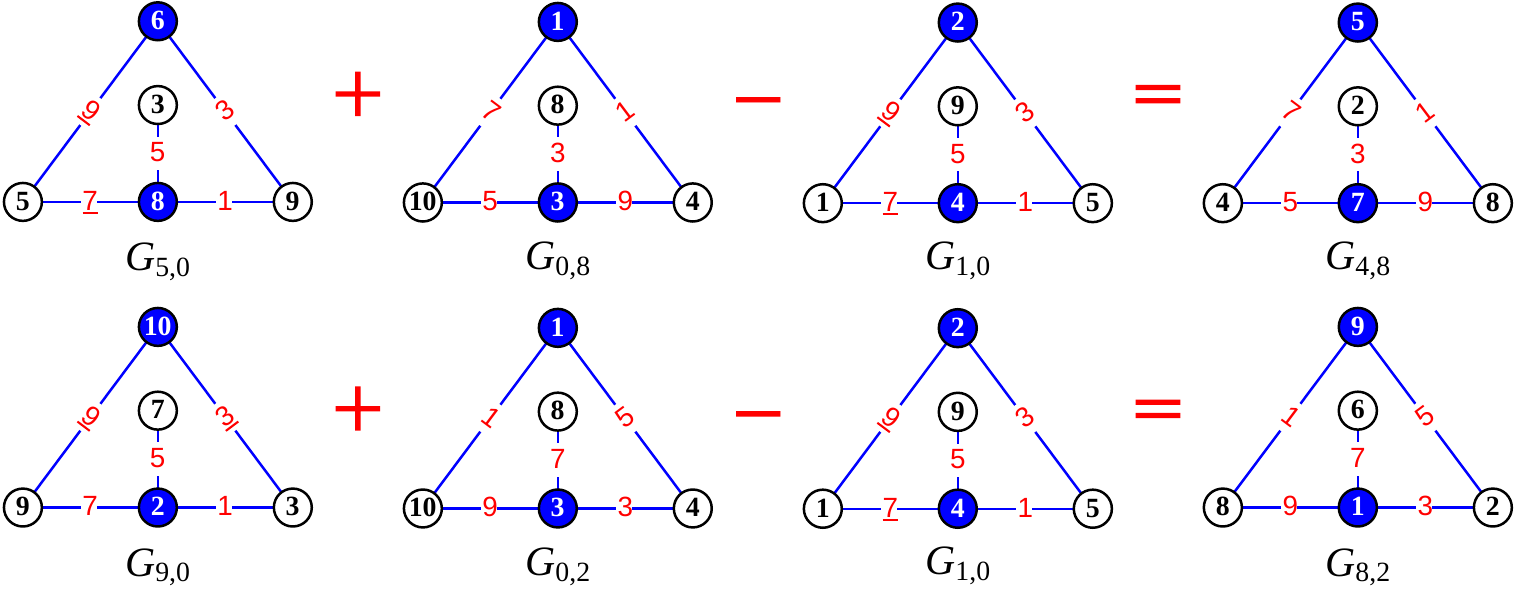}\\
\caption{\label{fig:odd-graceful-group-mixed} {\small Two examples for illustrating the formula (\ref{eqa:mixed-infinite-graphic-group}).}}
\end{figure}

\begin{rem}\label{rem:infinite-graphic-sequence-group}
Let $F^*(G,f))=F(\{\{G_{s,k}\}^{+\infty}_{-\infty}\}^{+\infty}_{-\infty};\oplus;(G,f))$. The elements of an every-zero infinite graphic sequence group $F^*(G,f)$ can tile fully each point $(x,y)$ of $xOy$-plane. And moreover, $F^*(G,f)$ contains infinite every-zero graphic groups having finite elements, such as $F(\{G_{s+i,k}\}^{p_W}_{i=1}$; $\oplus;(G,f))$ and $F(\{G_{s,k+j}\}^{q_W}_{i=1};\oplus;(G,f))$. Also, $F^*(G,f)$ contains infinite every-zero graphic groups having infinite elements.

Clearly, particular every-zero graphic groups having infinite elements or finite elements can be used easily to encrypt randomly networks. Suppose that the coloring $f$ of $G$ in $F^*(G,f)$ is equivalent to another $W_g$-type total coloring $g$ of $G$. Then we get another every-zero infinite graphic sequence group $F^*(G,g)=F(\{\{H_{i,j}\}^{+\infty}_{-\infty}\}^{+\infty}_{-\infty};\oplus;(G,g))$ with $G\cong H_{i,j}$. Thereby, the every-zero infinite graphic sequence group $F^*(G,f)$ is a \emph{public-key graphic sequence group}, the every-zero infinite graphic sequence group $F^*(G,g)$ is a \emph{private-key graphic sequence group} accordingly.

Since there exists a mapping $\varphi: V(G)\cup E(G)\rightarrow V(G)\cup E(G)$ such that $g(w)=\varphi(f(w))$ for $w\in V(G)\cup E(G)$, we claim that $F^*(G,f)$ admits an every-zero graphic sequence homomorphism to $F^*(G,g)$, and moreover $F^*(G,f)\leftrightarrow F^*(G,g)$, a pair of \emph{homomorphically equivalent every-zero graphic sequence homomorphisms}.

Let $G$ admit a \emph{graph homomorphism} to a connected graph $H$ under a mapping $\varphi:V(G)\rightarrow V(H)$, where $H$ admits a $W'$-type total coloring $g$. By the construction of an infinite graphic sequence $\{\{G_{s,k}\}^{+\infty}_{-\infty}\}^{+\infty}_{-\infty}$, we get another infinite graphic sequence $\{\{H_{a,b}\}^{+\infty}_{-\infty}\}^{+\infty}_{-\infty}$, and moreover, we have an \emph{infinite graphic sequence homomorphism} as follows:
\begin{equation}\label{eqa:c3xxxxx}
\{\{G_{s,k}\}^{+\infty}_{-\infty}\}^{+\infty}_{-\infty}\rightarrow \{\{H_{a,b}\}^{+\infty}_{-\infty}\}^{+\infty}_{-\infty},
\end{equation}
where each $H_{a,b}$ is a copy of $H$ and admit a $W'$-type total colorings $g_{a,b}$ defined by $g_{a,b}(x)=g(x)+a$ for every vertex $x\in V(H)$, and $g_{a,b}(xy)=g(xy)+b$ for each edge $xy\in E(H)$, and as two integers $a,b\in Z$.\qqed
\end{rem}

\subsection{Matching graphic groups}

In \cite{Yao-Sun-Zhang-Mu-Sun-Wang-Su-Zhang-Yang-Yang-2018arXiv} the authors introduce many matching colorings/labellings of graphs, and point out matching diversity: configuration matching partition, coloring matching partition, set matching partition, matching chain; and in the number of matching partitions: one-vs-more and more-vs-more styles of matching partitions, configuration-vs-configuration, configuration-vs-labelling, labelling-vs-labelling and (configuration, labelling)-vs-(configuration, labelling).

\subsubsection{Twin graphic groups}

Wang \emph{et al.} in \cite{Wang-Xu-Yao-2017-Twin2017} introduced the \emph{twin odd-graceful labellings}: Suppose $f:V(G)\rightarrow [0,2q-1]$ is an odd-graceful labelling of a graph $G$ with $p$ vertices and $q$ edges, and $g:V(H)\rightarrow [1,2q]$ is a labelling of another graph $H$ with $p'$ vertices and $q'$ edges such that each edge $uv\in E(H)$ has its own label defined as $h(uv)=|h(u)-h(v)|$ and the edge label set $f(E(H))=[1,2q-1]^o$. We say $(f,g)$ to be a \emph{twin odd-graceful labelling}, $H$ a \emph{twin odd-graceful matching} of $G$. Thereby, we get two \emph{twin odd-graceful graphic groups} $\{F_f(G);\oplus\}$ and $\{F_g(H);\oplus\}$ based on a twin odd-graceful labelling $(f,g)$. Notice that $G\not \cong H$, in general. See some examples shown in the section of Graphic Lattices.

\subsubsection{Dual-coloring/lacelling graphic groups} Suppose that a graph $G$ with $p$ vertices and $q$ edges admits a $W$-type coloring $f$. Let $\max f=\max\{f(w): w\in S\subseteq V(G)\cup E(G)\}$ and $\min f=\min\{f(w): w\in S\subseteq V(G)\cup E(G)\}$. We call $g(w)=\max f+\min f-f(w)$ for each element $w\in S\subseteq V(G)\cup E(G)$ the dual $W$-type coloring of the coloring $f$. Then, $\{F_g(G);\oplus\}$ is called the \emph{dual graphic group} of the graphic group $\{F_f(G);\oplus\}$ based on a pair of mutually dual $W$-type colorings $f$ and $g$. Notice that these two graphic groups were built up on the same graph $G$. See four dual total colorings defined in Definition \ref{defn:4-dual-total-coloring}.

\subsubsection{Other matching graphic groups}

If a graph $G$ is bipartite and admits a set-ordered graceful labelling $f$, then there are dozen labellings $g_i$ equivalent with $f$, so we get a dozen \emph{matching-labelling graphic groups} $\{F_f(G);\oplus\}$ and $\{F_{g_i}(H_i);\oplus\}$ with $i\in [1,m]$ for $m\geq 2$. For example, these labellings $g_i$ are odd-graceful labelling, odd-elegant labelling, edge-magic total labelling, image-labelling, 6C-labelling, odd-6C-labelling, even-odd separated 6C-labelling, and so on (Ref. \cite{Yao-Sun-Zhang-Mu-Sun-Wang-Su-Zhang-Yang-Yang-2018arXiv}). Here, we refer to $\{F_f(G);\oplus\}$ as a \emph{public key}, and each $\{F_{g_i}(H_i);\oplus\}$ as a \emph{private key} in encrypting networks. Let $G^c$ be the complementary graph of $G$, that is, $V(G)=V(G^c)=V(K_n)$, $E(G)\cup E(G^c)=E(K_n)$ and $E(G)\cap E(G^c)=\emptyset$. So, we have $\{F_f(G);\oplus\}$ and $\{F_g(G^c);\oplus\}$ as a pair of matching graphic groups.

\subsection{Graphic group sequences} Let $G^{(1)}_{ro}(H)=\{F_f(G^{(1)})$; $\oplus\}$ be an every-zero graphic group. We get an encrypted graph $G^{(2)}(H)=H\triangleleft G^{(1)}_{ro}(H)$ to be one of set
$$\left \{H\triangleleft_{a,b} |^{(p,q)}_{(s,k)}a^{(2)}_{s,k}G^{(2)}_{s,k}: a^{(2)}_{s,k}\in Z^0,~G^{(2)}_{s,k}\in F_f(G^{(1)})\right \}$$ with $\sum a^{(2)}_{s,k}\geq 1$ after encrypting a graph $H$ by the every-zero graphic group $G^{(1)}_{ro}(H)$. Immediately, we get an every-zero graphic group $G^{(2)}_{ro}(H)=\{F_f(G^{(2)});\oplus\}$ made by the graph $G^{(2)}(H)$, the operation ``$\oplus$'' and the MIXED Graphic-group Algorithm. Go on in this way, we get an \emph{every-zero $H$-graphic group sequence} $\{G^{(t)}_{ro}(H)\}$ based on the initial every-zero graphic group $G^{(1)}_{ro}(H)=\{F_f(G^{(1)});\oplus\}$ and the graph $H$, where $G^{(t)}_{ro}(H)=H\triangleleft G^{(t-1)}_{ro}(H)$. Clearly, each $G^{(t)}_{ro}(H)$ is a network at time step $t$.

\begin{problem}\label{qeu:mixed-Graphic-group-Algorithm}
We have the following problems:
\begin{asparaenum}[\textrm{Seq}-1. ]
\item \textbf{Characterize} the topological structure of $\left \{G^{(t)}_{ro}(H)\right \}$. \textbf{Is} $G^{(t)}_{ro}(H)$ scale-free? \textbf{Is} $G^{(t)}_{ro}(H)$ self-similar?
\item \textbf{Determine} colorings admitted by each element of $\;\left \{G^{(t)}_{ro}(H)\right \}$.
\item \textbf{Estimate} the cardinality of $\left \{G^{(t)}_{ro}(H)\right \}$.
\item For $\textrm{\textbf{H}}=(H_1,H_2,\dots ,H_m)$, \textbf{study} \emph{every-zero $\textrm{\textbf{H}}$-graphic group sequence} $\left \{G^{(t)}_{ro}(\textrm{\textbf{H}})\right \}$.
\end{asparaenum}
\end{problem}

%\section{Topcode-matrix lattices, topological coding lattices}
%%\input{6-section/Topcode-matrix-6}

\section{Topcode-matrix lattices, topological coding lattices}

In \cite{Yao-Zhao-Zhang-Mu-Sun-Zhang-Yang-Ma-Su-Wang-Wang-Sun-arXiv2019} the authors introduce topological coding matrices (Topcode-matrices) and topological matrices. Topcode-matrices are matrices of order $3\times q$ and differ
from popular matrices applied in linear algebra and computer science. Topcode-matrices can use numbers, letters, Chinese characters, sets, graphs, algebraic groups \emph{etc.} as their elements. One important thing is that Topcode-matrices of numbers can derive easily number strings, since number strings are text-based passwords used in information security. Topcode-matrices can be used to describe topological graphic passwords (Topsnut-gpws) used for solving some problems coming from the investigation of Graph Networks and Graph Neural Networks proposed by GoogleBrain and DeepMind \cite{Battaglia-27-authors-arXiv1806-01261v2}.

\subsection{Topcode-matrix lattices}

Since Topsnut-gpws are related with algebraic matrices, we will introduce \emph{Topcode-matrix lattices}, and then show \emph{topological coding lattices} defined from Topcode-matrix lattices.

\subsubsection{Topcode-matrices}
\begin{defn}\label{defn:topcode-matrix-definition}
\cite{Yao-Zhao-Zhang-Mu-Sun-Zhang-Yang-Ma-Su-Wang-Wang-Sun-arXiv2019} A \emph{Topcode-matrix} (or \emph{topological coding matrix}) is defined as
\begin{equation}\label{eqa:Topcode-matrix}
\centering
{
\begin{split}
T_{code}= \left(
\begin{array}{ccccc}
x_{1} & x_{2} & \cdots & x_{q}\\
e_{1} & e_{2} & \cdots & e_{q}\\
y_{1} & y_{2} & \cdots & y_{q}
\end{array}
\right)=
\left(\begin{array}{c}
X\\
E\\
Y
\end{array} \right)=(X,~E,~Y)^{T}_{3\times q}
\end{split}}
\end{equation}\\
where \emph{v-vector} $X=(x_1, x_2, \cdots, x_q)$, \emph{e-vector} $E=(e_1$, $e_2 $, $ \cdots $, $e_q)$, and \emph{v-vector} $Y=(y_1, y_2, \cdots, y_q)$ consist of non-negative integers $e_i$, $x_i$ and $y_i$ for $i\in [1,q]$. We say $T_{code}$ to be \emph{evaluated} if there exists a function $f$ such that $e_i=f(x_i,y_i)$ for $i\in [1,q]$, and call $x_i$ and $y_i$ to be the \emph{ends} of $e_i$, and $q$ is the \emph{size} of $T_{code}$.\qqed
\end{defn}

\begin{rem}\label{rem:a-topcode-matrix-vs-more-graphs}
A Topcode-matrix $T_{code}$ corresponds more graphs with different topological structures, see an example shown in Fig.\ref{fig:matrix-vs-more-graphs}, where two graphs $(k')$ and $(j')$ hold $(k')\not\cong (j')$ if $k\neq j$ and $k,j\in \{$a, b, c, d, e, f$\}$. In fact, the Topcode-matrix $T_{code}$ corresponds other Topsnut-gpws differing from these six Topsnut-gpws shown in Fig.\ref{fig:matrix-vs-more-graphs}.

\begin{equation}\label{eqa:Topcode-matrix-vs-6-Topsnut-gpws}
\centering
{
\begin{split}
T_{code}= \left(
\begin{array}{ccccccccccc}
6&5&6&6&6&1&1&1&1&1\\
1&2&3&4&5&6&7&8&9&10\\
7&7&9&10&11&7&8&9&10&11
\end{array}
\right)
\end{split}}
\end{equation}

\begin{figure}[h]
\centering
\includegraphics[width=16.2cm]{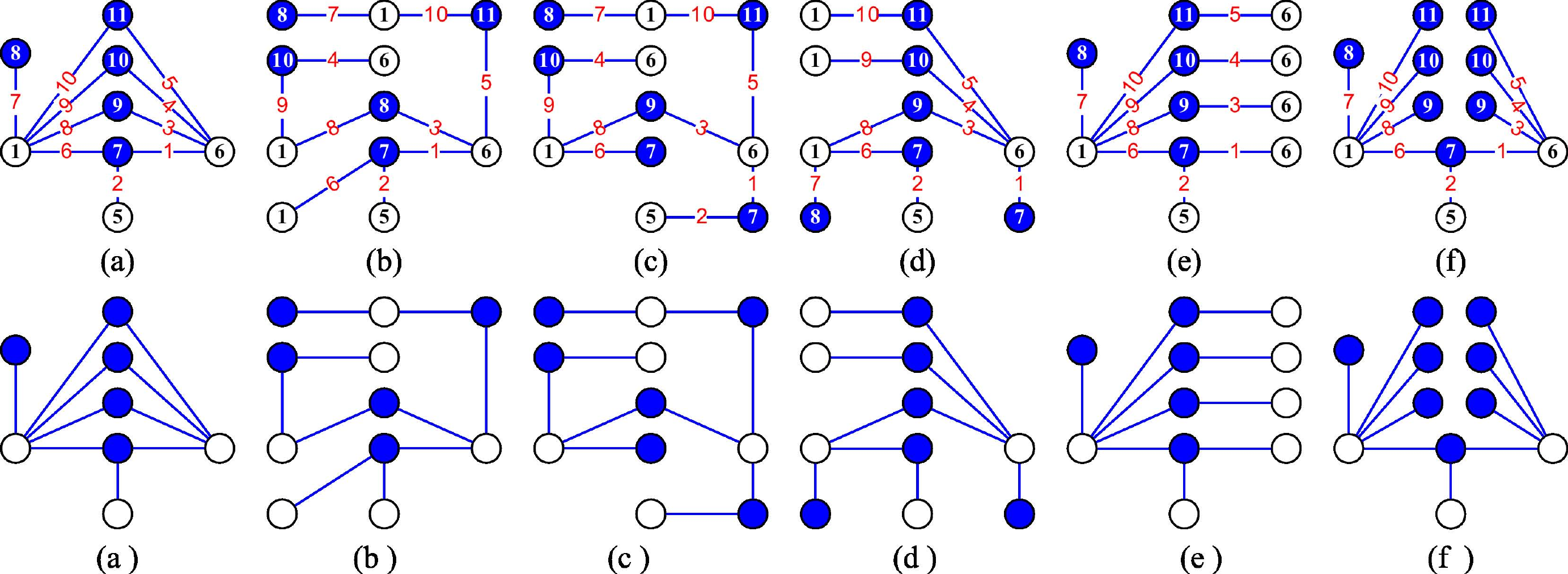}\\
\caption{\label{fig:matrix-vs-more-graphs} {\small $T_{code}$ shown in (\ref{eqa:Topcode-matrix-vs-6-Topsnut-gpws}) is the Topcode-matrix of each Topsnut-gpw $G_{(k)}$ with $k\in \{$a, b, c, d, e, f$\}$.}}
\end{figure}

Each Topsnut-gpw $G_{(k)}$ with $k\in \{$b, c, d, e, f$\}$ can be vertex-coincided into the Topsnut-gpw $G_{(\textrm{a})}$, in other word, two Topsnut-gpws $G_{(k)}$ and $G_{(j)}$ can be transformed to each other by the vertex-coinciding operation and the vertex-splitting operation.

In the language of graph homomorphism, each Topsnut-gpw $G_{(k)}$ admits a graph homomorphism to the Topsnut-gpw $G_{(a)}$ under a mapping $\theta_k:V(G_{(k)})\rightarrow V(G_{(a)})$ for $k\in \{$b, c, d, e, f$\}$.\qqed
\end{rem}

\begin{thm} \label{thm:5-trees-topcode-matrices}
We say $T$ to be \emph{vertex-equivalent} with $H$ if the resultant tree obtained by only doing some vertex-coinciding and vertex-splitting operations to $T$ is isomorphic to $H$, denoted this fact as $T\odot \wedge \odot H$. For any two trees $T$ and $H$ with the same number of vertices, we have $T\odot \wedge \odot H$.
\end{thm}

\subsubsection{Topcode-matrix lattices}
For a given Topcode-matrix $T_{code}$ if a graph $G$ admits a $W$-type coloring $f$, such that each $x_i=f(x)$ for some vertex $x\in V(G)$ and each $y_j=f(y)$ for some vertex $y\in V(G)$ and every $e_i=f(vu)$ for some edge $uv\in E(G)$, then we say $T_{code}$ corresponds the graph $G$, conversely, $G$ has its own Topcode-matrix, denoted as $T_{code}(G)$. Let $T^i_{code}=(X_i,~E_i,~Y_i)^{T}_{3\times q_i}$, where $X_i=(x^i_1, x^i_2, \cdots, x^i_{q_i})$, $E_i=(e^i_1, e^i_2, \cdots, e^i_{q_i})$ and $Y_i=(y^i_1, y^i_2, \cdots, y^i_{q_i})$ with $i=1,2$. The union operation ``$\uplus$'' of two Topcode-matrices $T^1_{code}$ and $T^2_{code}$ is defined by
\begin{equation}\label{eqa:xxx}
T^1_{code}\uplus T^2_{code}=(X_1\cup X_2,~E_1\cup E_2,~Y_1\cup Y_2)^{T}_{3\times (q_1+q_2)}.
\end{equation}
with $X_1\cup X_2=(x^1_1, x^1_2, \cdots, x^1_{q_1},x^2_1, x^2_2, \cdots, x^2_{q_2})$, $Y_1\cup T_2=(y^1_1, y^1_2, \cdots, y^1_{q_1},y^2_1, y^2_2, \cdots, y^2_{q_2})$ and $E_1\cup E_2=(e^1_1, e^1_2, \cdots, e^1_{q_1},e^2_1, e^2_2, \cdots, e^2_{q_2})$. Moreover, each Topcode-matrix $T^i_{code}$ corresponds a graph $G_i$, and each of $X_1\cap X_2\neq \emptyset $ and $Y_1\cap Y_2\neq \emptyset $ holds true, thus the Topcode-matrix of the vertex-coincided graph $G_1\odot G_2$ is just $T_{code}(G_1\odot G_2)=T^1_{code}\uplus T^2_{code}$. So, we vertex-split the vertex-coincided graph $H\odot |^n_{i=1}a_iT_i$ with $a_i\in \{0,1\}$ into disjoint graphs $T_1,T_2,\dots ,T_n,H$. In the expression of Topcode-matrices, we have the Topcode-matrix of a vertex-coincided graph $H\odot |^n_{i=1}a_iT_i$ as follows
\begin{equation}\label{eqa:xxx}
{
\begin{split}
T_{code}(H\odot |^n_{i=1}T_i)=T_{code}(H)\uplus ^n_{i=1}T^i_{code}=T_{code}(H)\uplus T^1_{code}\uplus T^2_{code}\uplus \cdots \uplus T^n_{code},
\end{split}}
\end{equation}
where $T^i_{code}=(X_i,~E_i,~Y_i)^{T}_{3\times q_i}=T_{code}(T_i)$ with $i\in [1,n]$, such that $T_{code}(H)=(X_H,~E_H,~Y_H)^{T}_{3\times q_H}$ and one of $X_H\cap X_i\neq \emptyset $ and $Y_H\cap Y_i\neq \emptyset $ for $i\in [1,n]$ holds true, and the size of $T_{code}(H\odot |^n_{i=1}T_i)$ is equal to $q_H+\sum^n_{i=1} q_i$. We call the following set
\begin{equation}\label{eqa:Topcode-matrices-lattice}
{
\begin{split}
\textbf{\textrm{L}}(\textbf{\textrm{T}}_{\textrm{code}}\uplus F_{p,q})=\{T_{code}(H)\uplus ^n_{i=1}a_iT^i_{code}:~a_i\in Z^0,~H\in F_{p,q}\}
\end{split}}
\end{equation}
a \emph{Topcode-matrix lattice} with $\sum^n_{i=1}a_i\geq 1$, where $\textbf{\textrm{T}}_{\textrm{code}}=(T^1_{code}$, $T^2_{code}$, $\dots $, $T^n_{code})$ is a group of \emph{linearly independent Topcode-matrix vectors} under the vertex-coinciding operation. Notice that the number of graphs corresponding to $T_{code}(H)$ (resp. $T^i_{code}$) is not one, in general. So, the cardinality of a Topcode-matrix lattice corresponding the graphic lattices (resp. colored graphic lattices) is very larger.

\begin{problem}\label{qeu:topcode-matrix-lattices}
About Topcode-matrix lattices, we have the following questions:
\begin{asparaenum}[\textbf{\textrm{Pro}}-1. ]
\item If $\{T^i_{code}\}^n_1$ is an every-zero graphic group, show properties of the Topcode-matrix lattice $\textbf{\textrm{L}}(\textbf{\textrm{T}}_{\textrm{code}}\uplus F_{p,q})$.

\item \textbf{Decompose} an evaluated Topcode-matrix $T_{code}$ defined in Definition \ref{defn:topcode-matrix-definition} into submatrices $T^1_{code}$, $T^2_{code}$, $\dots $, $T^m_{code}$, such that each $T^i_{code}$ is just a Topcode-matrix of a connected graph $H_i$ for $i\in [1,m]$.
\item \textbf{Define} Topcode-matrix lattices for other graphic lattices.\qqed
\end{asparaenum}
\end{problem}

\subsubsection{Text-based strings from Topcode-matrices}

We have three reciprocals from $T_{code}=(X,E,Y)^T$: $X^{-1}=(x_q,x_{q-1} , \cdots ,x_1)$, $E^{-1}=(e_q $, $ e_{q-1} $, $\cdots $, $e_1)$ and $Y^{-1}=(y_q, y_{q-1},\cdots ,y_1)$, then we get the \emph{reciprocal} of $T_{code}$, denoted as $T^{-1}_{code}=(X^{-1},E^{-1},Y^{-1})^{T}$ (Ref. \cite{Yao-Zhang-Sun-Mu-Sun-Wang-Wang-Ma-Su-Yang-Yang-Zhang-2018arXiv}). For a fixed Topcode-matrix $T_{code}$ and its reciprocal $T^{-1}_{code}$, there are basic algorithmic routes for generating text-based passwords (TB-paws) from Topcode-matrices:
\begin{asparaenum}[\textbf{\textrm{Route}}-1. ]
\item $D_1=x_1x_2\cdots x_qe_q e_{q-1} \cdots e_2e_1y_1 y_2\cdots y_q$ with its reciprocal $D^{-1}_1$, see Fig.\ref{fig:matrix-tb-text-1} (a) and (b).

\item $D_2=x_q x_{q-1} \cdots x_1e_1 e_2\cdots e_{q-1}e_qy_q y_{q-1} \cdots y_1$ with its reciprocal $D^{-1}_2$.

\item $D_3=x_1e_1y_1y_2e_2x_2x_3e_3y_3\cdots x_qe_qy_q$ with its reciprocal $D^{-1}_3$, see Fig.\ref{fig:matrix-tb-text-1} (c) and (d).

\item $D_4=x_q e_q y_q y_{q-1}e_{q-1}x_{q-1} \cdots y_2e_2x_2x_1e_1y_1$ with its reciprocal $D^{-1}_4$.

\item $D_5=y_2y_1e_1x_1e_2y_3y_4e_3x_2\cdots x_{q-2}e_{q-1}y_qe_{q}x_{q}$ $x_{q-1}$ with its reciprocal $D^{-1}_5$, see Fig.\ref{fig:matrix-tb-text-1} (e) and (f).

\item $D_6= y_{q-1}y_qe_{q}x_{q}e_{q-1}y_{q-2}y_{q-3}e_{q-2}x_{q-1}$ $\cdots $ $x_{2}e_{2}y_1e_{1}x_{1}x_{2}$ with its reciprocal $D^{-1}_6$.

\item Suppose $g: \{x_i,e_i,y_i:~x_i\in X,~e_i\in E,~y_i\in Y\}\rightarrow \{a_i:i\in [1,3q]\}$ is a bijection on the Topcode-matrix $T_{code}$, so it induces a TB-paw
\begin{equation}\label{eqa:Topcode-matrix-vs-TB-paws}
g=g^{-1}(a_{i_1})g^{-1}(a_{i_2})\cdots g^{-1}(a_{i_{3q}})
\end{equation} with its reciprocal $g^{-1}$, where $a_{i_1},a_{i_2},\dots,a_{i_{3q}}$ is a permutation of $a_1a_2\dots a_{3q}$. So, there are $(3q)!$ TB-paws by (\ref{eqa:Topcode-matrix-vs-TB-paws}), in general. Clearly, there are many random routes for inducing TB-paws from Topcode-matrices (see Fig.\ref{fig:matrix-tb-text-2}).
\begin{figure}[h]
\centering
\includegraphics[width=15.6cm]{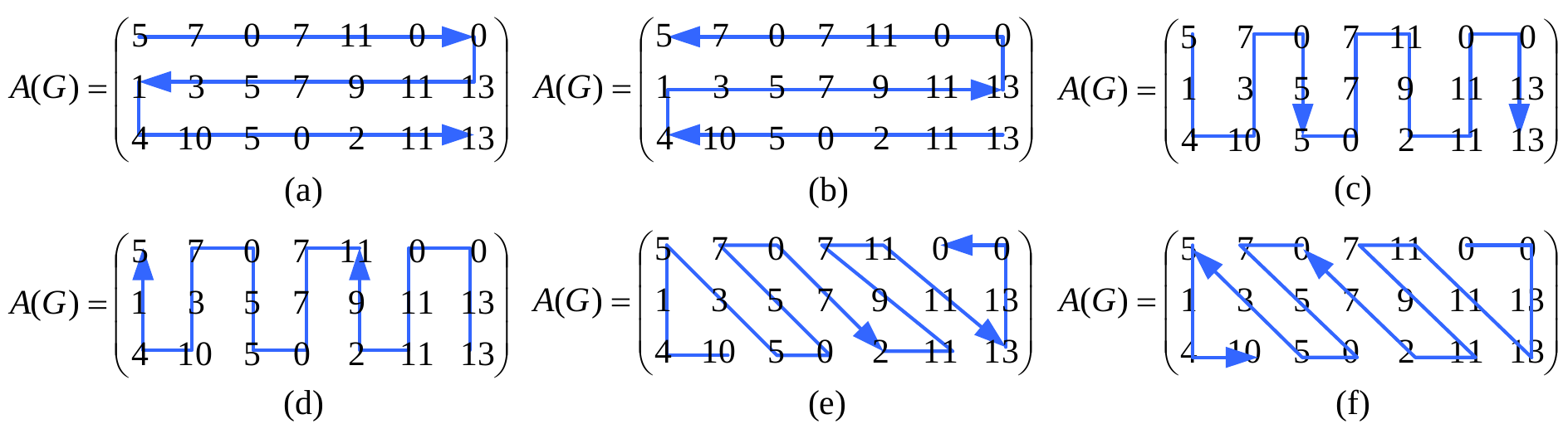}
\caption{\label{fig:matrix-tb-text-1}{\small Examples for illustrating basic algorithmic routes of generating TB-paws from Topcode-matrices.}}
\end{figure}
\end{asparaenum}

\begin{figure}[h]
\centering
\includegraphics[width=16.4cm]{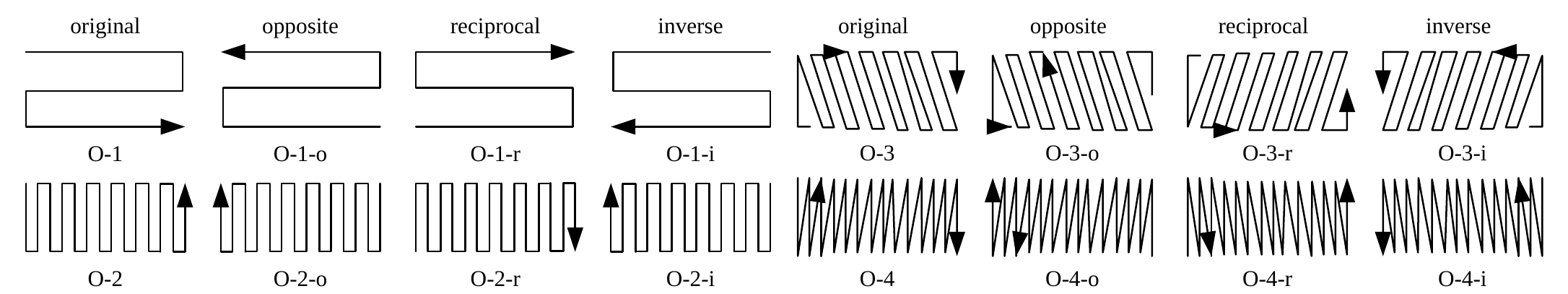}
\caption{\label{fig:matrix-tb-text-2}{\small Four kinds of algorithms for generating TB-paws from Topcode-matrices cited from \cite{Yao-Mu-Sun-Sun-Zhang-Wang-Su-Zhang-Yang-Zhao-Wang-Ma-Yao-Yang-Xie2019}.}}
\end{figure}

The number $N_{tbp}$ of all TB-paws generated from a Topcode-matrix $T_{code}$ can be computed in the formula (\ref{eqa:TB-paws-number}):

\begin{thm}\label{thm:TB-paws-numbers}
\cite{Yao-Zhang-Sun-Mu-Sun-Wang-Wang-Ma-Su-Yang-Yang-Zhang-2018arXiv} A Topcode-matrix $T_{code}$ of size $q$ distributes us
\begin{equation}\label{eqa:TB-paws-number}
N_{tbp}=(2q)\cdot (3q)!\cdot q!
\end{equation} different TB-paws in total.
\end{thm}

\subsubsection{Algebraic operation of real-valued Topcode-matrices}

We define an algebraic operation on Topcode-matrices \cite{Yao-Wang-Ma-Su-Wang-Sun-ITNEC-020}. In Definition \ref{defn:topcode-matrix-definition}, the Topcode-matrix $I_{code}=(X, E, Y)^{T}$ with $x_i=1$, $e_i=1$ and $y_i=1$ for $i\in [1,q]$ is called the \emph{unit Topcode-matrix}. For two Topcode-matrices $T^j_{code}=(X^j, E^j, Y^j)^{T}$ with $j=1,2$, where $X^j=(x^j_1, x^j_2, \cdots , x^j_q)$, $E^j=(e^j_1, e^j_2, \cdots , e^j_q)$ and $Y^j=(y^j_1, y^j_2, \cdots , y^j_q)$, the \emph{coefficient multiplication} of a function $f(x)$ and a Topcode-matrix $T^j_{code}$ is defined by
$${
\begin{split}
f(x)\cdot T^j_{code}=f(x)\cdot(X^j, E^j, Y^j)^{T}=(f(x)\cdot X^j, f(x)\cdot E^j, f(x)\cdot Y^j)^{T}
\end{split}}
$$
where $f(x)\cdot X^j=(f(x)\cdot x^j_1, f(x)\cdot x^j_2, \cdots , f(x)\cdot x^j_q)$, $f(x)\cdot E^j=(f(x)\cdot e^j_1, f(x)\cdot e^j_2, \cdots , f(x)\cdot e^j_q)$ and $f(x)\cdot Y^j=(f(x)\cdot y^j_1, f(x)\cdot y^j_2, \cdots , f(x)\cdot y^j_q)$. And the addition between two Topcode-matrices $T^1_{code}$ and $T^2_{code}$ is denoted as $T^1_{code}+T^2_{code}$, and
$$T^1_{code}+T^2_{code}=(X^1+X^2, E^1+E^2, Y^1+Y^2)^{T},$$
where $X^1+X^2=(x^1_1+x^2_1, x^1_2+x^2_2, \cdots , x^1_q+x^2_q)$, $E^1+E^2=(e^1_1+e^2_1, e^1_2+e^2_2, \cdots , e^1_q+e^2_q)$ and $Y^1+Y^2=(y^1_1+y^2_1, y^1_2+y^2_2, \cdots , y^1_q+y^2_q)$. We have a \emph{real-valued Topcode-matrix} $R_{code}$ defined as:
$R_{code}=\alpha(\varepsilon)T^1_{code}+\beta(\varepsilon)T^2_{code}$
and another real-valued Topcode-matrix
\begin{equation}\label{eqa:real-valued-topcode-matrix}
R_{code}(f_\varepsilon,G)=\alpha(\varepsilon)I_{code}+\beta(\varepsilon)T_{code}(G)
\end{equation}where $I_{code}$ is the unit Topcode-matrix, $T_{code}(G)$ is a Topcode-matrix of $G$.

Clearly, the text-based passwords induced by the real-valued Topcode-matrix $R_{code}(f_\varepsilon,G)$ are complex than that induced by a Topcode-matrix of $G$, and have huge numbers, since two functions $\alpha(\varepsilon)$ and $\beta(\varepsilon)$ are real and various. We have the following relationships between a Topcode-matrix $T_{code}(G)$ and a real-valued Topcode-matrix $R_{code}(f_\varepsilon,G)$:

(1) $e_i=|x_i-y_i|$ in a Topcode-matrix $T_{code}(G)$ of a $(p,q)$-graph $G$ corresponds $\alpha(\varepsilon)+\beta(\varepsilon)|x_i-y_i|$ of the real-valued Topcode-matrix $R_{code}(f_\varepsilon,G)$;

(2) $x_i+e_i+y_i=k$ in $T_{code}(G)$ corresponds $3\alpha(\varepsilon)+\beta(\varepsilon)\cdot k$ of $R_{code}(f_\varepsilon,G)$;

(3) $e_i+|x_i-y_i|=k$ in $T_{code}(G)$ corresponds $\alpha(\varepsilon)+\beta(\varepsilon)\cdot k$ of $R_{code}(f_\varepsilon,G)$;

(4) $|x_i+y_i-e_i|=k$ in $T_{code}(G)$ corresponds $\alpha(\varepsilon)+\beta(\varepsilon)\cdot k$ of $R_{code}(f_\varepsilon,G)$ if $e_i-(x_i+y_i)\geq 0$, otherwise $|x_i+y_i-e_i|=k$ corresponds $|\beta(\varepsilon)\cdot k-\alpha(\varepsilon)|$;

(5) $\big ||x_i-y_i|-e_i\big |=k$ in $T_{code}(G)$ corresponds $\alpha(\varepsilon)+\beta(\varepsilon)\cdot k$ of $R_{code}(f_\varepsilon,G)$ if $|x_i+y_i|-e_i< 0$, otherwise $\big ||x_i-y_i|-e_i\big |=k$ corresponds $|\beta(\varepsilon)\cdot k-\alpha(\varepsilon)|$.

\subsection{Connection between graphic lattices and traditional lattices}

\subsubsection{Topological coding lattice and traditional lattices}

Yao \emph{et al.} in \cite{Yao-Mu-Sun-Zhang-Wang-Su-2018} discussed the connection between text-based passwords and topological graphic passwords. We will construct a kind of lattices made by Topcode-matrices in the following.

For a Topcode-matrix lattice $\textbf{\textrm{L}}(\textbf{\textrm{T}}_{\textrm{code}}\uplus F_{p,q})$ with a group of linearly independent Topcode-matrix vectors $\textbf{\textrm{T}}_{\textrm{code}}=(T^1_{code},T^2_{code},\dots ,T^n_{code})$ under the vertex-coinciding operation, we do:

\emph{Step 1.} Take a determined bijection $\beta_i$ from $\{x_{i,j}$, $e_{i,j}$, $y_{i,j}:$ $~x_{i,j}\in X_i,~e_{i,j}\in E_i$, $y_{i,j}\in Y_i\}$ of each $T^i_{code}=(X_i$, $E_i$, $Y_i)^{T}$ to $\{b_{i,j}:~j\in [1,3q]\}$ to obtain a TB-paw $b_{i,j_1}b_{i,j_2}\dots b_{i,j_{3q}}$, which is a permutation of the TB-paw $b_{i,1}b_{i,2}\dots b_{i,3q}$ with $i\in [1,n]$, and we write this proceeding as $\beta_i(T^i_{code})=b_{i,1}b_{i,2}\dots b_{i,3q}$.

\emph{Step 2.} Similarly, we have another determined bijection $\alpha$ to translate $T_{code}(H)$ for $H\in F_{p,q}$ into a determined TB-paw $\alpha(T_{code}(H))=a_{1}a_{2}\dots a_{3q}$.

\emph{Step 3.} We cut $\alpha(T_{code}(H))$ into $n$ fragments $A_1$, $A_2$, $\dots $, $A_n$, correspondingly, we cut each TB-paw $b_{i,j_1}$ $b_{i,j_2}$ $\dots b_{i,j_{3q}}$ into $n$ fragments $B_{i,1},B_{i,2},\dots ,B_{i,n}$ with $i\in [1,n]$.

\emph{Step 4.} Suppose that all $a_{s}$ and $b_{i,j}$ are non-negative integers. Thereby, we get a traditional lattice defined as follows
\begin{equation}\label{eqa:graph-matric-numbers-lattice}
\textrm{\textbf{L}}(\textbf{\textrm{V}},F_{p,q})=\left \{\sum^n_{k=1}A_k \textbf{\textrm{V}}_k: ~A_k\in Z,~H\in F_{p,q}\right\}
\end{equation} where $\sum^n_{k=1}A_k\geq 1$, and $\textbf{\textrm{V}}_i=(B_{i,1},B_{i,2},\dots ,B_{i,n})$ is a vector, $\textbf{\textrm{V}}=(\textbf{\textrm{V}}_1$, $\textbf{\textrm{V}}_2$, $\dots $, $\textbf{\textrm{V}}_n)$ is a group of \emph{linearly independent vectors}, or a \emph{lattice base}. Clearly, our lattice $\textrm{\textbf{L}}(\textbf{\textrm{V}},F_{p,q})$ defined in (\ref{eqa:graph-matric-numbers-lattice}) is the same as a traditional lattice $\textrm{\textbf{L}}(\textbf{B})$ defined in (\ref{eqa:popular-lattice}), but $\textrm{\textbf{L}}(\textbf{\textrm{V}},F_{p,q})$ generated from the topological structure $H$ (also a graph) and the mathematical restrictions, that is, Topcode-matrices. So, we call $\textrm{\textbf{L}}(\textbf{\textrm{V}},F_{p,q})$ a \emph{topological coding lattice} for distinguishable purpose.

\subsubsection{Star-type graphic lattices and traditional lattices}

Notice that the graceful-difference star-graphic lattices $\textrm{\textbf{L}}(\ominus \textbf{I}_{ce}(GD))$ defined in (\ref{eqa:graceful-difference-stars-lattice}) and $\textrm{\textbf{L}}(\ominus \textbf{I}_{ce}(LGD))$ defined in (\ref{eqa:general-graceful-difference-stars-lattice}) construct more Topsnut-gpws admitting graceful-difference proper total colorings \cite{Wang-Yao-Star-type-Lattices-2020}. We come to build up a connection between star-graphic lattices and traditional lattices.

We call a tree to be a \emph{caterpillar} if the remainder after removing all leaves of this tree is just a \emph{path}, call this path the \emph{ridge} of the caterpillar, see a general caterpillar shown in Fig.\ref{fig:caterpillar-big}(a).

Let $L(T)$ be the set of leaves of a caterpillar $T$, and $P=u_1u_2\cdots u_n$ be the remainder $T-L(T)$ after deleting $L(T)$ from $T$. Furthermore, let $L_{eaf}(u_i)$ be the set of leaves adjacent with a vertex $u_i$ of the ridge $P=u_1u_2\cdots u_n$ of the caterpillar $T$. Thereby, we define $V_{ec}(T)=(a_1,a_2,\dots ,a_n)$ to be the \emph{topological vector} of the caterpillar $T$, where $a_i=|L_{eaf}(u_i)|$ with $i\in [1,n]$. See a topological vector shown in Fig.\ref{fig:caterpillar-big}(b). Obviously, each caterpillar can be expressed by $\overline{\ominus}^{\Delta}_{k=1}a_k K_{1,k}$, see Fig.\ref{fig:caterpillar-big}(b) and Fig.\ref{fig:caterpillar-decom-big}.

\begin{figure}[h]
\centering
\includegraphics[width=16.4cm]{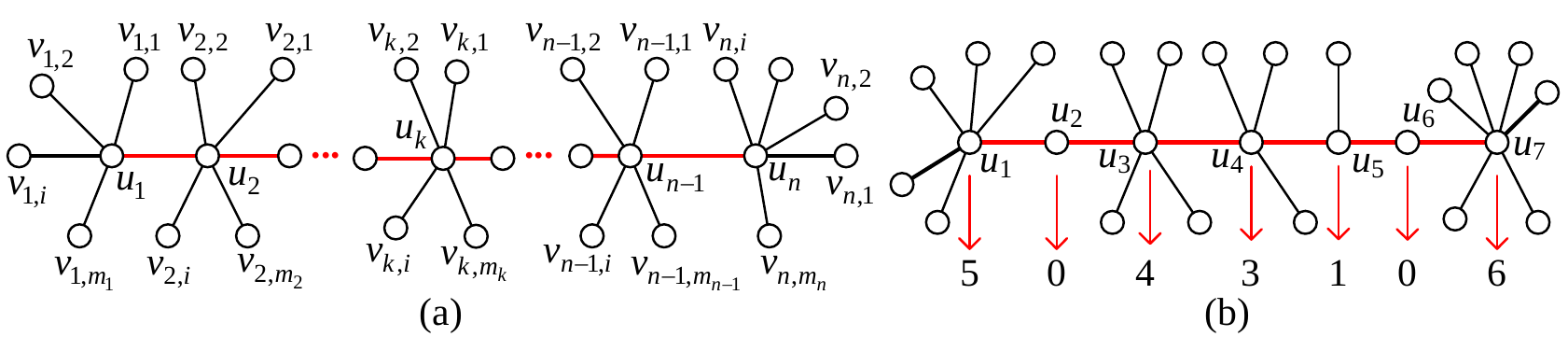}\\
\caption{\label{fig:caterpillar-big} {\small (a) A general caterpillar; (b) a caterpillar $H$ with its topological vector $V_{ec}(H)=(5,0,4,3,1,0,6)$.}}
\end{figure}

\begin{figure}[h]
\centering
\includegraphics[width=16.2cm]{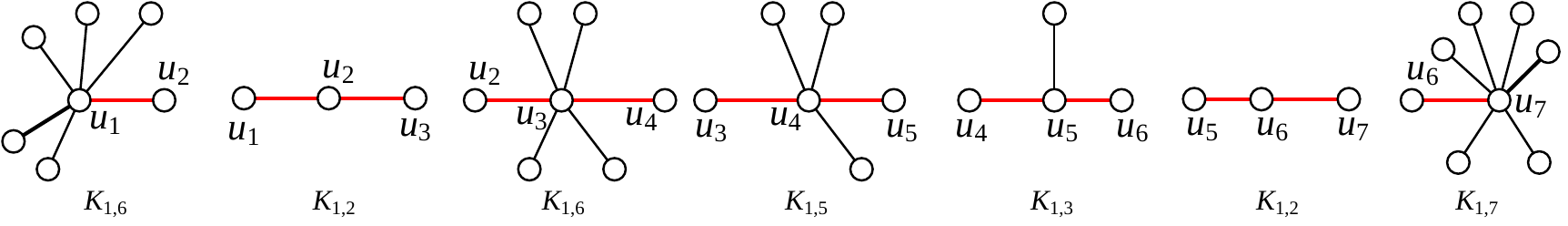}\\
\caption{\label{fig:caterpillar-decom-big} {\small The star decomposition of a caterpillar $H$ shown in Fig.\ref{fig:caterpillar-big}(b).}}
\end{figure}

Let each $T_k$ be a caterpillar with its topological vector $V_{ec}(T_k)$ and its ridge $P_k=u_{k,1}u_{k,2}\cdots u_{k,n}$ for $k\in [1,m]$. We call the following set
\begin{equation}\label{eqa:topological-vector-lattice}
\textrm{\textbf{L}}(\textbf{T}) =\left \{\sum^{m}_{k=1}a_k V_{ec}(T_k): a_k\in Z^0, k\in [1,m]\right \}
\end{equation} a \emph{topological coding lattice} with its \emph{lattice base} $\textbf{T}=\{V_{ec}(T_k)\}^m_1$, where $\sum^{m}_{k=1} a_k\geq 1$, $T_k$ belongs to the set $C_{ater}$ of caterpillars.

We provide a graph corresponding a topological vector $\sum^{m}_{k=1}a_k V_{ec}(T_k)$ with $a_k\in Z^0$ and $\sum^{m}_{k=1} a_k\geq 1$ as follows: Let $P^{j}_k=u^{j}_{k,1}u^{j}_{k,2}\cdots u^{j}_{k,n}$ be the $j$th copy of the ridge $P_k=u_{k,1}u_{k,2}\cdots u_{k,n}$ of a caterpillar $T_k$ for $j\in [1,a_k]$. So, we get a caterpillar $T^{j}_k$ with ridge $P^{j}_k=u^{j}_{k,1}u^{j}_{k,2}\cdots u^{j}_{k,n}$ for $j\in [1,a_k]$, clearly, $T^{j}_k$ is the $j$th copy of the caterpillar $T_k$. There are ways:

Way-1. We take a caterpillar $G$ with its ridge $P(G)=y_1y_2\cdots y_n$ such that the leaf set $L_{eaf}(y_i)$ of each $y_i$ holding $|L_{eaf}(y_i)|=\sum^{m}_{k=1}\sum ^{a_k}_{j=1}|L_{eaf}(u^{j}_{k,i})|$ for $i\in [1,n]$. Thereby, this caterpillar $G$ has its own topological vector $V_{ec}(G)=\sum^{m}_{k=1}a_k V_{ec}(T_k)$.

Way-2. We take a new vertex $w$, and join $w$ with the initial vertex $u^{j}_{k,1}$ of each ridge $P^j_k$ by an edge $wu^{j}_{k,1}$ with $j\in [1,a_k]$ and $k\in [1,m]$, the resulting graph is a \emph{super spider}, denoted by $w\ominus ^{m}_{k=1}\big (\bigcup ^{a_k}_{j=1} T^{j}_k\big )$. We call each caterpillar $T^{j}_k$ to be a \emph{super leg}, and $w$ the \emph{body}. Moreover, let $H=w\ominus ^{m}_{k=1}\big (\bigcup ^{a_k}_{j=1} T^{j}_k\big )$, we have the topological vector $V_{ec}(H)=\sum^{m}_{k=1}a_k V_{ec}(T_k)$.

\subsection{Directed Topcode-matrix lattices}

In Fig.\ref{fig:directed-topcode-matrix}, $\overrightarrow{T}$ is a \emph{directed Topsnut-gpw}, and it corresponds a \emph{directed Topcode-matrix} $A(\overrightarrow{T})$. In directed graph theory, the out-degree is denoted by ``$+$'', and the in-degree is denoted by ``$-$''. So, $\overrightarrow{T}$ has $d^+(22)=5$, $d^-(22)=0$, $d^-(13)=3$ and $d^+(13)=0$, and so on.

\begin{figure}[h]
\centering
\includegraphics[width=16cm]{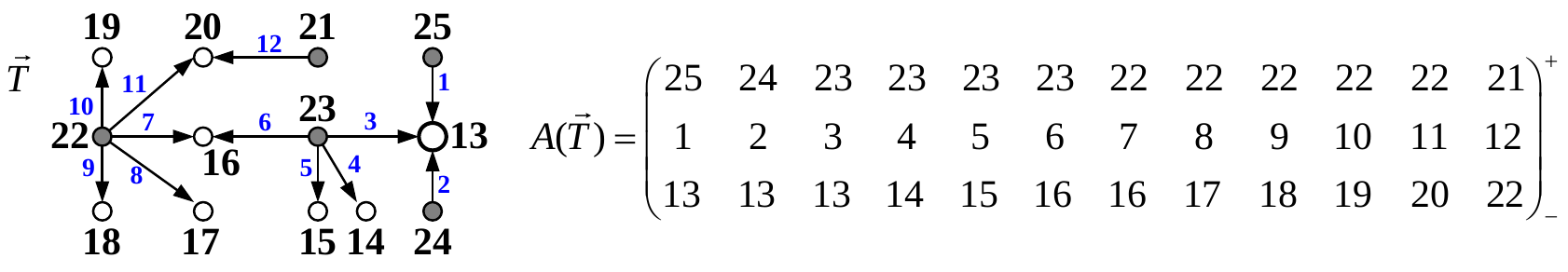}
\caption{\label{fig:directed-topcode-matrix}{\small A directed Topsnut-gpw with a directed Topcode-matrix.}}
\end{figure}

We show the definition of a directed Topcode-matrix as follows:
\begin{defn}\label{defn:directed-Topcode-matrix}
\cite{Yao-Zhao-Zhang-Mu-Sun-Zhang-Yang-Ma-Su-Wang-Wang-Sun-arXiv2019} A \emph{directed Topcode-matrix} is defined as
\begin{equation}\label{eqa:Topcode-dimatrix}
\centering
{
\begin{split}
\overrightarrow{T}_{code}= \left(
\begin{array}{ccccc}
x_{1} & x_{2} & \cdots & x_{q}\\
e_{1} & e_{2} & \cdots & e_{q}\\
y_{1} & y_{2} & \cdots & y_{q}
\end{array}
\right)^{+}_{-}=
\left(\begin{array}{c}
X\\
\overrightarrow{E}\\
Y
\end{array} \right)^{+}_{-}=[(X~\overrightarrow{E}~Y)^{+}_{-}]^{T}
\end{split}}
\end{equation}\\
where \emph{v-vector} $X=(x_1 ~ x_2 ~ \cdots ~x_q)$, \emph{v-vector} $Y=(y_1 $ ~ $y_2$ ~ $\cdots $ ~ $y_q)$ and \emph{directed-e-vector} $\overrightarrow{E}=(e_1$ ~ $e_2 $ ~ $ \cdots $ ~ $e_q)$, such that each arc $e_i$ has its head $x_i$ and its tail $y_i$ with $i\in [1,q]$, and $q$ is the \emph{size} of $\overrightarrow{T}_{code}$.\qqed
\end{defn}

A digraph book \cite{Bang-Jensen-Gutin-digraphs-2007} is very good and useful for studying digraphs. Since digraphs are useful and powerful in real applications, we believe that digraphs and their directed Topcode-matrices gradually are applied to Graph Networks and Graph Neural Networks \cite{Battaglia-27-authors-arXiv1806-01261v2}.

A \emph{directed Topcode-matrix lattice} is defined as
\begin{equation}\label{eqa:directed-Topcode-matrix-lattice}
{
\begin{split}
\overrightarrow{\textbf{\textrm{L}}}\left (\overrightarrow{\textbf{\textrm{T}}}_{\textrm{code}}\uplus \overrightarrow{F}_{p,q}\right )=\left \{\overrightarrow{T}_{code}(\overrightarrow{H})\uplus ^n_{i=1}a_i\overrightarrow{T}^i_{code}:~a_i\in Z^0,~\overrightarrow{H}\in \overrightarrow{F}_{p,q}\right \}.
\end{split}}
\end{equation} with a group of \emph{linearly independent directed Topcode-matrix vectors} $\overrightarrow{\textbf{\textrm{T}}}_{\textrm{code}}=\big (\overrightarrow{T}^i_{code}$, $\overrightarrow{T}^2_{code}$, $\dots $, $\overrightarrow{T}^n_{code}\big )$, and $\overrightarrow{F}_{p,q}$ being a set of directed graphs of $\lambda $ vertices and $\mu $ arcs with respect to $\lambda \leq p$, $\mu \leq q$ and $2n-2\leq p$, as well as $\sum^n_{i=1}a_i\geq 1$. Moreover, let $\overrightarrow{T}^i_{\textrm{colored}}$ be a colored directed Topcode-matrix and let $\overrightarrow{F}^c_{p,q}\}$ contain the colored directed graphs of $p$ vertices and $q$ arcs. We get a \emph{colored directed Topcode-matrix lattice}
\begin{equation}\label{eqa:colored-directed-Topcode-matrix-lattice}
{
\begin{split}
\overrightarrow{\textbf{\textrm{L}}}^c\left (\overrightarrow{\textbf{\textrm{T}}}^c_{\textrm{code}}\uplus \overrightarrow{F}^c_{p,q}\right )
=\left \{\overrightarrow{T}_{code}(\overrightarrow{H^c})\uplus ^n_{i=1}a_i\overrightarrow{T}^i_{\textrm{colored}}:~a_i\in Z^0,~\overrightarrow{H^c}\in \overrightarrow{F}^c_{p,q}\right \},
\end{split}}
\end{equation}
with $\sum^n_{i=1}a_i\geq 1$.

\begin{defn}\label{defn:flawed-directed-graceful}
\cite{Yao-Mu-Sun-Sun-Zhang-Wang-Su-Zhang-Yang-Zhao-Wang-Ma-Yao-Yang-Xie2019} Suppose that the underlying graph of a $(p,q)$-digraph $\overrightarrow{G}$ is disconnected, and $\overrightarrow{G}+E^*$ is a connected directed $(p,q+q')$-graph, where $q'=|E^*|$. Let $f:V(\overrightarrow{G}+E^*) \rightarrow [0,q+q']$ (resp. $[0,2(q+q')-1]$) be a directed graceful labelling (resp. a directed odd-graceful labelling) $f$ of $\overrightarrow{G}+E^*$, then $f$ is called a \emph{flawed directed graceful labelling} (resp. \emph{flawed directed odd-graceful labelling}) of the $(p,q)$-digraph $\overrightarrow{G}$.\qqed
\end{defn}

Let $T^*_k$ be a half-directed caterpillar with its topological vector $V_{ec}(T^*_k)$ and its undirected ridge $P_k=u_{k,1}u_{k,2}\cdots u_{k,n}$ for $k\in [1,m]$, see an example shown in Fig.\ref{fig:directed-caterpillar-coloring}(a). The following set
\begin{equation}\label{eqa:directed-topological-vector-lattice}
\textrm{\textbf{L}}(\overrightarrow{\textbf{T}}) =\left \{\sum^{m}_{k=1}a_k V_{ec}(T^*_k): a_k\in Z^0, k\in [1,m]\right \}
\end{equation} is called a \emph{directed topological coding lattice} with its \emph{base} $\overrightarrow{\textbf{T}}=\{V_{ec}(T^*_k)\}^m_1$, where $\sum^{m}_{k=1} a_k\geq 1$, $T^*_k$ belongs to the set $\overrightarrow{C}_{ater}$ of half-directed caterpillars.

Thereby, a directed topological coding lattice $\textrm{\textbf{L}}(\overrightarrow{\textbf{T}}) $ is equal to a traditional lattice $\textrm{\textbf{L}}(\textbf{B})$ defined in (\ref{eqa:popular-lattice}). By the way, we present the directed gracefully total coloring as follows:

\begin{defn}\label{defn:directed-graceful-coloring}
$^*$ Let $\overrightarrow{G}$ be a directed connected graph with $p$ vertices and $Q$ arcs. If $G$ admits a proper total coloring $f:V(\overrightarrow{G})\cup A(\overrightarrow{G})\rightarrow [1,M]$ such that $f(\overrightarrow{uv})=f(u)-f(v)$ for each arc $\overrightarrow{uv}\in A(\overrightarrow{G})$ and $\{|f(\overrightarrow{uv})|:\overrightarrow{uv}\in A(\overrightarrow{G})\}=[1,Q]$, then we call $f$ a \emph{directed gracefully total coloring} of $\overrightarrow{G}$, and moreover $f$ a proper directed gracefully total coloring if $M=Q$. (see an example shown in Fig.\ref{fig:directed-caterpillar-coloring}(b))\qqed
\end{defn}

\begin{figure}[h]
\centering
\includegraphics[width=16cm]{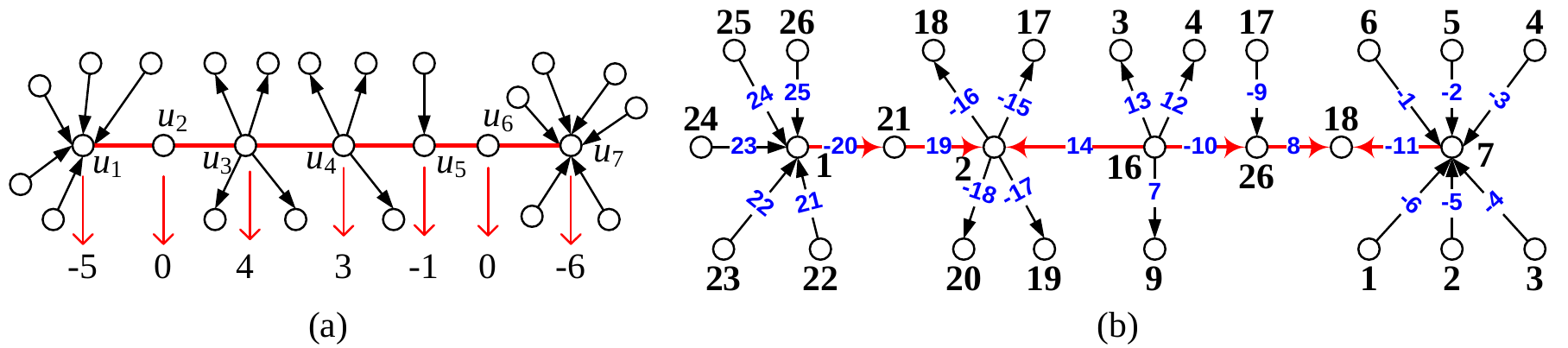}
\caption{\label{fig:directed-caterpillar-coloring}{\small (a) A half-directed caterpillar $T^*$ with its topological vector $V_{ec}(T^*)=(-5,0,4,3,-1,0,-6)$; (b) a directed gracefully total coloring of another directed caterpillar.}}
\end{figure}

\begin{problem}\label{qeu:standard-vs-graphic-lattices}
For the research of various topological coding lattices and directed topological coding lattices, we present the following questions:
\begin{asparaenum}[\textrm{D}-1. ]
\item \textbf{Determine} the number of trees corresponding a common Topcode-matrix $T_{code}$.
\item Since $T\odot \wedge \odot H$ for any two trees $T$ (as a public key) and $H$ (as a private key) with the same number of vertices (see Theorem \ref{thm:5-trees-topcode-matrices}), \textbf{determine} the smallest number of the vertex-coinciding and vertex-splitting operations. Let $IS(n)$ be the set of trees of $n$ vertices. We define a graph $G_{IS}$ with vertex set $V(G_{IS})=IS(n)$, two vertices $T_x$ and $T_y$ of $G_{IS}$ are adjacent to each other if they can be did only one operation of the vertex-coinciding operation and the vertex-splitting operation to be transformed to each other. \textbf{Find} the shortest path connecting two vertices $H_x$ (as a public key) and $H_y$ (as a private key) of $G_{IS}$. Consider this question about the graph $G_{IS^c}$ having the vertex set $V(G_{IS^c})=IS^c(n)$, where $IS^c(n)$ is the set of colored trees of $n$ vertices.
\item \textbf{Connections between traditional lattices and graphic lattices.} \textbf{Can} we translate a traditional lattice $\textrm{\textbf{L}}(\textbf{B})$ defined in (\ref{eqa:popular-lattice}) into a (colored) graphic lattice $\textbf{\textrm{L}}(\textbf{\textrm{H}},F_{p,q})$ defined in (\ref{eqa:graphic-lattice-graph-operation})?
\item \textbf{Translate} some problems of traditional lattices into graphic lattices, such as: Shortest Vector Problem (SVP, NP-hard), Shortest Independent Vector Problem (IVP), Unique Shortest Vector Problem, Closest Vector Problem (CVP, NP-C), Bounded Distance Decoding (BDD), Shortest Independent Vector Problem (SIVP, NP-hard), and so on.
\item We can provide many methods to build up topological vectors of graphs, for example, a spider tree $S_{pider}$ with $m$ legs $P_i$ of length $p_i$ for $i\in [1,m]$, so this spider tree $S_{pider}$ has its own topological vector $V_{ec}(S_{pider})=(p_1,p_2,\dots ,p_m)$; directly, a graph $G$ has its own topological vector $V_{ec}(G)=(d_1,d_2,\dots ,d_n)$, where $d_1,d_2,\dots ,d_n$ is the \emph{degree sequence} of $G$. \textbf{Find} other ways for making topological vectors of graphs.
\item \cite{Yao-Wang-Su-Sun-ITOEC-2020}~\textbf{Number String Decomposition Problem.} Suppose that a number string $S(n)=c_1c_2\cdots c_n$ with $c_j\in [0,9]$ was generated from some Topcode-matrix, \textbf{cut} $S(n)$ into $3q$ groups of substrings $a_{1},a_{2},\dots ,a_{3q}$ holding $a_{1}=c_1c_2\cdots c_{j_1}$, $a_{2}=c_{j_1+1}c_{j_1+2}\cdots c_{j_1+j_2}$, $\dots $, $a_{3q}=c_{a+1}c_{a+2}\cdots c_n$, where each $j_k\geq 1$ and $n=\sum ^{3q-1}_{k=1}j_k$, such that there exists at least a colored graph $H$ with its own Topcode-matrix $T_{code}(H)$ defined in Definition \ref{defn:topcode-matrix-definition}, which contains each substring $a_{i}$ with $i\in [1,3q]$ as its own elements and deduces a Topcode-string $a_{1}a_{2}\dots a_{3q}=S(n)$.
\item $^*$ \textbf{Number Strings and Matrices Problem.} In general, we want to cut a number string $D=c_1c_2\cdots c_n$ with $c_j\in [0,9]$ into $m\times n$ segments $a_{ij}$ such that these segments $a_{ij}$ are just the elements of adjacency matrix $A(a_{ij})_{n\times n}$ of a graph of $n$ vertices.
\item $^*$ Let $I_{string}=c_1c_2\cdots c_n\cdots $ be an infinite number string with $c_j\in [0,9]$, and let $T_{code}(G)$ ba a Topcode-matrix of a $(p,q)$-graph $G$ admitting a gracefully total coloring. For each finite number string $D_i$ induced from $T_{code}(G)$, \textbf{does} $D_i$ appear in $I_{string}$?\qqed
\item \textbf{Determine} $v_{gra}(\overrightarrow{G})=\min_f\left \{|f(V(\overrightarrow{G}))|\right \}$ over all of directed gracefully total colorings of $\overrightarrow{G}$.
\item \textbf{Define} other directed $W$-type total colorings.
\end{asparaenum}
\end{problem}

\section{Conclusion}

We have defined various graphic lattices and matrix lattices by means of knowledge of graph theory and topological coding, such as:
various (colored) graphic lattices, matching-type graphic lattices, star-graphic lattices, graphic lattice sequences, and so on. We have expressed some facts and objects of graph theory to be some kinds of graphic lattices. As known, many problems of graph theory can be expressed or illustrated by (colored) star-graphic lattices, graph homomorphism lattice and graphic lattice homomorphisms.

We have defined parameterized $W$-type total colorings: parameterized edge-magic proper total coloring, parameterized edge-difference proper total coloring, parameterized felicitous-difference proper total coloring and parameterized graceful-difference proper total coloring. Also, we have combined colorings and labellings to define: (set-ordered) gracefully total coloring, (set-ordered) odd-gracefully total coloring, (set-ordered) felicitous total coloring, (set-ordered) odd-elegant total coloring, (set-ordered) harmonious total coloring, (set-ordered) $c$-harmonious total coloring, (set-ordered) graceful edge-magic total coloring, (set-ordered) edge-difference magic total coloring, (set-ordered) graceful edge-difference magic total coloring, and so on. Importantly, we have defined the topological coloring isomorphism that consists of graph isomorphism and coloring isomorphism, and a new pair of the leaf-splitting operation and the leaf-coinciding operation.

In researching graphic lattices, we have met many mathematical problems, such as: Decompose graphs into Hanzi-graphs, $J$-graphic isomorphic Problem, Color-valued graphic authentication problem, Splitting-coinciding problem,
Prove any given planar graph in, or not in one of all 4-colorable planes $P_{\textrm{4C}}$, Tree and planar graph authentication, Tree topological authentication, Decompose evaluated Topcode-matrices, Number String Decomposition Problem, Translate a traditional lattice into a (colored) graphic lattice, Develop the investigation of the parameterized $W$-type proper total colorings, $(p,s)$-gracefully total numbers, $(p,s)$-gracefully total authentications \emph{etc.} However, determining the cardinality of a graphic lattice is not slight, since one will meet the Graph Isomorphic Problem, a NP-hard problem. The difficulty in solving these mathematical problems is useful for cryptographers, because they can apply this intractability to protect information. We need to apply graphic lattices in cryptosystems and the real world. There are complex problems in Number String Decomposition Problem:

\emph{First of all}, a number string can be composed of thousands of numbers, so it is difficult to divide it into $(3q)!$ pieces and write it into a Topcode-matrix $T_{code}$. The number string string may also correspond to other matrices, such as adjacency matrix, Topcode-matrix, Hanzi-matrix and so on.

\emph{Secondly}, since a large scale of Topcode-matrix $T_{code}$ corresponds to hundreds of colored graphs, so it is very difficult to find the specially appointed colored graph, which involves the NP-hard problem of graph isomorphism.

\emph{Thirdly}, this number string will involve hundreds of graph colorings and graph labellings as well as many problems in number theory.

\emph{Fourth}, because the number string is not an integer, so the well-known technology of integer decomposition can not be used to solve the Number String Decomposition Problem.

\emph{Fifthly}, because topological coding is made up of two different kinds of mathematics: topological structure and algebraic relation, it makes attackers switch back and forth in two different languages, unable to convey useful information. It is known that listening to two languages at the same time is forbidden by the basic laws of physics.

We have explored the construction of graphic group lattices and Topcode-matrix lattices, these lattices enable us to build up connections between traditional lattices and graphic lattices by topological vectors defined here, and try to find more deep relationships between them two, since we have believed algebraic technique will help us to do more interesting works on graphic lattices. Thereby, our techniques are not to enrich topological coding, but also can be applied to encryption networks, since our various graphic lattices (homomorphisms) can be used to encrypt a network wholly resisting full-scale attacks and sabotage by classical computers and quantum computers.

A graph in various graphic lattices (homomorphisms) is stored and run in the computer by various matrices, and the main theoretical technology of various graphic lattices (homomorphisms) comes from discrete mathematics, number theory, algebra, etc. Graphic lattice is an interdisciplinary product, which is expected to become the research content in the field of cryptography, or be concerned by the field of mathematics and computer science. It is known that there is no polynomial quantum algorithm to solve some lattice problems, so that the Number String Decomposition Problem in the topological coding may be the theoretical basis of the topological coding against supercomputer and quantum computing, because the Number String Decomposition Problem is irreversible, and various graphic lattice contains a lot of NP-hard problems. Moreover, a Topcode-matrix in the topological coding is either a homomorphic property, or it will be potential applicable. We clearly realize that we are far from the normal orbit of researching graphic lattices, so we must grope going on and work hard on graphic lattices.

\section*{Acknowledgment}

I thank gratefully the National Natural Science Foundation of China under grants No. 61163054, No. 61363060 and No. 61662066.

My students have done a lot of hard works on new labels and new colorings, they are: Dr. Xiangqian Zhou (School of Mathematics and Statistics, Huanghuai University, Zhumadian); Dr. Hongyu Wang, Dr. Xiaomin Wang, Dr. Fei Ma, Dr. Jing Su, Dr. Hui Sun (School of Electronics Engineering and Computer Science, Peking University, Beijing); Dr. Xia Liu (School of Mathematics and Statistics, Beijing Institute of Technology, Beijing); Dr. Chao Yang (School of Mathematics, Physics and Statistics, Shanghai University of Engineering Science, Shanghai); Meimei Zhao (College of Science, Gansu Agricultural University, Lanzhou); Sihua Yang (School of Information Engineering, Lanzhou University of Finance and Economics, Lanzhou); Jingxia Guo (Lanzhou University of technology, Lanzhou); Wanjia Zhang (College of Mathematics and Statistics, Hotan Teachers College, Hetian); Xiaohui Zhang (College of Mathematics and Statistics, Jishou University, Jishou, Hunan); Dr. Lina Ba (School of Mathematics and Statistics, Lanzhou University, Lanzhou); Lingfang Jiang, Tao haixia, Zhang jiajuan, Xiyang Zhao, Yaru Wang, Yarong Mu.

Thanks for the teachers of Graph Labelling Group: Prof. Mingjun Zhang (School of Information Engineering, Lanzhou University of Finance and Economics, Lanzhou); Prof. Ming Yao (Department of Information Process and Control Engineering, Lanzhou Petrochemical College of Vocational Technology, Lanzhou); Prof. Lijuan Qi (Department of basic courses, Lanzhou Institute of Technology, Lanzhou); Prof. Jianmin Xie (College of Mathematics, Lanzhou City University, Lanzhou).

\end{document}